\documentclass[12 pts]{article}
\usepackage[a4paper, total={6in, 10in}]{geometry}
\usepackage[utf8]{inputenc}
\usepackage{amsmath}
\usepackage{amssymb}
\usepackage{amsthm}
\usepackage{bbm}
\usepackage{subcaption}
\usepackage{comment}

\usepackage{mathrsfs}
\usepackage{bbm}
\usepackage[utf8]{inputenc}
\usepackage[T1]{fontenc}
\usepackage{color}
\usepackage{graphicx}
\usepackage{natbib}
\usepackage{enumitem}

\usepackage{multirow}

\usepackage{hyperref}
\usepackage{theoremref}

\numberwithin{equation}{section}
\theoremstyle{plain}
\newtheorem{theorem}{Theorem}[section]
\newtheorem{lemma}{Lemma}[section]
\newtheorem{corollary}{Corollary}[section]

\newtheorem{assumption}{Assumption}
\newtheorem{remark}{Remark}[section]
\usepackage{amsthm}
 
\definecolor{blue-violet}{rgb}{0,0,0}
\definecolor{Red}{rgb}{0.00, 0.00, 0.00}
\newcommand{\Red}{\color{Red}}
\definecolor{DR}{rgb}{0,0,0}
\newcommand{\DR}{\color{DR}}
\definecolor{Blue}{rgb}{0,0,0}
\newcommand{\Blue}{\color{Blue}}
\definecolor{Green}{rgb}{0,0,0}

\definecolor{DG}{rgb}{0,0,0}
\newcommand{\DG}{\color{DG}}

\definecolor{Grey}{rgb}{0.5,0.5,0.5}

\title{Estimation of Integrated Volatility Functionals with Kernel Spot Volatility Estimators}
\author{Jos\'e E. Figueroa-L\'opez\thanks{{Department of Statistics and Data Science, Washington University in St. Louis, St. Louis, MO 63130, USA ({\tt figueroa-lopez@wustl.edu}). Research supported in part by the NSF grant: DMS-2413557.}} \and Jincheng Pang\thanks{{Department of Statistics and Data Science, Washington University in St. Louis, St. Louis, MO 63130, USA ({\tt jinchengp@wustl.edu}).}} \and Bei Wu\thanks{{Upstart, 2950 S Delaware St \#410, San Mateo, CA 94403, USA ({\tt Bei.wu@upstart.com}).}}}
\date{August, 2025}

\begin{document}
\maketitle
\begin{abstract}
For a multidimensional It\^o semimartingale, we consider the problem of estimating integrated volatility functionals. \cite{jacod2013quarticity} studied a plug-in type of estimator based on a Riemann sum approximation of the integrated functional and a spot volatility estimator with a forward uniform kernel. Motivated by recent results that show that spot volatility estimators with general {two-sided} kernels of unbounded support are more accurate, in this paper an estimator using a general kernel spot volatility estimator as the plug-in is considered. A biased central limit theorem for estimating the integrated functional is established with an optimal convergence rate. Central limit theorems for properly de-biased estimators are also obtained both at the optimal convergence regime for the bandwidth and when applying undersmoothing. Our results show that one can significantly reduce the estimator's bias by adopting a general kernel instead of the standard uniform kernel. Our proposed bias-corrected estimators are found to maintain remarkable robustness against bandwidth selection in a variety of sampling frequencies and functions.
\smallskip
\end{abstract} \hspace{10pt}

\section{Introduction}
In this work, we are concerned with the  estimation of integrated volatility functionals of the form  $V_{T}(g):=\int_{0}^{T}g(c_s)ds$, where $g$ is a smooth function and $c_t:= \sigma_t\sigma_t^*$ is the spot {\Blue variance-covariance} process of a $d$-dimensional It\^o semimartingale given by
\begin{equation}
dX_t = \mu_t dt + \sigma_t d W_t +dJ_t.
\end{equation}
Here, $\{W_{t}\}_{t\geq{}0}$ is $d$-dimensional standard Brownian Motion (BM) and $\{J_{t}\}_{t\geq{}0}$ is a pure-jump process of bounded variation.
Such functionals have many applications in financial econometrics. In the simplest case, when $d=1$ and $g(x) =x$, this is the well-studied integrated volatility (IV) or variance $IV_T=\int_0^{T}\sigma^2_sds$, which measures the overall variability or {\Blue uncertainty} latent in the process $X$ during $[0,T]$.
Again, when $d=1$ and $g(x) = x^2$, $V_T(g)$ corresponds to the integrated quarticity, which appears in the limiting distribution of the realized variance estimator $\widehat{IV}_T=\sum_{i=1}^n (\Delta_i^nX)^2$ and, hence, it is crucial to construct feasible confidence intervals for IV. Another function appearing in the literature  is $g(x)=\log(x)$, when researching the statistical features of a factor $Y$ driving the volatility in a model where $\sigma$ is assumed to be the exponential of $Y$. Volatility functionals have also been used in principal component analysis by \cite{ait2019principal}  and in estimating integrated moments for option pricing purposes by \cite{li2016generalized}. 

One of the fundamental methods to estimate $V_{T}(g)=\int_{0}^{T}g(c_s)ds$ consists of approximating the integral by a Riemman sum and {\Blue plugging} in a suitable spot volatility estimator $\hat{c}$ in place of $c$:
\begin{equation} \label{eq:V_jacod}
\widehat{V}_{T}(g)=\Delta_{n} \sum_{i=1}^{n-k_{n}+1} g\left(\hat{c}_{i}^{n}\right).
\end{equation}
This method was pioneered by \cite{jacod2013quarticity}, where,  as the spot volatility estimator, a local rolling window estimator with uniform weights was used: 
 \begin{equation} \label{eq:V_jacodb}
	\hat{c}_{i}^{n,unif}=\frac{1}{k_{n} \Delta_{n}} \sum_{j=0}^{k_{n}-1}\left( \Delta^n_{i+j}X\right)^{2}.
\end{equation}
Two interesting facts emerged from their theory. {\Red For the ease of exposition, we focus on the $d=1$-case in the introduction, though our main results are stated for general dimension $d\geq 1$.} First of all, the rate of convergence of the estimation error is unexpectedly of order $\Delta_n^{\frac{1}{2}}$ even though $\hat{c}_{i}^{n}$ can only attain the rate $\Delta_n^{\frac{1}{4}}$\footnote{{\Red  Intuitively, by {\Blue Taylor} expansion of the estimator $g(\hat{c})$, we can write $\Delta_n\sum_{i=1}^n g(\hat{c}_i^n) - \Delta_n\sum_{i=1}^{n}g(c_{t_{i-1}})= \Delta_n\sum_{i=1}^{n}g'(c_{t_{i-1}}) \left(\hat{c}_{i}^n - c_{t_{i-1}}\right) + \frac{1}{2}\Delta_n\sum_{i=1}^{n}g''(c_{t_{i-1}}) \left(\hat{c}_{i}^n - c_{t_{i-1}}\right)^2+\text{higher order terms}$. The estimation errors in the first term of the right-hand side are ``averaged out" so that the convergence rate changes from $n^{-1/4}$ to $n^{-1/2}$.}}. Secondly, when $k_n \sim \theta/\sqrt{\Delta_n} $ for some $\theta\in(0,\infty)$, which yields the best convergence rate for the spot volatility estimator $\hat{c}_{i}^{n}$, $\frac{1}{\sqrt{\Delta_n}}\left( \widehat{V}_{T}(g) - V_{T}(g) \right) $ converges to a mixed Gaussian distribution with four bias terms:
\begin{enumerate}
\item A term of the form $-\frac{\theta}{2}\left(g(c_0) + g(c_T) \right) $ due to border effects;
\item A {\Blue second} term of the form $\frac{1}{\theta} \int_0^T g^{\prime \prime}(c_s) c_s^2ds $ due to the nonlinearity of $g$;
\item A third term of the form $-\frac{\theta}{12} \int_0^T g^{\prime \prime}(\tilde{c}_s) $, where $\tilde{c} $ is the volatility of volatility, due to target error $c_i - \bar{c}_{i}^n$, where $\bar{c}_{i}^{n}=\frac{1}{k_{n} \Delta_{n}} \sum_{j=0}^{k_{n}-1}c_{i+j}$;
\item A fourth term involving the jump component of the volatility process $\theta \sum_{s\le t} G(c_{s-}, c_s) $ for a suitable function of $G$ depending on g. This error disappears when $c$ is continuous. 
\end{enumerate}

In principle, the bias terms can be estimated and eliminated from $\widehat{V}_{T}(g)$ to obtain a feasible CLT with centered Gaussian limiting distribution. This strategy was explored by \cite{jacod2015estimation}.
An alternative solution, however, was adopted in \cite{jacod2013quarticity}. 
Specifically, the window length $k_n$ was made to converge to $\infty$ at a slower rate (under-smoothing) so that $k_n \sqrt{\Delta_n}\to{}0$. This corresponds to the case $\theta= 0 $ and, hence, the first, third, and fourth bias terms would vanish, while the second term becomes dominant. To keep the convergence rate $\sqrt{\Delta_n} $ of the estimator $\widehat{V}_T(g) $, this second term needs to be estimated and the estimator needs to be de-biased. The resulting estimator takes the form 
\begin{equation}
\Delta_n \sum_{i=1}^{n - k_n + 1}\left( g\left(\hat{c}_{i}^{n}\right) - \frac{1}{k_n} g^{\prime \prime }(\hat{c}_i)\hat{c}_i^2 \right),
\end{equation}
which indeed achieves the rate $\Delta_n^{\frac{1}{2}}$ with semi-parametrically optimal asymptotic variance. {\Red A similar approach to deal with {\Blue edge} effects and biases was also considered in \cite{kristensen2010nonparametric}, who also proposed kernel estimators and obtained a central limit theorem (though the model considered therein lacks leverage effects).} Alternatively, as mentioned before, in the optimal case $\theta \in (0,\infty)$ and employing forward uniform kernels, estimators for all the bias terms were proposed in \cite{jacod2015estimation}, where an unbiased estimator with centered limiting distribution was developed.

\cite{li2019efficient} took a different approach towards bias correction. They proposed a type of `jackknife' estimators, which are formed as a linear combination of a few uncorrected estimators of the form \eqref{eq:V_jacodb} associated with different windows $k_n$. 
Since these estimators have similar bias terms but different coefficients depending on the {\Blue window} sizes, the biases can be eliminated by a proper linear combination. Both the `two-scale' jackknife (i.e., when combining only two uncorrected {\Blue estimators}) and \cite{jacod2013quarticity}'s estimator 
use under-smoothing, so that the biases due to boundary effects, volatility of volatility, and volatility jumps are rendered asymptotically negligible. One drawback of this approach is that one needs to determine the degree of undersmoothing. In other words, if, for instance, we take $k_n \sim \beta\Delta_n^{\alpha}$ with $\alpha>-1/2$, then one needs to tune both parameters $\beta$ and $\alpha$, which is nontrivial.
In contrast, a multiscale jackknife estimator, formed as a linear combination of three (or more) estimators, can in principle cancel all the biases regardless of the speed of convergence of the window size. However, again, this involves more tuning parameters that need to be carefully calibrated.  

Other related work on this topic includes \cite{mykland2009inference}, who proposed using the same method as in \cite{jacod2013quarticity}, but {\Red keeping $k_n = k$ constant}; for the quarticity or more generally for estimating $\int_0^T g(c_s) ds  $ with  a power functional $g(x) = x^r$, they obtained a CLT with rate $1/\sqrt{\Delta_n}$ and an asymptotic variance that is not efficient, {\Red but can arbitrarily approach the optimal variance} when $k$ is large.  {\Red \cite{Glotter} established a Locally Asymptotically
Mixed Normality (LAMN) condition for the estimation of integrated volatility functionals in a continuous diffusion process where the volatility is driven by an independent It\^o process. The result therein shows that the optimal rate of convergence is $n^{-1/2}$. Another related result in the same direction is \cite{Renault}}. {\Red More recently, \cite{LiLiu2021} generalized \cite{jacod2013quarticity} 's result to allow a long-memory component such as {\Blue an} fBM with Hurst parameter $H>1/2$. They also established a semiparametric
efficiency lower bound, generalizing \cite{Renault}'s results.} \cite{chen2018inference} developed an estimator for $\widehat{V}_{T}(g)$ in the presence of microstructure noise,  based on a forward finite difference approximation of the standard pre-averaging estimator of the integrated volatility.

Our work aims to apply a general kernel spot volatility estimator to the plug-in estimator of integrated volatility functionals. The general kernel estimator (\cite{fan2008spot}, \cite{kristensen2010nonparametric}) is defined as 
\begin{align}\label{MDKNN00}
	{\tilde{c}^{n}_{i}} :=\sum_{j=1}^{n} K_{{k}_{n}\Delta_n}\left(t_{j-1}-t_i\right)\left(\Delta_{j}^{n} X\right)^{2},
\end{align}
where $K_{{b}}(x):=K(x/{b})/{b}$ {\Red for some bandwidth $b>0$}, and $k_n$ is an appropriate sequence that controls the bandwidth $b_n=k_n\Delta_n$ and that needs to be calibrated.
Of course, one can see \eqref{eq:V_jacodb} as a particular case of \eqref{MDKNN00} with $K(x)={\bf 1}_{[0,1]}(x)$. {\Red \cite{kristensen2010nonparametric} developed an asymptotic theory for \eqref{MDKNN00} when $t_i$ is a fixed time $\tau\in (0,T)$ and showed asymptotic normality to the spot variance $c_{\tau}$ under certain path-wise H\"older continuity and non-leverage conditions on $c$. See also \cite{fan2008spot} and \cite{mancini2015estimation} for other related results. All these results applied undersmoothing, which allowed them to neglect the ``target error" coming from approximating the spot volatility by a kernel weighted integrated volatility.}

In this paper, we consider the following version: 
\begin{align}\label{MDKNNa0}
	{\hat{c}^{n}_{i}} :=\frac{\sum_{j=1}^{n} K_{{k}_{n}\Delta_n}\left(t_{j-1}-t_i\right)\left(\Delta_{j}^{n} X\right)^2}{\bar{K}_{n}\left(t_i\right)}, 
\end{align}
with $\bar{K}_{n}\left(t_i\right):=\Delta_{n}\sum_{j=1}^{n} K_{{k}_{n}\Delta_n}\left(t_{j-1}-t_i\right)$,
which is known to be more robust than \eqref{MDKNN00} {\Blue against} edge effects {\Red (see also \cite{kristensen2010nonparametric} for other edge-robust alternatives)}. As a result, we can, in principle, consider the whole range $i=1,\dots, n$ in \eqref{eq:V_jacod}. Furthermore, we find it more convenient for our proofs to consider the estimator 
\begin{equation} \label{eq:V_jacodBis}
\widehat{V}_{T}(g)=\Delta_{n} \sum_{i=1}^{n} \bar{K}_n(t_i)g\left(\hat{c}_{i}^{n}\right).
\end{equation}
The coefficient  $\bar{K}_n(t_i)$ attenuates the contribution of $g\left(\hat{c}_{i}^{n}\right)$ for $t_i$ near $0$ and $T$ and, thus, serves as an additional edge effect correction. {\Red \cite{kristensen2010nonparametric} also considered a kernel-based estimator similar to \eqref{eq:V_jacodBis}, but applied an edge effect correction, similar to that in \cite{jacod2015estimation}, by eliminating the terms corresponding to $i$'s close to $0$ and $n$,  and did not consider the adjustments $\bar{K}_n(t_i)$. Their asymptotic theory also only considered the undersmoothing asymptotic regime mentioned above.}

The motivation for considering general kernels comes from \cite{FigLi} and \cite{FigWu}, where it has been shown, both theoretically and by Monte Carlo experiments, that a general two-sided kernel with unbounded support could significantly reduce the variance of the spot volatility estimator $\hat{c}$. More specifically, \cite{FigLi} studied the infill asymptotic behavior of the mean-square error (MSE) of \eqref{MDKNN00} {for continuous It\^o semimartingales} and showed that, when setting $k_n$ `optimally' (i.e., in order to minimize the leading order terms of the MSE), the double exponential kernel $K(x)=.5 e^{-|x|}$ minimizes the resulting MSE over all kernels. This result was established under the absence of leverage effects. \cite{FigWu} extended this result, by first proving a CLT for a general {two-sided} kernel, with optimal convergence rate and in the presence of jumps and leverage {effects}. It was then shown that exponential kernels minimize the asymptotic variance of the CLT. It is worth noting that, even when applying undersmoothing, the asymptotic variance could be reduced when choosing a kernel of unbounded support. Specifically, when applying undersmoothing, the asymptotic variance of the spot volatility estimator is proportional to $\|K\|^2=\int_{-\infty}^{\infty} K^2(x) dx$. When $K$ is constrained to have support on $[0,1]$, the optimal kernel is indeed the uniform kernel $K_{unif}(x)={\bf 1}_{[0,1]}(x)$ with $\|K_{unif}\|^2=1$\footnote{{\Red Indeed, by Jensen's inequality, $K_{unif}(x)={\bf 1}_{[0,1]}(x)$ satisfies $\int K^2_{unif}(x)dx=1=\left(\int K(x)dx\right)^2\leq \int K^2(x)dx$, which shows that a uniform kernel $K_{unif}(x)$ minimizes the asymptotic variance when applying undersmoothing.}}. However, picking the kernels $.5e^{-|x|}$ {\Red and $e^{-x}{\bf 1}_{x>0}$ will result in estimators whose asymptotic variances are respectively a quarter and a half of} that obtained by a uniform kernel $K_{unif}(x)={\bf 1}_{[0,1]}(x)$. {\Red Similar comments apply when considering only backward looking kernels such as ${\bf 1}_{[-1,0)}(x)$ and $e^{-|x|}{\bf 1}_{x<0}$.}

A natural question is whether the estimation accuracy gained when estimating the spot volatility by a general kernel transfers into gains in estimating integrated functionals. Intuitively, by the Taylor expansion of the estimator $g(\hat{c})$, we can write $g(\hat{c}) - g(c)= g'(c) \left(\hat{c} - c\right) + \frac{1}{2}g''(c)\left(\hat{c} - c\right)^2 + O_p\left(|\hat{c} - c|^3 \right)$. After taking expectations, the variance of the spot volatility estimator, which appears in the second term,  becomes a bias term of the integrated volatility estimator. Therefore, the bias of the estimator $\widehat{V}_{T}(g)$ could in principle be reduced by adopting a general kernel. 

As it turns out, the CLT of \eqref{MDKNN00} again exhibits several biases, which depend on a rather intricate and nontrivial manner on the kernel. The bias terms compared to the uniform kernel case involves a highly nontrivial transformation of the kernel function. The second bias term, involving $\frac{1}{\theta} \int_0^T g^{\prime \prime}(c_s) c_s^2ds$, contains the $L_2$ norm of the kernel and, thus, can be reduced because certain kernels of unbounded support {\Blue possess} smaller $L_2$ norm. We then propose bias-corrected estimator similar to those in \cite{jacod2015estimation}, and derive the corresponding CLT. We show by Monte Carlo experiments that our general kernel estimator typically has superior performance compared with the jackknife estimators of \cite{li2019efficient}. Of course, our estimator is simpler as it does not involve as many tuning parameters, which are hard to tune-up. In fact, our Monte Carlo experiments suggest our estimators are significantly more stable in the sense that they exhibit better performance for most values of the bandwidth, while the jackknife estimator requires accurate tuning of the `bandwidth' to match our estimator's performance. We also observe that even without bias corrections, under certain conditions, the plug-in estimator with general kernel spot volatility estimator has good performance, which could be a result of the increased accuracy of {\Blue the} spot volatility estimator and/or the reduction in bias attained by using a general kernel. 

While working on the final stages of writing the present manuscript, we became aware of a recent work by \cite{benvenutifunctionals}, who also studied a type of plug-in estimator with a general two-sided kernel spot volatility estimator. However, unlike our work, only a suboptimal bandwidth asymptotic regime (i.e., assuming under-smoothing) was considered and only consistency was established. 

The rest of this paper is organized as follows. Section \ref{setting} introduces the framework, assumptions, and  main results.  Section \ref{SimulationSect} illustrates the performance of our method via Monte Carlo simulations and compares it to the jackknife estimator of \cite{li2019efficient}.  The proofs are deferred to an appendix section.

\textbf{Notation:} We shall use the following notation:
\begin{itemize}
\item The set of (nonnegative definite) $d\times{}d$ real matrices $x=[x^{\ell m}]_{\ell,m=1}^{d}$ is denoted ($\mathcal{M}_+^d$) $\mathbb{R}^{d\times d}$;
\item For a $C^1$ function $g:\mathbb{R}^{d\times d}\to\mathbb{R}$, its gradient {\Red $\nabla g(x)$} is the $d\times d$ matrix with $(\ell,m)$ entry $\partial_{\ell m}g(x):=\frac{\partial g(x)}{\partial x^{\ell m}}$;
\item For $x,y\in\mathbb{R}^{d\times d}$, its inner product is defined as $x\cdot y=\sum_{\ell,m}x^{\ell m}y^{\ell m}$. Also, $x^{*}$ is the transpose of $x$ and $\|x\|=\sqrt{\sum_{\ell,m=1}^n (x^{\ell m})^2}$ is the Euclidian norm. 

\end{itemize}

\section{The Setting, Estimator, and Main Results} \label{setting}
Throughout, we consider a $d$-dimensional It\^o semimartingale $X=[X^1,\dots,X^d]^*\in\mathbb{R}^{d\times{}1}$ of the form:
\begin{equation} \label{eq:X}
\begin{split} 
X_{t}=&X_0 + \int_0^t\mu_{s} ds +\int_0^t\sigma_{s} d W_{s}\\
&+\int_0^t \int_E \delta(s, z) \mathbbm{1}_{\{|\delta(s, z)| \leq 1\}}(\mathfrak{p}-\mathfrak{q})(d s, d z)+\int_0^t \int_E \delta(s, z) \mathbbm{1}_{\{|\delta(s, z)|>1\}} \mathfrak{p}(d s, d z),
\end{split}
\end{equation} 
where all stochastic processes ($\mu :=\left\{\mu_{t}\right\}_{t \geq 0}$, $\sigma :=\left\{\sigma_{t}\right\}_{t \geq 0}$,  $W :=\left\{W_{t}\right\}_{t \geq 0}$, $\delta:=\left\{\delta(t, z):t\geq 0, z\in E\right\}$, $\mathfrak{p}:=\{\mathfrak{p}(B):B\in\mathcal{B}(\mathbb{R}_{+}\times E)\}$, etc.) are defined on a complete filtered probability space  $\left(\Omega, \mathcal{F}, \mathbb{F}, \mathbb{P}\right)$  with filtration $ \mathbb{F}=(\mathcal{F}_{t})_{t \geq 0}$. Here, {\DG $W=[W^1,\dots,W^d]^*$} is a $d$-dimensional standard Brownian Motion (BM)  adapted to the filtration \(\mathbb{F}\), $\delta$ is a predictable $\mathbb{R}^d$-valued function on $\Omega \times \mathbb{R}_{+} \times E$, and $\mathfrak{p}$ is a Poisson random measure on $\mathbb{R}_{+} \times E$ for some arbitrary Polish space $E$ with compensator $\mathfrak{q}(\mathrm{d} u, \mathrm{~d} x)=\mathrm{d} u \otimes \lambda(\mathrm{d} x)$, where $\lambda$ is a $\sigma-$ finite measure on $E$ having no atom. 
We also denote the spot variance-covariance process as $c_t := \sigma_t \sigma_t^*$, which takes values in the set $\mathcal{M}_d^+$ of nonnegative definite {\DG $d\times d$} {\DR symmetric} matrices. {\Blue As common} in the literature, $c$ is called the spot volatility of the process. We assume $c$ is also an It\^o semimartingale following the dynamics:
\begin{equation} \label{eq:sigma}
{\DG c_{t}^{\ell m}=c_0^{\ell m} + {\int_0^t{\tilde{\mu}^{\ell m}_{s}} \mathrm{d} s+\int_{0}^{t}\tilde{\sigma}^{\ell m}_{s} \mathrm{d}W_{s} }, \quad 1\leq \ell, m\leq d,}
\end{equation}
where {$W:=\{W_t\}_{t\geq0}$} is the same $d$-dimensional Brownian Motion driving the dynamics of $X$. {\Red The assumption \eqref{eq:sigma} is common in the literature (see, e.g., the monographs \cite{JacodProtter,jacodaitsahalia}), and was also assumed in the work of \cite{jacod2013quarticity}{\Blue )}.} {\DG For each $1\leq \ell, m\leq d$, $\{\tilde{\mu}^{\ell m}_t\}_{t\geq{}0}$ is adapted locally bounded {\DR and}  $\{\tilde{\sigma}^{\ell m}_t\}_{t\geq{}0}$ is adapted c\`adl\`ag.}

We now state the main assumption on the process $X$:
\begin{assumption} \label{process} 
The process $X$ follows the dynamics (\ref{eq:X}) with $c_t=\sigma_t\sigma_t^*$ satisfying (\ref{eq:sigma}) and, for some $r \in[0,1)$, a measurable function $\Gamma_m: E \rightarrow \mathbb{R}_{+}$, a constant $C_m<\infty$, and a localizing sequence of stopping times $\left(\tau_m\right)_{m \geq 1}$ such that $\tau_m \rightarrow \infty$, we have
$$
t \in\left[0, \tau_m\right] \Longrightarrow\left\{\begin{array}{l}
\left|\mu_t\right|+\left|\sigma_t\right|+\left|\tilde{\mu}_t\right|+\left|\tilde{\sigma}_t\right| \leq C_m, \\
|\delta(t, z)| \wedge 1 \leq \Gamma_m(z), \quad \text { where } \int \Gamma_m(z)^r \lambda(d z)<\infty.
\end{array}\right.
$$
\end{assumption}
The parameter $r$ determines the jump activity of the process: the larger $r$ is, the more active or frequent are the small jumps of the process. When $r=0$, the process {\Blue exhibits} {\DR finitely} many jumps in any bounded time interval (in that case, we say that the jumps are of finite activity). When $r<1$, the jump component of the process is of bounded variation, which is a standard {\Blue constraint} for the truncated realized quadratic variation to achieve efficiency (\cite{JacodProtter}).

Next, we give the conditions on the kernel function $K$ {\DG for which we need to define the function} 
\[
	{\DG L(t) := \int_{t}^{\infty} K(u) d u \mathbf{1}_{\{t>0\}}-\int_{-\infty}^{t} K(u) d u \mathbf{1}_{\{t \leq 0\}}.}
\]
As explained in the introduction, we aim to consider two-sided kernels of unbounded support.  It is important to remark that this type of {\Blue kernel} presents some challenges to our analysis (see comments before \eqref{eq:left}).
\begin{assumption} \label{kernel} 
The kernel function $K : \mathbb{R} \rightarrow \mathbb{R}$ is bounded such that
\begin{enumerate}
\item[{\rm a.}] $\int K(x) d x=1$;
\item[{\rm b.}] $K$ is Lipschitz and piecewise $C^2$  on its support $(A,B)$,  where $A < B$, $A \in [-\infty, 0]$, and $B \in [0, \infty]$; 
\item[{\rm c.}] (i) $\int|K(x)  x| d x<\infty$; (ii) $K(x) |x|^{2+\epsilon} \rightarrow 0$, $K^{\prime}(x) |x|^{2+\epsilon} \rightarrow 0$, with $\epsilon>1/2 $ as $|x| \rightarrow \infty$ ; (iii) $\int |K^{(j)}(x)|dx < \infty$, for $j=0,1,2,3$; (iv) $\int_0^\infty\int_{w}^\infty |K'(x)|dxdw<\infty$; (v) $\int K^{2}(x)dx<\infty$; {\DG (vi) $\int_{-\infty}^{\infty}|L(u)u|du<\infty$};
\item[{\rm d.}] (i) $V_{-\infty}^{\infty}(K^{(j)}):=\lim_{m\to\infty}V_{-m}^{m}(K^{(j)})<\infty$, where $V_{-m}^{m}(K^{(j)})$ is the total variation of $K^{(j)}$ on the interval $[-m,m]$, $j=0,1,2,3$; (ii) $V_{-\infty}^{\infty}(K^2)<\infty$;
\item[{\rm e.}] $K^{\prime}$ and $K^{\prime\prime}$ exist piecewise on $(A,B)\subset\mathbb{R}$ and are absolutely continuous on {\Red their} corresponding domain $\operatorname{dom}(K^{\prime})$ and $\operatorname{dom}(K^{\prime\prime})\subset(A,B)$, respectively;
\item[{\rm f.}] $K^{\prime\prime\prime}$ exists piecewise on $(A,B)\subset\mathbb{R}$.
\end{enumerate}
\end{assumption}

\begin{remark}
	{\Red Even though the list of assumptions above seems long, it is important to remark that all kernels commonly used in practice satisfy these conditions. The assumptions on the regularity of the high-order derivatives $K^{(j)}$ for $j=2,3$ are only used to show the following limit:
	\begin{equation}\label{eq:v5}
    \sqrt{\Delta_{n}} \sum_{i=1}^n\left[\bar{K}_{n}\left(t_i\right)-1\right] g\left(c_i^n\right)\longrightarrow \theta \int_{-\infty}^{0}L(u)du g\left(c_{0}\right)-\theta \int_{0}^{\infty}L(u)du g\left(c_{T}\right),
\end{equation}
where recall that $\bar{K}_{n}\left(t_i\right):=\Delta_{n}\sum_{j=1}^{n} K_{{k}_{n}\Delta_n}\left(t_{j-1}-t_i\right)$. This limit is easy to deal with in the case of a uniform kernel because, in that case, $\bar{K}_{n}\left(t_i\right)\equiv 1$ for almost all $i$ except if $i$ is close to $0$ or $n$.}
\end{remark}

For an arbitrary process $\{U_{t}\}_{t\geq{}0}$,  a filtration  $\mathbb{F}=(\mathcal{F}_{t})_{t \geq 0}$,  and  a given time span $\Delta_{n}>0$, we shall use the notation 
\[
	U_{i}^{n}:=U_{i\Delta_{n}},\qquad 
	\Delta_{i}^{n} U:=U_{i}^{n}-U_{i-1}^{n}, \qquad \mathcal{F}^n_{i} := \mathcal{F}_{i \Delta_n}.
\]
Stable convergence in law is denoted by \( \stackrel{st}{\longrightarrow}\). See (2.2.4) in \cite{JacodProtter} for the definition of this type of convergence. As usual, $a_{n}\sim c_{n}$  means that $a_{n}/c_{n}\to{}1$ as $n\to\infty$.

Throughout, we assume that we sample the process $X$ in (\ref{eq:X}) at discrete times   $t_{i} :=t_{i, n} :=i \Delta_{n}$,  where $\Delta_{n} :=T / n$ and $T\in(0,\infty)$ is a given fixed time horizon. 
To estimate the spot volatility $c_{t_i}$, we adopt a Nadaraya-Watson type of kernel estimator of the form:
\begin{align}\label{MDKNN}
	{\hat{c}^{n,X}_{i}} :=\frac{\sum_{j=1}^{n} K_{{k}_{n}\Delta_n}\left(t_{j-1}-t_i\right)\left(\Delta_{j}^{n} X\right)\left(\Delta_{j}^{n} X\right)^*}{\Delta_{n}\sum_{j=1}^{n} K_{{k}_{n}\Delta_n}\left(t_{j-1}-t_i\right)},
\end{align}
where \(K_{{b}}(x):=K(x / b) / b\), $k_{n}\in\mathbb{N}$, and ${b_n}:={k}_{n} \Delta_n$ is the bandwidth of the kernel function. This is a variation of the estimator
\begin{align}\label{MDKNNb}
	{\tilde{c}^{n,X}_{\tau}} :=\sum_{j=1}^{n} K_{{k}_{n}\Delta_n}\left(t_{j-1}-\tau\right)\left(\Delta_{j}^{n} X\right)\left(\Delta_{j}^{n} X\right)^*,
\end{align}
first proposed in \cite{fan2008spot} and \cite{kristensen2010nonparametric},  and studied in recent papers such as \cite{FigLi}, \cite{FigWu}. The denominator 
\[
	\bar{K}_{n}\left(t_i\right):=\Delta_{n}\sum_{j=1}^{n} K_{{k}_{n}\Delta_n}\left(t_{j-1}-t_i\right),
\]
{\DR in \eqref{MDKNN} is known to} improve the performance \eqref{MDKNNb} for values of $\tau=t_i$ near the edge of the estimation interval $[0,T]$. The variation \eqref{MDKNN} was also considered in \cite{kristensen2010nonparametric} and \cite{FigLi} {\DR when $t_i$ is {replaced} with a fix time $\tau\in(0,T)$}. In that case, $\bar{K}_{n}\left(\tau\right):=\Delta_{n}\sum_{j=1}^{n} K_{{k}_{n}\Delta_n}\left(t_{j-1}-\tau\right)\to1$ and the asymptotic behavior of ${\tilde{c}^{n,X}_{\tau}}$ and 
\[
	{\hat{c}^{n,X}_{\tau}} :=\sum_{j=1}^{n} K_{{k}_{n}\Delta_n}\left(t_{j-1}-\tau\right)\left(\Delta_{j}^{n} X\right)\left(\Delta_{j}^{n} X\right)^*/\bar{K}_{n}\left(\tau\right)
\] 
are equivalent. However, for our estimator (see \eqref{Target0a} below), we need to consider $t_i$ for all $i=1,\dots, n$ and, thus, it is not direct that we can simply omit $\bar{K}_{n}\left(t_i\right)$ (in fact, its presence helps with the analysis).

To handle jumps in the process $X$, we also consider a truncated version {\DR of \eqref{MDKNN}:}
\begin{equation}\label{estivoltrun}
\hat{c}_i^{n,X,v_n}:=\frac{\sum_{j=1}^n K_{k_n \Delta_n}\left(t_{j-1}-t_{i}\right)\left(\Delta_j X\right)\left(\Delta_j X\right)^*\mathbbm{1}_{\left\{\left\|\Delta_j X\right\| \leq v_n\right\}}}{\Delta_n \sum_{j=1}^n K_{k_{n}\Delta_n}\left(t_{j-1}-t_{i}\right)}{,}
\end{equation}
with a suitable sequence of truncation levels $v_n\in(0, \infty]$. {\DR More specifically, we consider two types of conditions on $v_n$. In the case that $X$ is continuous, we take
    \begin{equation}\label{moreassum1}
        v_n=\alpha \Delta_n^{\varpi}, \quad \text { for some arbitrary fixed } \alpha\in(0,\infty]\text { and } \varpi < \frac{1}{2},
    \end{equation}
while if $X$ is discontinuous, we take     
\begin{equation}\label{moreassum}
	v_n=\alpha \Delta_n^{\varpi}, \quad \text { for some fixed } \alpha\in(0,\infty)\text{ and } \varpi \in\left[\frac{2{\DG \ell}-1}{2(2{\DG \ell}-r)}, \frac{1}{2}\right),
\end{equation}
where $\ell\geq{}4$ is a positive integer that will be specified below. Note that when $\alpha=\infty$ in \eqref{moreassum1}, we recover the non-truncated version (\ref{MDKNN}) and, in that case, the value of $\varpi$ is irrelevant.}
For simplicity, we {\DR often omit the dependence on $X$ and $v_n$ on \eqref{MDKNN} and \eqref{estivoltrun} and simply write $\hat{c}_{i}^{n}$}.

{Our estimation target is the integrated volatility functional, defined as
\begin{equation}\label{Target0a}
V(g)_{t}:=\int_{0}^{t} g\left(c_{s}\right) d s,
\end{equation} 
for a continuous function $g:\mathcal{M}_d^+\to\mathbb{R}$. 
As proposed in \cite{jacod2013quarticity}, a natural estimator for this integral is given by {\DR $\Delta_{n} \sum_{i=1}^{\left[t / \Delta_{n}\right]} g\left(\hat{c}_{i}^{n}\right)$, which can be interpreted as a Riemann sum approximation of \eqref{Target0a}.} However, as we shall see, it will be easier to derive the asymptotic properties of the estimator
\begin{equation} \label{eq:simple_estimator}
V(g)^{n}_{t}:=\Delta_{n} \sum_{i=1}^{\left[t / \Delta_{n}\right]} g\left(\hat{c}_{i}^{n}\right){\bar{K}_n\left(t_i\right)}.
\end{equation}
Furthermore, we can see \eqref{eq:simple_estimator} as a weighted average version of $\Delta_{n} \sum_{i=1}^{\left[t / \Delta_{n}\right]} g\left(\hat{c}_{i}^{n}\right)$, which gives less weight to $g\left(\hat{c}_{i}^{n}\right)$ for values of $t_i$ near $0$ and $T$, where the estimator $\hat{c}_{i}^{n}$ is less accurate due to edge effects.}

{Our first result characterizes the asymptotic behavior of the estimation error of $V(g)^{n}_{t}$} when the bandwidth converges to $0$ at {an} "optimal" rate {\DR (c.f. Theorem \eqref{clt_suboptimal} below and the remarks before).}

\begin{theorem} \label{clt}
Let $g$ be a $C^{3}$ function such that, for all $x \in \mathbb{R}^{d \times d}$,
\begin{equation} \label{eq:g_prime_bound}
|g(x)| \leq C\left(1+\|x\|^{\DG \ell}\right),\quad\left|  \partial _{p_1 q_1,\dots,p_j q_j}g(x)\right| \leq C\left(1+\|x\|^{{\DG \ell}-j}\right), \quad j=1,2,3{,}
\end{equation} for $p_1, q_1,\dots,p_j, q_j\in\{1, \ldots, d\}$ and for some constants $C > 0$ and ${\DG \ell} \ge 4$. Assume $k_n$ satisfies 
\begin{equation} \label{eq: k_n}
k_{n} \sim \frac{\theta}{\sqrt{\Delta_{n}}} \qquad (i.e.,\;\; \text{the bandwidth}\ {{b}_n}:=k_n\Delta_n\sim \theta \sqrt{\Delta_n}),
\end{equation}
for some $\theta \in(0, \infty)$. Then, under (\ref{eq:X}) and Assumptions \ref{process} and \ref{kernel}, the following assertions hold true:
\begin{itemize}
    \item[a)] Suppose that {\DR $X$ is continuous} (i.e., $\delta\equiv 0$ in (\ref{eq:X})). Then, the estimator \eqref{eq:simple_estimator} of the integrated volatility functional $V(g)_T$ with {\DR $\hat{c}_{i}^{n}$ given by (\ref{estivoltrun}), {under} the condition \eqref{moreassum1}\footnote{In particular, by taking $\alpha=\infty$, we recover the untruncated estimator \eqref{MDKNN}.},}
    satisfies the following stable convergence in law, as $n \rightarrow \infty$:
\begin{equation}\label{eq:convergence_in_law}
{\DR \frac{1}{\sqrt{\Delta_{n}}}\Big(V(g)_T^{n}-V(g)_T\Big)
\stackrel{st}{\longrightarrow} A^{1} + A^{2}+A^{3}+Z},
\end{equation}
where $Z$ is {\Blue an} r.v. defined on an extension $\left(\widetilde{\Omega}, \widetilde{\mathcal{F}},(\widetilde{\mathcal{F}}_{t})_{t \geq 0}, \widetilde{\mathbb{P}}\right)$ of $\left(\Omega, \mathcal{F}, (\mathcal{F}_{t})_{t \geq 0}, \mathbb{P}\right)$, which, conditionally on $\mathcal{F}$, is a centered Gaussian variable with variance
\begin{equation}\label{VarForNL}
{\widetilde{\mathbb{E}}\left(Z^{2} \mid \mathcal{F}\right) =\sum_{j,k,l,m} \int_0^T\partial_{jk}g(c_s)\partial_{lm}g(c_s)(c_s^{jl}c_s^{km}+c_s^{jm}c_s^{kl})ds,}
\end{equation}
{and
\begin{equation}\label{MDfnOfAs}
\begin{split}
A^{1}&:=-\theta \int_{0}^{\infty}L(u)dug\left(c_{0}\right)+\theta \int_{-\infty}^{0}L(u)dug\left(c_{T}\right), \\
A^{2}&:=\frac{1}{2\theta} \int_0^T \sum_{p,q,u,v=1}^{d}\partial^{2}_{p q, u v} g\left(c_s\right) \check{c}_s^{p q, u v} d s \int_{-\infty}^{\infty} K^2(u) d u,\\
A^{3}&:=\frac{\theta}{2} \int_0^T \sum_{p,q,u,v=1}^{d}\partial^{2}_{p q, u v} g\left(c_s\right)\tilde{c}_s^{p q, u v}  d s\left[\int_{-\infty}^{\infty} L(u)^2 d u+\left(\int_{-\infty}^{0}-\int_{0}^{\infty}\right) L(u) d u\right]{\DR ,}
\end{split}
\end{equation}
where} $\check{c}_{s}^{pq,uv}:=c_{s}^{pu}c_{s}^{qv} +c_{s}^{pv}c_{s}^{qu}$ and $\tilde{c}^{pq,uv}_{s}:=\sum_{r=1}^d\tilde{\sigma}_s^{p q, r} \tilde{\sigma}_s^{u v, r}$.
\item[b)] 
{\DR When $X$ is discontinuous, we {\DR still} have (\ref{eq:convergence_in_law}) for the estimator $V\left(g\right)^{n}_{T}$ with the truncated version (\ref{estivoltrun}) provided that the stronger condition \eqref{moreassum} is satisfied}.
\end{itemize}
\end{theorem}

\begin{remark} 
Note that the bias terms depend on the kernel in a rather intricate manner (especially, the term $A^3$). The {\Blue signs} of the coefficients are always the same regardless of $K$. For instance, $\int_{-\infty}^{\infty} L(u)^2 d u+\left(\int_{-\infty}^{0}-\int_{0}^{\infty}\right) L(u) d u$  is negative because $|L(u)|\leq 1$ for all $u$, and $L(u)\leq{}0$ ($L(u)\geq 0$) for any $u<0$ ($u>0$). Thus, $A^3$ goes in {\Blue the} opposite direction to $A^2$.
\end{remark}
\begin{remark}
If we use the uniform {\DR right-sided} kernel $K(x)={\bf 1}_{[0,1)}(x)$ in our estimators, {as} in \cite{jacod2013quarticity,jacod2015estimation}, then
\begin{equation}\label{eq:K_bar_equal_1}
\bar{K}_n\left(t_{i}\right)=\Delta_{n}\sum_{j=1}^{n}{K_{b_n}}\left(t_{j-1}-t_{i}\right)=\Delta_{n}\sum_{j=i+1}^{{k_n}+i}\frac{1}{{b_n}}=1, \quad \text{for} \ 0\leq i\leq n-k_{n}.
\end{equation}
{\DR In that case,} the estimator in \cite{jacod2015estimation}, defined as 
\[
	V\left(g\right)^{n,JR}_{T}:=\Delta_{n}\sum_{i=1}^{n-k_{n}-1}g\left(\hat{c}_{i}^{n,JR}\right),
\]
where
\begin{equation}\label{vola_JR}
    \hat{c}_i^{n,{JR}}:=\frac{1}{k_n \Delta_n} \sum_{j=0}^{k_n-1} \left(\Delta_{i+j}^n X\right)\left(\Delta_{i+j}^n X\right)^* \mathbbm{1}_{\left\{\left\|\Delta_{i+j}^n X\right\| \leq v_n\right\}},
\end{equation}
can be written {\DR in terms of our estimator \eqref{eq:simple_estimator} as follows:}
\begin{equation}\label{eq:right_V_estimator}
    V\left(g\right)^{n,JR}_{T}=V\left(g\right)^{n}_{T}-\Delta_{n}\sum_{i=n-k_{n}-1}^{n}g\left(\hat{c}_{i}^{n}\right)\bar{K}\left(t_{i}\right)+\Delta_n g\left(\hat{c}_{0}^{n}\right).
\end{equation}
{The formula \eqref{eq:right_V_estimator} holds because} (\ref{eq:K_bar_equal_1}) ensures that $\hat{c}_{i}^{n}=\hat{c}_{i+1}^{n,JR}$, for $0\leq i\leq n-k_{n}$. {It is easy to see that, for {the uniform kernel $K(x)={\bf 1}_{[0,1)}(x)$,} $\Delta_{n}^{1/2}\sum_{i=n-k_{n}-1}^{n}g\left(\hat{c}_{i}^{n}\right)\bar{K}\left(t_{i}\right)\to \frac{\theta}{2}g\left(c_{T}\right)$ and the bias terms of Theorem \ref{clt} reduce to
\begin{equation}
\begin{array}{l}
A^{1}:=-\frac{\theta}{2}g\left(c_{0}\right){,} \\
A^{2}:=\frac{1}{2\theta} \int_0^T \sum_{p,q,u,v=1}^{d}\partial^{2}_{p q, u v} g\left(c_s\right) \check{c}_s^{p q, u v} d s,\\
A^{3}:=-\frac{\theta}{12} \int_0^T \sum_{p,q,u,v=1}^{d}\partial^{2}_{p q, u v} g\left(c_s\right)\tilde{c}_s^{p q, u v}  d s.
\end{array}
\end{equation}
Thus, Theorem \ref{clt} implies that 
\begin{equation} \label{eq:clt_right}
\frac{1}{\sqrt{\Delta_{n}}}\left(V(g)_T^{n,JR}-V(g)_T\right) \stackrel{st}{\longrightarrow}\tilde{A}^{1} +  A^{2}+A^{3}+Z,
\end{equation} where $\tilde{A}^{1}=-\frac{\theta}{2}\left[g\left(c_{0}\right)+g\left(c_{T}\right)\right]$. We then recover the same result as \cite{jacod2015estimation}.}
\end{remark}

%

To make this CLT ``feasible'' {\Blue in practical work}, the bias terms need to be estimated. 
{For $A^1$, we use that
\begin{equation} \label{eq:2side_edge_bias}
k_n \sqrt{\Delta_n}g\left(\hat{c}^n_1 \right) \stackrel{\mathbb{P}}{\to} \theta g(c_0), \quad k_n \sqrt{\Delta_n}g\left(\hat{c}^n_n \right) \stackrel{\mathbb{P}}{\to} \theta g(c_T).
\end{equation}
For $A^2$, we can apply Theorem \ref{clt} to the functional 
$$
    h(x):=\sum_{p, q, u, v=1}^d \partial_{p q, u v}^2 g(x)\left(x^{pu}x^{qv}+x^{pv}x^{qu}\right),
$$
for any $x\in \mathbb{R}^{d\times d}$, to get
\begin{equation}\label{eq:h}
V\left(h\right)^{n}_{T}=\Delta_{n}\sum_{i=1}^{n}h\left(\hat{c}_{i}^{n}\right)\bar{K}\left(t_{i}\right)\stackrel{\mathbb{P}}{\longrightarrow}  \int_0^T \sum_{p,q,u,v=1}^{d}\partial^{2}_{p q, u v} g\left(c_s\right) \check{c}_s^{p q, u v} d s.
\end{equation}
The term $A^3$ is the hardest to estimate since it involves the {\Blue `volvol'} $\tilde{c}$. In the case of a uniform {\DR right-sided} kernel, {\DR $K(x)={\bf 1}_{[0,1)}(\chi)$}, \cite{jacod2015estimation} proposed the estimator
\begin{equation}\label{An3}
    {\DR \widehat{A}^{n, 3}}:=\frac{\sqrt{\Delta_n}}{8} \sum_{i=1}^{n-2 k_n+1} \sum_{j, k, l, m} \partial_{j k, l m}^2 g(\hat{c}_i^n)\left(\hat{c}_{i+k_n}^{n,j k}-\hat{c}_i^{n,j k}\right)\left(\hat{c}_{i+k_n}^{n,l m}-\hat{c}_i^{n,l m}\right),
\end{equation}
with $\hat{c}$ given by (\ref{vola_JR}) and showed that  ${\DR \widehat{A}^{n, 3}}$ converges to a linear combination of the bias terms $A^{2}$ and $A^3$. The following theorem shows the corresponding result for a general kernel function. The proof of Theorem \ref{thm2.2} is {given} in Appendix \ref{proofthm2.2}.}
\begin{theorem}\label{thm2.2}
    Under the notations {\DR and conditions} of Theorem \ref{clt}, we have, {as $n\to\infty$},
\begin{equation}\label{eq:bias_estimates}
    \begin{aligned}
    {\DR \widehat{A}^{n, 3}}\stackrel{\mathbb{P}}{\longrightarrow}&\,\frac{1}{4\theta} \int_0^T \sum_{j, k, l, m}\partial^2_{jk, lm} g\left(c_s\right) \check{c}_s^{jk, lm} d s
    {\int^{\infty}_{-\infty} K(z)(K(z)-K(z-1))d z}\\
    &+\frac{\theta}{4} \int_0^T \sum_{j, k, l, m}\partial^2_{jk, lm} g\left(c_s\right) \tilde{c}_s^{jk, lm} d s\\
    &\quad\times \Bigg({\int_{-\infty}^{\infty} L(z)(L(z)-L(z-1))d z+\frac{1}{2}-\int_{-\infty}^{\infty}K(z)(|z|\wedge{}1)dz\Bigg),}
    \end{aligned}
\end{equation}
when $\hat{c}$ is given by (\ref{estivoltrun}) {\DR with the condition \eqref{moreassum}}. Furthermore, {\DR when $X$ is continuous (i.e., $\delta\equiv 0$ in (\ref{eq:X})),} \eqref{eq:bias_estimates} still holds {\DR under the weaker condition \eqref{moreassum1}} {(in particular, we can use the untruncated estimator \eqref{MDKNNb} by taking $\alpha=\infty$)}.
\end{theorem}
Under the Assumption (\ref{eq: k_n}), based on  (\ref{eq:2side_edge_bias}), (\ref{eq:h}), and (\ref{eq:bias_estimates}), we {\DR then} propose {a bias-corrected} estimator of the form:
\begin{equation} \label{eq:unbiased_V}
\begin{split}
{\widetilde{V}(g)^{n}_T} &:= V\left(g\right)^{n}_{T} +k_n \Delta_n g\left(\hat{c}_1^n\right) {\int_0^{\infty} L(u)d u}-k_n \Delta_n g\left(\hat{c}_n^n\right){\int_{-\infty}^0 L(u) d u}\\
&\quad-\frac{\Delta_n}{4} \sum_{i=1}^{n-2k_n+1}\sum_{p, q, u, v=1}^d \partial_{p q, u v}^2 g\left(\hat{c}_i^n\right)\left(\hat{c}_{i+k_n}^{n, p q}-\hat{c}_i^{n, p q}\right)\left(\hat{c}_{i+k_n}^{n, u v}-\hat{c}_i^{n, u v}\right)C_{K,1}\\
&\quad+\frac{1}{2k_n}V\left(h\right)^{n}_{T}\left[C_{K,2}-\int K^2(u) du\right]{\DR ,}
\end{split}
\end{equation}
where
$$
C_{K,1}:=\frac{\int_{-\infty}^{\infty} L(z)^2 d z+\left(\int_{-\infty}^{0}-\int_{0}^{\infty}\right) L(z) d z}{{\int_{-\infty}^{\infty} L(z)(L(z)-L(z-1))d z+\frac{1}{2}-\int_{-\infty}^{\infty}K(z)(|z|\wedge{}1)dz}}{},
$$
and
$$
C_{K,2}:={\int^{\infty}_{-\infty} K(z)(K(z)-K(z-1))d zC_{K,1}}.
$$
Next, we show that indeed ${\widetilde{V}(g)^{n}_T}$ enjoys a centered CLT. Its simple proof is given in Section \ref{SmplCorUnb}. 
\begin{corollary}\label{BiasCorrectedCLT0}
Under the assumptions {\DR and notation} of Theorem \ref{clt},  {as $n\to\infty$}, we have \begin{equation}\label{UnBsCLTb}
\frac{1}{\sqrt{\Delta_n}}\left( {\widetilde{V}(g)^{n}_T} - V(g)\right) \stackrel{st}{\to} Z,
\end{equation}
{\DR where ${\widetilde{V}(g)^{n}_T}$ is defined as in \eqref{eq:unbiased_V} with $\hat{c}$ given by (\ref{estivoltrun}) {\DR with the} condition \eqref{moreassum}. Furthermore, {\DR when $X$ is continuous, \eqref{UnBsCLTb} holds under the weaker condition {\eqref{moreassum1}}.}} 
\end{corollary}

\begin{remark}
\cite{jacod2015estimation} proposed the following bias-corrected estimator:
\begin{align}\label{CRJEN}
    \widetilde{V}(g)^{n,{JR}}_T&:=V(g)^{n,{JR}}_T+\frac{k_n \Delta_n}{2}\left(g\left(\hat{c}_1^{n,{JR}}\right)+g\left(\hat{c}^{n,{JR}}_{n-k_n+1}\right)\right)-\frac{3}{4k_n}V(h)_T^{n,{JR}}\\
    &\quad+\frac{\Delta_n}{8} \sum_{i=1}^{n-2k_n+1}\sum_{p, q, u, v=1}^d \partial_{p q, u v}^2 g\left(\hat{c}_i^{n,{JR}}\right)\left(\hat{c}_{i+k_n}^{n,{JR}, p q}-\hat{c}_i^{n,{JR}, p q}\right)\left(\hat{c}_{i+k_n}^{n,{JR}, u v}-\hat{c}_i^{n,{JR}, u v}\right),
    \nonumber
\end{align}
where $\hat{c}_i^{n,{JR}}$ and     $V\left(g\right)^{n,JR}_{T}$
are defined as in (\ref{vola_JR}) and (\ref{eq:right_V_estimator}), respectively. {\DR There is a connection between \eqref{CRJEN} and our estimator \eqref{eq:unbiased_V} taking a right-sided uniform kernel $K(x)={\bf 1}_{[0,1)}(x)$. Indeed, note that, with that kernel,} (\ref{eq:unbiased_V}) reduces to
\begin{equation}
\begin{split}
{\widetilde{V}(g)^{n}_T} &:= V\left(g\right)^{n}_{T} +\frac{k_n \Delta_n}{2} g\left(\hat{c}_1^n\right)-\frac{3}{4k_n}V\left(h\right)^{n}_{T}\\
&\quad+\frac{\Delta_n}{8} \sum_{i=1}^{n-2k_n+1}\sum_{p, q, u, v=1}^d \partial_{p q, u v}^2 g\left(\hat{c}_i^n\right)\left(\hat{c}_{i+k_n}^{n, p q}-\hat{c}_i^{n, p q}\right)\left(\hat{c}_{i+k_n}^{n, u v}-\hat{c}_i^{n, u v}\right).
\end{split}
\end{equation}
Then, by {\DR \eqref{eq:K_bar_equal_1} and} (\ref{eq:right_V_estimator}), we {\DR conclude that} 
\begin{equation}\label{relationtoJR}
\begin{split}
{ \widetilde{V}(g)^{n}_T} &=\widetilde{V}(g)^{n,{JR}}_T +\Delta_{n}\sum_{i=n-k_{n}-1}^{n}g\left(\hat{c}_{i}^{n}\right)\bar{K}\left(t_{i}\right)-\Delta_n g\left(\hat{c}_{0}^{n}\right)-\frac{k_n \Delta_n}{2} g\left(\hat{c}^{n,{JR}}_{n-k_n+1}\right)\\
&\quad-\frac{3}{4k_n}\Delta_n\sum_{i=n-k_{n}-1}^{n}h\left(\hat{c}_{i}^{n}\right)\bar{K}\left(t_{i}\right)+\frac{3\Delta_n}{4k_n} h\left(\hat{c}_{0}^{n}\right)\\
&\quad+\frac{\Delta_n}{8}\sum_{p, q, u, v=1}^d \partial_{p q, u v}^2 g\left(\hat{c}_{n-2k_n+1}^{n}\right)\left(\hat{c}_{n-k_n+1}^{n, p q}-\hat{c}_{n-2k_n+1}^{n, p q}\right)\left(\hat{c}_{n-k_n+1}^{n, u v}-\hat{c}_{n-2k_n+1}^{n, u v}\right)\\
&\quad-\frac{\Delta_n}{8}\sum_{p, q, u, v=1}^d \partial_{p q, u v}^2 g\left(\hat{c}_{0}^{n}\right)\left(\hat{c}_{k_n}^{n, p q}-\hat{c}_0^{n, p q}\right)\left(\hat{c}_{k_n}^{n, u v}-\hat{c}_0^{n, u v}\right).
\end{split}
\end{equation}
Since
$$
\sqrt{\Delta_n}\sum_{i=n-k_{n}-1}^{n}g\left(\hat{c}_{i}^{n}\right)\bar{K}\left(t_{i}\right)\stackrel{\mathbb{P}}{\longrightarrow} \frac{\theta}{2}g\left(c_T\right),\qquad 
\frac{k_n \sqrt{\Delta_n}}{2} g\left(\hat{c}^{n,{JR}}_{n-k_n+1}\right)\stackrel{\mathbb{P}}{\longrightarrow} \frac{\theta}{2}g\left(c_T\right),
$$
we know that $\Delta_n\sum_{i=n-k_{n}-1}^{n}g\left(\hat{c}_{i}^{n}\right)\bar{K}\left(t_{i}\right)-\frac{k_n \Delta_n}{2} g\left(\hat{c}^{n,{JR}}_{n-k_n+1}\right)=o_{p}\left(\sqrt{\Delta_n}\right)$. Also note that 
$$\sqrt{\Delta_n}\sum_{i=n-k_{n}-1}^{n}h\left(\hat{c}_{i}^{n}\right)\bar{K}\left(t_{i}\right)\stackrel{\mathbb{P}}{\longrightarrow} \frac{\theta}{2}h\left(c_T\right),
$$ 
which implies that $\frac{3}{4k_n}\Delta_n\sum_{i=n-k_{n}-1}^{n}h\left(\hat{c}_{i}^{n}\right)\bar{K}\left(t_{i}\right)=O_{p}\left(\frac{\sqrt{\Delta_n}}{k_n}\right)$. The remaining terms in (\ref{relationtoJR}) can {\DR directly} be proved {\DR to be} $o_{p}\left(\sqrt{\Delta_n}\right)$. We then conclude that
$$
\frac{1}{\sqrt{\Delta_n}}\left(\widetilde{V}(g)^{n}_T -\widetilde{V}(g)^{n,{JR}}_T\right)\stackrel{\mathbb{P}}{\longrightarrow} 0.
$$
Hence, by Corollary \ref{BiasCorrectedCLT0}, we can recover the {\DR stable} convergence in law for the bias-corrected estimator in \cite{jacod2015estimation}:
 \begin{equation}
\frac{1}{\sqrt{\Delta_n}}\left( \widetilde{V}(g)^{n,{JR}}_T - V(g)\right) \stackrel{st}{\to} Z.
\end{equation} 
\end{remark}

As mentioned above, the bias term $A^3 $ contains the `volvol' $\tilde{c} $ and the estimator of $A^3$ in (\ref{eq:bias_estimates}) could introduce extra variance to the {asymptotically} unbiased estimator (\ref{eq:unbiased_V}). {An alternative} way to eliminate {the} bias is {undersmoothing, i.e., through the selection of a bandwidth sequence $b_n$ converging to $0$ at a faster rate than $\sqrt{\Delta_n}$ ($b_n \ll \sqrt{\Delta_n}$) or, equivalently, picking $k_n$ such that $\sqrt{\Delta_n}k_n  \stackrel{n \to \infty}{\longrightarrow} 0$ (meaning that $\theta = 0$ in (\ref{eq: k_n}))}. This is the approach put forward in \cite{jacod2013quarticity}.
{In that case, it is expected that the bias terms $A^1$ and $A^3$ will vanish and the bias term $A^2 $ will dominate}. {{\DR Following the same idea,} we can devise an asymptotically unbiased {\DR estimator as follows}:
\begin{equation}\label{v_bar}
    \bar{V}(g)^n_T := \Delta_n \sum_{i= 1}^{n} \left( g(\hat{c}^n_i) - \frac{1}{2k_n} h(\hat{c}^n_i) \int K^2(u) du\right)\bar{K}\left(t_{i}\right).
\end{equation}
The following result establishes the asymptotic behavior of $\bar{V}(g)^n_T $. {Its proof is {given} in Appendix \ref{PrfSubs23}.}
\begin{theorem} \label{clt_suboptimal}
Assume $k_n$ satisfies \begin{equation} \label{eq:k_n_suboptimal}
k_{n}^2 \Delta_n \to 0,\quad  k_n^3 \Delta_n \to \infty{.}
\end{equation}
Then, under (\ref{eq:X}) and Assumptions \ref{process} and \ref{kernel}, we have the following stable convergence in law
\begin{equation} \label{eq:convergence_in_law_unbiased}
\frac{1}{\sqrt{\Delta_{n}}}\left(\bar{V}(g)_T^{n}-V(g)_T\right) \stackrel{st}{\longrightarrow} Z,
\end{equation}where $Z$ is as in Theorem \ref{clt} {\DR and ${\bar{V}(g)^{n}_T}$ is defined as in \eqref{v_bar} with $\hat{c}$ given by (\ref{estivoltrun}) under condition \eqref{moreassum}. Furthermore, {\DR when $X$ is continuous, \eqref{eq:convergence_in_law_unbiased} holds under the weaker condition \eqref{moreassum1}.}} 

\end{theorem}

\section{Simulation Study}\label{SimulationSect}
In this section we analyze the performance of {\DR our} kernel-based estimators. To isolate the effect of the bandwidth on the estimators' performance, we only consider continuous processes (i.e., $\delta\equiv 0$ in \eqref{eq:X} and, thus, $X\equiv X'$) and focus on the untruncated estimator \eqref{MDKNN} and the corresponding estimators $V(g)^{n}_{t}$ of \eqref{eq:simple_estimator} and $\widetilde{V}(g)^{n}_T$ of \eqref{eq:unbiased_V}.

\subsection{Simulation design}
We consider the data generating model in \cite{li2019efficient}:
\begin{equation} \label{eq:expOU_model}
\left\{\begin{array}{l}
d X_t = \sigma_t d W_t  \\
\sigma_t = exp(-1.6+F_t), \quad d F_t = -5 F_t dt + 2dB_t,
\end{array}\right.
\end{equation} with $\mathbb{E}\left[d W_t d B_t\right]  = -0.75 dt $. Throughout, we use the function $g(c) = c^2$, which corresponds to the integrated quarticity{,} and also the function $g(c) = \log(c)$.  We simulate data for  $T= 5$ days and  $T = 21$ days, and assume the process $X$ is observed once {\Red every 1 minute or every 5 minutes, with 6.5 trading hours per day, for all the results of Subsection \ref{Sec33}, but not for Subsection  \ref{Sec32}, where we assume a frequency of 1-second}. 

\subsection{Validity of the asymptotic theory}\label{Sec32}
We first show the finite sample behavior of the estimator is consistent with the Central Limit Theorem of Theorem \ref{clt}. Under the notation of Theorem \ref{clt}, we choose $\DR{ \theta = 0.2}$, $g(x) = x^2 $, and estimate the functional volatility using the exponential kernel $K(x) = \frac{1}{2}e^{-|x|}$. The choice of $\theta$ was obtained from the formula $b^*=\theta^*\sqrt{\Delta_n}$, where $b^*$ was chosen using the iterative method in \cite{FigLi}. The sample frequency of the data is 1 second. Based on 5,000 simulated paths from the model (\ref{eq:expOU_model}), we {calculate the 
standardized estimation error for each path $i=1,\dots, 5000$} according to Theorem \ref{clt}: \begin{equation}
z_i = \frac{\left(\Delta_n\sum_{j=1}^n g(\hat{c}^{n,i}_j)\bar{K}\left(t_j\right) - \int_0^T g(c^{i}_s) ds\right)/\sqrt{\Delta_n} - (A^{1,i}+A^{2,i} +A^{3,i})}{\sqrt{{\rm Var}(Z)^{i}}},
\end{equation}
where $c^i_s = (\sigma^i_s)^2 $, $\{\hat{c}^{n,i}_j\}_j $ are the kernel spot volatility estimates of the $i$th simulated path, and the variance ${\rm Var}(Z)^{i}=2 \int_0^T(g''(c^i_s)c^i_s)^2ds$ and the biases $A^{\ell,i}$ ($\ell=1,2,3$), as defined in \eqref{MDfnOfAs}, are computed via Riemann sum approximations  using the true simulated volatility $\{c^{n,i}_j\}_{j=0,\dots,n}$ and volvol observations $\{\tilde{c}^{n,i}_j\}_{j=0,\dots,n}$.
We then compare the histogram of the standardized estimates $z_i$ with a standard normal distribution (see Figure \ref{fig:hist_exp}). As shown therein, the distribution of the  estimation error is consistent with the asymptotic result in Theorem \ref{clt}.\\
\begin{figure}[h]
    \centering
    \includegraphics[scale = 0.6]{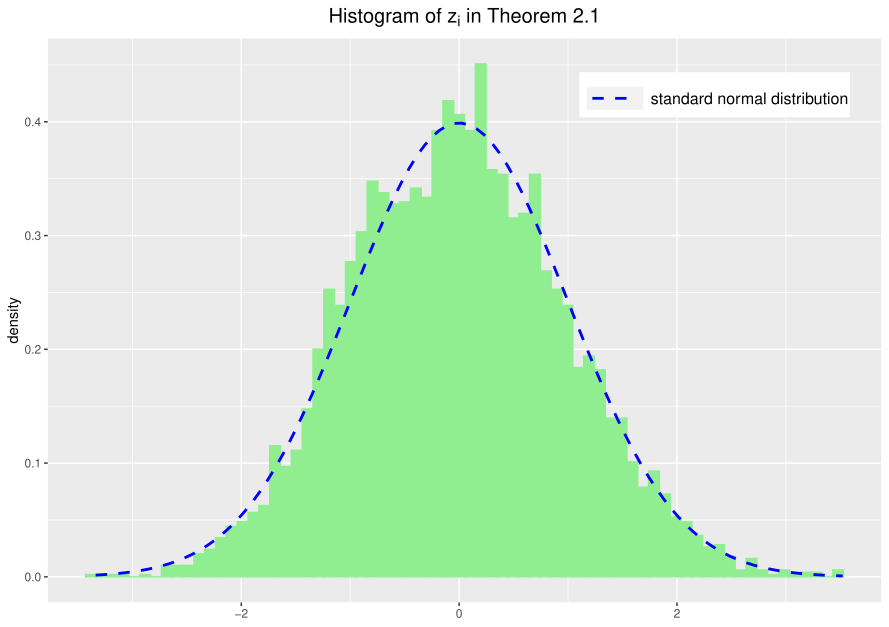}
    \caption{Histogram of $z_i$ and the density of standard normal distribution in Theorem \ref{clt}.  }
    \label{fig:hist_exp}
\end{figure} 
Similarly, we verify the behavior of the bias-corrected estimator proposed in Corollary \ref{BiasCorrectedCLT0}. With the same value of $\theta$, $g(x)$, exponential kernel, and simulated data, the standardized bias-corrected estimator takes the form:
\begin{equation}
    z_i = \frac{\frac{1}{\sqrt{\Delta_n}}\left(\widetilde{V}(g)^{n,i}_T - \int_0^T g(c^{i}_s) ds\right)}{\sqrt{Var(Z^i)}}.
\end{equation}
The corresponding histogram is shown in Figure \ref{fig:2}, which again suggest the CLT stated in Corollary \ref{BiasCorrectedCLT0} is valid.
\begin{figure}[h]
    \centering
    \includegraphics[scale = 0.6]{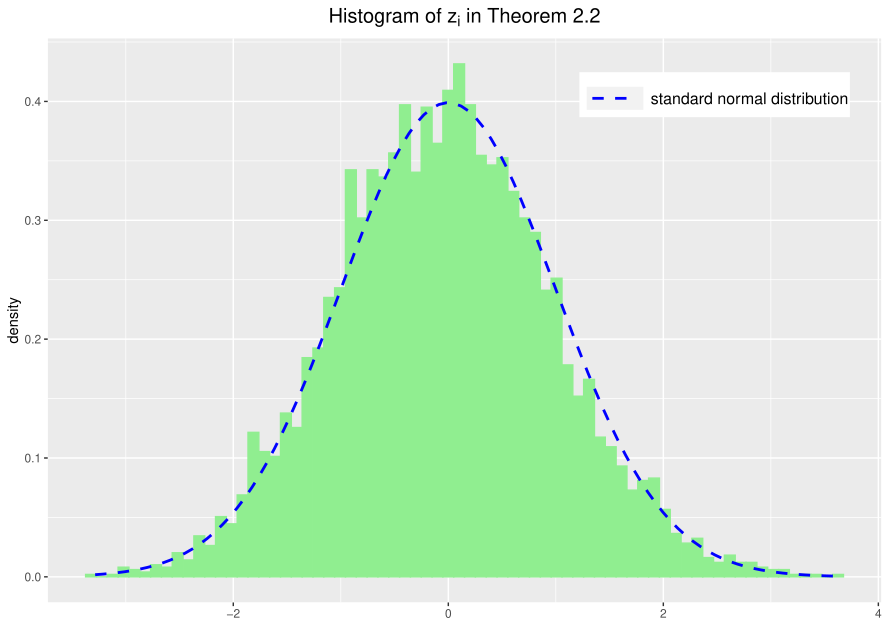}
    \caption{Histogram of $z_i$ and the density of standard normal distribution in Corollary \ref{BiasCorrectedCLT0}.  }
    \label{fig:2}
\end{figure} 
\begin{table}[h!]
\setlength\tabcolsep{3pt}
\begin{center} 
\begin{tabular}{  c c c c c c c c c }
\hline
& \multicolumn{4}{c}{$T=21$ days } &\multicolumn{4}{c}{$T=5$ days }\\
\hline
  & \multicolumn{2}{c}{$\Delta_n = 5 $ minutes } &\multicolumn{2}{c}{$\Delta_n = 1$ minute }& \multicolumn{2}{c}{$\Delta_n = 5 $ minutes } &\multicolumn{2}{c}{$\Delta_n = 1$ minute }\\
  & BIAS(\%) & RMSE(\%) & BIAS(\%) & RMSE(\%) & BIAS(\%) & RMSE(\%) & BIAS(\%) & RMSE(\%)\\
\hline 
\multicolumn{9}{c}{exponential kernel}\\
$V ^n_T$ & 5.884 & 9.890 & 2.345 & 3.941 & 13.291 & 17.627 & 6.718 &8.537 \\
$\widetilde{V}^{n}_T$ & 1.557 & 8.214 & 0.042 & 3.336 & 1.531 & 16.096 & 0.949 &6.861 \\
\multicolumn{9}{c}{two-sided uniform kernel}\\
$V ^n_T$ & 4.824 & 9.497 & 1.821 & 3.729 & 11.332 & 16.557 & 5.783 & 7.903\\
$\widetilde{V}^{n}_T$ & 1.804 & 8.357 & 0.084 & 3.352 & 2.288 & 15.843 & 1.183 & 6.874\\
\multicolumn{9}{c}{right-sided uniform kernel}\\
$V ^{n}_T$ & 1.727 & 8.237 & 0.208 & 3.384 & 8.380 & 16.339 & 4.638 & 7.086\\
$\widetilde{V}^{n}_T$ & 1.083 & 8.334 & -0.267 & 3.361 & 1.393 & 16.509 & 0.665 & 6.965\\
\multicolumn{9}{c}{Jacod \& Rosenbaum's Estimators}\\
$V ^{n,JR}_T$ & 5.096 & 11.127 & 1.656 & 4.115 & 19.714 & 23.313 & 8.574 & 10.459\\
$\widetilde{V}^{n,JR}_T$ & 1.795 & 8.469 & 0.061 & 3.370 & 2.002 & 16.562 & 1.139 &7.067 \\
\multicolumn{9}{c}{Jackknife Estimators}\\
$TS_n$ & 4.365 & 8.223 & 2.195 & 3.435 & 3.802 & 14.834 & 2.991 &6.779 \\
$TS_n'$ & 1.323 & 8.162 & 0.213 & 3.308 & -8.019 & 20.895 & -3.855 &8.714 \\
$MS_n$ & 1.989 & 8.220 & 0.098 & 3.347 & 1.618 & 16.641 & 1.041 &6.930 \\
\hline 
\end{tabular}
\caption{$g(c) = c^2$.}
\label{table:Jackknife}
\end{center}
\end{table}
\begin{table}[h!]
\setlength\tabcolsep{3pt}
\begin{center} 
\begin{tabular}{  c c c c c c c c c }
\hline
& \multicolumn{4}{c}{$T=21$ days } &\multicolumn{4}{c}{$T=5$ days }\\
\hline
  & \multicolumn{2}{c}{$\Delta_n = 5 $ minutes } &\multicolumn{2}{c}{$\Delta_n = 1$ minute }& \multicolumn{2}{c}{$\Delta_n = 5 $ minutes } &\multicolumn{2}{c}{$\Delta_n = 1$ minute }\\
  & BIAS(\%) & RMSE(\%) & BIAS(\%) & RMSE(\%) & BIAS(\%) & RMSE(\%) & BIAS(\%) & RMSE(\%)\\
\hline 
\multicolumn{9}{c}{exponential kernel}\\
$V ^n_T$ & 3.922 & 4.716 & 1.819 & 2.192 & 11.166 & 11.597 & 5.035 & 5.335\\
$\widetilde{V}^{n}_T$ & 0.486 & 2.218 & 0.120 & 0.908 & 1.007 & 4.427 & 0.169 & 2.111\\
\multicolumn{9}{c}{two-sided uniform kernel}\\
$V ^n_T$ & 3.313 & 4.016 & 1.463 & 1.818 & 9.620 & 10.219 & 4.309 & 4.673\\
$\widetilde{V}^{n}_T$ & 0.300 & 2.790 & -0.025 & 0.993 & 1.127 & 4.515 & 0.190 & 2.122\\
\multicolumn{9}{c}{right-sided uniform kernel}\\
$V ^{n}_T$ & 1.224 & 2.373 & 0.485 & 1.045 & 7.605 & 8.281 & 3.379 & 3.800\\
$\widetilde{V}^{n}_T$ & -0.753 & 3.440 & -0.432 & 1.277 & 0.091 & 4.463 & -0.216 & 2.122\\
\multicolumn{9}{c}{Jacod \& Rosenbaum's Estimators}\\
$V ^{n,{JR}}_T$ & 4.239 & 4.565 & 1.790 & 2.018 & 17.397 & 17.148 & 7.763 &7.785 \\
$\widetilde{V}^{n,JR}_T$ & -0.059 & 3.228 & -0.126 & 1.118 & 0.820 & 4.478 & 0.135 & 2.124\\
\multicolumn{9}{c}{Jackknife Estimators}\\
$TS_n$ & 2.269 & 3.962 & 1.296 & 2.152 & 1.913 & 5.099 & 0.957 &2.472 \\
$TS_n'$ & -0.425 & 2.931 & -0.142 & 1.472 & -9.518 & 10.686 & -5.169 &5.547 \\
$MS_n$ & 0.681 & 2.487 & 0.282 & 1.064 & 0.988 & 4.540 & 0.202 & 2.143\\
\hline 
\end{tabular}
\caption{$g(c) = log(c)$.}
\label{table:Jackknife_log}
\end{center}
\end{table}
\subsection{Comparison with Jackknife estimator}\label{Sec33}
We now compare our bias-corrected estimator $\widetilde{V}(g)_T^{n}$ to the Jackknife estimator of \cite{li2019efficient}. For the Jackknife method, we consider the two-scale estimator $TS_n$, its boundary-adjusted version $TS_n'${,} and the multiscale estimator $MS_n$ as defined in \cite{li2019efficient}, Section 3.1. The parameters for the Jackknife estimators are chosen according to \cite{li2019efficient}. {\DR By running numerous simulations, we determine that the values chosen by \cite{li2019efficient} broadly optimize their estimator's performance}. For the multiscale estimator $MS_n$, $\left(\psi_1, \psi_2, \psi_3\right)=(-2.5,8,-4.5)$, and $\left(k_{1, n}, k_{2, n}, k_{3, n}\right)=(15,30,45)$ for data sampled every 5 minutes, or $(40,80,120)$ for data sampled every minute. For the two-scale estimator, $\left(\psi_1, \psi_2\right)=(-1,2)$ along with the same $\left(k_{1, n}, k_{2, n}\right)$ as above. For our kernel estimator, we use 3 kernels: a two-sided exponential kernel $K(x)=\frac{1}{2}e^{-|x|}$, a two-sided uniform kernel $K(x)=\frac{1}{2}{\bf 1}_{[-1,1]}(x)$, and a right-sided uniform kernel $K(x)={\bf 1}_{[0,1]}(x)$. We analyze the performance of the simple biased estimator $V(g)^n_T$ in (\ref{eq:simple_estimator}) and the bias-corrected estimator $\widetilde{V}(g)^{n}_T$ in (\ref{eq:unbiased_V}). We also consider the simple biased estimator $V(g)^{n,{JR}}_T$ and the bias-corrected estimator $\tilde{V}(g)^{n,{JR}}_T$ of \eqref{eq:right_V_estimator} and \eqref{CRJEN}, respectively, which were proposed in \cite{jacod2015estimation}.

The bandwidths of all the kernel estimators are chosen using the iterative method in \cite{FigLi}, Section 5. We iterate the method only once. For a two-sided kernel, the integrated volatility of volatility used when updating the bandwidth is estimated by the Two-time Scale Realized Volatility (TSRV) estimator proposed in \cite{FigLi}. For the right-sided uniform kernel,  the TSRV estimator no longer applies. We then use the realized variance estimator on the estimated volatility path as a replacement. 
As finite-sample corrections, in (\ref{eq:unbiased_V}), at {each} $t_i$, we replace $\int_{-\infty}^{\infty}K^2(u) du $ with $
\sum_{j=1}^n K^2_h(t_i - t_j) \Delta_n h$, and replace $\int_{-\infty}^{\infty}L(s) ds  $ and $\int_{-\infty}^{\infty}L^2(s)ds $ with $
 \frac{\Delta_n}{b}\sum_{j=1}^n L\left(\frac{t_{j-1}-t_i}{b} \right)$ and $\frac{\Delta_n}{b}\sum_{j=1}^n L^2\left(\frac{t_j-t_i}{b} \right)$, respectively, where 
 $ L\left(\frac{t_j - t_i}{b}  \right) = 
 -\sum_{l=1}^{j}K_{{b}}(t_{l-1}-t_i) $ if $j \leq i$ and $\sum_{l=j+1}^{n}K_{{b}}(t_{l-1}-t_i)$ if $ j > i$.

In Table \ref{table:Jackknife}, we report the relative bias and relative root-mean-square-errors in percentage unit for the considered estimators. Our results of the Jackknife estimators are consistent with the ones in \cite{li2019efficient}, Table 1 therein. As shown in Table \ref{table:Jackknife}, the kernel estimators {\DR exhibit a similar performance to} the Jackknife estimators. However, the Jackknife estimators have more tuning parameters, while ours requires only one tuning parameter (namely, $\theta$). 
As we will show in Section \ref{sensitivity}, the {\Red most important} advantage of our estimators {\Red lies} on their stability relative to the choice of the bandwidth. That is, our estimators seem to exhibit satisfactory performance for a large range of the tuning parameters, while other estimators are more sensitive to choosing `good' values for their tuning parameters. 

{\DR Returning to} the results of Table \ref{table:Jackknife},} we also observe that in some cases the kernel estimators have smaller relative root mean square error (RMSE) compared with the Jackknife estimators, especially when the  time span $T$ is short ($T = 5$ days). Intuitively, the superior performance of the spot volatility estimators with exponential  kernel (shown in \cite{FigLi}, Section 7) {\DR is expected to be more evident} when the time span is relatively short, e.g. 5 days.  For reference, we include the tables for the function $g(c) = log(c)$ in Table \ref{table:Jackknife_log}. These tables confirm our previous conclusion.
It is worth noting that in general, the bias-corrected kernel estimator $\widetilde{V}^{n}_T$ has {\DR significantly} smaller RMSE compared with the simple biased estimator $V(g)^{n}_T$, which also happens among the Jacod \& Rosenbaum’s estimators, $V(g)^{n,{JR}}_T$ and $\widetilde{V}^{n,JR}_T$. 

\subsubsection{{\DR  Bandwidth Sensitivity}\label{sensitivity}}
In this subsection, we investigate how sensitive the estimators are relative to the bandwidth parameter. For the Jackknife estimators, we focus on the boundary adjusted multi-scale estimator $\text{MS}$ (which, as shown in Tables \ref{table:Jackknife} and \ref{table:Jackknife_log}, typically has the best performance among all {\Red the considered} Jackknife estimators), and for the kernel estimators, we {\DR first consider an} exponential kernel. We {\DR try} a sequence of bandwidths $k_n$ for the kernel estimator, and a sequence of triplets {$(k_{1,n}, k_{2,n}, k_{3,n}) = ([\frac{k_n}{2}], k_n, [\frac{3k_n}{2}])$} and $(\psi_1, \psi_2, \psi_3) = (-2.5,8,-4.5)$ for $MS$, which include the parameters chosen in \cite{li2019efficient} {\DR when $k_n=30$ (see Table 1 therein)}. The results are shown in the Figures \ref{fig:bdw_square_21}-\ref{fig:bdw_log_5}, where $\text{S0 exp}$ denotes the simple biased exponential kernel estimator in (\ref{eq:simple_estimator})  and $\text{S3 exp}$ denotes the bias-corrected estimators with exponential kernel in (\ref{eq:unbiased_V}). {\DR We can conclude that the estimator $\text{S3 exp}$ exhibits the most stable or consistent performance for} either 5 or 21 days data with 1 or 5 minutes sampling frequency, and both $g(c)=c^2$ or $\log (c)$. {\DR In other words, the estimator's values are {\DG in general} relatively invariant to changes in the bandwidth, {\DR especially} when working with low frequency and $g(c)=c^2$. Under those conditions, the others estimators require more careful tuning of the bandwidth in order to show good performance. 
We can also observe that the bias-corrected exponential kernel estimator almost always outperforms the boundary-adjusted multi-scale Jackknife estimator $\text{MS}$ and the simple biased exponential kernel estimator $\text{S0 exp}$.}

{\DR Finally,} in Figure \ref{fig:mixed}, we analyze the bandwidth stability of three estimators: the bias-corrected estimators \eqref{eq:unbiased_V} with exponential and right-sided uniform kernels and the bias-corrected estimator \eqref{CRJEN} proposed in \cite{jacod2015estimation}. Again, since we are not introducing jumps in the process $X$, the estimator $\hat{c}_i^{n}$ used for \eqref{eq:unbiased_V} and the estimator $\hat{c}_i^{n,{JR}}$ used for \eqref{CRJEN} do not have any truncation (i.e., we fix $v_n=\infty$ in \eqref{estivoltrun} and \eqref{vola_JR}). Figure \ref{fig:mixed} shows that our proposed estimator $\widetilde{V}^{n}_T$ with either exponential or right-sided uniform kernel almost always outperforms $\widetilde{V}^{n, JR}_T$. The exponential kernel appears to be less sensitive to the choice of $k_n$ than the right-sided uniform kernel. The estimator $\widetilde{V}^{n, JR}_T$ also seems to require a careful tuning of the bandwidth in order to achieve a good performance.

\begin{figure}
\centering
\begin{subfigure}{.5\textwidth}
  \centering
  \includegraphics[width=1\linewidth]{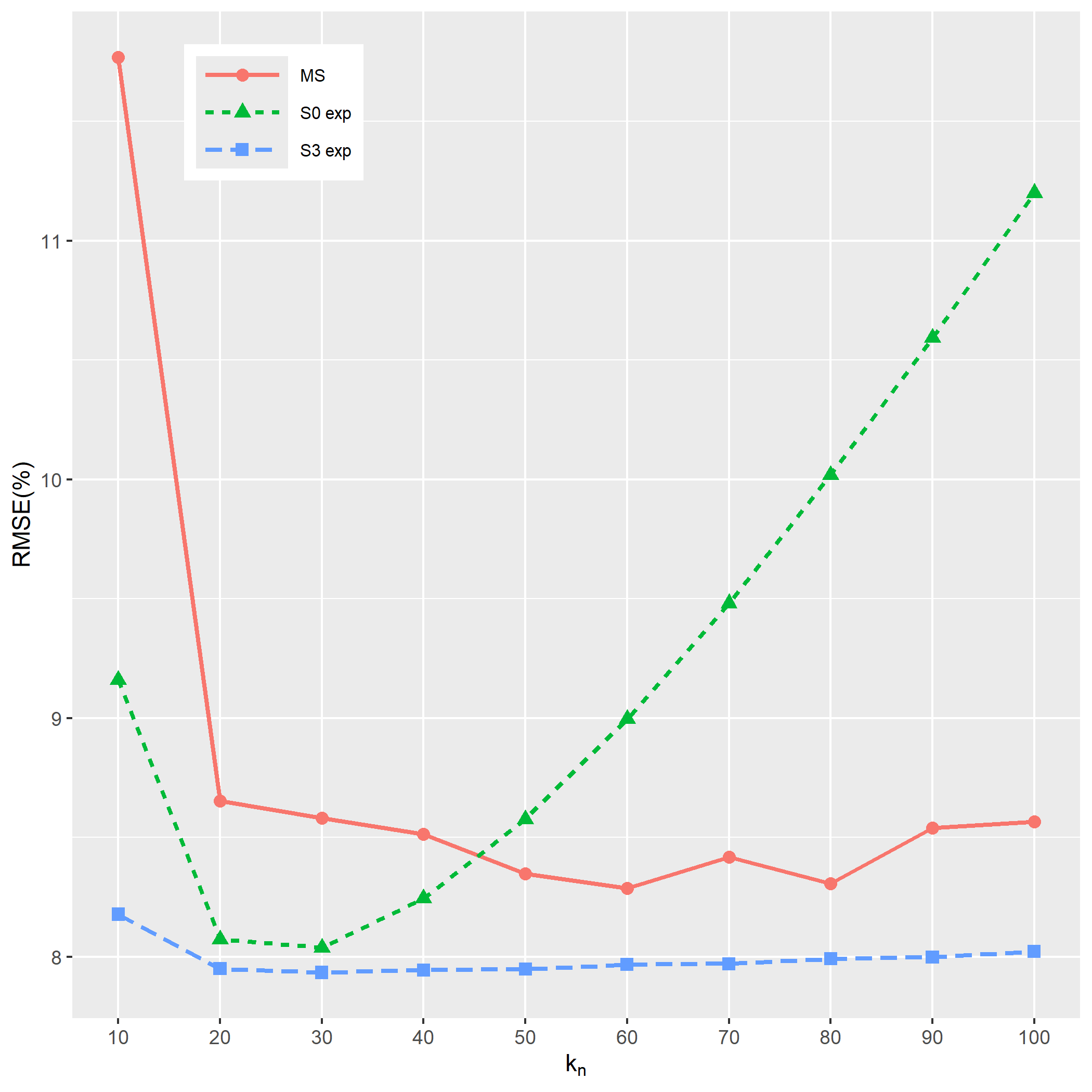}
  \caption{data with 5 minute sample frequency. }
  \label{fig:sub1}
\end{subfigure}%
\begin{subfigure}{.5\textwidth}
  \centering
  \includegraphics[width=1\linewidth]{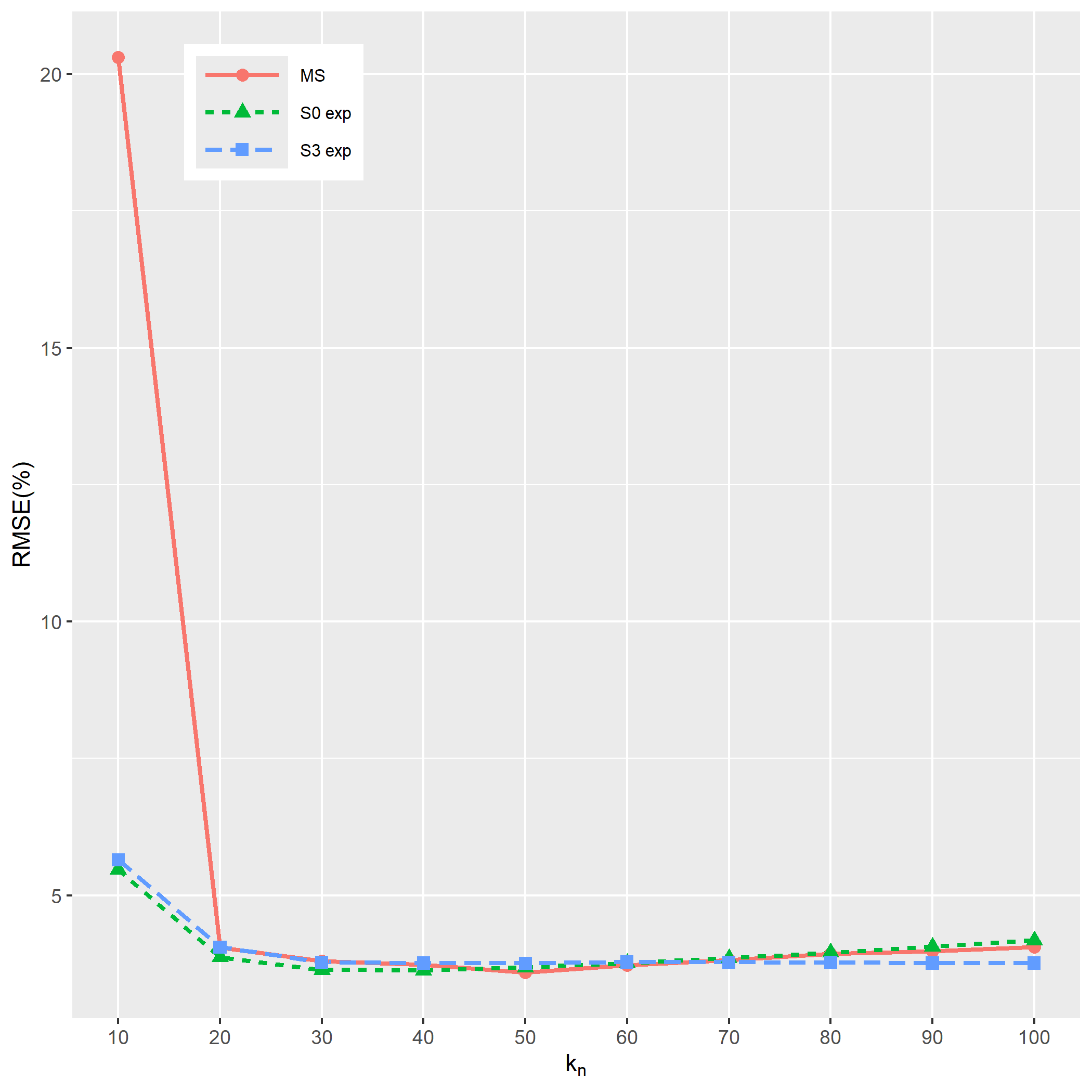}
  \caption{data with 1 minute sample frequency.}
  \label{fig:sub2}
\end{subfigure}
\caption{ RMSE(\%) v.s. the bandwidth $k_n$, for 21 days data, $g(c) = c^2$ and 1000 simulated path. $S0 $ exp denotes the simple biased estimator with exponential kernel and $S3$ exp denotes the unbiased estimator with exponential kernel. }
\label{fig:bdw_square_21}
\end{figure}

\begin{figure}
\centering
\begin{subfigure}{.5\textwidth}
  \centering
  \includegraphics[width=1\linewidth]{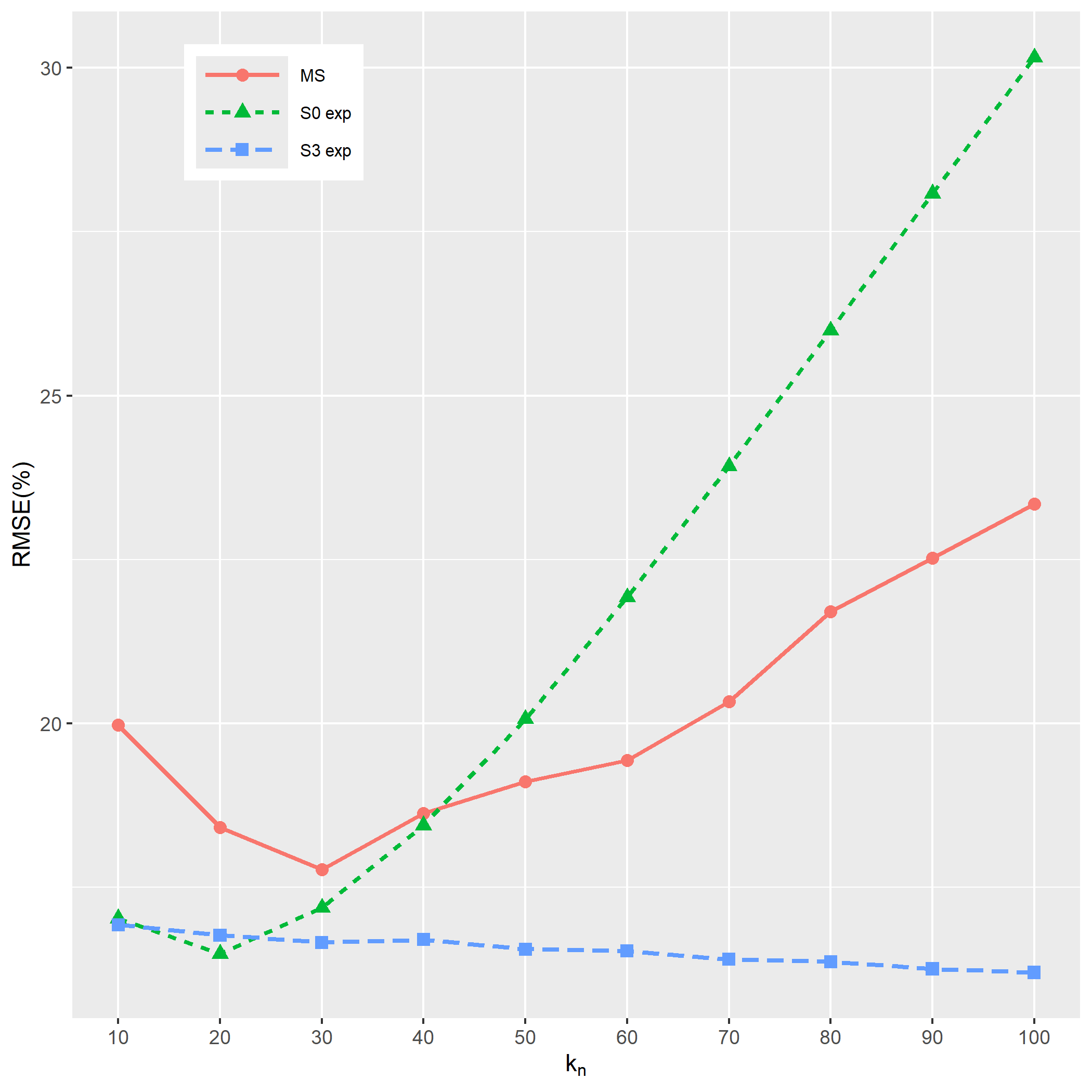}
  \caption{data with 5 minute sample frequency. }
  \label{fig:sub1}
\end{subfigure}%
\begin{subfigure}{.5\textwidth}
  \centering
  \includegraphics[width=1\linewidth]{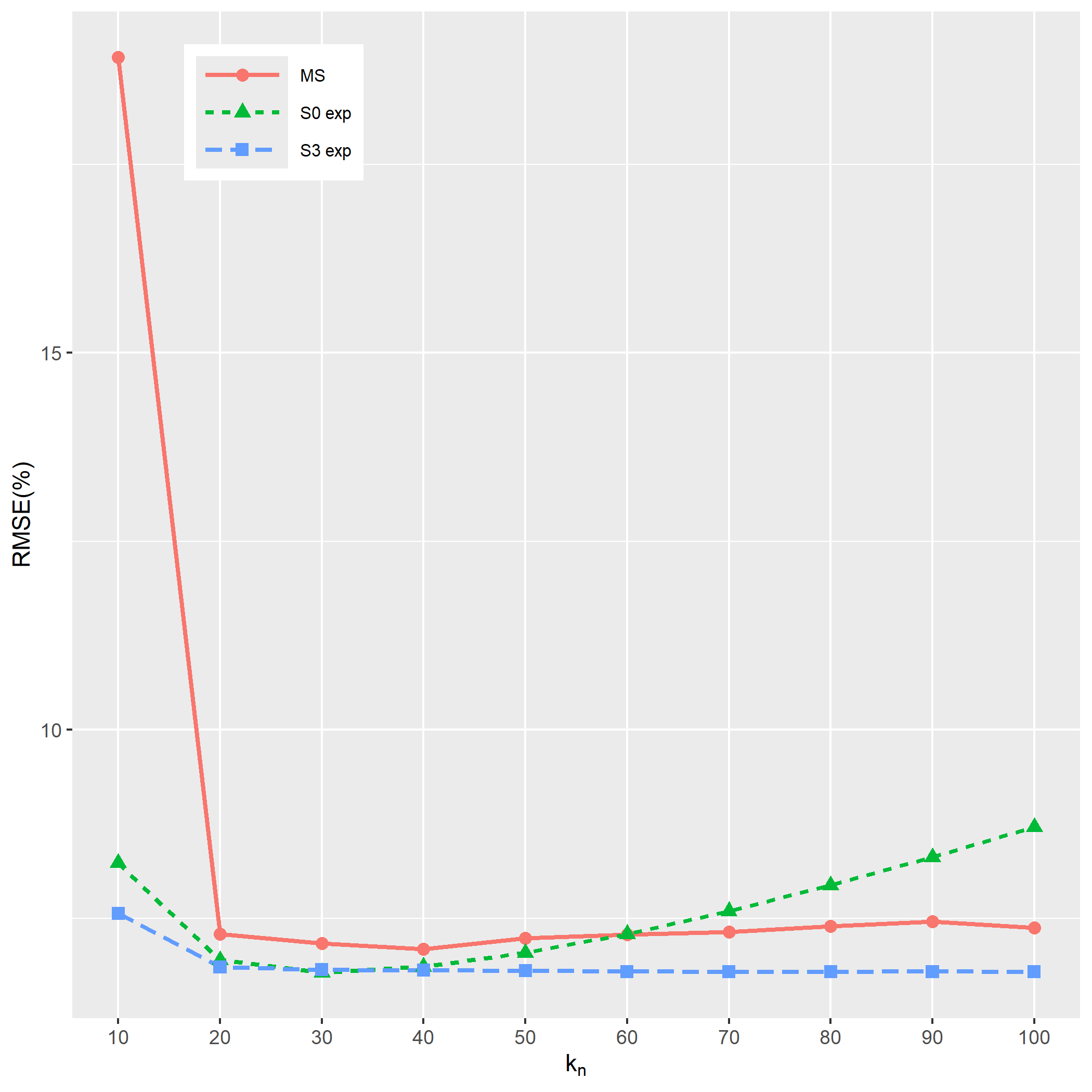}
  \caption{data with 1 minute sample frequency.}
  \label{fig:sub2}
\end{subfigure}
\caption{ RMSE(\%) v.s. the bandwidth $k_n$, for 5 days data, $g(c) = c^2$ and 1000 simulated path. $S0 $ exp denotes the simple biased estimator with exponential kernel and $S3$ exp denotes the unbiased estimator with exponential kernel. }
\label{fig:bdw_square_5}
\end{figure}

\begin{figure}
\centering
\begin{subfigure}{.5\textwidth}
  \centering
  \includegraphics[width=1\linewidth]{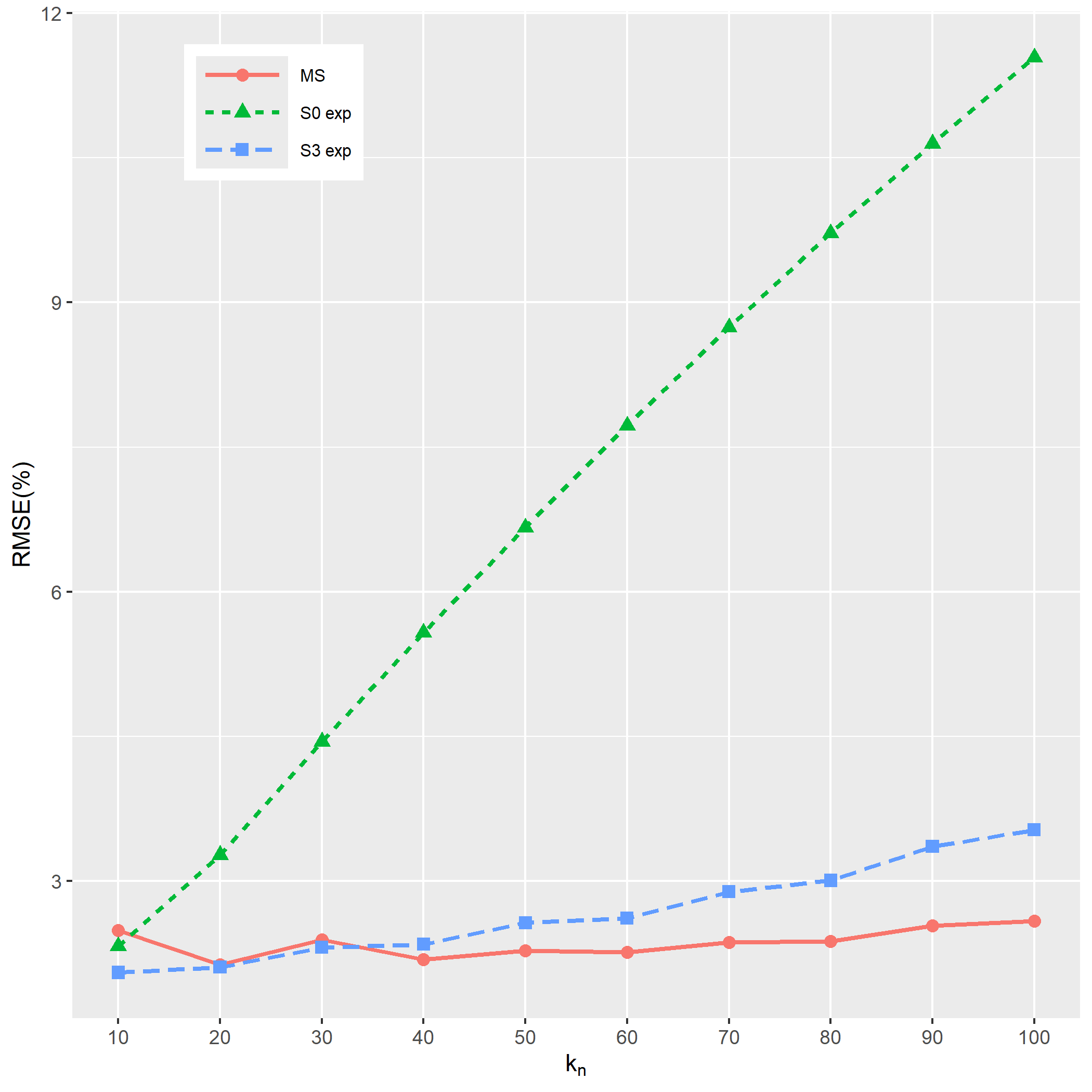}
  \caption{data with 5 minute sample frequency. }
  \label{fig:sub1}
\end{subfigure}%
\begin{subfigure}{.5\textwidth}
  \centering
  \includegraphics[width=1\linewidth]{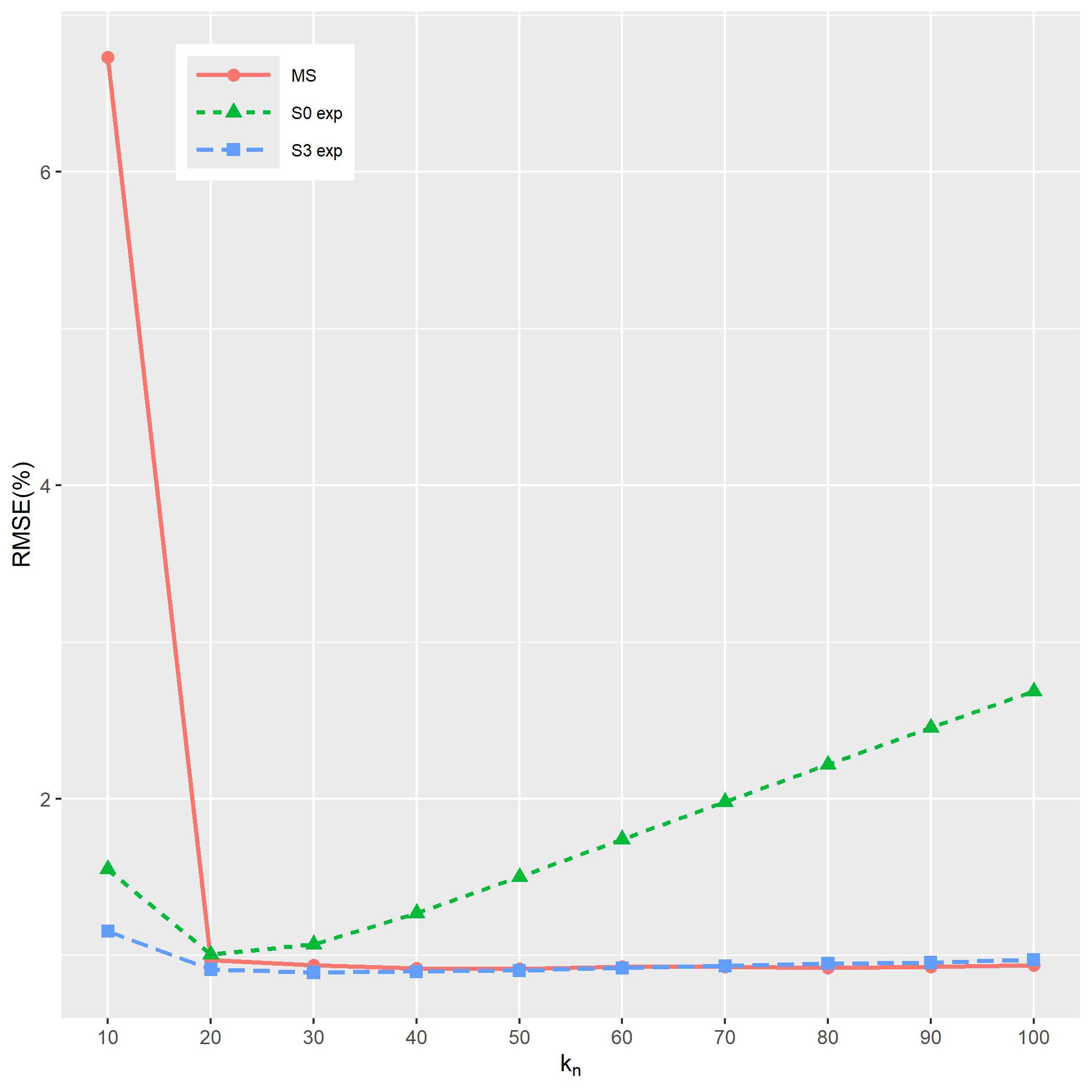}
  \caption{data with 1 minute sample frequency.}
  \label{fig:sub2}
\end{subfigure}
\caption{ RMSE(\%) v.s. the bandwidth $k_n$, for 21 days data, $g(c) = log(c)$ and 1000 simulated path. $S0$   exp denotes the simple biased estimator with exponential kernel and $S3 $ exp denotes the unbiased estimator with exponential kernel. }
\label{fig:bdw_log_21}
\end{figure}
\begin{figure}
\centering
\begin{subfigure}{.5\textwidth}
  \centering
  \includegraphics[width=1\linewidth]{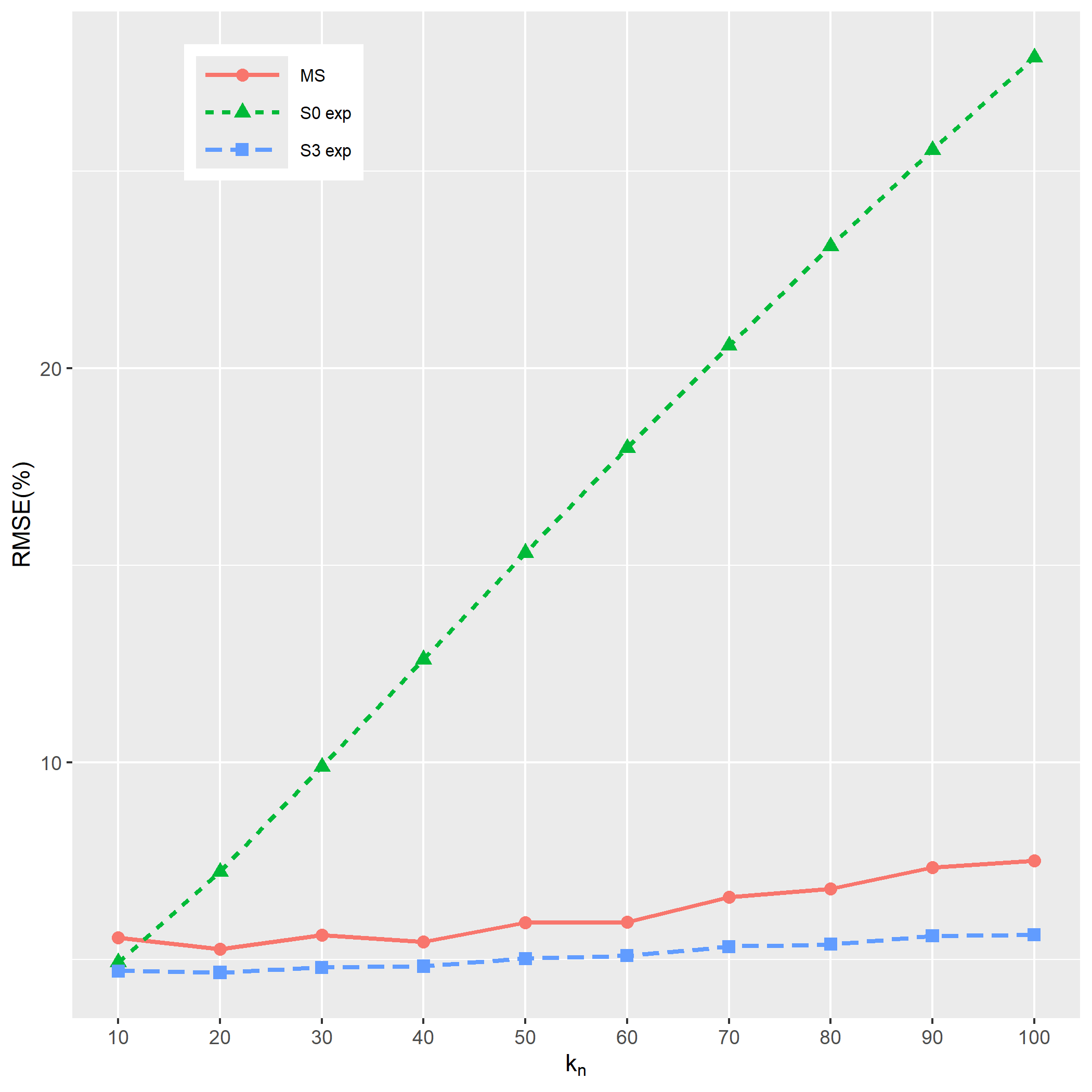}
  \caption{data with 5 minute sample frequency. }
  \label{fig:sub1}
\end{subfigure}%
\begin{subfigure}{.5\textwidth}
  \centering
  \includegraphics[width=1\linewidth]{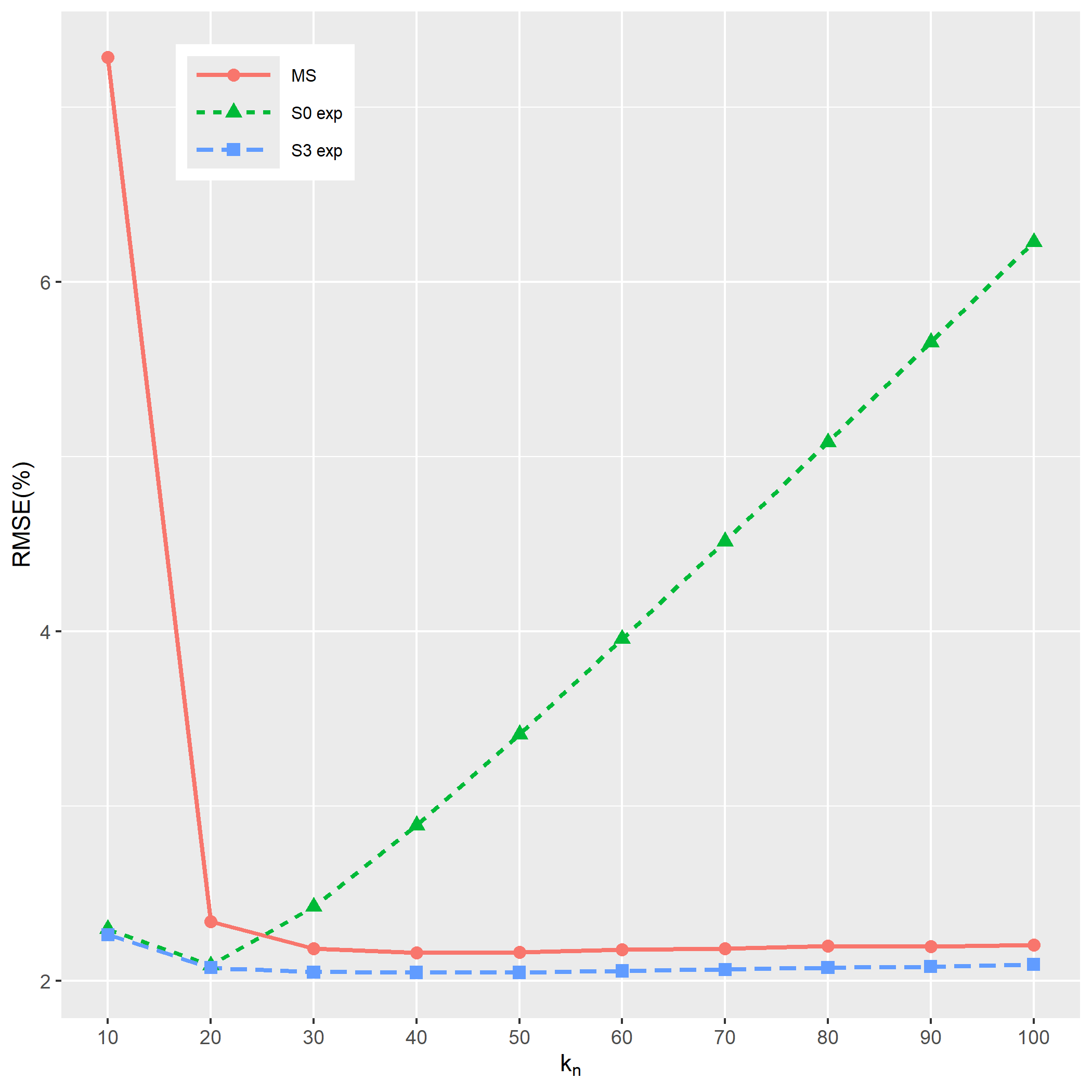}
  \caption{data with 1 minute sample frequency.}
  \label{fig:sub2}
\end{subfigure}
\caption{ RMSE(\%) v.s. the bandwidth $k_n$, for 5 days data, $g(c) = log(c)$ and 1000 simulated path. $S0$   exp denotes the simple biased estimator with exponential kernel and $S3 $ exp denotes the unbiased estimator with exponential kernel. }
\label{fig:bdw_log_5}
\end{figure}

\begin{figure}
\centering
\begin{subfigure}{.5\textwidth}
  \centering
  \includegraphics[width=1\linewidth]{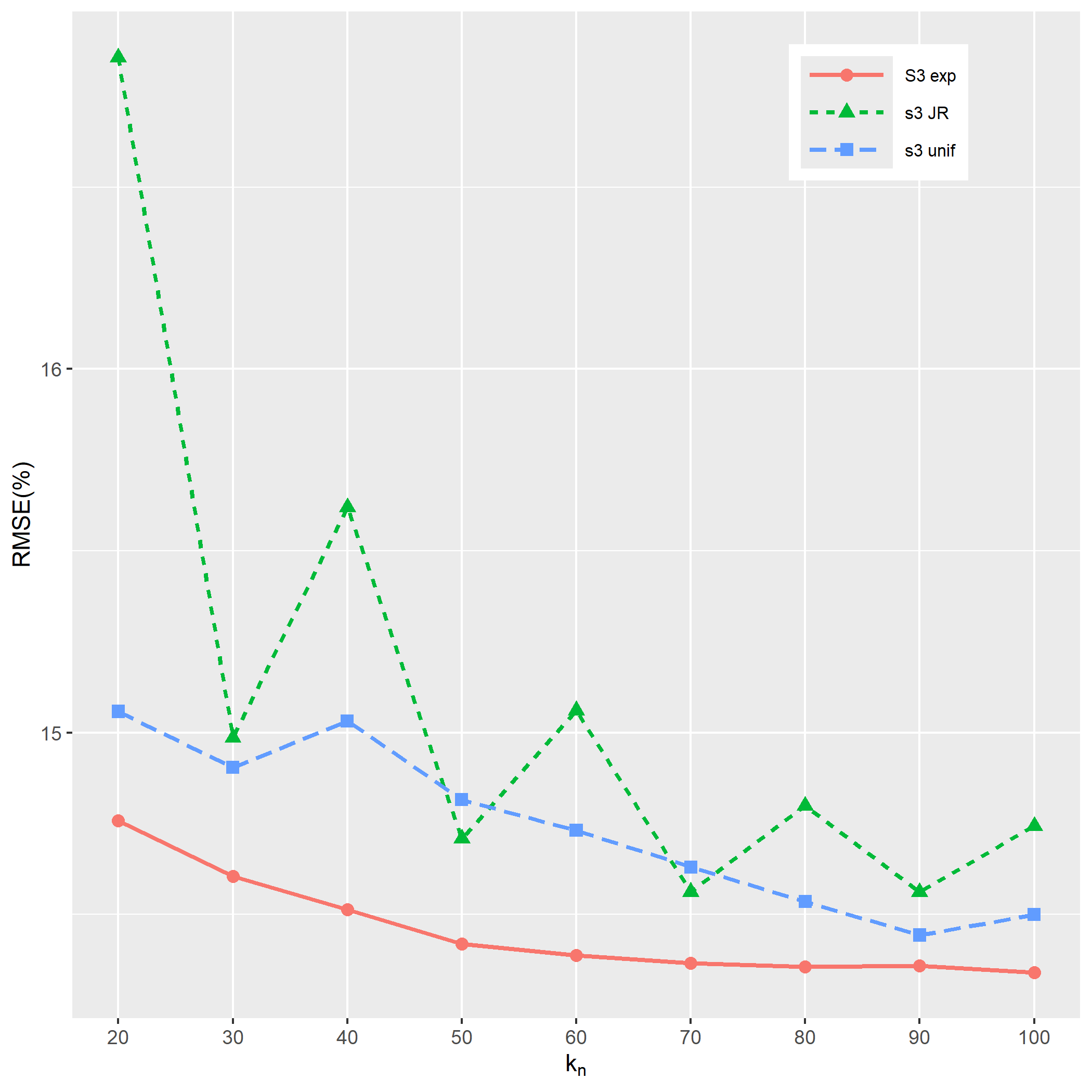}
  \caption{data with 5 minute sample frequency. }
  \label{fig:sub1}
\end{subfigure}%
\begin{subfigure}{.5\textwidth}
  \centering
  \includegraphics[width=1\linewidth]{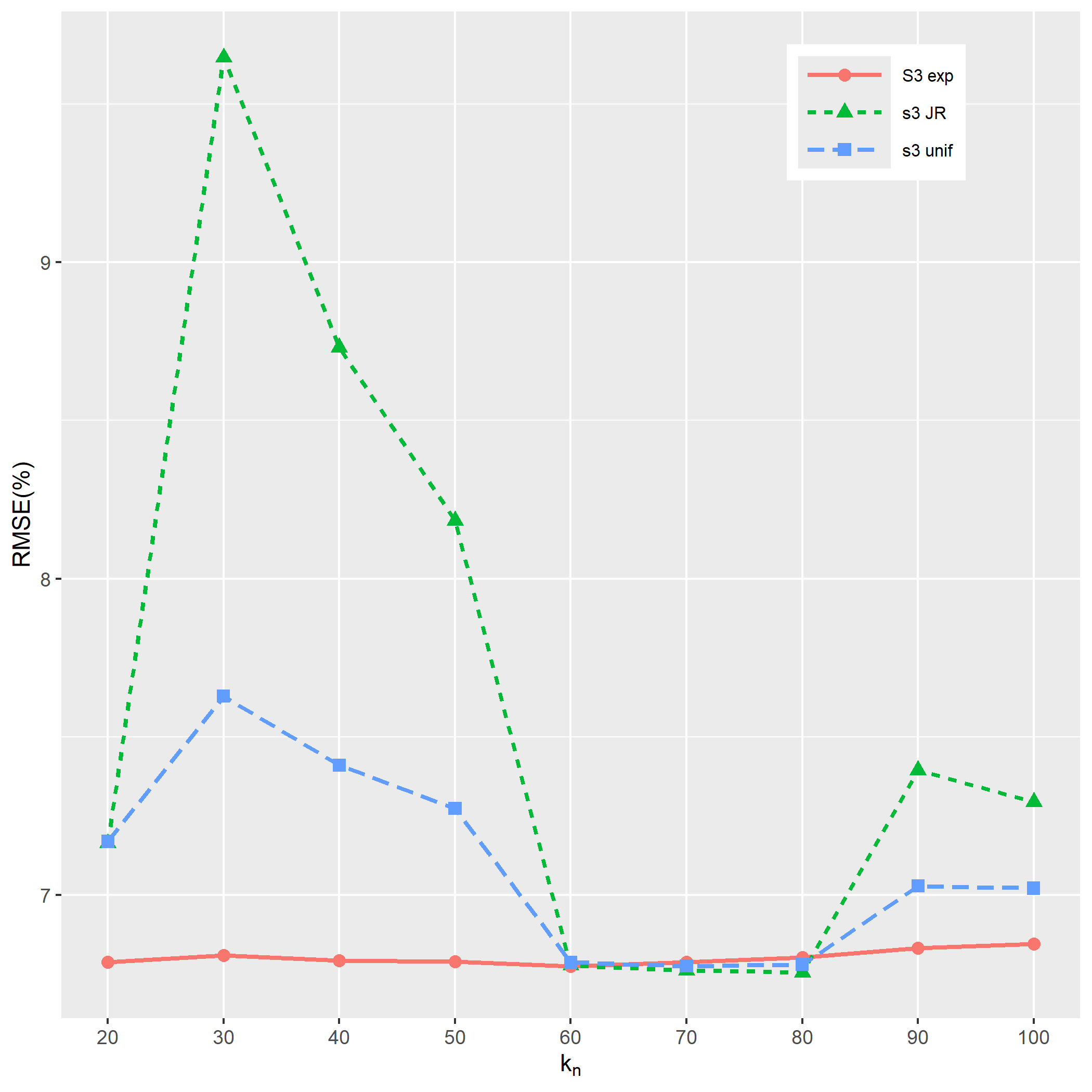}
  \caption{data with 1 minute sample frequency.}
  \label{fig:sub2}
\end{subfigure}
\caption{ RMSE(\%) v.s. the bandwidth $k_n$, for 5 days data, $g(c) = c^2$ and 1000 simulated path. $S3$ exp and $S3$ unif denotes the unbiased estimator with exponential kernel and right-sided uniform kernel, respectively. $S3$ JR corresponds to the bias-corrected estimator (\ref{CRJEN}) in \cite{jacod2015estimation}.}
\label{fig:mixed}
\end{figure}

\newpage
\appendix \label{appendix}

\section{Proof of Main Results }
By virtue of localization, without loss of generality, we assume throughout the {\Blue proofs} that, for some constant $A$,
\begin{equation} \label{eq:localization}
\Gamma(z)+\int \Gamma(z)^r \lambda(d z)+\left|\mu_{t}\right|+\left|\sigma_{t}\right|+\left|X_{t}\right| +\left|\mu_t^{\prime}\right| + \left|\tilde{\sigma}_{t}\right| + \left|\tilde{\mu}_{t}\right|\leq A, \text{ for any } t\le T,
 \end{equation} 
 (see Section 4.4.1 and (6.2.1) in \cite{JacodProtter} and Appendix A.5 in \cite{jacodaitsahalia} for details). 
 
 We use \(C\) to represent a generic constant that may change from line to line.
We first need some notations:
\begin{equation}\label{Dfnalphabeta0}
\begin{split}
\alpha_{i}^{n}&:=\left(\Delta_{i}^{n} X^{\prime}\right)\left(\Delta_{i}^{n} X^{\prime}\right)^* -c_{i-1}^{n} \Delta_{n}\in\mathbb{R}^{d\times d},\\
\beta_{i}^{n}:=\hat{c}_{i}^{\prime n}-c_{i}^{n}&=\frac{\sum_{j=1}^{n} K_{{k}_{n}\Delta_n}\left(t_{j-1}-t_i\right)\left(\Delta_{j}^{n} X^{\prime}\right)\left(\Delta_{j}^{n} X^{\prime}\right)^*}{\Delta_{n}\sum_{j=1}^{n} K_{{k}_{n}\Delta_n}\left(t_{j-1}-t_i\right)}-c_{i}^n\\
&=\frac{1}{\bar{K}\left(t_i\right)}\sum_{j=1}^n K_{{b}}\left(t_{j-1}-t_i\right)\left[\alpha_j^n+\left(c_{j-1}^n-c_i^n\right) \Delta_n\right] ,
\end{split}
\end{equation} 
and $\tilde{\beta}_i^n:=\beta_i^n\bar{K}\left(t_i\right)$. The nice decomposition of the estimation error $\beta_{i}^{n}=\hat{c}_{i}^{\prime n}-c_{i}^{n}$ in terms of the $\alpha_j^n$ and $c_{j-1}^n-c_i$ was possible because of the denominator $\bar{K}(t_i)$ in the estimator $\hat{c}_{i}^{\prime n}$ and this is one of the reasons why this was added. For an uniform kernel, $\bar{K}(t_i)$ equal to $1$ for almost all $i$ and is not needed if one considers the truncated version \eqref{eq:V_jacod} as in \cite{jacod2013quarticity}. However, it seems more natural to include the ``standardized'' term $\bar{K}(t_i)$.

For any process $Z$, we set:
\begin{equation} \label{eq:eta}
\begin{split}
\eta(Z)_{t, s} & =\sqrt{\mathbb{E}\left(\sup _{v \in(t, t+s]}\left\|Z_{{v}}-Z_{t}\right\|^{2} \Big| \mathcal{F}_{t}\right)},\\
\eta(Z)^{n}_{i,j} & =\eta\left(Z\right)_{i \Delta_{n},j\Delta_n},\\
\eta(Z)^n_{i-1 }& = \eta(Z)^n_{(i-1),1 }.
\end{split}
\end{equation} 
By Lemma 16.5.15 in \cite{JacodProtter}, if $U$ is a c\`adl\`ag bounded process, we have for all $q > 0$, \begin{align} \label{eq:AN}
\Delta_{n} \mathbb{E}\left(\sum_{i=1}^{\left[t / \Delta_{n}\right]} \left|\eta(U)^n_{i-1 }\right|^q\right) \rightarrow 0.
\end{align}
{\DG For future reference, we hereafter denote the continuous version of the process \eqref{eq:X} as
\begin{equation}\label{eq:X_conti}
    X^{\prime}_t:=X_0+\int_0^t \mu^{\prime}_s d s+\int_0^t \sigma_s d W_s,
\end{equation}
where $\mu_t^{\prime}=\mu_t-\int_{\{z:\|\delta(t, z)\| \leq 1\}} \delta(t, z) \lambda(d z)$. We also let $\hat{c}_{i}^{\prime n}:=\hat{c}_{i}^{n,X^{\prime}}$ be the Nadaraya-Watson kernel estimator \eqref{MDKNN} based on the continuous process $X'$ defined in (\ref{eq:X_conti}). Finally, we define the estimators of the integrated volatility functionals based on the continuous process as:}
\begin{equation}\label{CntVerThE}
{\DG V(g)_t^{\prime n}:=\Delta_n \sum_{i=1}^{\left[t / \Delta_{n}\right]} g\left(\hat{c}_{i}^{\prime n}\right) \bar{K}\left(t_i\right)}.
\end{equation}

\subsection{Preliminary Result}\label{SectionA.1}
{\DG Let} us recall the following estimates:
\begin{enumerate}
\item By (3.10) in \cite{jacod2015estimation}, under  (\ref{eq:localization}),  for any $s, t \geq 0$ and $q \geq 0$ :  \begin{equation} \label{eq:X_c_bound}
\begin{array}{ll}
{\rm (i)}\;\mathbb{E}\left(\sup _{w \in[0, s]}\left\|X^{\prime}_{t+w}-X^{\prime}_{t}\right\|^{q} \mid \mathcal{F}_{t}\right) \leq C_{q} s^{q / 2}, & {\rm (ii)}\;\left\|\mathbb{E}\left(X^{\prime}_{t+s}-X^{\prime}_{t} \mid \mathcal{F}_{t}\right)\right\| \leq C s, \\
{\rm (iii)}\;\mathbb{E}\left(\sup _{w \in[0, s]}\left\|c_{t+w}-c_{t}\right\|^{q} \mid \mathcal{F}_{t}\right) \leq C_{q} s^{1 \wedge(q / 2)}, & {\rm (iv)}\;\left\|\mathbb{E}\left(c_{t+s}-c_{t} \mid \mathcal{F}_{t}\right)\right\| \leq C s,
\end{array}
\end{equation}
where recall that above $\|x\|$ is the Euclidian norm of a $d\times d$-matrix $x$, considered as a $\mathbb{R}^{d\times d}$ vector.
\item From \cite{jacod2015estimation} (3.20) and (3.24), or BDG and H\"{o}lder inequalities. We have  
\begin{equation} \label{eq:bounds_alpha}
\begin{array}{l}
{\rm (i)}\;\mathbb{E}\left(\left\|\alpha_{i}^{n}\right\|^{q} \mid \mathcal{F}_{i-1}^{n }\right) \leq C_{q} \Delta_{n}^{q}, \quad
{\rm (ii)}\;
\left\|\mathbb{E}\left(\alpha_{i}^{n} \mid \mathcal{F}_{i-1}^{n}\right)\right\| \leq C \Delta_{n}^{3/2},\quad{\rm (iii)}\;
\mathbb{E}\left(\left\|\hat{c}_{i}^{\prime n}\right\|^{q} \mid \mathcal{F}_{i-1}^{n}\right) \leq C_{q}.
\end{array}
\end{equation} 
\item From (3.22) and (3.24) in \cite{jacod2015estimation} and notation (\ref{eq:eta}), we have  
\begin{equation} \label{eq:better_bound_alpha}
\begin{array}{l}
\left|\mathbb{E}\left(\alpha_{i}^{n,\ell m} \mid \mathcal{F}_{i-1}^{n}\right)\right| \leq C \Delta_{n}^{3 / 2}\left(\sqrt{\Delta_{n}}+\eta_{i-1}^{n}\right)\\
\left|\mathbb{E}\left(\alpha_{i}^{n,jk}\alpha_{i}^{n,lm}\mid \mathcal{F}_{i-1}^{n}\right)-\left(c_{i-1}^{n,jl}c_{i-1}^{n,km} +c_{i-1}^{n,jm}c_{i-1}^{n,kl}\right)\Delta_{n}^{2}\right| \leq C\Delta_{n}^{5 / 2}\\
\left|\mathbb{E}\left(\alpha_{i}^{n,jk} \Delta_{i}^{n} c^{lm} \mid \mathcal{F}_{i-1}^{n}\right)\right| \leq C \Delta_{n}^{3 / 2}\left(\sqrt{\Delta_{n}}+\eta_{i-1}^{n}\right)\\
\left|\mathbb{E}\left(\alpha_{i}^{n,jk} \Delta_{i}^{n} W^{l} \mid \mathcal{F}_{i-1}^{n}\right)\right| \leq C \Delta_{n}^{3 / 2}\left(\sqrt{\Delta_{n}}+\eta_{i-1}^{n}\right), 
\end{array}
\end{equation}
where 
\begin{equation}\label{etadef}
\eta_{i-1}^{n}=\max \left(\eta\left(Y\right)^n_{i-1 }: Y=\mu, \widetilde{\mu}, c, \tilde{c}, \widehat{c}\right).
\end{equation}
\item By (\ref{eq:bounds_alpha}) and (\ref{eq:better_bound_alpha}), {we have for $q\geq 2,$} 
\begin{equation} \label{eq:beta_bound}
    \mathbb{E}\left[\left\|\beta^n_i \right\|^q \right] ={\mathbb{E}\left[\left\|\hat{c}_{i}^{\prime n}-c_{i}^{n}\right\|^q \right]}\leq C_q \left(k_n^{-q/2} + (k_n \Delta_n)^{q/2}\right).
\end{equation}
The proof is included in Appendix \ref{TechnicalA_1} {\DG for completeness}.
\item {We shall use often the following asymptotic approximation:}
\begin{equation} \label{eq:K_discrete}
   \Delta_n \sum_{i=1}^n K_{{b}}\left(t_{i-1} - \tau \right) = \int K(u) du + O(\Delta_n/{b})=1+O(\Delta_n/{b}),
\end{equation}
where the error is uniform on $\tau\in {\DG B}$, for any compact set ${\DG B}\subseteq (0,T)$. Furthermore, we can show that there exists $\delta>0$ such that for any $n$ and $t_{i}$,
\begin{equation}\label{lowerboundofdelta}
    \delta<\left|\bar{K}_{n}\left(t_i\right)\right|=\left|\Delta_{n}\sum_{j=1}^{n} K_{{k}_{n}\Delta_n}\left(t_{j-1}-t_i\right)\right|<\frac{1}{\delta}.
\end{equation}
The proofs of (\ref{eq:K_discrete}) and (\ref{lowerboundofdelta}) are given in Appendix \ref{TechnicalA_1}.

\end{enumerate}
The following result will be used several times. Its proof is similar to that of Lemma 3.1  in \cite{FigLiSumplement1} and is given in Appendix \ref{TechnicalA_1} for completeness.
 \begin{lemma} \label{K_discrete} 
 Let $f: \mathcal{M}_d^+ \rightarrow \mathbb{R}$ and $ K: \mathbb{R} \rightarrow \mathbb{R}$  be such that 
    \begin{enumerate}
\item $f$ is $C^1$ and $|f(x)| \leq C (1+ \|x\|^p) $ for some $p\ge 3$;
\item $K$ is absolutely continuous and $K^{\prime}$ exists piecewise on $(A,B)$ (in particular, recall that its support is $(A, B) $ for some $A< B$ such that $A \in [-\infty,0] $ and $B \in [0,\infty]$); $\int \left|K(x)\right|dx<\infty$; $\int \left|K^{\prime}(x)\right|dx<\infty$; $V_{-\infty}^{\infty}(K^{\prime})<\infty$.
    \end{enumerate}
Suppose also that {\DG $\Delta_n,b_n\to0$} such that {\DG $\Delta_n/{b}_n\to{}0$}. Then, the following assertions hold true:
\begin{enumerate}[label=(\roman*)]
\item Let 
\begin{equation*} 
D_{1}(f):=\sum_{i=1}^{n} K_{{b}}\left(t_{i-1}-\tau\right) \int_{t_{i-1}}^{t_{i}}  f\left(c_s\right) d s
-\int_0^T K_{{b}}\left(s-\tau\right) f\left(c_s\right) d s.
\end{equation*}
Then, for each $\tau \in(0, T)$, 
\begin{equation}\label{eq:D1f}
D_{1}(f)=\frac{1}{2} f(c_\tau) \left( K\left(A^{+}\right)-K\left(B^{-}\right) \right)\frac{\Delta_n}{b}+o\left(\frac{\Delta_n}{b}\right).
\end{equation}
Furthermore, for $\tau=0$ or $T$, we have $D_{1}(f)=O_p\left(\frac{\Delta_n}{b}\right)$.
\item Let 
\begin{equation*} 
D_{2}(f):=\sum_{i=1}^{n} K_{{b}}\left(t_{i-1}-\tau\right) f\left(c_i^n\right) \Delta_n 
-\int_0^T K_{{b}}\left(s-\tau \right) f\left(c_s\right) d s.
\end{equation*} 
Then, for each $\tau \in [0,T]$,
\begin{equation} \label{eq:D2}
D_2(f) = D_1(f) + O_p(\Delta_n^{1/2}).
\end{equation}
\end{enumerate}
\end{lemma}

\begin{remark}\label{SpecLemmaA1}
In Lemma \ref{K_discrete}-(i), $\tau$ is a constant time in $(0,T)$. However, {\DR in the definition of our estimator \eqref{estivoltrun} and} in some of the proofs below, we need to take $\tau=t_{j}$ for some $j=0,\dots,n$. In that case, following the same proof as that of Lemma \ref{K_discrete}-(i), one can show that for a constant $C\in (0,\infty)$,
\[
	\limsup_{n\to\infty}\frac{b_n}{\Delta_n}|D_1(f)|\leq C|f(c_{\tau})|.
\]
\end{remark}
\subsection{Proof of Theorem \ref{clt}}
\subsubsection{Elimination of the jumps and the truncation}

\begin{lemma}\label{jump_1}
    Under (\ref{eq:X}),  {\DG (\ref{eq:g_prime_bound}),} and Assumptions \ref{process} \& \ref{kernel}, {\DG the following assertions hold:}
    \begin{itemize}
        \item[a)] {\DR For any} $q\geq 1$, $q^{\prime}> 0$, and a sequence $a_n$ going to 0, {\DG we have:} 
\begin{equation}\label{boundofdiffc}
    \mathbb{E}\left[\left\|\hat{c}_i^n-\hat{c}_i^{\prime n}\right\|^q\right] \leq C a_n \Delta_n^{(2 q-r) \varpi+1-q}+C \Delta_n^{q^{\prime}(1-2 \varpi)},
\end{equation}
where $\hat{c}$ is defined as in (\ref{estivoltrun}) with $v_n=\alpha \Delta_n^{\varpi}$ {\DG and $0<\varpi<1/2$}, and $\hat{c}_{i}^{\prime n}:=\hat{c}_{i}^{n,X^{\prime}}$ is the Nadaraya-Watson kernel estimator \eqref{MDKNN} based on the continuous process $X'$ defined in (\ref{eq:X_conti}). Based on (\ref{boundofdiffc}), under the condition (\ref{moreassum}), we have the following:
\begin{equation}\label{boundofdiffv}
    \frac{1}{\sqrt{\Delta_n}}\left[V(g)_T^{n}-V(g)_T^{\prime n}\right]\xrightarrow{\mathbb{P}} 0,
\end{equation}
where recall that $V(g)_T^{\prime n}$ is defined as in \eqref{CntVerThE}.
\item[b)] {\DG Furthermore,} if $X$ is continuous (i.e., $\delta\equiv 0$ in (\ref{eq:X})), then for any $q^{\prime}>0$, (\ref{boundofdiffc}) becomes
\begin{equation}\label{boundofdiffc1}
    \mathbb{E}\left[\left\|\hat{c}_i^n-\hat{c}_i^{\prime n}\right\|^q\right] \leq C \Delta_n^{q^{\prime}(1-2 \varpi)},
\end{equation}
based on which, under (\ref{moreassum1}), (\ref{boundofdiffv}) holds true as well.
    \end{itemize} 
\end{lemma}
\begin{proof}
We consider the case of discontinuous $X$ first. By {\DG Minkowski's} and Jensen's inequality, for any $q\geq 1$,
\begin{align*}
    &\mathbb{E}\left[\left\|\hat{c}_i^n-\hat{c}_i^{\prime n}\right\|^q\right]\\
    &=\mathbb{E}\left[\left\|\bar{K}_{n}\left(t_i\right)^{-1}\sum_{j=1}^n K_{k_{n} \Delta_n}\left(t_{j-1}-t_{i}\right)\left[\left(\Delta_j X\right)\left(\Delta_j X\right)^* \mathbbm{1}_{\left\{\left\|\Delta_j X\right\| \leq v_n\right\}}-\left(\Delta_j X^{\prime}\right)\left(\Delta_j X^{\prime}\right)^*\right]\right\|^q\right]\\
    &\leq C_q\mathbb{E}\left[\left\|\sum_{j=1}^n K_{k_{n} \Delta_n}\left(t_{j-1}-t_{i}\right)\left(\Delta_j X^{\prime}\right)\left(\Delta_j X^{\prime}\right)^* \mathbbm{1}_{\left\{\left\|\Delta_j X^{\prime}\right\| > v_n\right\}} \right\|^q\right]\\
    &+C_q\mathbb{E}\left[\left\|\sum_{j=1}^n K_{k_{n} \Delta_n}\left(t_{j-1}-t_{i}\right)\left[\left(\Delta_j X\right)\left(\Delta_j X\right)^* \mathbbm{1}_{\left\{\left\|\Delta_j X\right\| \leq v_n\right\}}-\left(\Delta_j X^{\prime}\right)\left(\Delta_j X^{\prime}\right)^*\mathbbm{1}_{\left\{\left\|\Delta_j X^{\prime}\right\| \leq v_n\right\}}\right] \right\|^q\right]\\
    &\leq C_q \left[\sum_{j=1}^n\left|K_{k_n \Delta_n}\left(t_{j-1}-t_{i-1}\right)\right|\right]^{q-1} \sum_{j=1}^n\left|K_{k_n \Delta_n}\left(t_{j-1}-t_{i-1}\right)\right| \mathbb{E}\left[\left\|\left(\Delta_j X^{\prime}\right)\right\|^{2 q} \mathbbm{1}_{\left\{\left\|\Delta_j X^{\prime}\right\|>v_n\right\}}\right]\\
    &+C_q \left[\sum_{j=1}^n\left|K_{k_n \Delta_n}\left(t_{j-1}-t_{i-1}\right)\right|\right]^{q-1} \sum_{j=1}^n\left|K_{k_n \Delta_n}\left(t_{j-1}-t_{i-1}\right)\right|\\
    &\quad\quad\quad\quad\quad\quad\quad\quad\quad\quad\quad\quad\mathbb{E}\left[\left\|\left(\Delta_j X\right)\left(\Delta_j X\right)^* \mathbbm{1}_{\left\{\left\|\Delta_j X\right\| \leq v_n\right\}}-\left(\Delta_j X^{\prime}\right)\left(\Delta_j X^{\prime}\right)^* \mathbbm{1}_{\left\{\left\|\Delta_j X^{\prime}\right\| \leq v_n\right\}}\right\|^q\right].
\end{align*}
Furthermore, we have
$$
\mathbb{E}\left[\left\|\left(\Delta_j X^{\prime}\right)\right\|^{2 q} \mathbbm{1}_{\left\{\left\|\Delta_j X^{\prime}\right\|>v_n\right\}}\right]\leq \left\{\mathbb{E}\left[\left\|\left(\Delta_j X^{\prime}\right)\right\|^{4q}\right]\right\}^{1 / 2}\left\{\mathbb{E}\left[\mathbbm{1}_{\left\{\left\|\Delta_j X^{\prime}\right\|>v_n\right\}}\right]\right\}^{1 / 2},
$$
where{\DG , by Markov's inequality and (\ref{eq:X_c_bound})},
$$
\mathbb{E}\left[\mathbbm{1}_{\left\{\left\|\Delta_j X^{\prime}\right\|>v_n\right\}}\right]\leq \frac{\mathbb{E}\left[\left\|\left(\Delta_j X^{\prime}\right)\right\|^{4 q^{\prime}}\right]}{v_n^{4 q^{\prime}}}\leq \frac{K_{q} \Delta_n^{2 q^{\prime}}}{\Delta_{n}^{4q^{\prime}\varpi}},
$$
for any $q^{\prime}>0$ by Markov inequality. Thus,
$$
\mathbb{E}\left[\left\|\left(\Delta_j X^{\prime}\right)\right\|^{2q} \mathbbm{1}_{\left\{\left\|\Delta_j X^{\prime}\right\|>v_n\right\}}\right]\leq C\Delta_n^{q+q^{\prime}(1-2 \varpi)}.
$$
Meanwhile, by Lemma 13.2.6 of \cite{JacodProtter}, with $F(x)=xx^*, m=q, p^{\prime}=s^{\prime}=2,s=1,k=1,j=0$, we have for some sequence $a_n$ going to 0,
\begin{equation}\label{qthmoment}
\mathbb{E}\left(\left\|\left(\Delta_i^n X \right) \left(\Delta_i^n X^*\right) \mathbbm{1}_{\left\{\left\|\Delta_i^n X\right\| \leq v_n\right\}}-\left(\Delta_i^n X^{\prime} \right) \left(\Delta_i^n X^{\prime *}\right) \mathbbm{1}_{\left\{\left\|\Delta_i^n X^{\prime}\right\| \leq v_n\right\}}\right\|^q\right)\leq C a_n \Delta_n^{(2q-r) \varpi+1},
\end{equation}
where, {\DG from Assumption \ref{process}, $r< 1$}. Hence, the difference between $\hat{c}_i^n$ and $\hat{c}_i^{\prime n}$ can be estimated by
$$
\mathbb{E}\left[\left\|\hat{c}_i^n-\hat{c}_i^{\prime n}\right\|^q\right]\leq C a_n \Delta_n^{(2q-r) \varpi+1-q}+C\Delta_n^{q^{\prime}(1-2 \varpi)},
$$
{\DG and} (\ref{boundofdiffc}) is proved. Note that
$$
V(g)_T^{n}-V(g)_T^{\prime n}=\Delta \sum_{i=1}^n\left[g\left(\hat{c}_i^n\right)-g\left(\hat{c}_i^{\prime n}\right)\right] \bar{K}\left(t_i\right),
$$
and, by (\ref{eq:g_prime_bound}) and {\DG the} mean value theorem, for some $\xi\in[0,1]$,
\begin{equation}\label{eq:A.18}
    \begin{aligned}
    &\left|g\left(\hat{c}_i^n\right)-g\left(\hat{c}_i^{\prime n}\right)\right|\leq \left\|\nabla g\left(\xi\hat{c}_i^n+\left(1-\xi\right)\hat{c}_i^{\prime n}\right)\right\|\left\|\hat{c}_i^n-\hat{c}_i^{\prime n}\right\|\\
    &\leq C\left(1+\left\|\xi\hat{c}_i^n+\left(1-\xi\right)\hat{c}_i^{\prime n}\right\|^{p-1}\right)\left\|\hat{c}_i^n-\hat{c}_i^{\prime n}\right\|.
\end{aligned}
\end{equation}
By (\ref{boundofdiffc}), if we let $q=1$ and $q^{\prime}>\frac{1}{1-2 \varpi}$, since $\varpi \in\left[\frac{2p-1}{2(2p-r)}, \frac{1}{2}\right)$, we have
$$
\mathbb{E}\left[\left\|\hat{c}_i^n-\hat{c}_i^{\prime n}\right\|\right] \leq C a_n \Delta_n^{\frac{(2-r)(2p-1)}{2(2p-r)}}+C \Delta_n,
$$
which converges to zero at the rate of $o(\sqrt{\Delta_n})$. {\Blue Plugging} it back in (\ref{eq:A.18}), we have
\begin{equation}\label{diffgprime}
    \mathbb{E}\left[\left|g\left(\hat{c}_i^n\right)-g\left(\hat{c}_i^{\prime n}\right)\right|\right]=o(\sqrt{\Delta_n}),
\end{equation}
and (\ref{boundofdiffv}) is obtained. 

{\DG We now proceed to show (b). In that case, $X$ is continuous (i.e., $X=X^{\prime}$) and clearly the $q$th moment in (\ref{qthmoment}) is zero}. Thus, (\ref{boundofdiffc}) simplifies {as \eqref{boundofdiffc1}.
Let $q^{\prime}>\frac{1}{1-2\varpi}$ and} (\ref{boundofdiffv}) follows by (\ref{eq:A.18}), (\ref{boundofdiffc1}), and (\ref{moreassum1}).
\end{proof}
With Lemma \ref{jump_1} and the previous arguments, it remains to prove {\DG Theorem} \ref{clt} in the continuous case with $\hat{c}$ as (\ref{MDKNN}) under the following assumption:
\begin{assumption}\label{assum3}
    We have (\ref{eq:X}) with $X$ continuous, {\DG $c_t := \sigma_t \sigma_t^*$} satisfies (\ref{eq:sigma}), the processes $\mu, \tilde{\mu}, \sigma, \tilde{\sigma}$ are bounded.
\end{assumption}
We shall exploit the following decomposition \begin{equation} \label{eq:decomposition}
\frac{1}{\sqrt{\Delta_{n}}}\left(V(g)_T^{n}-V(g)_T\right)=\sum_{j=1}^{5} V^{n}_{j},
\end{equation} where
\begin{equation} \label{eq:V1234}
\begin{array}{l}
V^{n}_1=\frac{1}{\sqrt{\Delta_{n}}} \sum_{i=1}^{n} \int_{t_{i-1}}^{t_{i}}\left(g\left(c_{i}^{n}\right)-g\left(c_{s}\right)\right) d s, \\
V^{n}_2=\sqrt{\Delta_{n}} \sum_{i=1}^{n}  \sum_{j=1}^{n} K_{{b}}\left(t_{j-1} - t_{i} \right) \nabla g\left(c_{i}^{n}\right)\cdot\alpha_{j}^{n}, \\
V^{n}_3=\sqrt{\Delta_{n}}\sum_{i=1}^{n}  \sum_{j=1}^{n}K_{{b}}\left(t_{j-1} - t_{i} \right) \nabla g\left(c_{i}^{n}\right)\cdot \left(c_{j-1}^{n}-c_{i}^{n}\right)\Delta_n, \\
V^{n}_4=\sqrt{\Delta_{n}} \sum_{i=1}^{n}\left[(g\left(c_{i}^{n}+\beta_{i}^{n}\right)-g\left(c_{i}^{n}\right)-\nabla g(c_i^n)\cdot\beta_{i}^{n}\right]{\DG \bar{K}_n\left(t_i\right)},\\
V^{n}_{5}=\sqrt{\Delta_{n}} \sum_{i=1}^n\left[\bar{K}_{n}\left(t_i\right)-1\right] g\left(c_i^n\right).
\end{array}
\end{equation}
{Above,} recall that $x\cdot y=\sum_{\ell,m}x^{\ell m}y^{\ell m}$ for $x,y\in\mathbb{R}^{d\times d}$ (thus, e.g., $\nabla g(c_i^n)\cdot \beta_i^n:=\sum_{\ell, m}\partial_{\ell m} g(c_i^n)\beta_{i}^{n,\ell m}$).

The first term $V^{n}_{1}$ {\Blue converges} to 0 since $g$ is $C^2$ and $c_t$ is an It\^o semimartingale with bounded characteristics (see, e.g., (5.3.24) in \cite{JacodProtter}).
{In the subsequence subsections, we show} the convergence for each of the other terms in the decomposition.

\subsubsection{The leading term $V^{n}_{2}$} \label{leading_section}
In this section, we show  
\begin{equation}\label{eq:leading}
V^{n}_{2} \stackrel{st}{\Longrightarrow} Z,
\end{equation} where $Z$ is defined as in the statement of Theorem \ref{clt}.
Denoting $w_j= \Delta_n \sum_{i = 1}^{n}K_{{b}}(t_{j-1} - t_i) \nabla g(c^n_i)  \in\mathbb{R}^{d\times d}$, we can rewrite $V^{n}_{2}$ as 
\begin{equation}
\begin{split}
V^{n}_{2}  &= \frac{1}{\sqrt{\Delta_n}} \sum_{j = 1}^{n} w_j \cdot \alpha^n_j=\frac{1}{\sqrt{\Delta_n}} \sum_{j = 1}^{n} w^{R}_j\cdot \alpha^n_j +  \frac{1}{\sqrt{\Delta_n}} \sum_{j = 1}^{n} w^{L}_j\cdot \alpha^n_j,
\end{split}
\end{equation} 
where $w^{R}_j = \Delta_n\sum_{i = 1}^{j-1} K_{{b}}\left(t_{j-1} - t_i\right)\nabla g\left(c_i^n \right)\in \mathbb{R}^{d\times d}$ and $w^{L}_j = \Delta_n\sum_{i = j}^{n} K_{{b}}\left(t_{j-1} - t_i\right)\nabla g\left(c_i^n \right)\in\mathbb{R}^{d\times d} $.

In contrast to the $w_j$ in \cite{jacod2015estimation}, Section 3.6, $w_j$ here is not measurable with respect to $\mathcal{F}_{j}^n:=\mathcal{F}_{t_{j}} $ since the kernel is in general \emph{two-sided}. Therefore, we cannot directly apply {\DG the triangular array CLT of  \cite{JacodProtter} (Theorem 2.2.14 therein)}. To solve this problem,  
we aim to show the following:
\begin{align}
{\rm (i)}\;\frac{1}{\sqrt{\Delta_n}} \sum_{j = 1}^{n} 	\bar{w}^{R}_j\cdot \alpha^n_j \stackrel{st}{\Longrightarrow}  Z,   \label{eq:left}    
\qquad{\rm (ii)}\;\frac{1}{\sqrt{\Delta_n}} \sum_{j = 1}^{n} \bar{w}^{L}_j\cdot    \alpha^n_j\longrightarrow 0,    
\end{align}
where, with the notation $\bar{K}_j^n(t):=\sum_{i = j}^n \Delta_n K_{{b}}\left( t- t_i \right)$, 
\begin{align*}
	\bar{w}^{R}_j&:=w^{R}_j  + \bar{K}_j^n(t_{j-1})\nabla g\left(c^n_{j-1} \right)=\Delta_n\sum_{i = 1}^{j-1} K_{{b}}\left(t_{j-1} - t_i\right)\nabla g\left(c_i^n \right)+ \bar{K}_j^n(t_{j-1})\nabla g\left(c^n_{j-1} \right),\\
	\bar{w}^{L}_j&:=w^{L}_j  - \bar{K}_j^n(t_{j-1})\nabla g\left(c^n_{j-1} \right)=\Delta_n\sum_{i = j}^{n} K_{{b}}\left(t_{j-1} - t_i\right)\nabla g\left(c_i^n \right)- \bar{K}_j^n(t_{j-1})\nabla g\left(c^n_{j-1} \right).
\end{align*}

For (\ref{eq:left}-i), since $w^{R}_j, \nabla g\left(c^n_{{j-1}} \right) \in \mathcal{F}_{j-1}^n$  and $\alpha^n_{j} \in \mathcal{F}_{j}^n  $, we can follow arguments similar to those in  \cite{jacod2015estimation} {\DG (Section 3.6) to apply the triangular array CLT of \cite{JacodProtter}}. In particular, it suffices to show that, for any $\ell,m,p,q\in\{1,\dots,m\}$, \begin{align} \label{eq:limit_variance}
&\Delta_{n} \sum_{j=1}^{n} \bar{w}^{R,\ell m}_j{\DG \bar{w}^{R,pq}_j}(c_{j-1}^{n,\ell p}c_{j-1}^{n,mq}+c_{j-1}^{n,\ell q}c_{j-1}^{n,mp}) \stackrel{\mathbb{P}}{\longrightarrow} \int_{0}^{T}    \partial_{\ell m}g\left(c_{s} \right)\partial_{pq}g\left(c_{s} \right)  \left(c_{s}^{\ell p}c_{s}^{m q} +c_{s}^{\ell q}c_{s}^{mp}\right)d s.
\end{align}
{To this end, first note that, for $ |s - t_{j-1}| \leq \Delta_n $, $c^n_{j-1} \to c_s $ a.s.. If we also have   $\bar{w}^{R,\ell m}_j   \to \partial_{\ell m}g(c_s)$ a.s., whenever $ |s - t_{j-1}| \leq \Delta_n $, then (\ref{eq:limit_variance}) would follow from dominated convergence theorem. 
For $\left|s - t_{j-1} \right| \leq \Delta_n $, note that \begin{equation}
\begin{split}
 \bar{w}^{R,\ell m}_j- \partial_{\ell m}g(c_s)&=  \Delta_n \sum_{i = 1}^{j-1} \left(\partial_{\ell m}g\left(c^n_i \right) - \partial_{\ell m}g\left(c_s \right) \right) K_{{b}}\left( t_{j-1} - t_i \right)\\
&\quad  + \left(\partial_{\ell m}g\left( c^n_{j-1}\right) - \partial_{\ell m}g(c_s) \right)\Delta_n \sum_{i = j}^n K_{{b}}\left( t_{j-1} - t_i \right)\\
& \quad+ \Delta_n\sum_{i=1}^n  K_{{b}}\left( t_{j-1} - t_i \right) \partial_{\ell m}g(c_s) -  \partial_{\ell m}g(c_s) \\
& =:  P_1+ P_2 + O(\Delta_n/{b}),
\end{split}
\end{equation}
where the order $O(\Delta_n/{b})$ follows from (\ref{eq:K_discrete}).
Since, for $\left|s - t_{j-1} \right| \leq \Delta_n $, $\partial_{\ell m}g\left( c^n_{j-1}\right) \to \partial_{\ell m}g(c_s) $ a.s and $\Delta_n \sum_{i = j}^n K_{{b}} \left( t_{j-1} -t_i\right) $ is bounded,  $P_2 \stackrel{\mathbb{P}}{\to} 0$, as $n\to\infty$. We now deal with $P_1$. Since $c$ is continuous and $\partial_{\ell m}g$ is locally continuous, for an arbitrary, but fixed, $\varepsilon>0$, there exists a $\delta=\delta(\varepsilon)\in(0, \min(T-s, s)) $ such that 
$|\partial_{\ell m}g(c_u)-\partial_{\ell m}g(c_s)|<\varepsilon$, whenever $|u-s|<\delta$.  Next, consider the decomposition 
\begin{equation}
P_1 = \Delta_n \left( \sum_{0<s- t_i  < \delta}  + \sum_{T>s-t_i > \delta}\right) \left(\partial_{\ell m}g\left(c^n_i \right) - \partial_{\ell m}g\left(c_s \right) \right) K_{{b}}\left( t_{j-1} - t_i \right) =: P_{11} + P_{12}.
\end{equation} 
Since $\partial_{\ell m}g$ is locally bounded and $c$ is bounded,  for some $s_{i-1} \in\left(t_{i-1}, t_{i}\right)$, {\DG we have:}
\begin{align} \nonumber
\left|P_{12} \right| &
\leq  C \Delta_n \sum_{i:T>s-t_i \ge \delta} \left|K_{{b}}\left( t_{j-1} - t_i \right)\right|\\
\nonumber
&\leq  C \sum_{i:T>s-t_{i} \geq \delta} \int^{\left(t_{j-1} - t_{i-1}\right) / {\DG b}}_{\left(t_{j-1} - t_{i}\right) / {\DG b}}\left|K(t)\right| d t+\frac{\Delta_n}{b} \sum_{i:T> s - t_{i}\geq \delta}\left|K\left(\frac{t_{j-1}-t_i}{b}\right)-K\left(\frac{t_{j-1}-s_{i-1}}{b}\right)\right|\\
&\leq C\int_{\frac{\delta-\Delta_n}{b}}^{+\infty}\left|K(t)\right| d t+\frac{\Delta_n}{b} V_{-\infty}^{\infty}\left(K\right)\stackrel{n\to\infty}{\longrightarrow}0,\label{eq:farfrom_tau}
\end{align}
where we used $\int^{\left(t_{j-1} - t_{i-1}\right) / h}_{\left(t_{j-1} - t_{i}\right) / h}\left|K(t)\right| d t =\frac{\Delta_n}{b} |K\left(\frac{t_{j-1}-s_{i-1}}{b}\right)| $ for  some $s_{i-1} \in\left(t_{i-1}, t_{i}\right)$, and $V_{-\infty}^{\infty}(K)  $ is the total variation of $K$.
For $P_{11}$, 
 \begin{equation*} 
 \begin{split}
 \left|P_{11} \right|&\leq \Delta_n \sum_{{\DG i}:0<s- t_i  < \delta}  \left|\partial_{\ell m}g\left(c^n_i \right) - \partial_{\ell m}g\left(c_s \right) \right|\left| K_{{b}}\left( t_{j-1} - t_i \right) \right| 
  \leq \varepsilon \sum_{r=0}^\infty\frac{\Delta_n}{b}\left| K\left(\frac{\Delta_n r}{b}\right)\right|\stackrel{n\to\infty}{\longrightarrow} \varepsilon.
 \end{split}
 \end{equation*}
Therefore, $\limsup_{n\to\infty}|P_{1}|\leq{}\varepsilon$ and, since $\varepsilon$ is arbitrary, we conclude that $\bar{w}^{R,\ell m}_j   \to \partial_{\ell m}g(c_s)$ and, thus, \eqref{eq:limit_variance} and (\ref{eq:left}-i).

Next, we prove (\ref{eq:left}-ii). Note the expression therein can be written as \begin{equation}
\begin{split}
 \frac{1}{\sqrt{\Delta_n}} \sum_{j = 1}^{n} \bar{w}^{L}_j\alpha_{j}^n& =\frac{1}{\sqrt{\Delta_n}}\sum_{j=1}^n \sum_{i=j}^n \Delta_n K_{{b}}\left( t_{j-1} - t_i\right) \sum_{u = j}^i \left( \nabla g(c_u) - \nabla g(c_{u-1})\right)\cdot \alpha^n_j\\
 &=\sum_{\ell, m} \sum_{u = 1}^n \left( \partial_{\ell m} g(c_u) - \partial_{\ell m} g(c_{u-1})\right) \left[ \sum_{j=1}^u\sum_{i = u}^n \sqrt{\Delta_n} K_{{b}}( t_{j-1}-t_i) \alpha^{n,\ell m}_j  \right]\\
&=\sum_{\ell,m} \int_0^T A^{n,\ell m}_s d \upsilon_{s},
\end{split}
\end{equation} 
where $s\to \upsilon_s:=\partial_{\ell m} g(c_s)$ is a real-valued function and $A^{n,\ell m}_s :=  \sum_{j=1}^u\sum_{i = u}^n \sqrt{\Delta_n} K_{{b}}( t_{j-1}-t_i) \alpha^{n,\ell m}_j $ for $ s \in (t_{u-1}, t_u) $. 
We now proceed to show that $A^{n,\ell m}_s \to 0$, for all $s\in(0,T)$. Note that $A^{n,\ell m}_s $ is a triangular array. By (\ref{eq:bounds_alpha}), (\ref{eq:better_bound_alpha}), (\ref{eq:AN}), and (\ref{eq:K_discrete}), we have 
\begin{equation}
\begin{split}
& \mathbb{E}\left|\sum_{j=1}^u \sum_{i=u}^n\sqrt{\Delta_n}K_{{b}}(t_{j-1}-t_i)\mathbb{E}\left[\alpha^{n,\ell m}_j \Big|\mathcal{F}_{j-1} \right]\right|\\
&\leq  C\sum_{j=1}^u \sum_{i=u}^n\sqrt{\Delta_n}|K_{{b}}(t_{j-1}-t_i)|\Delta_n^{3/2}\left(\sqrt{\Delta_n} + \mathbb{E}\eta^n_{j-1} \right)\\
&\leq C\sum_{j=1}^u \sum_{i=u}^n\Delta_n|K_{{b}}(t_{j-1}-t_i)|\Delta_n^{3/2}+C\sum_{j=1}^u  \mathbb{E}\eta^n_{j-1} \Delta_n\to  0.
\end{split}
\end{equation} 
Furthermore, recalling that $s \in (t_{u-1}, t_u) $ and using (\ref{eq:bounds_alpha}), 
\begin{equation}
\begin{split}
& \sum_{j=1}^u \left(\sum_{i=u}^n\sqrt{\Delta_n}K_{{b}}(t_{j-1}-t_i)\right)^2 \mathbb{E}\left[\left(\alpha^{n,\ell m}_j\right)^2 \Big| \mathcal{F}_{j-1} \right]\\
&\leq C\sum_{j=1}^u \left(\sum_{i=u}^n\sqrt{\Delta_n}K_{{b}}(t_{j-1}-t_i)\right)^2 \Delta_n^2\\
&\leq  C\int_0^s \left(\int_s^T K_{{b}}(t - v) dv \right)^2 dt \\
&\leq  C {\DG b}\int_{-\infty}^0 \left(\int_{-\infty}^{x} |K(w)| dw \right)^2 d x\to  0,
\end{split}
\end{equation}
where we used that $\int_{-\infty}^0 \left(\int_{-\infty}^{x} |K(w)| dw \right)^2 d x\leq\int_{-\infty}^0 \int_{-\infty}^{x} |K(w)| dw d x=\int_{-\infty}^{0}|K(w)|(-w)dw<\infty$ by Assumption \ref{kernel}-c. 
Therefore, $A^n_s \stackrel{\mathbb{P}}{\longrightarrow} 0$ and $ \int_0^T A^n_s d g'(c_s) \to \int_0^T A_s d g'(c_s)  =0$. This finishes the proof of (\ref{eq:left}-(ii).

\subsubsection{The bias term $V^{n}_{3}$} \label{main_bias_section}
We proceed to prove that
\begin{equation} \label{eq:v4}
\begin{split}
V^{n,\ell m}_{3}&:=\Delta_{n}^{3/2}\sum_{i=1}^{n} \partial_{\ell m}g\left(c_{i}^{n}\right) \sum_{j=1}^{n}K_{{b}}\left(t_{j-1} - t_{i} \right) \left(c_{j-1}^{n,\ell m}-c_{i}^{n,\ell m}\right) \\
&  \stackrel{\mathbb{P}}{\longrightarrow}\, \theta \int_0^T  \partial_{\ell m} g\left(c_s\right) d c^{lm}_s \int_{-\infty}^{\infty} L(u) d u +
\theta\sum_{p,q}\int_0^T \partial_{pq,\ell m}g(c_{s})\tilde{c}_{s}^{\ell m,pq} ds \int_{-\infty}^0 L(u) du.
\end{split}
\end{equation} 
Define 
\begin{equation}
\widetilde{V}^{n,\ell m}_{3} =\Delta_{n}^{3/2} \sum_{j=1}^{n} \sum_{i=j}^{n} K_{{\DG b}}\left(t_{j-1}-t_{i}\right) \partial_{\ell m}g\left(c_{j-1}^{n}\right)\left[c_{j-1}^{n,\ell m}-c_{i}^{n,\ell m}\right].
\end{equation}
Similarly to the proof of the leading term $V^{n}_{2}$, we write $V^{n,\ell m}_{3}$ as \begin{equation}
\begin{aligned}
V^{n,\ell m}_{3}&=
\Delta_{n}^{3/2} \sum_{j=1}^{n} \left(\sum_{i=j}^{n} +\sum_{i=1}^{j-1} \right)K_{{b}}\left(t_{j-1}-t_{i}\right) \partial_{\ell m}g\left(c_{i}^{n}\right)\left[c_{j-1}^{n,\ell m}-c_{i}^{n,\ell m}\right] \\
&=: L + R,
\end{aligned}
\end{equation}and aim to prove that \begin{align} \label{eq:right4}
& R + \widetilde{V}^{n,\ell m}_{3}\to \theta \int_0^T \partial_{\ell m}g\left(c_t\right) dc^{lm}_t \int_{-\infty}^{\infty} L(u) d u,\\
&L - \widetilde{V}^{n,\ell m}_{3}\to \theta\sum_{p,q}\int_0^T \partial_{pq,\ell m}g(c_{s})\tilde{c}_{s}^{\ell m,pq} ds \int_{-\infty}^0 L(u) du.\label{eq:left4}
\end{align}
We first show (\ref{eq:right4}). Writing $c^{n,\ell m}_{j-1} - c^{n,\ell m}_i = \sum_{u = i+1}^{j-1}\Delta^n_u c^{\ell m} $, when $j-1\ge i $, and $c^{n,\ell m}_{j-1} - {\DG c^{n,\ell m}_i} = -\sum_{u = j}^i \Delta^n_u c^{\ell m}$, when $j-1 < i$, and using the usual convention that $\sum_{j=n+1}^{n}=\sum_{i=1}^{0}=0$, we have
\begin{align}
\nonumber
R + \widetilde{V}^{n,\ell m}_{3}
&=\Delta_{n}^{3/2}\sum_{j=1}^{n} \sum_{i=1}^{j-1} K_{{b}}\left(t_{j-1}-t_{i}\right) \partial_{\ell m}g\left(c_{i}^{n}\right)\left[c_{j-1}^{n,\ell m}-c_{i}^{n,\ell m}\right]\\
\nonumber
&\quad+
\Delta_{n}^{3/2} \sum_{j=1}^{n} \sum_{i=j}^{n} K_{b}\left(t_{j-1}-t_{i}\right) \partial_{\ell m}g\left(c_{j-1}^{n}\right)\left[c_{j-1}^{n,\ell m}-c_{i}^{n,\ell m}\right]  \\
\nonumber
&=\Delta_{n}^{3/2}\sum_{j=2}^{n} \sum_{i=1}^{j-2} K_{{b}}\left(t_{j-1}-t_{i}\right) \partial_{\ell m}g\left(c_{i}^{n}\right)\sum_{u = i+1}^{j-1}\Delta^n_u c^{\ell m} \\
\nonumber
&\quad-
\Delta_{n}^{3/2} \sum_{j=1}^{n} \sum_{i=j}^{n} K_{{b}}\left(t_{j-1}-t_{i}\right) \partial_{\ell m}g\left(c_{j-1}^{n}\right)\sum_{u = j}^i \Delta^n_u c^{\ell m} \\
\nonumber
&= \Delta_{n}^{3/2} \sum_{u=1}^{n} \Delta_{u}^{n} c^{\ell m} \left[\sum_{j=u+1}^{n} \sum_{i=1}^{u-1} K_{{b}}\left(t_{j-1}-t_{i}\right) \partial_{\ell m}g\left(c_{i}^{n}\right) - \sum_{j=1}^{u} \sum_{i=u}^{n} K_{{b}}\left(t_{j-1}-t_{i}\right)  \partial_{\ell m}g\left(c_{j-1}^{n}\right) \right]\\
\label{DcmpRtldV3}
& =: \int_0^T H(n)_t d c^{\ell m}_t,
\end{align}
where, 
for $t \in \left ((u-1)\Delta_n, u\Delta_n \right) $,  
\begin{equation} \label{eq:H_n(t)}
H(n)_t := \Delta_n^{3/2} \left[\sum_{j=u+1}^{n} \sum_{i=1}^{u-1} K_{{b}}\left(t_{j-1}-t_{i}\right) \partial_{\ell m}g\left(c_{i}^{n}\right) - \sum_{j=1}^{u} \sum_{i=u}^{n} K_{{b}}\left(t_{j-1}-t_{i}\right)  \partial_{\ell m}g\left(c_{j-1}^{n}\right) \right].
\end{equation} 
Note $H(n)_t$ is adapted to $\mathcal{F}_t$. We first consider the first term in (\ref{eq:H_n(t)}). 
{\DG It can be shown (see Appendix \ref{Technical2_1}) that}
\begin{align}\label{ApprmRSWInt}
	\Delta_n^{\frac{3}{2}}\sum_{j=u+1}^{n} \sum_{i=1}^{u-1} K_{{b}}\left(t_{j-1}-t_{i}\right) \partial_{\ell m}g\left(c_{i}^{n}\right) - \Delta_n^{-\frac{1}{2}}\int_{t_{u+1}}^{T}\int_0^{t_u}K_{{b}}(s-\varpi) \partial_{\ell m}g(c_\varpi) d\varpi ds=o_p(1).
\end{align}
{\DG Next, we} show that $\Delta_n^{-1/2}\int_{t_{u+1}}^T \int_0^{t_u} K_{{b}}(s-\varpi)\partial_{\ell m}g(c_\varpi) d\varpi ds$ converges to $\theta \partial_{\ell m}g(c_{t})\int_0^{\infty} L(v) dv $.  {\DG Indeed, since} 
\[
	\Delta_n^{-\frac{1}{2}}\int_{t_{u}}^{t_{u+1}} \int_0^{t_u} |K_{{b}}(s-\varpi)\partial_{\ell m}g(c_\varpi)| d\varpi ds\leq
	C\Delta_n^{-\frac{1}{2}}\int_{t_{u}}^{t_{u+1}} \int_{\frac{s-t_{u}}{b}}^{\frac{s}{b}} |K(v)| dvds=O(\Delta_n^{\frac{1}{2}}),
\]
we only need to consider $\Delta_n^{-1/2}\int_{t_{u}}^T \int_0^{t_u} K_{{b}}(s-\varpi)\partial_{\ell m}g(c_\varpi) d\varpi ds$. Let us write $\tau = t_{u} $ for simplicity. We exploit the decomposition \begin{equation*}
\Delta_n^{-\frac{1}{2}}\int_{\tau}^T \int_0^{\tau} K_{{b}}(s-\varpi)\partial_{\ell m}g(c_\varpi) d\varpi ds= \Delta_n^{-\frac{1}{2}} \left( \iint_{S_{\delta}} + \iint_{S_{\delta}^c}\right) K_{{b}}(s-\varpi)\partial_{\ell m}g(c_\varpi) d\varpi ds,
\end{equation*}
where $S_{\delta} = \{(s,\varpi): s-\varpi<\delta, s\in (\tau,T), \varpi \in(0,\tau)  \} $ for $\delta>0$ and $S_{\delta}^c = [\tau, T]\times [0,\tau] - S_{\delta}$. Letting $w = \frac{s-\varpi}{b}$ first and then $v = \frac{s}{b}$, we can write:
\begin{equation}
\begin{split}
&\Delta_n^{-1/2} \left|\iint_{S_{\delta}^c} K_{{b}}(s-\varpi)\partial_{\ell m}g(c_\varpi) d\varpi ds\right|\\
& =\Delta_n^{-1/2} \left|\int_{\tau}^{\tau+\delta}\int_{0}^{s-\delta}K_{{b}}(s-\varpi)\partial_{\ell m}g(c_\varpi) d\varpi ds + \int_{\tau+\delta}^{T}\int_{0}^{\tau}K_{{b}}(s-\varpi) \partial_{\ell m}g(c_\varpi) d\varpi ds  \right|\\
&\leq C \Delta_n^{-1/2}\left[\int_{\tau}^{\tau+\delta} \left(\int_{\delta/{b}}^{s/{b}} |K(w)| dw\right) ds  + \int_{\tau+\delta}^{T}\left( \int_{(s-\tau)/{b}}^{s/{b}}|K(w)|dw\right)ds\right]\\
&\leq C \Delta_n^{-1/2} {\DG b} \left[\int_{\tau/{b}}^{(\tau+\delta)/{b}} \left(\int_{\delta/{b}}^{v} |K(w)| dw\right) dv  + \int_{(\tau+\delta)/{b}}^{T/{b}}\left( \int_{v-\tau/{b}}^{v}|K(w)|dw\right)dv\right]\\
 & \leq C \left[\bar{L}(\delta/{b})\frac{\delta}{b}+ \bar{L}(\delta/{b})\frac{T}{b}\right]\to  0 , 
\end{split}
\end{equation}
for any $\delta > 0$, where $\bar{L}(x) = \int_x^{\infty}\left|K(u) \right| d u$ and the last equation uses $\bar{L}(x)x\stackrel{x\to\infty}{\longrightarrow} 0$ derived from $K(x)x^{2+\epsilon} \to 0$, $\epsilon>1/2$ as $|x| \to \infty $.

Next, we consider the integral over $S_{\delta}. $ Clearly, for any $\epsilon>0$, there {\Blue exists} $\delta>0$ such that $\left| \partial_{\ell m}g(c_\varpi) - \partial_{\ell m} g(c_{\tau})\right|< \epsilon $, for $|\varpi-\tau|<\delta $. Letting $w = \frac{s-\varpi}{b}$ first and then $v = \frac{s-\tau}{b}$, we can write
\begin{equation}
\begin{split}
& \Delta_n^{-1/2}\int_{S_{\delta}} K_{{b}}(s-\varpi)\partial_{\ell m}g(c_\varpi) d\varpi ds \\
&= \Delta_n^{-1/2}\int_{\tau}^{\tau+\delta}\int_{s-\delta}^{\tau}K_{{b}}(s-\varpi)\partial_{\ell m}g(c_\varpi) d\varpi ds\\
& \leq \Delta_n^{-1/2} \int_{\tau}^{\tau + \delta} \left( \int_{(s-\tau)/{b}}^{\delta/{b}} K(w)\left( \partial_{\ell m}g(c_{\tau}) + \epsilon\right)\mathrm{1}_{K(w)>0} dw\right)  ds \\
& \quad +\Delta_n^{-1/2} \int_{\tau}^{\tau + \delta} \left( \int_{(s-\tau)/{b}}^{\delta/{b}} K(w)\left( \partial_{\ell m}g(c_{\tau}) - \epsilon\right)\mathrm{1}_{K(w)<0} dw\right)  ds \\
& \leq {\DG b}\Delta_n^{-1/2}\left[\int_{0}^{\delta/{b}}\int_{v}^{\delta/{b}}K(w) \partial_{\ell m}g(c_{\tau})  dw dv + \epsilon\int_{0}^{\infty} \int_v^{\infty}|K(w)|dw dv  \right].
\end{split}
\end{equation} Similarly, we can find a lower bound of the form:
\begin{equation}
{\DG b}\Delta_n^{-1/2}  \left[\int_{0}^{\delta/{b}}\int_{v}^{\delta/{b}}K(w) \partial_{\ell m}g(c_{\tau})  dw dv - \epsilon\int_{0}^{\infty} \int_v^{\infty}|K(w)|dw dv  \right].
\end{equation} 
Recall ${\DG b}\Delta_n^{-1/2} \sim \theta$ and $c_{\tau}=c_{t_u}\stackrel{n\to\infty}{\longrightarrow} c_{t}$. Taking $ n \to \infty$ and then  $\epsilon \to 0$, we have 
\begin{equation}\label{LmINTR}
\Delta_n^{-1/2}\int_{\tau}^T \int_0^{\tau} K_{{b}}(s-\varpi)\partial_{\ell m}g(c_\varpi) d\varpi ds \,\stackrel{n\to\infty}{\longrightarrow}\, \theta \partial_{\ell m}g(c_{t}) \int_0^{\infty}L(v) dv.
\end{equation} 

For the second term in (\ref{eq:H_n(t)}), {\DG similar to \eqref{ApprmRSWInt}, we can show that}
\begin{align}\label{ApprmRSWIntT2}
	\Delta_n^{\frac{3}{2}}\sum_{j=1}^{u} \sum_{i=u}^{n} K_{{b}}\left(t_{j-1}-t_{i}\right) \partial_{\ell m}g\left(c_{j-1}^{n}\right) - \Delta_n^{-\frac{1}{2}}\int_{0}^{t_{u}}\int_{t_u}^{T}K_{{b}}(\varpi-s) \partial_{\ell m}g(c_\varpi) ds d\varpi=o_P(1).
\end{align}
Note that $\Delta_n^{-\frac{1}{2}}\int_{0}^{t_{u}}\int_{t_u}^{T}K_{{b}}(\varpi-s) \partial_{\ell m}g(c_\varpi) ds d\varpi$ is the same as the left-hand side in (\ref{LmINTR}) but with $K(-w)$ instead of $K(w)$. Therefore, 
\[
 \Delta_n^{-\frac{1}{2}}\int_{0}^{t_{u}}\int_{t_u}^{T}K_{{b}}(\varpi-s) \partial_{\ell m}g(c_\varpi) ds d\varpi
\, \stackrel{n\to\infty}{\longrightarrow}\,\theta \partial_{\ell m}g(c_{t}) \int_0^{\infty}\int_{v}^{\infty}K(-w)dw dv.
 \]
We conclude that 
 \begin{equation}
H(n)_t \to \theta \partial_{\ell m}g(c_{t})\int_{-\infty}^{\infty} L(v) dv,
\end{equation}
{\DG where we recall that \(L(t) := \int_{t}^{\infty} K(u) d u \mathbf{1}_{\{t>0\}}-\int_{-\infty}^{t} K(u) d u \mathbf{1}_{\{t \leq 0\}}\).} 
Since $c $ is a.s. continuous, (\ref{eq:right4}) follows by dominated convergence theorem for stochastic integrals.

{{Next, we {\DG proceed} to show (\ref{eq:left4}). Note that 
\begin{align}\label{TOTrmNd0}
L-\widetilde{V}^{n,\ell m}_{3} &=\Delta_{n}^{3/2} \sum_{j=1}^{n} \sum_{i=j}^{n} K_{{b}}\left(t_{j-1}-t_{i}\right) \left(\partial_{\ell m}g\left(c_{i}^{n}\right)-\partial_{\ell m}g\left(c_{j-1}^{n}\right)\right)\left[c_{j-1}^{n,\ell m}-c_{i}^{n,\ell m}\right].
\end{align}
Let us start by rewriting the inner summation in (\ref{TOTrmNd0}). Applying the conclusion in Remark \ref{SpecLemmaA1} with $\tau = t_{j-1} $,  $f(x) = \left(\partial_{\ell m}g(c^n_{j-1}) - \partial_{\ell m}g(x)\right) \left( c^{n,\ell m}_{j-1}  - x^{\ell m}\right)$, and $K(-x) \mathrm{1}_{x>0}$ instead of $K(x)$, we have $f(c_{j-1}^n) = 0 $ and}}
\begin{align}\nonumber
&\sum_{i=j}^{n}K_{{b}}\left(t_{j-1}-t_{i}\right) \left(\partial_{\ell m}g\left(c_{j-1}^{n}\right) -\partial_{\ell m}g\left(c_{i}^{n}\right) \right)\left[\left(c_{j-1}^{n,\ell m}-c_{i}^{n,\ell m}\right) \Delta_{n}\right]\\
\nonumber
&\quad =\sum_{i=j+1}^{n+1}K_{{b}}\left(t_{j-1}-t_{i-1}\right) f(c^n_{i-1})\Delta_n\\
\nonumber
&\quad =\int_{t_{j-1}}^T K_{{b}}\left(t_{j-1}-s\right) f(c_s) ds + D_1(f)  + \sum_{i=j+1}^{n+1}K_{{b}}\left(t_{j-1}-t_{i-1}\right)  \int_{t_{i-1}}^{t_i} \left( f(c^n_{i-1}) - f(c_s)\right)   ds\\
\label{SmplHA}
&\quad= \int_{t_{j-1}}^T K_{{b}}\left(t_{j-1}-s\right) f(c_s) ds
+ o_p\left(\frac{\Delta_n}{b}\right) + O_p\left( \sqrt{\Delta_n {\DG b}}\right),
\end{align} 
where the order $O_p(\sqrt{\Delta_n {\DG b}})$ can be justified since, for $s \in ((i-1)\Delta_n, i \Delta_n] $ and $i>j $,
\begin{equation*}
\begin{split}
\mathbb{E}\left| f(c^n_{i-1}) - f(c_s)\right|&\leq \sqrt{\mathbb{E}(c^{n,\ell m}_{j-1} - c^{\ell m}_s)^2  \mathbb{E} (\partial_{\ell m}g(c^n_{i-1}) - \partial_{\ell m}g(c_{s})  )^2 }\\
&\quad+ \sqrt{\mathbb{E}(\partial_{\ell m}g(c^n_{j-1}) - \partial_{\ell m}g(c^n_{i-1}))^2  \mathbb{E} (c^{n,\ell m}_{i-1} - c^{\ell m}_{s}  )^2 }\\
&\leq C\sqrt{\Delta_n (t_{i}-t_{j-1} )},
\end{split}
\end{equation*} 
for some constant $C$, and, therefore, for some possibly different constant $C$,
\begin{equation*}
\begin{split}
\sum_{i=j+1}^{n+1}|K_{{b}}\left(t_{j-1}-t_{i-1}\right) | \sqrt{\Delta_n|t_{j-1} - t_{i}|}\Delta_n 
&\leq  C\int_{t_{j+1}}^{T}|K_{{b}}\left(t_{j-1}-s\right) | \sqrt{\Delta_n|t_{j-1} - s|}ds \\
& \leq C \int_{-\infty}^0|K\left(u\right) | \sqrt{\Delta_n |u| {\DG b}}ds =  O_p\left(\sqrt{\Delta_n{\DG b}} \right).
\end{split}
\end{equation*}
Plugging (\ref{SmplHA}) back into (\ref{TOTrmNd0}), recalling that $f(x) = \left(\partial_{\ell m}g(c^n_{j-1}) - \partial_{\ell m}g(x)\right) \left( c^{n,\ell m}_{j-1}  - x^{\ell m}\right) $, and applying integration by parts, we get
\begin{align}
\nonumber
 L -{\DG \widetilde{V}^{n,\ell m}_{3}}
&=-\sqrt{\Delta}_{n} \sum_{j=1}^{n}\left[\int_{t_{j-1}}^T K_{{b}}\left(t_{j-1}-t\right) \left(\partial_{\ell m}g\left(c_{j-1}^{n}\right) - \partial_{\ell m}g\left({\DG c_s}\right) \right) \left(c_{j-1}^{n,\ell m}-{\DG c_s}^{\ell m}\right) ds \right] + o_p(1) \\
\nonumber
&=-\sqrt{\Delta}_{n} \sum_{j=1}^{n}L\left(\frac{t_{j-1}- T}{b} \right)
 \left(\partial_{\ell m}g\left(c_{j-1}^{n}\right) - \partial_{\ell m}g\left(c_T\right) \right) \left(c_{j-1}^{n,\ell m}-c_T^{\ell m}\right)  \\
 \nonumber
& +\sqrt{\Delta}_{n} \sum_{j=1}^{n}\int_{t_{j-1}}^T L\left(\frac{t_{j-1}- t}{b} \right) b(j)_s dW_s \\
\nonumber
& +\sqrt{\Delta}_{n} \sum_{j=1}^{n} \int_{t_{j-1}}^T L\left(\frac{t_{j-1}- t}{b} \right) a(j)_s ds + o_p(1), \\
\label{SndTrmN}
 &=: P_1 + P_2 +P_3 +  o_p(1),
\end{align}	
where, recalling the notation $\tilde{c}_{s}^{uv,pq}=\sum_{j=1}^{d}\tilde{\sigma}^{uv,j}\tilde{\sigma}^{pq,j}$, 
\begin{align*}
	b(j)_s&:=\sum_{pq}(c^{\ell m}_s-c_{j-1}^{n,pq})
	\partial_{pq,\ell m}^2g(c_s)\tilde{\sigma}_s^{pq}+
	(\partial_{\ell m}g(c_s)-\partial_{\ell m}g(c_{j-1}^{n}))\tilde{\sigma}_s^{\ell m}\\
	a(j)_s&:=\sum_{pq}(c^{\ell m}_s-c_{j-1}^{n,pq})
	\partial_{pq,\ell m}^2g(c_s)\tilde{\mu}_s^{pq}+
	(\partial_{\ell m}g(c_s)-\partial_{\ell m}g(c_{j-1}^{n}))\tilde{\mu}_s^{\ell m}\\
	&\quad+\frac{1}{2}\sum_{uv,pq}(c_{s}^{\ell m}-c_{j-1}^{n,\ell m})\partial^3_{uv,pq,\ell m}g(c_s)\tilde{c}_s^{uv,pq}+\sum_{pq}\partial_{pq,\ell m}g(c_s)\tilde{c}_s^{\ell m,pq}.
\end{align*}
For the first term,  for some $\xi$ in between $c_T$ and $c_{j-1}^n$, we have:
\begin{equation*}
\begin{split}
\mathbb{E}\left|P_1 \right| &\leq \sqrt{\Delta_n} \sum_{j=1}^n \left|L\left(\frac{t_{j-1}- T}{b} \right)  \right|\mathbb{E}\left[\left\|\nabla \partial_{\ell m}g\left(\xi\right)\right\| \left\|c_{j-1}^{n}-c_T^{n}\right\|^2\right]\\
&\leq C\sqrt{\Delta_n} \sum_{j=1}^n \left|L\left(\frac{t_{j-1}- T}{b} \right) \right| \left(T- t_{j-1} \right)\\
\leq & C \frac{{b}^{2}}{\sqrt{\Delta_n}}\int_0^T| L(u)|u du \to   0,
\end{split}
\end{equation*}
where we used (\ref{eq:X_c_bound}-iii) in the second inequality above. 

For $P_2$, let $\zeta^n_j = \sqrt{\Delta}_{n} \int_{t_{j-1}}^T L\left(\frac{t_{j-1}- t}{b} \right) b(j)_s d{\DG W_s}  $ {\DG so that} $P_2 = \sum_{j=1}^n \zeta^n_j$. Observe that for $j'<j$, $\mathbb{E}\left[ \zeta^n_{j'}\zeta^n_j\right] = \mathbb{E}\left[\zeta^n_{j'}\mathbb{E}\left[ \zeta^n_j\mid \mathcal{F}^n_{j-1} \right] \right] =0 $ because $\mathbb{E}\left[ \zeta^n_j\mid \mathcal{F}^n_{j-1} \right] = 0 $ and, thus,
 \begin{equation} \label{eq:b_t_bound}
\begin{split}
 \mathbb{E}\left[P_2^2\right]&=\sum_{j=1 }^n \mathbb{E}\left[\left(\zeta^n_j \right)^2\right]\\
&= \Delta_n\sum_{j=1}^n\int_{t_{j-1}}^T L^2\left(\frac{t_{j-1}-s}{b} \right)\mathbb{E}\left(\|b(j)_s\|^2 \right)d s\\
& \leq C\Delta_n\sum_{j=1}^n\int_{t_{j-1}}^T L^2\left(\frac{t_{j-1}-s}{b} \right)\left(s-t_{j-1} \right) d s\\
& \leq C{b}^{2}\Delta_n\sum_{j=1}^n\int_{-\infty}^0 L^2\left(u\right)(-u)d u \to 0,
\end{split}
\end{equation}
where in the first inequality we used that $\mathbb{E}\left\|b(j)_s \right\|^2 \leq C \left|s-t_{j-1}\right|$, which follows from the definition of $b(j)_s$ and (\ref{eq:X_c_bound}).

For $P_3$, recall ${\DG b} \sim \theta\sqrt{\Delta_n} $ and note that $a(j)_s=O_p(\sqrt{\left|s-t_{j-1}\right|})+\sum_{pq}\partial_{pq,\ell m}g(c_s)\tilde{c}_s^{\ell m,pq}$ 
due to (\ref{eq:X_c_bound}). We then have 
\begin{equation}
\begin{split}
P_3 & = \sqrt{\Delta}_{n} \sum_{j=1}^{n} \int_{t_{j-1}}^T  L\left(\frac{t_{j-1}- s}{b} \right) a(j)_s ds\\
&= o_p(1) + \sqrt{\Delta}_{n} \sum_{pq}\sum_{j=1}^{n} \int_{t_{j-1}}^T  L\left(\frac{t_{j-1}- s}{b} \right)\partial_{pq,\ell m}g(c_s)\tilde{c}_s^{\ell m,pq}ds\\
&= o_p(1) + \sqrt{\Delta}_{n} {\DG b}\sum_{pq}\sum_{j=1}^{n} 
\partial_{pq,\ell m}g(c_{t_{j-1}})\tilde{c}_{t_{j-1}}^{\ell m,pq}\int_{\frac{t_{j-1} - T}{b}}^0  L\left(u \right) du\\
&\stackrel{n\to\infty}{\longrightarrow}  \theta\sum_{pq}\int_0^T \partial_{pq,\ell m}g(c_{s})\tilde{c}_{s}^{\ell m,pq} ds \int_{-\infty}^0 L(u) du,
\end{split}
\end{equation} 
since
\[
	\sqrt{\Delta}\sum_{j=1}^n\int_{t_{j-1}}^T |L\left(\frac{t_{j-1}-s}{b} \right)|\sqrt{s-t_{j-1}} d s \leq C{\DG b^{\frac{3}{2}}}\Delta^{-\frac{1}{2}}\int_{-\infty}^0 |L\left(u\right)|\sqrt{-u}d u=O(\Delta^{\frac{1}{4}}),
\]
and, applying again  (\ref{eq:X_c_bound}),   
\begin{equation}
\begin{split}
&\mathbb{E}\left|\int_{t_{j-1}}^T  L\left(\frac{t_{j-1}- s}{b} \right)\partial_{pq,\ell m}g(c_s)\tilde{c}_s^{\ell m,pq} ds- \partial_{pq,\ell m}g(c_{t_{j-1}})\tilde{c}_{t_{j-1}}^{\ell m,pq} h\int_{\frac{t_{j-1} - T}{b}}^0  L\left(u \right) du\right|\\
&\leq  \int_{t_{j-1}}^T  \left|L\left(\frac{t_{j-1}- s}{b} \right)\right| \mathbb{E}\left|\partial_{pq,\ell m}g(c_s)\tilde{c}_s^{\ell m,pq} - \partial_{pq,\ell m}g(c_{t_{j-1}})\tilde{c}_{t_{j-1}}^{\ell m,pq}\right| ds\\
&\leq  C\int_{t_{j-1}}^T  \left|L\left(\frac{t_{j-1}- s}{b} \right)\right|\sqrt{s - t_{j-1}} ds\\
 &\leq C {\DG b}\int_{-\infty}^0 |L(u)| \sqrt{-u{\DG b}} du = O({\DG b^{3/2}}),
\end{split}
\end{equation}
where $\int_{-\infty}^0 |L(u)|\sqrt{-u} du < \infty$   can be shown using the condition $K(x)=O(|x|^{-2-\epsilon})$, as $|x| \to \infty $, for some $\epsilon>1/2$. Indeed, for some $M<0$ and $B<\infty$,
\begin{equation*}
\begin{split}
 \int_{-\infty}^0 |L(u)| \sqrt{-u} du &\leq \int_M^0 |L(u)|\sqrt{-u} du + \int_{-\infty}^{M}\sqrt{-u}\int_{-\infty}^{u}|K(x)|dxdu\\
 &\leq  \int_M^0 |L(u)| \sqrt{-u} du + B\int_{-\infty}^{M}\sqrt{-u}\int_{-\infty}^{u}|x|^{-2-\epsilon}dxdu< \infty.
\end{split}
\end{equation*}
This concludes $L - {\DG \widetilde{V}^{n,\ell m}_{3}} \to \theta\sum_{pq}\int_0^T \partial_{pq,\ell m}g(c_{s})\tilde{c}_{s}^{\ell m,pq} ds \int_{-\infty}^0 L(u) du$ and, hence,  (\ref{eq:v4}).}
\subsubsection{The bias term {\DG $V^{n}_4$}} \label{last_bias_section}

We now show that
\begin{equation}\label{eq:last_bias}
    \begin{aligned} 
{\DG V^{n}_{4}} &\stackrel{\mathbb{P}}{\longrightarrow}\, {\frac{1}{2\theta} \sum_{p,q,u,v} \int_0^{T}   \partial_{pq,uv}g(c_{s}) \check{c}_s^{pq,uv} ds\int_{-\infty}^{\infty} K^2\left(u  \right) du}\\
&\qquad
+{\frac{\theta}{2}\sum_{p,q,u,v}\sum_{r=1}^{d}\int_0^{T}\partial_{pq,uv}g(c_{s})\,\tilde{\sigma}^{pq,r}_{s}\tilde{\sigma}^{uv,r}_{s} ds\left(\int_{-\infty}^{\infty} L\left(u\right)^2 du\right)}.
\end{aligned}
\end{equation}
We write ${\DG V^{n}_{4}} =\sum_{i=1}^{n} v_{i}^{n}\bar{K}_n\left(t_i\right)$, where $v_{i}^{n}=\sqrt{\Delta_{n}}\left(g\left(c_{i}^{n}+\beta_{i}^{n}\right)-g\left(c_{i}^{n}\right)-\nabla g\left(c_{i}^{n}\right) \cdot\beta_{i}^{n }\right).$
By Taylor's expansion and (\ref{eq:g_prime_bound}), for some constants ${\ell} \ge 3$, $v_{i}^{n}=v(1)_{i}^{n}+v(2)_{i}^{n}$, where 
\begin{align*}
	v(1)_i^n&=\frac{\sqrt{\Delta_n}}{2}\sum_{p,q,u,v}\partial_{pq,uv}g(c_i^n)\beta_i^{n,pq}\beta_i^{n,uv},\\
	|v(2)_i^n|&\leq{}C\sqrt{\Delta_n}(1+\|\beta_i^n\|^{{\ell}-3})\|\beta_{i}^{n}\|^3.
\end{align*}
For $v(2)^n_i $, applying (\ref{eq:beta_bound}) and recalling that $k_n \sim \frac{\theta}{\sqrt{\Delta_n} }$, 
\begin{equation}
\mathbb{E}\left[\sum_{i=1}^n |v(2)^n_i|\right] \leq \sum_{i=1}^n C\sqrt{\Delta_n}  \left( k_n^{-3/2} + (k_n \Delta_n)^{3/2}\right) \leq \sum_{i=1}^n C\sqrt{\Delta_n}\Delta_n^{3/4} \to 0.
\end{equation}
Thus,  we only need to consider $v(1)^n_i $. Recalling (\ref{Dfnalphabeta0}), write 
\begin{equation}\label{eq:beta2_decomp}
\beta_i^{n,pq}\beta_i^{n,uv}=\bar{K}\left(t_i\right)^{-2}\sum_{r=1}^{6}\xi^{n}_{{r
	},i},
\end{equation}
where 
\begin{align*}
\xi_{1,i}^{n}&:=\sum_{j=1}^{n} K_{{b}}^{2}\left(t_{j-1}-t_{i}\right)\alpha_{j}^{n,pq}\alpha_{j}^{n,uv},\\
\xi^{n}_{2,i} & = \Delta_{n}^2\left(\sum_{j=1}^{n} K_{{b}}\left(t_{j-1}-t_{i}\right) \left(c_{j-1}^{n,pq}-c_{i}^{n,pq}\right)\right)
\left(\sum_{j=1}^{n} K_{{b}}\left(t_{j-1}-t_{i}\right) \left(c_{j-1}^{n,uv}-c_{i}^{n,uv}\right)\right)\\
\xi^{n}_{3,i}  &=\sum_{j=1}^n\sum_{l=j+1}^{n} 
K_{{b}}\left(t_{j-1}-t_{i}\right) K_{{b}}\left(t_{l-1}-t_{i}\right) \alpha_{j}^{n,uv} \alpha_{l}^{n,pq},\\
\xi^{n}_{4,i}  &=\sum_{l=1}^{n-1}\sum_{j=l+1}^{n} 
K_{{b}}\left(t_{j-1}-t_{i}\right) K_{{b}}\left(t_{l-1}-t_{i}\right) \alpha_{j}^{n,uv} \alpha_{l}^{n,pq},\\
\xi^{n}_{5,i}  & =\sum_{j=1}^{n} \sum_{l=1}^{n} K_{{b}}\left(t_{j-1}-t_{i}\right) K_{{b}}\left(t_{l-1}-t_{i}\right) \alpha_{j}^{n,uv} \left(c_{l-1}^{n,pq}-c_{i}^{n,pq}\right) \Delta_{n},\\
\xi^{n}_{6,i}  & =\sum_{j=1}^{n} \sum_{l=1}^{n} K_{{b}}\left(t_{j-1}-t_{i}\right) K_{{b}}\left(t_{l-1}-t_{i}\right) \alpha_{j}^{n,pq} \left(c_{l-1}^{n,uv}-c_{i}^{n,uv}\right) \Delta_{n}.
\end{align*}
Therefore,  we only need to  consider the convergence of \begin{equation} \label{eq:part5_decomposition}
	S_{n,r}:=S_{n,r}^{(p,q,u,v)}:=\sum_{i=1}^n \frac{\sqrt{\Delta_n}}{2}\partial_{pq,uv}g(c_i^n)\bar{K}\left(t_i\right)^{-1}\xi^n_{r,i} ,   \quad r = 1, \dots, 6,
\end{equation} 
for each $p,q,u,v$, since
$$
{\DG V^{n}_{4}}=\sum_{r=1}^{6}\sum_{p,q,u,v}S_{n,r}^{(p,q,u,v)}+o_{p}(1).
$$
When ${r=1}$, 
\begin{equation}\label{eq:A.50}
\begin{split}
{S_{n,1}}
&=\sum_{j=1}^{n}\sum_{i=1}^{j-1} \frac{\sqrt{\Delta_{n}}}{2} \partial_{pq,uv}g(c_i^n) \bar{K}\left(t_i\right)^{-1} K_{{b}}^{2}\left(t_{j-1}-t_{i}\right) \alpha_{j}^{n,pq}\alpha_{j}^{n,uv}\\
&\quad+ 
\sum_{j=1}^{n} \sum_{i=j}^{n}\frac{\sqrt{\Delta_{n}}}{2} \partial_{pq,uv}g(c_i^n) \bar{K}\left(t_i\right)^{-1} K_{{b}}^{2}\left(t_{j-1}-t_{i}\right) \alpha_{j}^{n,pq}\alpha_{j}^{n,uv}\\
& =:R + L.
\end{split}
\end{equation} 
Similar to the idea {behind} (\ref{eq:left}), we define 
\begin{equation*}
M = \frac{\sqrt{\Delta_n}}{2} \sum_{j=1}^{n} \sum_{i=j}^{n} \bar{K}\left(t_i\right)^{-1} K_{{b}}^{2}\left(t_{j-1}-t_{i}\right) 
\partial_{pq,uv}g(c_{j-1}^n)  \alpha_{j}^{n,pq}\alpha_{j}^{n,uv}, 
\end{equation*} 
and aim to show that 
\begin{align}
&R +M \to {\frac{1}{2\theta}  \int_0^{T}   \partial_{pq,uv}g(c_{s}) \check{c}_s^{pq,uv} ds\int_{0}^{\infty} K^2\left(u  \right) du},
\label{eq:right5}\\
& L -M\to 0. \label{eq:left5}
\end{align}
For (\ref{eq:left5}), due to (\ref{eq:X_c_bound}), (\ref{eq:bounds_alpha}), and (\ref{eq:g_prime_bound}), we have  
\begin{align} \label{eq:L_M_v4}
\nonumber
&\mathbb{E}\left|L - M \right| = \frac{\sqrt{\Delta_n}}{2} \mathbb{E}\left| \sum_{j=1}^{n} \sum_{i=j}^{n} \bar{K}\left(t_i\right)^{-1}K_{{b}}^{2}\left(t_{j-1}-t_{i}\right) \left(\partial_{pq,uv}g(c_{i}^n)-\partial_{pq,uv}g(c_{j-1}^n)\right) \alpha_{j}^{n,pq}\alpha_{j}^{n,uv}\right|\\
\nonumber
& \leq \frac{\sqrt{\Delta_n}}{2} \sum_{j=1}^{n} \sum_{i=j}^{n} |\bar{K}\left(t_i\right)|^{-1}K_{{b}}^{2}\left(t_{j-1}-t_{i}\right) \sqrt{\mathbb{E}\left[ \left(\partial_{pq,uv}g(c_{i}^n)-\partial_{pq,uv}g(c_{j-1}^n)\right)^2\right]  \mathbb{E}\left[\left\| \alpha_{j}^{n}\right\|^4 \right]}\\
\nonumber
& \leq \frac{\sqrt{\Delta_n}}{2} \sum_{j=1}^{n} \sum_{i=j}^{n} |\bar{K}\left(t_i\right)|^{-1}K_{{b}}^{2}\left(t_{j-1}-t_{i}\right) \sqrt{\mathbb{E}\left[ \left\|\nabla \partial_{pq,uv}g\left(\xi\right) \right\|^2\left\|c_{i}^{n}-c_{j-1}^{n}\right\|^2  \right]\Delta_n^4 }\\
\nonumber
& \leq C\sqrt{\Delta_n} \sum_{j=1}^{n} \sum_{i=j}^{n} |\bar{K}\left(t_i\right)|
^{-1}K_{{b}}^{2}\left(t_{j-1}-t_{i}\right) \sqrt{\left|t_{j-1} - t_{i}\right|\Delta_n^4 }\\
& \leq \frac{C}{\delta} \sqrt{\Delta_n}\sum_{j=1}^n\Delta_n \int_{-\infty}^0 K^2(u) \sqrt{-u{b}}\, du\frac{1}{{b}} \leq \frac{C}{\delta}\sqrt{\frac{\Delta_n}{{b}}} \to 0,
\end{align} 
given $\left|\bar{K}\left(t_i\right)\right|>\delta$.
For (\ref{eq:right5}), {we write it as
\begin{align} 
R + M  
\label{eq:RandM}
& =: \frac{\sqrt{\Delta_n}}{2}\sum_{j=1}^{n} \zeta^n_{j} \alpha_{j}^{n,pq}\alpha_{j}^{n,uv},
\end{align}
where
\[
	\zeta^n_{j}:= \sum_{i=1}^{j-1} \partial_{pq,uv}g(c_{i}^n)\bar{K}\left(t_i\right)^{-1}  K_{{b}}^{2}\left(t_{j-1}-t_{i}\right) + \sum_{i=j}^{n}  \partial_{pq,uv}g(c_{j-1}^n)\bar{K}\left(t_i\right)^{-1} K_{{b}}^{2}\left(t_{j-1}-t_{i}\right)=:\zeta_{j}^{n,1}+\zeta_{j}^{n,2}.
\]
In particular, note that $R+M$ is an adapted} triangular array, since $ \zeta^n_j$ and $\alpha^n_j$ are measurable w.r.t. $\mathcal{F}_{j-1} $ and $\mathcal{F}_{j} $, respectively. By (\ref{eq:better_bound_alpha}), we have 
\begin{equation}\label{eq:K2_discrete00}
\begin{split}
&\sum_{j=1}^n \mathbb{E}\left[ \zeta^n_{j} \alpha_{j}^{n,pq}\alpha_{j}^{n,uv} \big| \mathcal{F}_{j-1}  \right]=  \sum_{j=1}^n \zeta^n_{j}\left(\check{c}_{j-1}^{n,pq,uv}\Delta_{n}^{2} + O_p\left(\Delta_n^{5/2}\right) \right),
\end{split}
\end{equation}
{\DG where recall that} $\check{c}_{s}^{pq,uv}=\left(c_{s}^{pu}c_{s}^{qv} +c_{s}^{pv}c_{s}^{qu}\right)$.
{We now consider the limits of 
\begin{equation}\label{SlEqN}
	A_{n}:=\sqrt{\Delta_n}\sum_{j=1}^{n}\zeta_{j}^{n,1}\check{c}_{j-1}^{n,pq,uv}\Delta_{n}^{2},
	\quad 
	B_{n}:=\sqrt{\Delta_n}\sum_{j=1}^{n}\zeta_{j}^{n,2}\check{c}_{j-1}^{n,pq,uv}\Delta_{n}^{2}.
\end{equation}
For the first term above, let us first note that we expect that}
\begin{align}\nonumber
\zeta_{j}^{n,1}\Delta_n&=\sum_{i=1}^{j-1}  K_{{b}}^{2}\left(t_{j-1}-t_{i}\right)\bar{K}\left(t_{i}\right)^{-1} \partial_{pq,uv}g(c_{i}^n) \Delta_n\\
&=\int_0^{t_{j-1}} K_{{b}}^2 \left(t_{j-1}-s\right)\bar{K}\left(s\right)^{-1}\partial_{pq,uv}g(c_{s})  ds  +O\left(\frac{\Delta_n}{b_n^2} \right) + O_p\left(\frac{\sqrt{\Delta_n}}{b_n}\right).
 \label{eq:K2_discrete}
\end{align}
The proof of (\ref{eq:K2_discrete}) can be found in Appendix \ref{Technical2_1}.
Therefore, recalling that $\frac{\sqrt{\Delta_n}}{{b_n}}\to \frac{1}{\theta} $, 
\begin{align} \label{eq:leftP1_V4}
\nonumber
A_{n}&= \sqrt{\Delta_n} \sum_{j=1}^{n} \check{c}^{n,pq,uv}_{j-1}\int_0^{t_{j-1}} K_{{b}}^2 \left(t_{j-1}-s\right)\bar{K}\left(s\right)^{-1}\partial_{pq,uv}g(c_{s})  ds\Delta_n + O_p(\sqrt{\Delta_n})\\
\nonumber
&= {\DR\sqrt{\Delta_n} \sum_{j=1}^{n} \Delta_n\check{c}^{n,pq,uv}_{j-1}\partial_{pq,uv}g(c_{j-1}^n) \int_0^{{\frac{t_{j-1}}{b}}}K^2(u)\left[\frac{\Delta_{n}}{b}\sum_{l=1}^{n}K\left(u+\frac{t_{l-1}-t_{j-1}}{b}\right)\right]^{-1}du \frac{1}{b}} \\
\nonumber
&\qquad+ O_p(\sqrt{\Delta_n})+O_p\left(\frac{\sqrt{\Delta_n}}{\sqrt{b_n}}\right)\\
\nonumber
&= {\frac{\sqrt{\Delta_n}}{{b}}}\int_0^{T}  \partial_{pq,uv}g(c_{s}) \check{c}_s^{pq,uv} \int_{0}^{\frac{s}{{b}}} K^2\left(u  \right)\left[\int_{u-\frac{{s}}{{b}}}^{u+\frac{T-{s}}{{b}}}K(v)dv\right]^{-1}du  ds +o_p(1) \\
&\quad\stackrel{n\to\infty}{\longrightarrow} \frac{1}{\theta}  \int_0^{T}   \partial_{pq,uv}g(c_{s}) \check{c}_s^{pq,uv} ds\int_{0}^{\infty} K^2\left(u  \right) du,
\end{align}
where in the second equality we used that, for $j= 1,2, \cdots, n$,
\begin{equation}\label{secndEq}
\begin{split}
   & \mathbb{E}\left|\int_0^{t_{j-1}} K_{{b}}^2 \left(t_{j-1}-s\right)\bar{K}(s)^{-1}\partial_{pq,uv}g(c_{s})ds - 
   \partial_{pq,uv}g(c_{j-1}^n) \int_0^{t_{j-1}} K_{{b}}^2 \left(t_{j-1}-s\right)\bar{K}(s)^{-1}ds \right|\\
    &\quad\leq  \int_0^{t_{j-1}} K_{{b}}^2 \left(t_{j-1}-s\right) \bar{K}(s)^{-1}\mathbb{E} \left|\partial_{pq,uv}g(c_{s})- \partial_{pq,uv}g(c_{j-1}^n) \right|ds \\
    &\quad\leq C  \int_0^{t_{j-1}} K_{{b}}^2\left( t_{j-1}-s\right)\bar{K}(s)^{-1}\sqrt{\left|s - t_{j-1}\right|} ds  \\
    &\quad \leq \frac{C}{\delta} \frac{1}{b}\sqrt{b} \int_0^{\infty} K^2(u) \sqrt{u} du  =  O\left(\frac{1}{\sqrt{b}}\right),
\end{split}
\end{equation}
where we used \eqref{eq:X_c_bound} and \eqref{eq:g_prime_bound}. 
{For the second term in (\ref{SlEqN}),} similar to (\ref{eq:K2_discrete}), we have 
\[
\sum_{i=j}^{n}  K_{{b}}^{2}\left(t_{j-1}-t_{i}\right)  \bar{K}\left(t_i\right)^{-1}\Delta_n -\int_{t_{j-1}}^{T} K_{{b}}^2 \left(t_{j-1}-t\right)\bar{K}\left(t\right)^{-1}  dt = O_p(1),
\] 
and, thus,
\begin{align}\label{eq:rightP1_V4}
\nonumber
{B_n}
&= \sqrt{\Delta_n} \sum_{j=1}^{n} \partial_{pq,uv}g(c_{j-1}^n)  \check{c}^{n,pq,uv}_{j-1} \int_{t_{j-1}}^{T} K_{{b}}^2 \left(t_{j-1}-s\right)\bar{K}\left(s\right)^{-1}  ds\Delta_n + o_p(1)\\
\nonumber
 &= \sqrt{\Delta_n} \sum_{j=1}^{n}  \partial_{pq,uv}g(c_{j-1}^n)  \check{c}^{n,pq,uv}_{j-1} \int_{\frac{t_{j-1}-T}{b}}^{0} K^2\left(u\right)\left[\frac{\Delta_n}{b} \sum_{l=1}^n K\left(u+\frac{t_{l-1}-t_{j-1}}{b}\right)\right]^{-1}  du\frac{\Delta_n}{b} + o_p(1)\\
 \nonumber
 &=\frac{\sqrt{\Delta_n}}{b}\int_0^{T}  \partial_{pq,uv}g(c_{s}) \check{c}_s^{pq,uv}\int_{\frac{s-T}{b}}^{0}K^2\left(u  \right)\left[\int_{u-\frac{s}{b}}^{u+\frac{T-s}{b}}K(v)dv\right]^{-1}du  ds +o_p(1)\\
 &\quad\stackrel{n\to\infty}{\longrightarrow} \frac{1}{\theta}  \int_0^{T}    \partial_{pq,uv}g(c_{s})  \check{c}^{pq,uv}_{s} ds\int_{-\infty}^{0} K^2\left(u  \right) du.
\end{align}
For the contribution of the second term on the right-hand side of (\ref{eq:K2_discrete00}), 
note that 
\begin{align*}
	|\zeta^n_j| \leq C\sum_{i=1}^{n}  K_{{b}}^{2}\left(t_{j-1}-t_{i}\right)\bar{K}\left(t_i\right)^{-1}&\leq \frac{C}{\delta}
	\frac{1}{b_n\Delta_{n}}\int_{-\infty}^{\infty} K^{2}\left(u\right)du.
\end{align*} 
Therefore, 
\begin{equation}\label{eq:A.62}
 \sqrt{\Delta_n}O\left(\Delta_n^{5/2}\right)\sum_{j=1}^n |\zeta^n_{j}|  \leq \sqrt{\Delta_n}O\left(\Delta_n^{5/2}\right)\frac{1}{b_n\Delta_n^2}\leq O_p\left( \frac{\Delta_n}{b_n}\right)=o_p(1).
\end{equation} 
Finally, again using $|\zeta_{j}^n|=O_p(b_n^{-1}\Delta_n^{-1})$ and (\ref{eq:bounds_alpha}), we can easily show that 
\begin{align}\label{ScndMntH}
	\Delta_n\sum_{j=1}^n \mathbb{E}\left[ (\zeta^n_{j} \alpha_{j}^{n,pq}\alpha_{j}^{n,uv})^2 \big| \mathcal{F}_{j-1}  \right]&\leq {C}
	\Delta_n|\zeta^n_{j}|^2\sum_{j=1}^n \mathbb{E}\left[ \|\alpha_{j}^{n}\|^4 \big| \mathcal{F}_{j-1}  \right]=O_p\left(\frac{\Delta_n^5}{b^2 \Delta_n^3}\right)=o_p(1).
\end{align}
Combining with (\ref{eq:L_M_v4}), (\ref{eq:RandM}), (\ref{eq:leftP1_V4}), (\ref{eq:rightP1_V4}), and (\ref{ScndMntH}), we conclude that
\begin{equation}\label{eq:A.64}
    S_{n,1}=\sum_{i=1}^n \frac{\sqrt{\Delta_n}}{2}\partial_{pq,uv}g(c_i^n)\bar{K}\left(t_i\right)^{-1} \xi^n_{1,i}
    \stackrel{n\to\infty}{\longrightarrow}  \frac{1}{2\theta}  \int_0^{T}   \partial_{pq,uv}g(c_{s}) \check{c}_s^{pq,uv} ds
    \int_{-\infty}^{\infty} K^2\left(u  \right) du.
\end{equation}

For ${r=2}$ in (\ref{eq:part5_decomposition}), first note that, by Lemma \ref{K_discrete},
\begin{align}\label{RFEANa}\nonumber
    &\sum_{j=1}^n K_b\left(t_{j-1}-t_i\right)\left(c_{j-1}^{n, p q}-c_i^{n, p q}\right) \Delta_n-\int_0^T K_b\left(s-t_i\right)\left(c_s^{n, p q}-c_i^{n, p q}\right) d s\\
    &=\frac{1}{2}\left(c_{i-1}^{n,pq}-c_i^{n, pq}\right){\DG \Delta K}\frac{\Delta_n}{b}+o\left(\frac{\Delta_n}{b}\right) + O_p\left(\Delta_n^{1 / 2}\right),
\end{align}
{\DG where $\Delta K:=\left(K\left(A^{+}\right)-K\left(B^{-}\right)\right)$}.
Thus,
\begin{align}
\nonumber
  S_{n,2}
\nonumber
    &=\sum_{i=1}^n \frac{\Delta_n^{\frac{5}{2}}}{2} \partial^2_{p q, u v} g\left(c_i^n\right) \bar{K}\left(t_i\right)^{-1}\sum_{j=1}^n K_b\left(t_{j-1}-t_i\right)\left(c_{j-1}^{n, p q}-c_i^{n, p q}\right)\sum_{j=1}^n K_b\left(t_{j-1}-t_i\right)\left(c_{j-1}^{n, uv}-c_i^{n, uv}\right)\\
    \nonumber
    &=\sum_{i=1}^n \frac{\sqrt{\Delta_n}}{2} \partial^2_{p q, u v} g\left(c_i^n\right) \bar{K}\left(t_i\right)^{-1}\\
    \nonumber
    &\qquad\times\Bigg[\int_0^T K_b\left(s-t_i\right)\left(c_s^{{p q}}-c_i^{n, p q}\right) d s+\frac{1}{2}\left(c_{i-1}^{n,pq}-c_i^{n, pq}\right){\DG \Delta K}\frac{\Delta_n}{b}+o\left(\frac{\Delta_n}{b}\right) + O_p\left(\Delta_n^{1 / 2}\right)\Bigg]\\
    \nonumber
    &\quad\times\left[\int_0^T K_b\left(s-t_i\right)\left(c_s^{{uv}}-c_i^{n, uv}\right) d s+\frac{1}{2}\left(c_{i-1}^{n,uv}-c_i^{n, uv}\right){\DG \Delta K} \frac{\Delta_n}{b}+o\left(\frac{\Delta_n}{b}\right) + O_p\left(\Delta_n^{1 / 2}\right)\right]\\
    \nonumber
    &= \sum_{i = 1}^n \frac{\sqrt{\Delta_n}}{2} \partial_{pq,uv}g(c_i^n)\bar{K}\left(t_i\right)^{-1} \left(\int_{0}^T K_{{b}}\left(t - t_i \right) \left(c_t^{pq} - c_i^{n,pq} \right)dt \right)\left(\int_{0}^T K_{{b}}\left(t - t_i \right) \left(c_t^{uv} - c_i^{n,uv} \right)dt \right) \\
    \nonumber
    & + \sum_{i=1}^n \Delta_n \int_0^T K_b\left(t-t_i\right)\left(c_t^{p q}-c_i^{n, p q}\right) d t \frac{\sqrt{\Delta_n}}{4b} \partial^2_{p q, u v} g\left(c_i^n\right) \bar{K}\left(t_i\right)^{-1}\left(c_{i-1}^{n, u v}-c_i^{n, u v}\right)
    {\DG \Delta K} \\
    \nonumber
     & + \sum_{i=1}^n \Delta_n \int_0^T K_b\left(t-t_i\right)\left(c_t^{uv}-c_i^{n, uv}\right) d t \frac{\sqrt{\Delta_n}}{4b} \partial^2_{p q, u v} g\left(c_i^n\right) \bar{K}\left(t_i\right)^{-1}\left(c_{i-1}^{n, pq}-c_i^{n, pq}\right){\DG \Delta K}\\
    \nonumber
    &+o_{p}(1)\\
    \label{RFEAN}
&=: P_1 + P_2+P_3+ o_p(1),
\end{align}
where
\begin{equation}\label{p2p3}
\begin{split}
    P_2&=O\left(\left|\sum_{i = 1}^n \Delta_n\int_{0}^T K_{{b}}\left(t - t_i \right) \left(c_t^{pq} - c_i^{n,pq} \right)dt  \right|\right),\\
    P_3&=O\left(\left|\sum_{i = 1}^n \Delta_n\int_{0}^T K_{{b}}\left(t - t_i \right) \left(c_t^{uv} - c_i^{n,uv} \right)dt  \right|\right).
\end{split}
\end{equation}
The above estimate follows from the boundedness of $K$, (\ref{eq:g_prime_bound}), (\ref{eq:localization}), (\ref{lowerboundofdelta}), and the convergence of $\frac{\sqrt{\Delta_n}}{b}$, which is $\frac{1}{\theta}$ according to (\ref{eq: k_n}). By integration by parts, 
\begin{align} \nonumber
& \int_{0}^T K_{{b}}\left(t - t_i \right) \left(c_t^{pq} - c_i^{n,pq} \right)dt\\
\nonumber
&\quad= -L\left(\frac{T-t_i}{{b}}\right)\left(c_T^{pq} - c_{t_i}^{pq} \right) + L\left(\frac{-t_i}{{b}}\right)\left(c_0^{pq} - c_{t_i}^{pq} \right) + \int_0^T L\left(\frac{t-t_i}{{b}} \right)\tilde{\sigma}^{pq}_t dW_t + o_p({b}^{\frac{1}{2}})\\
&\quad=:  R^{pq}_{1,i} + R^{pq}_{2,i} +\int_0^T L\left(\frac{t-t_i}{b} \right)\tilde{\sigma}^{pq}_t dW_t+ o_p({b}^{\frac{1}{2}}).
\label{eq:int_by_parts_decomp}
\end{align}
Now we consider each term in $P_1$. By (\ref{eq:X_c_bound}),
\begin{equation} \label{eq:R1}
\begin{split}
&\mathbb{E}\left|\sum_{i=1}^n \frac{\sqrt{\Delta_n}}{2} \partial_{pq,uv}g(c_i^n)\bar{K}\left(t_i\right)^{-1} R_{1,i}^{pq} R_{1,i}^{uv}  \right|\leq  \frac{C}{\delta}\sum_{i=1}^n \sqrt{\Delta_n}\sqrt{\mathbb{E}[(R_{1,i}^{pq})^2]\mathbb{E}[(R_{1,i}^{uv})^2]} \\
&\quad \leq \frac{C}{\delta}\sum_{i=1}^n \sqrt{\Delta_n}L\left(\frac{T-t_i}{b}\right)^2\left(T - {t_i} \right)\leq  \frac{C}{\delta}\frac{{b}^{2}}{\sqrt{\Delta_n}}\int_0^{\infty}L(u)^2 udu \stackrel{n\to\infty}{\longrightarrow}  0,
\end{split}
\end{equation}
where above we used that $\int_0^{\infty} L(u)^2 u du < \infty$, which can be deduced from the condition $K(x)x^{2+\epsilon} \to 0$, as $|x| \to \infty$, for $\epsilon>1/2$.
In the same way, we can show that 
\begin{equation} \label{eq:R2}
\mathbb{E}\left|\sum_{i=1}^n \frac{\sqrt{\Delta_n}}{2} \partial_{pq,uv}g(c_i^n)\bar{K}\left(t_i\right)^{-1} R_{2,i}^{pq} R_{2,i}^{uv}  \right| 
\leq 
\frac{C}{\delta}\frac{{b}^{2}}{\sqrt{\Delta_n}}\int_{-\infty}^0 L(u)^2 |u|du \stackrel{n\to\infty}{\longrightarrow}  0.
\end{equation}
For the mixed term $|R_{1,i}^{pq}R_{2,i}^{uv}|$, we simply upper bound it by $2^{-1}(|R_{1,i}^{pq}|^2+|R_{2,i}^{uv}|^2)$ and proceed as above. 
{For the term of $P_1$ corresponding to the third term in \eqref{eq:int_by_parts_decomp}, note that,}  using that $\tilde{\sigma}$ and $c$ are bounded, (\ref{eq:g_prime_bound}), and BDG inequality, for  $t_i <\sqrt{b} $,
\begin{equation*}
\begin{split}
&\mathbb{E}\left|\sum_{t_i < \sqrt{b}} \frac{\sqrt{\Delta_n}}{2} \partial_{pq,uv}g(c_i^n)\bar{K}\left(t_i\right)^{-1}  \left(\int_0^T L\left(\frac{t-t_i}{b} \right)\tilde{\sigma}^{pq}_t dW_t\right)\left(\int_0^T L\left(\frac{t-t_i}{b} \right)\tilde{\sigma}^{uv}_t dW_t\right)\right|\\  
&\leq \frac{C}{\delta} \sqrt{\Delta_n}\sum_{t_i < \sqrt{b}} \sqrt{\mathbb{E}\left(\int_0^T L\left(\frac{t-t_i}{b} \right)\tilde{\sigma}^{pq}_t d W_t \right)^2\mathbb{E}\left(\int_0^T L\left(\frac{t-t_i}{b} \right)\tilde{\sigma}^{pq}_t d W_t \right)^2}\\
& \leq \frac{C}{\delta} \sum_{t_i < \sqrt{b}} \sqrt{\Delta_n}\int_0^T L\left(\frac{t-t_i}{b} \right)^2 dt\\
&\leq  \frac{C}{\delta} \sum_{t_i < \sqrt{b}} \sqrt{\Delta_n}\int L\left(u\right)^2 du b\leq \frac{C}{\delta}\frac{b\sqrt{\Delta_n}\sqrt{b}}{\Delta_n}\leq\frac{C}{\delta}\sqrt{b} \to 0.
\end{split}
\end{equation*}
For $t_i > \sqrt{b}$, similar to the proof of Theorem 6.2 in \cite{FigLi}, $\int_0^T L\left(\frac{t-t_i}{b} \right)\tilde{\sigma}^{pq}_t dW_t=\sum_{r=1}^{d}\tilde{\sigma}^{pq,r}_{t_i - \sqrt{b}}\int_{t_i - \sqrt{b}}^T L\left(\frac{t-t_i}{b} \right) dW^{r}_t + o_p({b}^{\frac{1}{2}})$ and, thus, for any $r_1,r_2\in\{1,\dots,d\}$, it suffices to consider the convergence of  
\begin{align} \nonumber
&\sum_{i:t_i > \sqrt{b}} \frac{\sqrt{\Delta_n}}{2} \partial_{pq,uv}g(c_i^n)\bar{K}\left(t_i\right)^{-1}\tilde{\sigma}^{pq,r_1}_{t_i - \sqrt{b}}\,\tilde{\sigma}^{uv,r_2}_{t_i - \sqrt{b}}\int_{t_i - \sqrt{b}}^T L\left(\frac{t-t_i}{b} \right) dW^{r_1}_t
\int_{t_i - \sqrt{b}}^T L\left(\frac{t-t_i}{b} \right) dW^{r_2}_t\\  
&\quad =:
\sum_{i:t_i > \sqrt{b}} \zeta^{n,r_1,r_2}_i+o_{p}(1),
\label{eq:main_p2_v4}
\end{align}
where 
\[
	\zeta^{n,r_1,r_2}_i:=\frac{\sqrt{\Delta_n}}{2}  \partial_{pq,uv}g(c_{t_i-\sqrt{b}})\bar{K}\left(t_i\right)^{-1}\,\tilde{\sigma}^{pq,r_1}_{t_i - \sqrt{b}}\tilde{\sigma}^{uv,r_2}_{t_i - \sqrt{b}}\int_{t_i - \sqrt{b}}^T L\left(\frac{t-t_i}{b} \right) dW^{r_1}_t
\int_{t_i - \sqrt{b}}^T L\left(\frac{t-t_i}{b} \right) dW^{r_2}_t.
\]
For simplicity, we write $\zeta^{n}_i$ instead of $\zeta^{n,r_1,r_2}_i$. For the convergence of (\ref{eq:main_p2_v4}), note that $\mathbb{E}\left[\left.\zeta^n_i\right| \mathcal{F}_{t_i - \sqrt{b}}\right]=0$ {if}  $r_1\neq r_2$, {but, when $r_2=r_1$,} 
\begin{equation*}
\begin{split}
&\sum_{t_i > \sqrt{b}} \mathbb{E}\left[\left.\zeta^n_i\right| \mathcal{F}_{t_i - \sqrt{b}}\right]= \sum_{t_i > \sqrt{b}}
\frac{\sqrt{\Delta_n}}{2}  \partial_{pq,uv}g(c_{t_i-\sqrt{b}})\bar{K}\left(t_i\right)^{-1}\,\tilde{\sigma}^{pq,r_1}_{t_i - \sqrt{b}}\tilde{\sigma}^{uv,r_1}_{t_i - \sqrt{b}} \int_{t_i-\sqrt{b}}^{T} L\left(\frac{t-t_i}{b} \right)^2 dt \\
 &= \sum_{t_i > \sqrt{b}} \frac{\sqrt{\Delta_n}}{2}  \partial_{pq,uv}g(c_{t_i-\sqrt{b}})\bar{K}\left(t_i\right)^{-1}\,\tilde{\sigma}^{pq,r_1}_{t_i - \sqrt{b}}\tilde{\sigma}^{uv,r_1}_{t_i - \sqrt{b}} \int_{-\frac{1}{\sqrt{b}}}^{\frac{T-t_i}{{b}}} L\left(u\right)^2 du{b} \\
 &= \int_0^{T-\sqrt{b}}\frac{\theta}{2}\partial_{pq,uv}g(c_{s})\left(\int_{-\frac{s+\sqrt{b}}{{b}}}^{\frac{T-s-\sqrt{b}}{{b}}}K(v)dv\right)^{-1}\,\tilde{\sigma}^{pq,r_1}_{s}\tilde{\sigma}^{uv,r_1}_{s} \left(\int_{-\frac{1}{\sqrt{b}}}^{\frac{T-s-\sqrt{b}}{{b}}} L\left(u\right)^2 du\right) ds + o_p(1)\\
& \stackrel{n\to\infty}{\longrightarrow}\, \frac{\theta}{2}\int_0^{T}\partial_{pq,uv}g(c_{s})\,\tilde{\sigma}^{pq,r_1}_{s}\tilde{\sigma}^{uv,r_1}_{s} ds\left(\int_{-\infty}^{\infty} L\left(u\right)^2 du\right),
\end{split}
\end{equation*}
and
\begin{equation*}
\begin{split}
\mathbb{E} \left( \sum_{t_i > \sqrt{b}} \left(\zeta^n_i- \mathbb{E}\left[\zeta^n_i\mid \mathcal{F}_{t_i - \sqrt{b}}\right]  \right)\right)^2&=  \sum_{t_i > \sqrt{b}} \mathbb{E} \left(\zeta^n_i - \mathbb{E}\left[\zeta^n_i \mid \mathcal{F}_{t_i - \sqrt{b}}\right] \right)^2\\
 &\quad  \leq C\sum_{t_i > \sqrt{b}} \mathbb{E} \left[\mathbb{E}\left[\left(\zeta^n_i\right)^2 \mid \mathcal{F}_{t_i - \sqrt{b}}\right] \right]\\
&\quad \leq C \sum_{t_i > \sqrt{b}}\Delta_n \left(\int_{t_i - \sqrt{b}}^T L\left(\frac{t-t_i}{b} \right)^2 dt\right)^2\\
&\quad \leq  C \sum_{t_i > \sqrt{b}}\Delta_n \left(\int_{-\infty}^{\infty} L\left(u\right)^2 du\right)^2 {b}^{2}\stackrel{n\to\infty}{\longrightarrow}  0.
\end{split}
\end{equation*}
Therefore, we obtain that, as $n\to\infty$, 
\begin{equation}
\sum_{t_i > \sqrt{b}}\zeta^{n,r_1,r_2}_i\stackrel{\mathbb{P}}{\longrightarrow}\frac{\theta}{2}\int_0^{T}\partial_{pq,uv}g(c_{s})\,\tilde{\sigma}^{pq,r_1}_{s}\tilde{\sigma}^{uv,r_1}_{s} ds\left(\int_{-\infty}^{\infty} L\left(u\right)^2 du\right) {\bf 1}_{r_1=r_2}.
\end{equation} 
Therefore, combining the previous inequalities, we obtain \begin{equation}
P_1\stackrel{\mathbb{P}}{\longrightarrow}\,\frac{\theta}{2}\sum_{r=1}^{d}\int_0^{T}\partial_{pq,uv}g(c_{s})\,\tilde{\sigma}^{pq,r}_{s}\tilde{\sigma}^{uv,r}_{s} ds\left(\int_{-\infty}^{\infty} L\left(u\right)^2 du\right).
\end{equation}
As for $P_2$ and $P_3$, by (\ref{p2p3}) and (\ref{eq:int_by_parts_decomp}), since
\begin{align*}
    \mathbb{E}\left|\sum_{i=1}^n \Delta_n R_{1, i}^{p q}\right| &\leq \sum_{i=1}^n \Delta_n \mathbb{E}\left|L\left(\frac{T-t_i}{b}\right)\left(c_T^{p q}-c_{t_i}^{p q}\right)\right| \leq \sum_{i=1}^n \Delta_n\left|L\left(\frac{T-t_i}{b}\right)\right|\left(T-t_i\right)^{1/2}\\
    &=\int_0^{\infty}|L(u)| b \sqrt{b u} d u=b^{3 / 2} \int_0^{\infty}|L (u)| \sqrt{u} d u \rightarrow 0,
\end{align*}
and
\begin{align*}
    &\mathbb{E}\left|\sum_{i=1}^n \Delta_n \int_0^T L\left(\frac{t-t_i}{b}\right) \tilde{\sigma}_t^{p q} d W_t\right| \leq \sum_{i=1}^n \Delta_n \sqrt{\mathbb{E}\left(\int_0^T L\left(\frac{t-t_i}{b}\right) \tilde{\sigma}_t^{p q} d W_t\right)^2} \\
    &\leq \sum_{i=1}^n \Delta_n \sqrt{\int_0^T L\left(\frac{t-t_i}{b}\right)^2 d t}\leq \sum_{i=1}^n \Delta_n \sqrt{\int_{-\infty}^{\infty} L(u)^2 d u b} \rightarrow 0,
\end{align*}
we have $P_2 \to 0$. Along the same steps, $P_3 \to 0$ as well. Thus, as $n\to\infty$,
\begin{equation}\label{eq:A.73}
S_{n,2}\stackrel{\mathbb{P}}{\longrightarrow}\,\frac{\theta}{2}\sum_{r=1}^{d}\int_0^{T}\partial_{pq,uv}g(c_{s})\,\tilde{\sigma}^{pq,r}_{s}\tilde{\sigma}^{uv,r}_{s} ds\left(\int_{-\infty}^{\infty} L\left(u\right)^2 du\right).
\end{equation}

For {$r = 3$} in (\ref{eq:part5_decomposition}), we consider the decomposition 
\begin{equation}\label{TrmRLCrs}
\begin{split}
  S_{n,3}
   & =\sum_{i=1}^n \frac{\sqrt{\Delta_n}}{2}\partial_{pq,uv}g(c_i^n)\bar{K}\left(t_i\right)^{-1}\sum_{j=i+1}^n\sum_{l=j+1}^n K_{{b}}\left(t_{j-1} - t_{i} \right)K_{{b}}\left(t_{l-1} - t_{i} \right)\alpha^{n,uv}_j \alpha^{n,pq}_l\\
   & \quad+  \sum_{i=1}^n \frac{\sqrt{\Delta_n}}{2}\partial_{pq,uv}g(c_i^n)\bar{K}\left(t_i\right)^{-1}\sum_{j=1}^{i}\sum_{l=j+1}^{i} K_{{b}}\left(t_{j-1} - t_{i} \right)K_{{b}}\left(t_{l-1} - t_{i} \right)\alpha^{n,uv}_j \alpha^{n,pq}_l\\
   & \quad+ \sum_{i=1}^n \frac{\sqrt{\Delta_n}}{2}\partial_{pq,uv}g(c_i^n)\bar{K}\left(t_i\right)^{-1}\sum_{j=1}^{i}\sum_{l=i+1}^{n} K_{{b}}\left(t_{j-1} - t_{i} \right)K_{{b}}\left(t_{l-1} - t_{i} \right)\alpha^{n,uv}_j \alpha^{n,pq}_l\\
&  =:  R + L + Cross.
\end{split}
\end{equation}
We can write 
\begin{equation*}
\begin{split}
    R  =  \sum_{l=3}^n\alpha^{n,pq}_l\sum_{j = 2}^{l-1}\sum_{i=1}^{j-1} \frac{\sqrt{\Delta_n}}{2}\partial_{pq,uv}g(c_i^n)\bar{K}\left(t_i\right)^{-1}K_{{b}}\left(t_{j-1} - t_{i} \right)K_{{b}}\left(t_{l-1} - t_{i} \right)\alpha^{n,uv}_j=:  \sum_{l= 3}^n \alpha^{n,pq}_l \zeta^n_l,
\end{split}
\end{equation*}
which is now an adapted triangular array where $ \zeta^n_l$ is $\mathcal{F}^n_{l-1} $ measurable. 
Then, by (\ref{eq:bounds_alpha}), (\ref{eq:better_bound_alpha}), and (\ref{eq:AN}),
\begin{equation*}
\begin{split}
& \mathbb{E}\left|   \sum_{l=3}^n \mathbb{E}\left[\alpha^{n,pq}_l \zeta^{n}_l\mid \mathcal{F}^n_{l-1} \right] \right| \leq C\sum_{l=3}^n  \mathbb{E}\left[|\zeta^n_l|  \Delta_n^{3/2}\left(\sqrt{\Delta_n} + \eta^n_{l-1}\right)\right]\\
  &   \leq \frac{C}{\delta} \sum_{l=3}^n \Delta_n^{3/2}\sum_{j = 2}^{l-1}\sum_{i=1}^{j-1}\sqrt{\Delta_n}\left|K_{{b}}\left(t_{j-1} - t_{i} \right)K_{{b}}\left(t_{l-1} - t_{i} \right)\right|\sqrt{\mathbb{E}\left(\alpha^{n,uv}_j\right)^2 \mathbb{E}\left(\eta^n_{l-1}\right)^2 }\\
 &    \leq \frac{C}{\delta} \Delta_n \sum_{l=3}^n\sqrt{\mathbb{E}\left(\eta^n_{l-1}\right)^2} \left(\int_0^{\infty}|K(x)|dx\right)^2\\
 &   \leq  \frac{C}{\delta} \sqrt{\sum_{l=3}^n\Delta_n\mathbb{E}\left(\eta^n_{l-1}\right)^2 } \left(\int_0^{\infty}|K(x)|dx\right)^2\stackrel{n\to\infty}{\longrightarrow}0.
\end{split}
\end{equation*}
We also have, for $j < j'$, 
\begin{equation}\label{alphajjprime}
\begin{split}
\left|\mathbb{E}\left(\alpha^{n,uv}_j\alpha^{n,uv}_{j'}\right)\right| &= \left|\mathbb{E}\left[ \alpha^{n,uv}_j\mathbb{E}\left(\alpha^{n,uv}_{j'}\mid \mathcal{F}^n_{j'-1}\right)\right] \right|
\leq  \sqrt{\mathbb{E}\left(\alpha^{n,uv}_j\right)^2\mathbb{E}\left( \mathbb{E}\left(\alpha^{n,uv}_{j'}\mid \mathcal{F}^n_{j'-1}\right)^2\right)}\\
 &\leq C \sqrt{\Delta_n^2 \Delta_n^3\mathbb{E}\left(\sqrt{\Delta_n} + \eta^n_{j'-1} \right)^2}\leq  C  \Delta_n^{5/2}.
\end{split}
\end{equation} 
Using this inequality and (\ref{eq:bounds_alpha}), we obtain \begin{equation} \label{eq:alpha_cross_R}
    \begin{split}
& \mathbb{E}\left|   \sum_{l=3}^n \mathbb{E}\left[\left(\alpha^{n,pq}_l \zeta^n_l\right)^2\mid \mathcal{F}^n_{l-1} \right] \right|
\leq   C\Delta_n^2  \sum_{l=3}^n  \mathbb{E}[\left(\zeta^n_l\right)^2]\\
&\leq  C\Delta_n^2  \sum_{l=3}^n \sum_{j=2}^{l-1}\left (\sum_{i=1}^{j-1}\sqrt{\Delta_n}\left|K_{{b}}\left(t_{j-1} - t_{i} \right)K_{{b}}\left(t_{l-1} - t_{i} \right)\right|\right)^2 \mathbb{E}\left[\left(\alpha^{n,uv}_j\right)^2\right] \\
&\quad + C\Delta_n^2  \sum_{l=3}^n \left (\sum_{j=2}^{l-1}\sum_{i=1}^{j-1}\sqrt{\Delta_n}\left|K_{{b}}\left(t_{j-1} - t_{i} \right)K_{{b}}\left(t_{l-1} - t_{i} \right)\right|\right)^2\max_{j,j': j\neq j'} \left| \mathbb{E}\left(\alpha^{n,uv}_j\alpha^{n,uv}_{j'}\right)\right|\\
 & \leq C\Delta_n b\sum_{l=3}^n \int_0^{\infty}\left( \int_0^{\infty}\left|K(u) K(v+u)  \right|du\right)^2 dv\\
& \quad + C \sum_{l=3}^n\Delta_n^{3/2}\left( \int_0^{\infty} \int_0^{\infty}\left|K(u) K(v+u)\right|dudv  \right)^2\stackrel{n\to\infty}{\longrightarrow}0,
\end{split}
\end{equation}
where the finiteness of the integrals above can be shown using the condition $K(x)x^{2+\epsilon} \to 0$,  for $\epsilon> 1/2$ , as $|x| \to \infty$ and boundedness of $K$. 
Hence, by Lemma 2.2.12 in \cite{JacodProtter}, $R=  \sum_{l= 2}^n \alpha^{n,pq}_l \zeta^n_l \stackrel{\mathbb{P}}{\longrightarrow} 0$.
Similarly, we can write the term $L$ in (\ref{TrmRLCrs}) as
\begin{equation*}
\begin{split}
L &=  \sum_{j = 1}^{n}\sum_{l=j+1}^{n} \sum_{i = l}^{n}\sqrt{\Delta_n}  \partial^2_{pq,uv}g(c_i^n)\bar{K}\left(t_i\right)^{-1}K_{{b}}\left(t_{j-1} - t_{i} \right)K_{{b}}\left(t_{l-1} - t_{i} \right)\alpha^{n,uv}_j \alpha^{n,pq}_l =: \sum_{j=1}^{n} \zeta^n_j.
\end{split}
\end{equation*} 
Note that (see proof in Appendix \ref{Technical2_1})
\begin{equation} \label{eq:alphapartialg}
\left|\mathbb{E}\left[\left.\alpha_l^{n,pq} \partial_{p q, u v} g\left(c_i^n\right) \right| \mathcal{F}_{l-1}^n\right]\right| \leq C \Delta_n \eta_{l-1, i-l+1}^n\left(t_i-t_{l-1}\right), \text{ for } i \geq l,
\end{equation} 
with $\eta_{l-1, i-l+1}^n=\max \left(\eta(Y)_{l-1, i-l+1}^n , Y=\mu, \tilde{\mu}, c, \tilde{c}, \hat{c}\right).$
Now consider  
\begin{equation} \label{eq:eta_sum}
\begin{split}
&\mathbb{E}\left|\sum_{j=1}^n \mathbb{E}\left[\left.\zeta_j^n \right| \mathcal{F}_{j-1}^n\right]\right|\\
    &=\mathbb{E}\left|\sum_{j=1}^n \sum_{l=j+1}^n \sum_{i=l}^n \sqrt{\Delta_n} {\DG \bar{K}\left(t_i\right)^{-1}} K_b\left(t_{j-1}-t_i\right) K_b\left(t_{l-1}-t_i\right) \mathbb{E}\left[\left.\alpha_j^{n, u v} \mathbb{E}\left[\left.\alpha_l^{n, p q} \partial_{p q, u v} g\left(c_i^n\right) \right| \mathcal{F}_{l-1}^n\right] \right| \mathcal{F}_{j-1}^n\right]\right|\\
    &\leq C \sum_{i=1}^n \sum_{j=1}^{l-1} \sum_{l=2}^i \sqrt{\Delta_n}\left|K_b\left(t_{j-1}-t_i\right) K_b\left(t_{l-1}-t_i\right)\right|\mathbb{E}\left|\alpha_j^{n, n v} \Delta_n \eta_{l-1, i-l+1}^n\left(t_i-t_{l-1}\right)\right|\\
    &\leq C \sum_{i=1}^n \sum_{j=1}^{l-1} \sum_{l=2}^i \sqrt{\Delta_n}\left|K_b\left(t_{j-1}-t_i\right) K_b\left(t_{l-1}-t_i\right)\right|\Delta_n\left(t_i-t_{l-1}\right)\sqrt{\mathbb{E}\left[\left\|\alpha_j^n\right\|^2\right] \mathbb{E}\left[\left(\eta_{l-1, i-l+1}^n\right)^2\right]}\\
    &\leq C \sqrt{\Delta_n}\sum_{i=1}^n \Delta_n\sum_{j=1}^{l-1}\left|K_b\left(t_{j-1}-t_i\right)\right| \Delta_n\sum_{l=2}^i\left|K_b\left(t_{l-1}-t_i\right)\right|\left(t_i-t_{l-1}\right)\sqrt{\mathbb{E}\left[\left(\eta_{l-1, i-l+1}^n\right)^2\right]}\\
    &\leq C \sqrt{\Delta_n}\sum_{l=1}^n \Delta_n\left(\sum_{i=l}^{\left(l+\sqrt{b}/\Delta_n-1\right)\wedge n}+\sum_{i=\left(l+\sqrt{b}/\Delta_n\right)\wedge n}^{n}\right)\left|K_b\left(t_{l-1}-t_i\right)\right|\left(t_i-t_{l-1}\right)\sqrt{\mathbb{E}\left[\left(\eta_{l-1, i-l+1}^n\right)^2\right]}\\
    &\leq Cb\sqrt{\Delta_n}\sum_{l=1}^n\left\{\int_{-\frac{1}{\sqrt{b}}}^0\left|K(u)\right|(-u)du\sqrt{\mathbb{E}\left[\left(\eta_{l-1, \frac{\sqrt{b}}{\Delta_n}}^n\right)^2\right]}+\int_{\left(\frac{t_{l-1}-T}{b}\right)\wedge -\frac{1}{\sqrt{b}}}^{-\frac{1}{\sqrt{b}}}\left|K(u)\right|(-u)du\right\},
\end{split}
\end{equation} 
where 
$$
\int_{\left(\frac{t_{l-1}-T}{b}\right)\wedge -\frac{1}{\sqrt{b}}}^{-\frac{1}{\sqrt{b}}}\left|K(u)\right|(-u)du\to 0,\quad \int_{-\frac{1}{\sqrt{b}}}^0\left|K(u)\right|(-u)du<\infty.
$$
Observe that, from the notation in (\ref{eq:eta}), $$
\mathbb{E}\left[\left(\eta_{l-1, \frac{\sqrt{b}}{\Delta_n}}^n\right)^2\right]=\mathbb{E}\left[\left(\eta_{t_{l-1}, \sqrt{b}}\right)^2\right]\to 0,
$$
when $\sqrt{b}\to 0$, because $\eta_{t_{l-1}, \sqrt{b}}\leq C $ and the c\`adl\`ag property of $c, \tilde{\mu}, \tilde{c}$ implies that $\eta_{s,\sqrt{b}}\to 0 $ for all $s$ except for countably many positive values of $s$ (see also \cite{jacod2013quarticity}, proof of Lemma 4.2). {\DG Plugging} back in (\ref{eq:eta_sum}), we obtain 
\begin{equation}\label{eq:L4p3_mean}
\mathbb{E}\left|\sum_{j= 1}^n \mathbb{E}\left[\zeta^n_j\mid \mathcal{F}^n_{j-1}\right] \right|\to 0.
\end{equation}
{\DG Next, note that 
\[
	\mathbb{E}\left(\sum_{j= 1}^n \left( \zeta^n_j- \mathbb{E}\left[\zeta^n_j\mid \mathcal{F}^n_{j-1}\right] \right)  \right)^2= \sum_{j = 1}^n\mathbb{E}\left( \zeta^n_j - \mathbb{E}\left[\zeta^n_j\mid \mathcal{F}^n_{j-1}\right] \right)^2\leq \sum_{j= 1}^n\mathbb{E}\left[\left( \zeta^n_j \right)^2\right],
\]
and, therefore,}
\begin{equation} \label{eq:L4p3_2moment}
\resizebox{1.0\hsize}{!}{$
\begin{aligned}
&\mathbb{E}\left(\sum_{j= 1}^n \left( \zeta^n_j- \mathbb{E}\left[\zeta^n_j\mid \mathcal{F}^n_{j-1}\right] \right)  \right)^2\\
&\leq  \sum_{j = 1}^n\mathbb{E}\left[\left(\alpha^{n,uv}_j\right)^2\left(\sum_{l = j+1}^{n} \sum_{i = l}^{n}\sqrt{\Delta_n} \bar{K}\left(t_i\right)^{-1}  K_{{ b}}\left(t_{j-1} - t_{i} \right)K_{{ b}}\left(t_{l-1} - t_{i} \right)\alpha^{n,pq}_l  \partial^2_{pq,uv}g\left(c^n_i\right) \right)^2\right]\\
&\leq {\sum_{j = 1}^n\mathbb{E}\Bigg[\left( \alpha^{n,uv}_j\right)^2\sum_{l = j+1}^{n}   \left(\sum_{i = l}^{n}\sqrt{\Delta_n}\bar{K}\left(t_i\right)^{-1}  K_{{ b}}\left(t_{j-1} - t_{i} \right)K_{{ b}}\left(t_{l-1} - t_{i} \right)\partial^2_{pq,uv}g\left(c^n_i\right) \right)^2 \left(\alpha^{n,pq}_l\right)^2}\\
 &\quad+ (\alpha^{n,uv}_j)^2\sum_{\substack{l,l' = j+1 \\ l\neq l'}}^{n} \left( \sum_{i = l}^{n}\sqrt{\Delta_n} \bar{K}\left(t_i\right)^{-1}  K_{{ b}}\left(t_{j-1} - t_{i} \right)K_{{ b}}\left(t_{l-1} - t_{i} \right)\partial^2_{pq,uv}g\left(c^n_i\right)\right)   \\
&\quad\quad\quad\quad\quad\times\left( \sum_{i' = l'}^{n}\sqrt{\Delta_n} \bar{K}\left(t_i\right)^{-1}  K_{{ b}}\left(t_{j-1} - t_{i'} \right)K_{{ b}}\left(t_{l'-1} - t_{i'} \right) \partial^2_{pq,uv}g\left(c^n_{i'}\right)  \right)\alpha^{n,pq}_l \alpha^{n,pq}_{l'} \Bigg]\\
&\leq  C\sum_{j= 1}^n\Delta_n^4\sum_{l = j+1}^{n}   \left(\sum_{i = l}^{n}\sqrt{\Delta_n} \left|K_{{ b}}\left(t_{j-1} - t_{i} \right)K_{{ b}}\left(t_{l-1} - t_{i} \right) \right|\right)^2\\
& \quad+ C\sum_{j = 1}^n\Delta_n^5\sum_{\substack{l,l' = j+1 \\ l\neq l'}}^{n} \left( \sum_{i = l}^{n}K_{{ b}}\left(t_{j-1} - t_{i} \right)K_{{ b}}\left(t_{l-1} - t_{i} \right)   \right)\\
&\quad\quad\quad\quad\quad\times\left| \sum_{i' = l'}^{n}K_{{ b}}\left(t_{j-1} - {\DG t_{i'}} \right)K_{{ b}}\left({\DG t_{l'-1} - t_{i'}} \right)\sqrt{\mathbb{E}\left(\eta^n_{l-1,i'-l'+1}\right)^2} (t_{i'} - t_{l'-1})   \right| \\
&\leq C\Delta_n\int_{-\frac{T}{b}}^{0}{\DG \int_{s}^{0}}\left({\DG \int_{s-v}^{0}}\left|K(v+u)K(u)\right|du\right)^2dvds\\
&+Cb^2\int_{-\frac{T}{b}}^{0}\left\{{\DG \int_{s}^{0}}\left|{\DG \int_{s-v}^0}K(v+u)K(u)du\right|dv\right.\\
&\quad\quad\quad\quad\quad\times\left.{\DG \int_{s}^{0}}\left|{\DG \int_{s-v'}^{0}}K(v'+u')K(u')u'\sqrt{\mathbb{E}\left(\eta_{T+(s-v)b,|u'|b}\right)^2}du'\right|dv'\right\}ds\longrightarrow  0,
\end{aligned}$}
\end{equation}
where in the {\DG third} inequality we use the following estimate valid for $i'\geq i\geq l\geq l'\geq j $, 
\begin{equation}\label{eq:sucessive}
\begin{split}
& \left|\mathbb{E}\left[\left(  \alpha^{n,uv}_j\right)^2 \partial^2_{pq,uv}g\left(c^n_i\right) \partial^2_{pq,uv}g\left(c^n_{i'}\right) \alpha^{n,pq}_l \alpha^{n,pq}_{l'} \right]\right| \\
&\quad\leq  \mathbb{E}\left[\left|\left(  \alpha^{n,uv}_j\right)^2\alpha^{n,pq}_{l'}\right|  \left|\mathbb{E} \left[\partial^2_{pq,uv}g\left(c^n_i\right) \partial^2_{pq,uv}g\left(c^n_{i'}\right) \alpha^{n,pq}_l\mid \mathcal{F}^n_{l-1} \right]\right| \right] \\
 &\quad\leq C\mathbb{E}\left[\left|\left(  \alpha^{n,uv}_j\right)^2\alpha^{n,pq}_{l'}\right|  \Delta_n \eta^n_{l-1,i'-l+1} (t_{i'} - t_{l-1})\right] \\
 &\quad\leq C \Delta_n^4 \sqrt{\mathbb{E}\left(\eta^n_{l-1,i'-l+1}\right)^2} (t_{i'} - t_{l-1})\\
&\quad\leq C \Delta_n^4 \sqrt{\mathbb{E}\left(\eta^n_{l-1,i'-l'+1}\right)^2} (t_{i'} - t_{l'-1}).
\end{split}
\end{equation} 
The second inequality above is shown in Appendix \ref{Technical2_1} and the third inequality in (\ref{eq:sucessive}) follows by (\ref{eq:bounds_alpha}) and Cauchy–Schwarz inequality. Finally, the Cross term of (\ref{TrmRLCrs}) can be analyzed similarly to the case of $L$ in (\ref{TrmRLCrs}). We then conclude 
\begin{equation}\label{eq:snrto0}
    S_{n, r}\to_p 0, \quad \text{for} \ r=3 \ \text{and} \ r=4.
\end{equation}

For $ {r= 5}$, {\DR first note that,} by successive conditioning and It\^o's lemma, we {\DR can {prove} that}
\begin{equation}\label{eq:A.84}
\begin{aligned}
    &\left|\mathbb{E}\left[\left.\alpha^{n,uv}_j\left(c_{l-1}^{n, p q}-c_i^{n, p q} \right)\partial^2_{pq,uv}g(c^n_i) \right| \mathcal{F}^n_{j-1} \right]\right| \\
    &\leq
\begin{cases}
        C\Delta_n^{3/2}\left(\sqrt{\Delta_n} + \eta^n_{j-1} \right)& \text{ if } i,l-1< j,\\
    C\Delta_n^{3/2}\left(\sqrt{\Delta_n} + \eta^n_{j-1} \right)+C\Delta_n\eta^{n}_{j-1,l-j}\left(t_{l-1}-t_{j-1}\right) & \text{ if } i < j \leq l-1,\\
    C \Delta_n^{3 / 2}\left(\sqrt{\Delta_n}+\eta_{j-1}^n\right)+C\Delta_n\eta^{n}_{j-1,i-j+1}\left(t_{i}-t_{j-1}\right) & \text{ if } i \geq j>l-1,\\
    C \Delta_n \eta_{j-1, l-j}^n\left(t_{l-1}-t_{j-1}\right)+C\Delta_n\eta^{n}_{j-1,i-j+1}\left(t_{i}-t_{j-1}\right) & \text{ if } i, l-1 \geq j.
    \end{cases}
\end{aligned}
\end{equation}
The proof of (\ref{eq:A.84}) is also included in Appendix \ref{Technical2_1}. With these inequalities, when $r=5$,
\begin{align*}
    S_{n,5}&=\sum_{i=1}^n \frac{\sqrt{\Delta_n}}{2}\partial_{pq,uv}g(c_i^n)\bar{K}\left(t_i\right)^{-1}\sum_{j=1}^{n} \sum_{l=1}^{n} K_{{b}}\left(t_{j-1}-t_{i}\right) K_{{b}}\left(t_{l-1}-t_{i}\right) \alpha_{j}^{n,uv} \left(c_{l-1}^{n,pq}-c_{i}^{n,pq}\right) \Delta_{n}\\
    &=:\sum_{j=1}^{n}\zeta_{j}^{n}.
\end{align*}
By (\ref{eq:A.84}), we can estimate $M:=\mathbb{E}\left|\sum_{j= 1}^n \mathbb{E}\left[\zeta^n_j\mid \mathcal{F}^n_{j-1}\right] \right|$ as follows:
\begin{equation*}
\begin{aligned}
    M
    &\leq C \Delta_n^{3/2}\sum_{j= 1}^n\sum_{i=1}^{n}\sum_{l=1}^{n}\left|K_{{b}}\left(t_{j-1}-t_{i}\right) K_{{b}}\left(t_{l-1}-t_{i}\right)\right|\Big[\Delta_n^{3/2}\left(\sqrt{\Delta_n}+ \mathbb{E}\left[\eta^n_{j-1}\right] \right)\\
    & \qquad\qquad\qquad\qquad\qquad\qquad+\Delta_n \mathbb{E}\left[\eta_{j-1, \left|l-j\right|}^n\right]\left|t_{l-1}-t_{j-1}\right|+\Delta_n\mathbb{E}\left[\eta^{n}_{j-1,\left|i-j+1\right|}\right]\left|t_{i}-t_{j-1}\right|\Big]\\
    &\leq C \Delta_n^{3/2}\sum_{j= 1}^n\sum_{i=1}^{n} \left|K_{{b}}\left(t_{j-1}-t_{i}\right)\right|\int_{-t_{i}/b}^{(T-t_i)/b}\left|K(u)\right|\Big[\Delta_n+ \mathbb{E}\left[\eta^n_{j-1}\right]\sqrt{\Delta_n}\\
    &\qquad\qquad\qquad\qquad\qquad+ \mathbb{E}\left[\eta_{t_{j-1},\left|t_i-t_{j-1}+ub\right|}\right]\left|t_{i}-t_{j-1}+ub\right|+\mathbb{E}\left[\eta_{t_{j-1},\left|t_i-t_{j-1}\right|}\right]\left|t_{i}-t_{j-1}\right|\Big]du\\
    &\leq C\sqrt{\Delta_n}\sum_{j=1}^{n}\int_{(t_{j-1}-T)/b}^{t_{j-1}/b}\left|K(v)\right|\int_{v-\frac{t_{j-1}}{b}}^{v+\frac{T-t_{j-1}}{b}}\left|K(u)\right|\Big[\Delta_n+ \eta^n_{j-1}\sqrt{\Delta_n}\\
    &\qquad\qquad\qquad\qquad\qquad\qquad\qquad\qquad\qquad\qquad+\mathbb{E}\left[\eta_{t_{j-1},\left|u-v\right|b}\right]\left|u-v\right|b+\mathbb{E}\left[\eta_{t_{j-1},\left|v\right|b}\right]\left|v\right|b\Big]dudv\\
    &\leq C\sqrt{\Delta_n}b\sum_{j=1}^{n}\int \left|K(v)\right|\int \left|K(u)\right|\Bigg[\frac{\Delta_n}{b}+\eta^n_{j-1}\frac{\sqrt{\Delta_n}}{b}+\mathbb{E}\left[\eta_{t_{j-1},\left|u-v\right|b}\right]\left|u-v\right|+\mathbb{E}\left[\eta_{t_{j-1},\left|v\right|b}\right]\left|v\right|\Bigg]dudv\to 0.
\end{aligned}
\end{equation*}
Also, by (\ref{eq:X_c_bound}) and (\ref{eq:bounds_alpha}), and Cauchy-Schwartz inequality,
\begin{equation*}
\begin{aligned}
&\mathbb{E}\left(\sum_{j= 1}^n \left( \zeta^n_j- \mathbb{E}\left[\zeta^n_j\mid \mathcal{F}^n_{j-1}\right] \right)  \right)^2\leq \sum_{j= 1}^n\mathbb{E}\left[\left( \zeta^n_j \right)^2\right]\\
&\leq C\Delta_n^3\sum_{j= 1}^n \mathbb{E}\Bigg[\left(\alpha_{j}^{n,uv}\right)^2\Bigg(\sum_{i=1}^{n}\partial_{p q, u v}^2 g(c_i^n)\bar{K}\left(t_{i}\right)^{-1}K_{{b}}\left(t_{j-1}-t_{i}\right)\sum_{l=1}^{n} K_{{b}}\left(t_{l-1}-t_{i}\right)\left(c_{l-1}^{n,pq}-c_{i}^{n,pq}\right)\Bigg)^2\Bigg]\\
&\leq C\Delta_n^3\sum_{j= 1}^n \sum_{i,i^{\prime}}\sum_{l,l^{\prime}}\left|K_{{b}}\left(t_{j-1}-t_{i}\right)\right|\left| K_{{b}}\left(t_{l-1}-t_{i}\right)\right|\left|K_{{b}}\left(t_{j-1}-t_{i^{\prime}}\right)\right|\left| K_{{b}}\left(t_{l^{\prime}-1}-t_{i^{\prime}}\right)\right|\\
&\qquad\qquad\qquad\qquad\qquad\qquad\qquad\qquad\qquad\qquad\qquad\qquad\times
\mathbb{E}\left[\left(\alpha_{j}^{n,uv}\right)^2\left|c_{l-1}^{n,pq}-c_{i}^{n,pq}\right|\left|c_{l^{\prime}-1}^{n,pq}-c_{i^{\prime}}^{n,pq}\right|\right]\\
&\leq C\Delta_n^5\sum_{j= 1}^n \sum_{i,i^{\prime}}\sum_{l,l^{\prime}}\left|K_{{b}}\left(t_{j-1}-t_{i}\right)\right|\left| K_{{b}}\left(t_{l-1}-t_{i}\right)\right|\left|K_{{b}}\left(t_{j-1}-t_{i^{\prime}}\right)\right|\left| K_{{b}}\left(t_{l^{\prime}-1}-t_{i^{\prime}}\right)\right| \left|t_{l-1}-t_{i}\right|^{1/4}\left|t_{l^{\prime}-1}-t_{i^{\prime}}\right|^{1/4}\\
&\leq C\Delta_n\sum_{j= 1}^n\left(\int \left|K(v)\right|\left|v\right|^{1/4}dv\right)^2\left(\int \left|K(u)\right|du\right)^2\sqrt{b}\to 0.
\end{aligned}
\end{equation*}
Thus, ${A}_{\lambda_1,\lambda_2}^{n,3}(5)\to_p 0$.  We handle the case $r=6$ in the same way as $r=5$. Now, we can finally conclude (\ref{eq:snrto0}) for $r=5,6$ and then (\ref{eq:last_bias}).

\subsubsection{The bias term $V^{n}_{5}$}\label{a.2.5}
We now show that 
\begin{equation}\label{eq:v5}
    V^{n}_{5}\to \theta \int_{-\infty}^{0}L(u)du g\left(c_{0}\right)-\theta \int_{0}^{\infty}L(u)du g\left(c_{T}\right).
\end{equation}
By Assumption \ref{kernel}, $K^{\prime}$ and $K^{\prime\prime}$ are defined in {\DR all} $(A,B)\subset \mathbb{R}$, but finite number of points. Without loss of generality, suppose $\operatorname{dom}(K^{\prime})=\operatorname{dom}(K^{\prime\prime})=(A,B)\backslash \{0\}$.
Note that
\begin{align*}
    \bar{K}\left(t_{i}\right)-1&=-\left(\int_{-\infty}^{-t_{i}/{b}}+\int_{\left(T-t_{i}\right)/{b}}^{\infty}\right)K(u)du+\frac{\Delta_{n}}{b}\sum_{j=1}^{n}K\left(\frac{t_{j-1}-t_{i}}{b}\right)-\frac{1}{b}\int_{0}^{T}K\left(\frac{t-t_{i}}{b}\right)dt\\
    &=:-\left(\int_{-\infty}^{-t_{i}/{b}}+\int_{\left(T-t_{i}\right)/{b}}^{\infty}\right)K(u)du+D_{1}^{K,t_i}\left(1\right).
\end{align*}
Then, 
\begin{equation}\label{eq:v5decomp}
    \begin{aligned}
    V_5^n&=\Delta_{n}^{1/2}\sum_{i=1}^{n}\left[\bar{K}\left(t_{i}\right)-1\right]g\left(c_{i}^{n}\right)\\
    &=\Delta_{n}^{1/2}\sum_{i=1}^{n}\left[-\left(\int_{-\infty}^{-t_{i}/{b}}+\int_{\left(T-t_{i}\right)/{b}}^{\infty}\right)K(u)du+D_{1}^{K,t_i}\left(1\right)\right]g\left(c_{i}^{n}\right)\\
    &=:-\uppercase\expandafter{\romannumeral1}-\uppercase\expandafter{\romannumeral2}+\uppercase\expandafter{\romannumeral3},
\end{aligned}
\end{equation}
where note that $D_{1}^{K,t_i}\left(1\right)$ can {\DG be} further decomposed as follows, for some $s_t \in\left(t_{j-1}, t_j\right)$,
\begin{align*}
    D_{1}^{K,t_i}\left(1\right)
    &={\DG \frac{1}{b}\sum_{j=1}^{n}
    \int_{t_{j-1}}^{t_{j}}\left(K\Big(\frac{t_{j-1}-t_{i}}{b}\Big)
    -K\Big(\frac{t-t_i}{b}\Big)\right)dt}\\
    &=\frac{1}{b} \sum_{j=1,j\neq i,i+1}^n \int_{t_{j-1}}^{t_j} K^{\prime}\left(\frac{t_{j-1}-t_i}{b}\right) \frac{t_{j-1}-t}{b} d t\\
    &+\frac{1}{b} \sum_{j=1,j\neq i,i+1}^n \int_{t_{j-1}}^{t_j}\left[K^{\prime}\left(\frac{s_t-t_i}{b}\right)-K^{\prime}\left(\frac{t_{j-1}-t_i}{b}\right)\right] \frac{t_{j-1}-t}{b} d t\\
    &+\frac{\Delta_{n}}{b}K\left(-\frac{\Delta_n}{b}\right)+\frac{\Delta_{n}}{b}K\left(0\right)-\frac{1}{b}\int_{t_{i-1}}^{t_{i+1}}K\left(\frac{t-t_{i}}{b}\right)dt\\
    &=:D_{11}^{K,t_i}\left(1\right)+D_{12}^{K,t_i}\left(1\right)+D_{13}^{K,t_i}\left(1\right).
\end{align*}
With the notation $\tilde{K}(u)=\int_{-\infty}^{-u}K(v)dv$ {\DG and applying Lemma \ref{K_discrete}, we have:}
\begin{align*}
    \uppercase\expandafter{\romannumeral1}&=\Delta_{n}^{1/2}\sum_{i=1}^{n}\int_{-\infty}^{-t_{i}/{b}}K(u)dug\left(c_{i}^{n}\right)
    \\
    &=b\Delta_{n}^{-1/2}\int_{0}^{T}\tilde{K}_b\left(t\right)g\left(c_{t}\right)dt+{\DG
    b\Delta_{n}^{-1/2}O_P\Big(\frac{\Delta_n}{b}\Big)}\\
    &\stackrel{n\to\infty}{\longrightarrow}-\theta\int_\infty^0L(u)dug(c_0){\Blue ,}
\end{align*}
{\DG where in the last limit we used that $b\Delta_{n}^{-1/2}\to \theta$.}
Similarly, let $K^{*}(u)=\int^{\infty}_{-u}K(v)dv=1-\tilde{K}(u)$. {\DR Then,}
\begin{align*}
    \uppercase\expandafter{\romannumeral2}&=\Delta_{n}^{1/2}\sum_{i=1}^{n}\int_{\left(T-t_{i}\right)/{b}}^{\infty}K(u)dug\left(c_{i}^{n}\right)=b\Delta_{n}^{-1/2}\Delta_{n}\sum_{i=1}^{n}\frac{1}{b}K^{*}\left(\frac{t_{i}-T}{b}\right)g\left(c_{i}^{n}\right)\\
    &=b\Delta_{n}^{-1/2}\int_{0}^{T}K^{*}_{b}(t-T)g\left(c_{t}\right)dt+{\DG
    b\Delta_{n}^{-1/2}O_P\Big(\frac{\Delta_n}{b}\Big)}\\
     &\stackrel{n\to\infty}{\longrightarrow}\theta\int_0^\infty L(u)dug(c_T).
\end{align*}
As for $\uppercase\expandafter{\romannumeral3}$, by the same argument {\DR as} in (\ref{lemmaa.1.1}) \& (\ref{lemmaa.1.2}) with $f\equiv 1$, we have $D_{12}^{K,t_i}\left(1\right)=O\left(\frac{\Delta^2}{{b}^{2}}\right)$. For $D_{13}^{K,t_i}\left(1\right)$, by the continuity of $K(x)$ at $x=0$, we have $K\left(-\frac{\Delta_n}{b}\right)+K\left(0\right)-\frac{1}{\Delta_n}\int_{t_{i-1}}^{t_{i+1}}K\left(\frac{t-t_{i}}{b}\right)dt\to 0$ when $n\to\infty$. Thus,
$$
\Delta_{n}^{1/2}\sum_{i=1}^{n}D_{13}^{K,t_i}\left(1\right)g\left(c_{i}^{n}\right)\to_{p} 0.
$$
For $D_{11}^{K,t_i}\left(1\right)$, consider
\begin{equation}\label{v5eq3}
    \begin{aligned}
    &\Delta_{n}^{1/2}\sum_{i=1}^{n}D_{11}^{K,t_i}\left(1\right)g\left(c_{i}^{n}\right)=\frac{\Delta_{n}^{1/2}}{2}\frac{\Delta_{n}}{b}\sum_{i=1}^{n}\frac{\Delta_{n}}{b}\sum_{j=1,j\neq i,i+1}^{n}K^{\prime}\left(\frac{t_{j-1}-t_{i}}{b}\right)g\left(c_{i}^{n}\right)\\
    &=\frac{\Delta_{n}^{1/2}}{2}\frac{\Delta_{n}}{b}\sum_{i=1}^{n}\left[\frac{\Delta_{n}}{b}\sum_{j=1,j\neq i,i+1}^{n}K^{\prime}\left(\frac{t_{j-1}-t_{i}}{b}\right)-\left(\int_{0}^{t_{i-1}}+\int_{t_{i+1}}^{T}\right)\frac{1}{b}K^{\prime}\left(\frac{t-t_{i}}{b}\right)dt\right]g\left(c_{i}^{n}\right)\\
    &\quad+\frac{\Delta_{n}^{1/2}}{2}\frac{\Delta_{n}}{b}\sum_{i=1}^{n}\left(\int_{0}^{t_{i-1}}+\int_{t_{i+1}}^{T}\right)\frac{1}{b}K^{\prime}\left(\frac{t-t_{i}}{b}\right)dtg\left(c_{i}^{n}\right)=:\uppercase\expandafter{\romannumeral4}+\uppercase\expandafter{\romannumeral5}.
\end{aligned}
\end{equation}
Clearly,
\begin{align*}
    \uppercase\expandafter{\romannumeral5}&=\frac{\Delta_{n}^{1/2}}{2}\frac{\Delta_{n}}{b}\sum_{i=1}^{n}\left(\int_{-\frac{t_{i}}{b}}^{-\frac{\Delta_n}{b}}+\int_{\frac{\Delta_n}{b}}^{\frac{T-t_{i}}{b}}\right)K^{\prime}(v)dvg\left(c_{i}^{n}\right)\\
    &=\frac{\Delta_{n}^{1/2}}{2}\frac{\Delta_{n}}{b}\sum_{i=1}^{n}\left[K\left(\frac{T-t_{i}}{b}\right)-K\left(\frac{\Delta_n}{b}\right)+K\left(-\frac{\Delta_n}{b}\right)-K\left(\frac{-t_{i}}{b}\right)\right]g\left(c_{i}^{n}\right).
\end{align*}
Note that with the notation $\tilde{K}^*(u)=K(-u)$,
\begin{equation}\label{v5eq4}
    \begin{aligned}
    &\Delta_{n}^{1/2}\frac{\Delta_{n}}{b}\sum_{i=1}^{n}K\left(\frac{-t_{i}}{b}\right)g\left(c_{i}^{n}\right)=\Delta_{n}^{1/2}\frac{\Delta_{n}}{b}\sum_{i=1}^{n}\tilde{K}^*\left(\frac{t_{i}}{b}\right)g\left(c_{i}^{n}\right)\\
    &=\Delta_{n}^{1/2}\left[\int_{0}^{T}\tilde{K}^*_{{b}}(t)g\left(c_{t}\right)dt+D^{\tilde{K}^*,0}_{1}(g)\right]=\Delta_{n}^{1/2}\left[\int_{0}^{T/{b}}\tilde{K}(u)g\left(c_{ub}\right)du+D^{\tilde{K}^*,0}_{1}(g)\right]\to_{p} 0,
\end{aligned}
\end{equation}
and
\begin{align}\label{v5eq5}
    &\Delta_{n}^{1/2}\frac{\Delta_{n}}{b}\sum_{i=1}^{n}K\left(\frac{T-t_{i}}{b}\right)g\left(c_{i}^{n}\right)=\Delta_{n}^{1/2}\frac{\Delta_{n}}{b}\sum_{i=1}^{n}\tilde{K}^*\left(\frac{t_{i}-T}{b}\right)g\left(c_{i}^{n}\right)\\
    &=\Delta_{n}^{1/2}\left[\int_{0}^{T}\tilde{K}^*_{{b}}(t-T)g\left(c_{t}\right)dt+D^{\tilde{K}^*,T}_{1}(g)\right]=\Delta_{n}^{1/2}\left[\int^{0}_{-T/{b}}\tilde{K}(u)g\left(c_{uh+T}\right)du+D^{\tilde{K}^*,T}_{1}(g)\right]\to_{p} 0.\nonumber
\end{align}
By (\ref{v5eq4}) \& (\ref{v5eq5}), together with the fact that $K\left(-\frac{\Delta_n}{b}\right)-K\left(\frac{\Delta_n}{b}\right)\to 0$ when $n\to\infty$, we get $\uppercase\expandafter{\romannumeral5}\to_{p} 0$. For $\uppercase\expandafter{\romannumeral4}$, note that for $j\neq i,i+1$,
\begin{align*}
    {\DR |A_{j}|}&:=\left|\frac{1}{2}\left[K_{{b}}^{\prime}\left(t_{j-1}-t_{i}\right)+K_{{b}}^{\prime}\left(t_{j}-t_{i}\right)\right]-\frac{1}{\Delta_n}\int_{t_{j-1}}^{t_{j}}K_{{b}}^{\prime}\left(t-t_{i}\right)dt\right|\\
    &=\frac{1}{2\Delta_n}\left|\int_{t_{j-1}}^{t_{j}}\left[K_{{b}}^{\prime}\left(t_{j}-t_{i}\right)-K_{{b}}^{\prime}\left(t-t_{i}\right)\right]-\left[K_{{b}}^{\prime}\left(t-t_{i}\right)-K_{{b}}^{\prime}\left(t_{j-1}-t_{i}\right)\right]dt\right|\\
    &=\frac{1}{2\Delta_n}\left|\int_{t_{j-1}}^{t_{j}}\left[\int_{t}^{t_{j}}\frac{1}{b^{2}}K^{\prime\prime}\left(\frac{s-t_{i}}{b}\right)ds-\int_{t_{j-1}}^{t}\frac{1}{b^{2}}K^{\prime\prime}\left(\frac{s-t_{i}}{b}\right)ds\right]dt\right|\\
    &=\frac{1}{2\Delta_n}\left|\int_{t_{j-1}}^{t_{j}}\left[\int_{t_{j-1}}^{s}\frac{1}{b^{2}}K^{\prime\prime}\left(\frac{s-t_{i}}{b}\right)dt-\int_{s}^{t_{j}}\frac{1}{b^{2}}K^{\prime\prime}\left(\frac{s-t_{i}}{b}\right)dt\right]ds\right|\\
    &=\frac{1}{2\Delta_n}\left|\int_{t_{j-1}}^{t_{j}}\left[\left(s-t_{j-1}\right)\frac{1}{b^{2}}K^{\prime\prime}\left(\frac{s-t_{i}}{b}\right)-\left(t_{j}-s\right)\frac{1}{b^{2}}K^{\prime\prime}\left(\frac{s-t_{i}}{b}\right)\right]ds\right|\\
    &=\frac{1}{\Delta_n}\left|\int_{t_{j-1}}^{t_{j}}\left(s-\frac{t_{j-1}+t_{j}}{2}\right)\frac{1}{b^{2}}K^{\prime\prime}\left(\frac{s-t_{i}}{b}\right)ds\right|\\
    &=\frac{1}{\Delta_n}\left|\int_{t_{j-1}}^{t_{j}}\left(s-\frac{t_{j-1}+t_{j}}{2}\right)\left[\frac{1}{b^{2}}K^{\prime\prime}\left(\frac{s-t_{i}}{b}\right)-\frac{1}{b^{2}}K^{\prime\prime}\left(\frac{\frac{t_{j-1}+t_{j}}{2}-t_{i}}{b}\right)\right]ds\right|\\
    &\leq \frac{1}{b^{2}\Delta_n} \omega_{K^{\prime\prime}}\left(\frac{\Delta_{n}}{2b}\right)\int_{t_{j-1}}^{t_{j}}\left|s-\frac{t_{j-1}+t_{j}}{2}\right|ds=\frac{\Delta_{n}}{4b^{2}}\omega_{K^{\prime\prime}}\left(\frac{\Delta_{n}}{2b}\right),
\end{align*}
where $\omega_{K^{\prime\prime}}\left(\delta\right)=\sup \left\{\left|K^{\prime\prime}\left(s\right)-K^{\prime\prime}\left(t\right)\right|:\left|s-t\right|\leq \delta\right\}$ is the modulus of continuity of $K^{\prime\prime}$. Note that $\lim _{\delta \downarrow 0} \omega_{K^{\prime\prime}}(\delta)=0$ is true {\DG since we are assuming that} $K^{\prime\prime}$ is uniformly continuous on its domain $\operatorname{dom}(K^{\prime\prime})=(A,B)\backslash \{0\}$. Recalling Assumption \ref{kernel} and that $b_n \sim \theta \sqrt{\Delta_n}$, we have $\frac{\Delta_{n}}{2b}\to 0$, and thus $\omega_{K^{\prime\prime}}\left(\frac{\Delta_{n}}{2b}\right)\to 0$. Then, sum it up (except for $j=i,i+1$) and we will obtain
\begin{align*}
    {\DR \Big|\sum_{j=1,j\neq i, i+1}^{n}A_{j}\Big|}&=\Big|\sum_{j=1,j\neq i,i+1}^{n}K_{{b}}^{\prime}\left(t_{j-1}-t_{i}\right)-\frac{K_{{b}}^{\prime}\left(-t_{i}\right)-K_{b}^{\prime}\left(-\Delta_n\right)+K_{b}^{\prime}\left(\Delta_n\right)-K_{{b}}^{\prime}\left(T-t_{i}\right)}{2}\\
    &\qquad -\frac{1}{\Delta_n}\left(\int_{0}^{t_{i-1}}+\int_{t_{i+1}}^{T}\right)K_{{b}}^{\prime}\left(s-t_{i}\right)ds\Big|\\
    &\leq \frac{T}{4b^{2}}\omega_{K^{\prime\prime}}\left(\frac{\Delta_{n}}{2b}\right).
\end{align*}
Thus,
\begin{align*}
    &\left|\Delta_{n}\sum_{j=1,j\neq i,i+1}^{n}K_{{b}}^{\prime}\left(t_{j-1}-t_{i}\right)-\left(\int_{0}^{t_{i-1}}+\int_{t_{i+1}}^{T}\right)K_{{b}}^{\prime}\left(s-t_{i}\right)ds\right|\\
    &\leq \Delta_{n}\left|\frac{K_{{b}}^{\prime}\left(-t_{i}\right)-K_{b}^{\prime}\left(-\Delta_n\right)+K_{b}^{\prime}\left(\Delta_n\right)-K_{{b}}^{\prime}\left(T-t_{i}\right)}{2}\right|+\frac{T\Delta_{n}}{4b^{2}}\omega_{K^{\prime\prime}}\left(\frac{\Delta_{n}}{2b}\right)\to 0.
\end{align*}
Therefore, 
\begin{align*}
	|IV|&=\Big|\frac{\Delta_{n}^{3/2}}{2b}\sum_{i=1}^{n}\left[\Delta_{n}\sum_{j=1,j\neq i,i+1}^{n}K_b^{\prime}\left(\frac{t_{j-1}-t_{i}}{b}\right)-\left(\int_{0}^{t_{i-1}}+\int_{t_{i+1}}^{T}\right)\frac{1}{b}K_b^{\prime}\left(\frac{t-t_{i}}{b}\right)dt\right]g\left(c_{i}^{n}\right)\Big|\\
	&\stackrel{n\to\infty}{\longrightarrow}0.
\end{align*}
Combining all the results and plugging back in (\ref{eq:v5decomp}) \& (\ref{v5eq3}), (\ref{eq:v5}) is proved.

\subsubsection{{Eliminating $\theta \int_0^T \sum_{\ell, m=1}^{d} \partial_{\ell m} g\left(c_s\right) d c^{lm}_s$}} 
{From (\ref{eq:leading}), (\ref{eq:v4}), (\ref{eq:last_bias}), and (\ref{eq:v5}) we deduce that
\begin{equation} \label{eq:convergence_in_lawPre}
\frac{1}{\sqrt{\Delta_{n}}}\Big(V(g)_T^{n}-V(g)_T\Big) \stackrel{st}{\longrightarrow} \widetilde{A}^{1} + \widetilde{A}^{2}+\widetilde{A}^{3}+Z,
\end{equation}
with 
 \begin{equation*}
\begin{split}
\widetilde{A}^{1}&:=\theta {\int_{-\infty}^{0}L(u)du}g\left(c_{0}\right)-\theta {\int_{0}^{\infty}L(u)du}g\left(c_{T}\right),\\
\widetilde{A}^{2}&:=\frac{1}{2\theta} \int_0^T \sum_{p,q,u,v=1}^{d}\partial^{2}_{p q, u v} g\left(c_s\right) \check{c}_s^{p q, u v} d s \int_{-\infty}^{\infty} K^2(u) d u,\\
\widetilde{A}^{3}&:=\theta \int_0^T \sum_{\ell, m=1}^{d} \partial_{\ell m} g\left(c_s\right) d c^{lm}_s \int_{-\infty}^{\infty} L(u) d u \\
    &\quad+\theta \int_0^T \sum_{p,q,u,v=1}^{d}\partial^{2}_{p q, u v} g\left(c_s\right) \tilde{c}^{pq,uv}_{s} d s\left[\int_{-\infty}^0 L(u) d u+\frac{1}{2}\int_{-\infty}^{\infty} L(u)^2 d u\right]{\Blue .}
\end{split}
\end{equation*}
On the other hand, by Itô's formula, we have
\begin{align*}
    g\left(c_T\right)&=g\left(c_0\right)+ \int_0^T \sum_{l, m}\partial_{l m} g\left(c_{s}\right) d c_s^{l m}+\frac{1}{2}\int_{0}^{T}\sum_{p,q,u,v}\partial^{2}_{pq,uv}g\left(c_{s}\right)\tilde{c}_{s}^{pq,uv}ds.
\end{align*}
Hence, solving for the second term and plugging in the expression of $\widetilde{A}^{3}$, we can rewrite the bias $\widetilde{A}^1+\widetilde{A}^2+\widetilde{A}^3$ as ${A}^1+{A}^2+{A}^3$ with the terms $A^{j}$'s given as in 
\eqref{MDfnOfAs}.}
\subsection{Proof of Theorem \ref{thm2.2}}\label{proofthm2.2}
We {\DG first} discuss the elimination of the jumps. Suppose $X$ is discontinuous as in (\ref{eq:X}) and {\DR $\hat{c}^{n}_i$} in (\ref{An3}) is given by (\ref{estivoltrun}). Let
\begin{equation}\label{An3prime}
    \widehat{A}^{\prime n, 3}:=\frac{\sqrt{\Delta_n}}{8} \sum_{i=1}^{n-2 k_n+1} \sum_{p, q, u, v} \partial_{p q, u v}^2 g(\hat{c}_i^{\prime n})\left(\hat{c}_{i+k_n}^{\prime n,pq}-\hat{c}_i^{\prime n,pq}\right)\left(\hat{c}_{i+k_n}^{\prime n,uv}-\hat{c}_i^{\prime n,uv}\right),
\end{equation}
where $\hat{c}_{i}^{\prime n}=\hat{c}_{i}^{n,X^{\prime}}$ is defined in (\ref{MDKNN}) with {\DG the continuous process $X^{\prime}$ defined in \eqref{eq:X_conti}}. We prove here that 
\begin{equation}\label{diffAn3}
    \widehat A^{n,3}-\widehat A^{\prime n,3} \xrightarrow{\mathbb{P}} 0.
\end{equation}
Let
$$
R\left(\hat{c}_i^n, \hat{c}_{i+k_n}^n\right)=\partial_{p q, u v}^2 g(\hat{c}_i^{n})\left(\hat{c}_{i+k_n}^{n,pq}-\hat{c}_i^{n,pq}\right)\left(\hat{c}_{i+k_n}^{n,uv}-\hat{c}_i^{n,uv}\right).
$$
By (\ref{eq:g_prime_bound}) and {\DG the} mean value theorem, for some {\DG $\xi\in(0,1)$},
\begin{equation}\label{rdiffupper}
\begin{split}
        &\left|R\left(\hat{c}_i^n, \hat{c}_{i+k_n}^n\right)-R\left(\hat{c}_i^{\prime n}, \hat{c}_{i+k_n}^{\prime n}\right)\right|\\
    &\leq \left|\partial_{p q, u v}^2 g(\hat{c}_i^{n})\left[\left(\hat{c}_{i+k_n}^{n,p q}-\hat{c}_i^{n,p q}\right)\left(\hat{c}_{i+k_n}^{n,u v}-\hat{c}_i^{n,u v}\right)-\left(\hat{c}_{i+k_n}^{\prime n,p q}-\hat{c}_i^{\prime n,p q}\right)\left(\hat{c}_{i+k_n}^{\prime n,u v}-\hat{c}_i^{\prime n,u v}\right)\right]\right|\\
    &+\left|\left[\partial_{p q, u v}^2 g(\hat{c}_i^{n})-\partial_{p q, u v}^2 g(\hat{c}_i^{\prime n})\right]\left(\hat{c}_{i+k_n}^{n,p q}-\hat{c}_i^{n,p q}\right)\left(\hat{c}_{i+k_n}^{n,u v}-\hat{c}_i^{n,u v}\right)\right|\\
    &\leq C\left(1+\left\|\hat{c}_i^{n}\right\|^{{\DG \ell}-2}\right)\left(\left\|\hat{c}_i^{n}\right\|+\left\|\hat{c}_{i+k_n}^{n}\right\|\right)\left({\DG \left\|\hat{c}_i^n-\hat{c}_i^{\prime n}\right\|+\left\|\hat{c}_{i+k_n}^n-\hat{c}_{i+k_n}^{\prime n}\right\|}\right)\\
    &+C\left(1+\left\|\xi\hat{c}_i^n+\left(1-\xi\right)\hat{c}_i^{\prime n}\right\|^{{\DG \ell}-3}\right)\left\|\hat{c}_i^n-\hat{c}_i^{\prime n}\right\|\left(\left\|\hat{c}_i^{n}\right\|+\left\|\hat{c}_{i+k_n}^{n}\right\|\right)^2.
\end{split} 
\end{equation} 
{\DR By \eqref{boundofdiffc} with $q=1$ and $q'$ large enough, we have
\[
\mathbb{E}\left[\left\|\hat{c}_i^n-\hat{c}_i^{\prime n}\right\|^q\right] \leq C a_n \Delta_n^{(2 -r) \varpi}.
\]
It is easy to see that under (\ref{moreassum}), $(2 -r) \varpi\geq{}1/2$ when $\ell\geq{}4$ and thus, we conclude that}
$$
\mathbb{E}\left[\left|R\left(\hat{c}_i^n, \hat{c}_{i+k_n}^n\right)-R\left(\hat{c}_i^{\prime n}, \hat{c}_{i+k_n}^{\prime n}\right)\right|\right]=o(\sqrt{\Delta_n}),
$$
which implies (\ref{diffAn3}). Now suppose $X$ is continuous, i.e., $X=X^{\prime}$. {\DR Then,} (\ref{diffAn3}) can be obtained again with (\ref{boundofdiffc1}) and (\ref{rdiffupper}). Hence, {\DG in what follows,} we only need to consider $\hat{c}_i^n$ as (\ref{MDKNN}) {\DG under the Assumption} \ref{assum3}.

Set
$$
\gamma_i^n :=\hat{c}_{i+k_n}^{ n}-\hat{c}_i^n=\beta_{i+k_n}^n-\beta_i^n+c_{i+k_n}^n-c_i^n.
$$
The estimator $\widehat A^{n, 3}$ can then be written as
\begin{align*}
    \widehat A^{n, 3}
    &=\frac{\sqrt{\Delta_n}}{8} \sum_{i=1}^{n-2 k_n+1} \sum_{p, q, u, v} \partial_{p q, u v}^2 g(\hat{c}_i^n)\gamma_i^{n, p q} \gamma_i^{n, u v},
\end{align*}
where $\gamma_i^{n, p q} \gamma_i^{n, u v}$ can be decomposed as
\begin{equation}\label{eq:decomgamma}
\begin{aligned}
\gamma_i^{n, p q} \gamma_i^{n, u v}&=
\beta_i^{n, p q} \beta_i^{n, u v}+\beta_{i+k_n}^{n, p q} \beta_{i+k_n}^{n, u v}-\beta_{i+k_n}^{n, p q} \beta_i^{n, u v}-\beta_i^{n, p q} \beta_{i+k_n}^{n, u v}\\
&\quad+\left(c_{i+k_n}^{n, p q}-c_i^{n, p q}\right)\left(c_{i+k_n}^{n, u v}-c_i^{n, u v}\right)\\
&\quad -\beta_i^{n, p q}\left(c_{i+k_n}^{n, u v}-c_i^{n, u v}\right)-\beta_i^{n, u v}\left(c_{i+k_n}^{n, p q}-c_i^{n, p q}\right)\\
&\quad+\beta_{i+k_n}^{n, u v}\left(c_{i+k_n}^{n, p q}-c_i^{n, p q}\right)+\beta_{i+k_n}^{n, pq}\left(c_{i+k_n}^{n, uv}-c_i^{n, uv}\right).
\end{aligned}
\end{equation}
{We can then consider the decomposition:
$$
\widehat{A}^{n,3}=\widehat{A}_{0,0}^{n, 3}+\widehat{A}_{1,1}^{n, 3}-\widehat{A}_{0,1}^{n, 3}-\widehat{A}_{1,0}^{n, 3}+\widehat{A}_1^{n, 3}-\widehat{A}_2^{n, 3}-\widehat{A}_3^{n, 3}+\widehat{A}_4^{n, 3}+\widehat{A}_5^{n, 3},
$$
where 
\begin{align*}
\widehat{A}^{n,3}_{\lambda_1,\lambda_2}&:=\frac{\sqrt{\Delta_n}}{8} \sum_{i=1}^{n-2k_n+1}  \sum_{p, q, u, v} \partial_{p q, u v}^2 g(\hat{c}_i^n)
\beta_{i+\lambda_1 k_n}^{n, p q} \beta_{i+\lambda_2 k_n}^{n, u v},\\
\widehat{A}_{1}^{n,3}&:=\frac{\sqrt{\Delta_n}}{8} \sum_{i=1}^{n-2 k_n+1} \sum_{p, q, u, v} \partial_{p q, u v}^2 g(\hat{c}_i^n)\left(c_{i+k_n}^{n, p q}-c_i^{n, p q}\right)\left(c_{i+k_n}^{n, u v}-c_i^{n, u v}\right),\\
\widehat{A}_{2}^{n,3}&:=\frac{\sqrt{\Delta_n}}{8} \sum_{i=1}^{n-2k_n+1} \sum_{p, q, u, v} \partial_{p q, u v}^2 g(\hat{c}_i^n)\beta_i^{n, p q}\left(c_{i+k_n}^{n, u v}-c_i^{n, u v}\right),\\
\widehat{A}_{3}^{n,3}&:=\frac{\sqrt{\Delta_n}}{8} \sum_{i=1}^{n-2k_n+1} \sum_{p, q, u, v} \partial_{p q, u v}^2 g(\hat{c}_i^n)\beta_i^{n, u v}\left(c_{i+k_n}^{n, p q}-c_i^{n, p q}\right),\\
\widehat{A}_{4}^{n,3}&:=\frac{\sqrt{\Delta_n}}{8} \sum_{i=1}^{n-2 k_n+1} \sum_{p, q, u, v} \partial_{p q, u v}^2 g(\hat{c}_i^n)\beta_{i+k_n}^{n, p q}\left(c_{i+k_n}^{n, u v}-c_i^{n, u v}\right),\\
\widehat{A}_{5}^{n,3}&:=\frac{\sqrt{\Delta_n}}{8} \sum_{i=1}^{n-2 k_n+1} \sum_{p, q, u, v} \partial_{p q, u v}^2 g(\hat{c}_i^n)\beta_{i+k_n}^{n, uv}\left(c_{i+k_n}^{n, p q}-c_i^{n, p q}\right).
\end{align*}
In the {\Red Appendix \ref{Technical2_2}}, we show the following limits:
\begin{equation}\label{FrstStOfLmts}
\begin{aligned}
\widehat{A}_{0,0}^{n,3}, \widehat{A}_{1,1}^{n,3}&\to \frac{1}{8\theta} \int_0^T \sum_{p, q, u, v}\partial^2_{pq, uv} g\left(c_s\right) \check{c}_s^{pq, uv} d s \int^{\infty}_{-\infty} K(z)^2 d z\\
&\quad\quad\quad+\frac{\theta}{8} \int_0^T \sum_{p, q, u, v}\partial^2_{pq, uv} g\left(c_s\right) \tilde{c}_s^{pq, uv} d s\int_{-\infty}^{\infty} L(z)^2 d z,\\
\widehat{A}_{0,1}^{n,3},\widehat{A}_{1,0}^{n,3} &\to \frac{1}{8\theta} \int_0^T\sum_{p, q, u, v} \partial^2_{pq, uv} g\left(c_s\right) \check{c}_s^{pq, uv} d s \int^{\infty}_{-\infty} K(z)K(z-1) d z\\
&\quad\quad\quad+\frac{\theta}{8} \int_0^T\sum_{p, q, u, v} \partial^2_{pq, uv} g\left(c_s\right) \tilde{c}_s^{pq, uv} d s\int_{-\infty}^{\infty} L(z)L(z-1) d z,
\end{aligned}
\end{equation}
\begin{equation}\label{SndStOfLmts}
\begin{aligned}
\widehat{A}_{1}^{n,3}\to \frac{\theta}{8} \int_{0}^{T} \sum_{p, q, u, v} \partial_{p q, u v}^2 g(c_s) \tilde{c}_s^{p q, u v}ds,
\end{aligned}
\end{equation}
\begin{equation}\label{TrdStOfLmts}
\begin{aligned}
\widehat{A}_{2}^{n,3},\widehat{A}_{3}^{n,3}\to \frac{\theta}{8}\int_{0}^{T}\sum_{p, q, u, v} \partial_{p q, u v}^2 g(c_{s})\tilde{c}_{s}^{p q, u v}ds\left(\int_{1}^{\infty}K(z)dz+\int_{0}^{1}K(z)zdz\right),
\end{aligned}
\end{equation}
\begin{equation}\label{FourthStOfLmts}
\begin{aligned}
\widehat{A}_{4}^{n,3},\widehat{A}_{5}^{n,3}\to \frac{\theta}{8}\int_{0}^{T}\sum_{p, q, u, v} \partial_{p q, u v}^2 g(c_{s})\tilde{c}_{s}^{p q, u v}ds\left(\int_{0}^{1}K(z-1)(z-1)dz-\int_{-\infty}^{0}K(z-1)dz\right).
\end{aligned}
\end{equation}
%
We finally conclude \eqref{eq:bias_estimates}.}

\subsection{Proof of Corollary \ref{BiasCorrectedCLT0}}\label{SmplCorUnb}
To show (\ref{UnBsCLTb}), by Theorem \ref{clt}, we only need to eliminate the bias terms $A^{1}, A^{2}, A^{3}$. Note that
$$
k_n \sqrt{\Delta_n} g\left(\hat{c}_1^n\right) {\int_0^{\infty} L(u)d u}-k_n \sqrt{\Delta_n} g\left(\hat{c}_n^n\right){\int_{-\infty}^0 L(u) d u}\to \theta g\left(c_{0}\right)\int_{0}^{\infty}L(u)du-\theta g\left(c_{T}\right) \int_{-\infty}^{0}L(u)du,
$$
which cancels the term $A^{1}$. Then, by (\ref{eq:h}) and Theorem \ref{thm2.2}, we have
\begin{align*}
    &-\frac{\sqrt{\Delta_n}}{4} \sum_{i=1}^{n-2k_n+1}\sum_{p, q, u, v=1}^d \partial_{p q, u v}^2 g\left(\hat{c}_i^n\right)\left(\hat{c}_{i+k_n}^{n, p q}-\hat{c}_i^{n, p q}\right)\left(\hat{c}_{i+k_n}^{n, u v}-\hat{c}_i^{n, u v}\right)C_{K,1}=-2C_{K,1} \widehat{A}^{n, 3}\\
    & \to -\frac{1}{2\theta} \int_0^T \sum_{j, k, l, m}\partial^2_{jk, lm} g\left(c_s\right) \check{c}_s^{jk, lm} d s C_{K,2}\\
    &\quad-\frac{\theta}{2} \int_0^T \sum_{j, k, l, m}\partial^2_{jk, lm} g\left(c_s\right) \tilde{c}_s^{jk, lm} d s \left[\int_{-\infty}^{\infty} L(z)^2 d z+\left(\int_{-\infty}^{0}-\int_{0}^{\infty}\right) L(z) d z\right],
\end{align*}
and
\begin{align*}
    \frac{1}{2k_n\sqrt{\Delta_n}}V\left(h\right)^{n}_{T}\left[C_{K,2}-\int K^2(u) du\right]\to \frac{1}{2\theta}\int_0^T \sum_{p,q,u,v=1}^{d}\partial^{2}_{p q, u v} g\left(c_s\right) \check{c}_s^{p q, u v} d s\left[C_{K,2}-\int K^2(u) du\right].
\end{align*}
Combining both limits, the bias terms $A^{2}$ and $A^{3}$ can be eliminated.

\subsection{Proof of Theorem \ref{clt_suboptimal}}\label{PrfSubs23}
Suppose the $X$ is discontinuous and {$\hat{c}_i^n$ in (\ref{v_bar})} is given by (\ref{estivoltrun}). Let
\begin{equation}\label{v_bar_prime}
    \bar{V}(g)^{\prime n}_T = \Delta_n \sum_{i= 1}^{n} \left( g(\hat{c}^{\prime n}_i) - \frac{1}{2k_n} h(\hat{c}^{\prime n}_i) \int K^2(u) du\right)\bar{K}\left(t_{i}\right),
\end{equation}
where $\hat{c}_{i}^{\prime n}=\hat{c}_{i}^{n,X^{\prime}}$ is defined as in (\ref{MDKNN}) with {the continuous process $X^{\prime}$ defined as in \eqref{eq:X_conti}}. {As shown in \eqref{diffgprime}, we have $\mathbb{E}\left[\left|g\left(\hat{c}_i^n\right)-g\left(\hat{c}_i^{\prime n}\right)\right|\right]=o(\sqrt{\Delta_n})$. Following the same arguments, we can show that $\mathbb{E}\left[\left|h\left(\hat{c}_i^n\right)-h\left(\hat{c}_i^{\prime n}\right)\right|\right]=o(\sqrt{\Delta_n})$. Therefore, we {can conclude that}
$$
\frac{1}{\sqrt{\Delta_n}}\left[\bar{V}(g)^{n}_T-\bar{V}(g)^{\prime n}_T\right]\xrightarrow{\mathbb{P}} 0,
$$
which, in light of b) in Lemma \ref{jump_1}, also holds when $X$ is continuous under the weaker condition \eqref{moreassum1}.} Therefore, in what follows, we only need to consider $\hat{c}_i^n$ as in (\ref{MDKNN}) with Assumption \ref{assum3}. 

Similar to (\ref{eq:decomposition}), we have the decomposition \begin{equation*} 
\frac{1}{\sqrt{\Delta_{n}}}\left(\bar{V}(g)_T^{n}-V(g)_T\right)=\sum_{j=1}^{3} {V^{n}_{j}} + {\bar{V}^{n}_{4}} + {V^{n}_{5}},
\end{equation*} 
where {$V^{n}_{1}, V^{n}_{2}, V^{n}_{3}${,} and $V^{n}_{5}$} are defined as in (\ref{eq:V1234}) and  
\begin{equation*}
{\bar{V}^{n}_{4}}=\sqrt{\Delta_{n}} \sum_{i=1}^{n}\left(g\left(c_{i}^{n}+\beta_{i}^{n}\right)-g\left(c_{i}^{n}\right)-\nabla g\left(c_{i}^{n}\right) \beta_{i}^{n} - \frac{1}{2k_n} h\left(c^{n}_{i}+\beta^{n}_{i}\right) \int K^2(u) du\right).
\end{equation*}
Since $k_n^2\Delta_n \to {0}$ corresponds to (\ref{eq: k_n}) with $\theta = 0$, we can deduce {from \eqref{eq:v4}} that {$V^{n}_{3} \to 0 $}. The {asymptotic behavior of the terms $V^{n}_{1}$ and $V^{n}_{2}$ remain the same}: $V^{n}_{1} \to 0 $ {and} $V^{n}_{2} \stackrel{\mathcal{L}-\mathrm{s}}{\Longrightarrow} Z$. 

{We now consider the updated term $\bar{V}^{n}_{4}$.} We can write for some $\lambda_{i}\in [0, 1]$,
\begin{align*}
    \bar{V}^{n}_{4}&=\sqrt{\Delta_{n}} \sum_{i=1}^{n}\Bigg(\frac{1}{2}\sum_{p, q, u, v=1}^d\partial_{p q, u v}^2 g\left(c^{n}_{i}\right)\beta_i^{n, p q} \beta_i^{n, u v}-\frac{1}{2k_n}h\left(c^{n}_{i}+\beta^{n}_{i}\right)\int K^2(u) du\\
    &+\frac{1}{6}\sum_{p_{1}, q_{1}, p_{2}, q_{2},p_{3}, q_{3}=1}^d\partial_{p_{1} q_{1}, p_{2} q_{2},p_{3} q_{3}}^3 g\left(c^{n}_{i}+\lambda_{i}\beta^{n}_{i}\right)\beta_i^{n, p_{1} q_{1}} \beta_i^{n, p_{2} q_{2}}\beta_i^{n, p_{3} q_{3}}\Bigg)\\
    &=\sqrt{\Delta_{n}} \sum_{i=1}^{n}\Bigg(\frac{1}{2}\sum_{p, q, u, v=1}^d\partial_{p q, u v}^2 g\left(c^{n}_{i}\right)\beta_i^{n, p q} \beta_i^{n, u v}-\frac{1}{2k_n}h\left(c^{n}_{i}\right)\int K^2(u) du\\
    &\quad+\frac{1}{2k_n}\left[h\left(c^{n}_{i}\right)_{i}-h\left(c^{n}_{i}+\beta^{n}_{i}\right)\right]\int K^2(u) du\\
    &\quad+\frac{1}{6}\sum_{p_{1}, q_{1}, p_{2}, q_{2},p_{3}, q_{3}=1}^d\partial_{p_{1} q_{1}, p_{2} q_{2},p_{3} q_{3}}^3 g\left(c^{n}_{i}+\lambda_{i}\beta^{n}_{i}\right)\beta_i^{n, p_{1} q_{1}} \beta_i^{n, p_{2} q_{2}}\beta_i^{n, p_{3} q_{3}}\Bigg)\\
    &=:\sum_{i=1}^{n}{(v(1)_{i}^{n}+v(2)_{i}^{n})},
\end{align*}
where 
\begin{equation}
\begin{array}{l}
v(1)_{i}^{n}=\frac{\sqrt{\Delta_{n}}}{2}\sum_{p, q, u, v=1}^d\partial_{p q, u v}^2 g\left(c^{n}_{i}\right) \left[\beta_i^{n, p q} \beta_i^{n, u v}-\frac{1}{k_n}\check{c}^{n,pq,nv}_{i}\int K^2(u) du\right]\\
\left|v(2)_{i}^{n}\right| \leq C \sqrt{\Delta_{n}}\left(1+\left\|\beta_{i}^{n}\right\|\right)^{{\ell}-3}\left\|\beta_{i}^{n}\right\|^{3} + \frac{C\sqrt{\Delta_n}}{k_n}\left(1+ \left\|\beta^n_i \right\|^{{\ell}-1} \right)\left\|\beta^n_i \right\| .
\end{array}
\end{equation}Here, we used (\ref{eq:g_prime_bound}) and  the boundedness of $c$.  
For $v(2)^n_i $, by (\ref{eq:k_n_suboptimal}), $k^2_n\Delta_n\to 0$, we have $k^{-1}_n \gg k_n\Delta_n$. Then, (\ref{eq:beta_bound}) becomes
$$
\mathbb{E}\left[\left\|\beta_i^n\right\|^q\right]\leq C_qk_n^{-q / 2}.
$$
Applying this inequality and we have
\begin{align*}
\mathbb{E}\left|v(2)_{i}^{n}\right|&\leq C \sqrt{\Delta_n}\mathbb{E}\left[\left(1+\left\|\beta_i^n\right\|\right)^{\ell-3}\left\|\beta_i^n\right\|^3\right]+\frac{C \sqrt{\Delta_n}}{k_n}\mathbb{E}\left[\left(1+\left\|\beta_i^n\right\|^{\ell-1}\right)\left\|\beta_i^n\right\|\right]\\
&\leq C\sqrt{\Delta_n}\left(\mathbb{E}\left[\left\|\beta_i^n\right\|^6\right]\right)^{1/2}+\frac{C \sqrt{\Delta_n}}{k_n}\left(\mathbb{E}\left[\left\|\beta_i^n\right\|^2\right]\right)^{1/2}\leq C\sqrt{\Delta_n}k_n^{-3/2},
\end{align*}
based on which, we have  
\begin{equation*}
\sum_{i=1}^n v(2)^n_i \to 0
\end{equation*}
since by (\ref{eq:k_n_suboptimal}), $k^{-3/2}_n\Delta_n^{-1/2}\to 0$. For $v(1)^n_i$, recall (\ref{eq:beta2_decomp}), { and, since $\theta = 0$, \eqref{eq:A.73}, \eqref{eq:snrto0}, and the paragraph after \eqref{eq:A.84} implies that}
\begin{equation*}
\frac{\sqrt{\Delta_n}}{2}\sum_{i= 1}^n \sum_{p, q, u, v=1}^d\partial_{p q, u v}^2 g\left(c^{n}_{i}\right)\left[\xi^{n}_{2,i}+\xi^{n}_{3,i}+\xi^{n}_{4,i}+\xi^{n}_{5,i}+\xi^{n}_{6,i}\right] \to 0.
\end{equation*}
We are {left} with 
\begin{equation}\label{VLSE}
\begin{split}
&\frac{\sqrt{\Delta_n}}{2}\sum_{i= 1}^n \sum_{p, q, u, v=1}^d\partial_{p q, u v}^2 g\left(c^{n}_{i}\right)  \left[\xi^{n}_{1,i}- \frac{1}{k_n}\check{c}^{n,pq,nv}_{i}\int K^2(u) du\right] \\
&=\frac{\sqrt{\Delta_n}}{2}\sum_{i= 1}^n \sum_{p, q, u, v=1}^d\partial_{p q, u v}^2 g\left(c^{n}_{i}\right) \left[\sum_{j=1}^{n} K_{{ b}}^{2}\left(t_{j-1}-t_{i}\right)\alpha_j^{n, p q} \alpha_j^{n, u v}- \frac{1}{k_n}\check{c}^{n,pq,nv}_{i}\int K^2(u) du\right] \\
&=\frac{\sqrt{\Delta_n}}{2}\sum_{i= 1}^n \sum_{p, q, u, v=1}^d\partial_{p q, u v}^2 g\left(c^{n}_{i}\right) \left[\sum_{j=1}^{n} K_{{ b}}^{2}\left(t_{j-1}-t_{i}\right)\left( \alpha_j^{n, p q} \alpha_j^{n, u v} - \check{c}^{n,pq,nv}_{j-1}\Delta_n^2 \right)\right.\\
& \quad+\sum_{j=1}^{n} K_{{ b}}^{2}\left(t_{j-1}-t_{i}\right) \left(\check{c}^{n,pq,nv}_{j-1}\Delta_n^2 - \check{c}^{n,pq,nv}_{i}\Delta_n^2 \right)\\
&\quad+\left.\sum_{j=1}^{n} K_{{ b}}^{2}\left(t_{j-1}-t_{i}\right)  \check{c}^{n,pq,nv}_{i}\Delta_n^2 - \frac{1}{k_n}\check{c}^{n,pq,nv}_{i}\int K^2(u) du\right] \\
&=: P_1 + P_2 + P_3.
\end{split}
\end{equation}
For $P_1$, note that
\begin{align*}
    P_1&:=\frac{\sqrt{\Delta_n}}{2}\sum_{i= 1}^n \sum_{p, q, u, v=1}^d\partial_{p q, u v}^2 g\left(c^{n}_{i}\right)\sum_{j=1}^{n} K_{{ b}}^{2}\left(t_{j-1}-t_{i}\right)\left( \alpha_j^{n, p q} \alpha_j^{n, u v} - \check{c}^{n,pq,nv}_{j-1}\Delta_n^2 \right)\\
    &=\frac{\sqrt{\Delta_n}}{2}\sum_{i= 1}^n \sum_{j=1}^{n}\sum_{p, q, u, v=1}^d\left[\partial_{p q, u v}^2 g\left(c^{n}_{i}\right)-\partial_{p q, u v}^2 g\left(c^{n}_{j-1}\right)\right] K_{{ b}}^{2}\left(t_{j-1}-t_{i}\right)\left( \alpha_j^{n, p q} \alpha_j^{n, u v} - \check{c}^{n,pq,nv}_{j-1}\Delta_n^2 \right)\\
    &+\frac{\sqrt{\Delta_n}}{2}\sum_{i= 1}^n \sum_{j=1}^{n}\sum_{p, q, u, v=1}^d\partial_{p q, u v}^2 g\left(c^{n}_{j-1}\right) K_{{ b}}^{2}\left(t_{j-1}-t_{i}\right)\left( \alpha_j^{n, p q} \alpha_j^{n, u v} - \check{c}^{n,pq,nv}_{j-1}\Delta_n^2 \right)\\
    &=:P_{11}+P_{12},
\end{align*}
where
\begin{align*}
    &\mathbb{E}\left|P_{11}\right|\leq C\sqrt{\Delta_n}\sum_{i= 1}^n \sum_{j=1}^{n}\sum_{p, q, u, v=1}^d K_{{ b}}^{2}\left(t_{j-1}-t_{i}\right)\\
    &\qquad\qquad\qquad\qquad\sqrt{\mathbb{E}\left[\partial_{p q, u v}^2 g\left(c^{n}_{i}\right)-\partial_{p q, u v}^2 g\left(c^{n}_{j-1}\right)\right]^2\mathbb{E}\left[\left( \alpha_j^{n, p q} \alpha_j^{n, u v} - \check{c}^{n,pq,nv}_{j-1}\Delta_n^2 \right)^2\right]}\\
    &\leq C\sqrt{\Delta_n}\sum_{i= 1}^n \sum_{j=1}^{n}K_{{ b}}^{2}\left(t_{j-1}-t_{i}\right)\sqrt{\left|t_{j-1}-t_{i}\right|}\Delta_n^2\leq C\frac{\sqrt{\Delta_n}}{\sqrt{b}}\int K^2(u)\sqrt{\left|u\right|}du\to 0.
\end{align*}
Write $P_{12}=\sum_{j=1}^{n}\zeta_{j}^{n}$, with
$$
\zeta_{j}^{n}:=\frac{\sqrt{\Delta_n}}{2}\sum_{i= 1}^n\sum_{p, q, u, v=1}^d\partial_{p q, u v}^2 g\left(c^{n}_{j-1}\right) K_{{ b}}^{2}\left(t_{j-1}-t_{i}\right)\left( \alpha_j^{n, p q} \alpha_j^{n, u v} - \check{c}^{n,pq,nv}_{j-1}\Delta_n^2 \right).
$$
Applying  (\ref{eq:better_bound_alpha}),
\begin{align*}
    &\mathbb{E}\left|\sum_{j=1}^n \mathbb{E}\left[\zeta_j^n \mid \mathcal{F}_{j-1}^n\right]\right|\\
    &=\mathbb{E}\left|\frac{\sqrt{\Delta_n}}{2}\sum_{j=1}^n\sum_{i= 1}^n\sum_{p, q, u, v=1}^d\partial_{p q, u v}^2 g\left(c^{n}_{j-1}\right) K_{{ b}}^{2}\left(t_{j-1}-t_{i}\right)\left( \mathbb{E}\left[\left.\alpha_j^{n, p q} \alpha_j^{n, u v}\right|\mathcal{F}_{j-1}^n\right] - \check{c}^{n,pq,nv}_{j-1}\Delta_n^2 \right)\right|\\
    &\leq C\sqrt{\Delta_n}\sum_{j=1}^n\sum_{i= 1}^n K_{{ b}}^{2}\left(t_{j-1}-t_{i}\right)\Delta_n^{5 / 2}\leq C\frac{\Delta_n}{b} \int K^2(u)du\to 0.
\end{align*}
Also,
\begin{align*}
&\mathbb{E}\left(\sum_{j=1}^n\left(\zeta_j^n-\mathbb{E}\left[\zeta_j^n \mid \mathcal{F}_{j-1}^n\right]\right)\right)^2\leq \sum_{j=1}^n \mathbb{E}\left[\left(\zeta_j^n\right)^2\right]\\
&\leq C\Delta_n\sum_{j=1}^n \left(\sum_{i= 1}^n K_{{ b}}^{2}\left(t_{j-1}-t_{i}\right)\right)^2 \mathbb{E}\left[\left(\sum_{p, q, u, v=1}^d\partial_{p q, u v}^2 g\left(c^{n}_{j-1}\right)\right)^2\left( \alpha_j^{n, p q} \alpha_j^{n, u v} - \check{c}^{n,pq,nv}_{j-1}\Delta_n^2 \right)^2\right]\\
&\leq C\Delta_n^3 \sum_{j=1}^n \left(\Delta_n\sum_{i= 1}^n K_{{ b}}^{2}\left(t_{j-1}-t_{i}\right)\right)^2\leq C\frac{\Delta_n^3}{b^2}\sum_{j=1}^n \left(\int K^2(u)du\right)^2\to 0.
\end{align*}
Thus, $P_{12}\to 0$ and hence, we have that $P_1$ converges to 0. Applying Lemma \ref{K_discrete} with $\tilde{K} = K^2 $, then $\tilde{K}_{b}(x) = b K_{{ b}}^2(x) $ and 
\begin{equation*}
\Delta_n \sum_{j=1}^{n} K_{{ b}}^{2}\left(t_{j-1}-t_{i}\right) -\frac{1}{ b} \int K^2(u) du = O\left(\frac{\Delta_n}{{ b}^{2}} \right).
\end{equation*} 
Thus,
\begin{align*}
    P_3&:=\frac{\sqrt{\Delta_n}}{2}\sum_{i= 1}^n \sum_{p, q, u, v=1}^d\partial_{p q, u v}^2 g\left(c^{n}_{i}\right)\check{c}^{n,pq,nv}_{i}\left[\sum_{j=1}^{n} K_{{ b}}^{2}\left(t_{j-1}-t_{i}\right) \Delta_n^2 - \frac{1}{k_n}\int K^2(u) du\right]\\
    &=\frac{\sqrt{\Delta_n}}{2}\sum_{i= 1}^n \sum_{p, q, u, v=1}^d\partial_{p q, u v}^2 g\left(c^{n}_{i}\right)\check{c}^{n,pq,nv}_{i}O\left(\frac{\Delta^2_n}{{ b}^{2}} \right)=O_{p}\left(\frac{\Delta^{3/2}_n}{{b}^{2}} \right).
\end{align*}
{which implies that $P_{3}$ converges to $0$}. {For $P_2$, we proceed as follows:}
\begin{equation*}
\begin{split}
&\mathbb{E} \left|\sum_{j=1}^{n} K_{{ b}}^{2}\left(t_{j-1}-t_{i}\right) \left(\check{c}^{n,pq,nv}_{j-1}\Delta_n^2 - \check{c}^{n,pq,nv}_{i}\Delta_n^2 \right)\right|\\
& \leq C\Delta_n^2\sum_{j=1}^{n} K_{{ b}}^{2}\left(t_{j-1}-t_{i}\right) \sqrt{|t_{j-1}-t_{i}|}\\
&\leq C\frac{\Delta_n}{\sqrt{ b}} \int K^2(u) \sqrt{|u|} du.
\end{split}
\end{equation*}
Therefore, $\mathbb{E}|P_2|\leq C \Delta_n^{-1/2}\frac{\Delta_n}{\sqrt{ b}} \leq C \frac{1}{\sqrt{k_n}} \to 0$. 
{This shows that $ \bar{V}^{n,4}\to 0$ and we conclude \eqref{eq:convergence_in_law_unbiased}.}

\section{Other Technical Proofs}\label{ApOTPrfs}
\subsection{Technical Proofs related to Section \ref{SectionA.1}}\label{TechnicalA_1}
\begin{proof}[Proof of (\ref{eq:beta_bound})]
Suppose $q\geq 1$. Note that
\begin{align*}
    \alpha_j^n&=\left(\Delta_j^n X^{\prime}\right)\left(\Delta_j^n X^{\prime}\right)^*-c_{j-1}^n \Delta_n\\
    &=\int_{t_{j-1}}^{t_j}\left[2 \mu^{\prime}_s\left(X^{\prime}_s-X^{\prime}_{t_{j-1}}\right)+\left(c_s-c_{j-1}^n\right)\right] d s+2 \int_{t_{j-1}}^{t_j}\left(X^{\prime}_s-X^{\prime}_{t_{j-1}}\right) \sigma_s d W_s.
\end{align*}
We check first the following:
\begin{equation}\label{b_1}
\begin{aligned}
    &\mathbb{E}\left[\left\|B_{1}\right\|^q\right]:=\mathbb{E}\left[\left\|\sum_{j=1}^n K_b\left(t_{j-1}-t_i\right) \int_{t_{j-1}}^{t_j}2\mu^{\prime}_s\left(X^{\prime}_s-X^{\prime}_{t_{j-1}}\right)  ds\right\|^q\right]\\
    &\leq \left(\sum_{j=1}^n \left|K_b\left(t_{j-1}-t_i\right)\right|\right)^{q-1}\sum_{j=1}^n \left|K_b\left(t_{j-1}-t_i\right)\right|\mathbb{E}\left[\left\|\int_{t_{j-1}}^{t_j}2\mu^{\prime}_s\left(X^{\prime}_s-X^{\prime}_{t_{j-1}}\right)  ds\right\|^q\right]\\
    &\leq C\left(\sum_{j=1}^n \left|K_b\left(t_{j-1}-t_i\right)\right|\right)^{q-1}\sum_{j=1}^n \left|K_b\left(t_{j-1}-t_i\right)\right|\mathbb{E}\left[\sup _{s \in\left[0,\Delta_n\right]}\left\|X^{\prime}_{t_{j-1}+s}-X^{\prime}_{t_{j-1}}\right\|^q \right]\Delta_n^{q}\\
    &\leq C \left(\Delta_n\sum_{j=1}^n \left|K_b\left(t_{j-1}-t_i\right)\right|\right)^{q-1}\Delta_n\sum_{j=1}^n \left|K_b\left(t_{j-1}-t_i\right)\right|\Delta_n^{q/2}\leq C\left(\int_{-\infty}^{\infty}\left|K(u)\right|du\right)^{q}\Delta_n^{q / 2},
\end{aligned}
\end{equation}
where the third inequality holds with (\ref{eq:X_c_bound}) and the boundedness of $\mu^{\prime}_s$. Similarly, by BDG inequality, (\ref{eq:X_c_bound}), and the boundedness of $\sigma_s$,
\begin{equation} 
\begin{aligned}
&\mathbb{E}\left[\left\|B_{2}\right\|^q\right]:=\mathbb{E}\left[\left\|\sum_{j=1}^n K_b\left(t_{j-1}-t_i\right) \int_{t_{j-1}}^{t_j}2\left(X^{\prime}_s-X^{\prime}_{t_{j-1}}\right) \sigma_s d W_s\right\|^q\right]   \\
&=\mathbb{E}\left[\left\|\int_0^T \sum_{j=1}^n K_b\left(t_{j-1}-t_i\right) \mathbbm{1}_{\left\{t_{j-1} \leq s \leq t_j\right\}}2\left(X^{\prime}_s-X^{\prime}_{t_{j-1}}\right) \sigma_s d W_s\right\|^q\right]\\
&\leq \mathbb{E}\left[\left(\int_0^T\left\|\sum_{j=1}^n K_b\left(t_{j-1}-t_i\right) \mathbbm{1}_{\left\{t_{j-1} \leq s \leq t_j\right\}}2\left(X^{\prime}_s-X^{\prime}_{t_{j-1}}\right) \sigma_s\right\|^2 d s\right)^{q / 2}\right]\\
&=\mathbb{E}\left[\left(\int_0^T \sum_{j=1}^n K_b\left(t_{j-1}-t_i\right)^2 \mathbbm{1}_{\left\{t_{j-1} \leq s \leq t_j\right\}}\left\|2\left(X^{\prime}_s-X^{\prime}_{t_{j-1}}\right) \sigma_s\right\|^2 d s\right)^{q / 2}\right]\\
&\leq C \mathbb{E}\left[\left(\sum_{j=1}^n K_b\left(t_{j-1}-t_i\right)^2 \int_{t_{j-1}}^{t_j}\left\|X^{\prime}_s-X^{\prime}_{t_{j-1}}\right\|^2 d s\right)^{q / 2}\right]\\
&\leq C \left(\sum_{j=1}^n K_b\left(t_{j-1}-t_i\right)^2\right)^{q/2-1}\sum_{j=1}^n K_b\left(t_{j-1}-t_i\right)^2 \mathbb{E}\left[\left(\int_{t_{j-1}}^{t_j}\left\|X^{\prime}_s-X^{\prime}_{t_{j-1}}\right\|^2 d s\right)^{q / 2}\right]\\
&\leq C \left(\sum_{j=1}^n K_b\left(t_{j-1}-t_i\right)^2\right)^{q/2-1}\sum_{j=1}^n K_b\left(t_{j-1}-t_i\right)^2 \mathbb{E}\left[\sup _{s \in\left[0,\Delta_n\right]}\left\|X^{\prime}_{t_{j-1}+s}-X^{\prime}_{t_{j-1}}\right\|^q \right]\Delta_n^{q/2}\\
&\leq C \left(\sum_{j=1}^n K_b\left(t_{j-1}-t_i\right)^2\right)^{q/2}\Delta_n^{q}\leq C\left(\int_{-\infty}^{\infty}K(u)^2du\right)^{q/2}\Delta_n^{q / 2} b^{-q / 2}.
\end{aligned}
\end{equation}
Then, by the boundedness of $\tilde{\mu}_t$,
\begin{equation}
\begin{aligned}
&\mathbb{E}\left[\left\|B_{3}\right\|^q\right]:=\mathbb{E}\left[\left\|\sum_{j=1}^n K_b\left(t_{j-1}-t_i\right) \int_{t_{j-1}}^{t_j}\int_{t_{j-1}}^s \tilde{\mu}_t d t  ds\right\|^q\right]\\
&\leq C\left(\Delta_n^2 \sum_{j=1}^n \left|K_b\left(t_{j-1}-t_i\right)\right|\right)^q\leq C\left(\int \left|K(u)\right|du\right)^q \Delta_n^q.
\end{aligned}
\end{equation}
With the boundedness of $\tilde{\sigma}_s$, for some $s_{j}\in [t_{j-1},t_{j}]$,
\begin{equation}
\begin{aligned}
&\mathbb{E}\left[\left\|B_{4}\right\|^q\right]:=\mathbb{E}\left[\left\|\sum_{j=1}^n K_b\left(t_{j-1}-t_i\right) \int_{t_{j-1}}^{t_j}\int_{t_{j-1}}^s \tilde{\sigma}_t d W_t d s\right\|^q\right]\\
&=\mathbb{E}\left[\left\|\sum_{j=1}^n K_b\left(t_{j-1}-t_i\right) \int_{t_{j-1}}^{s_j} \tilde{\sigma}_s d W_s \Delta_n\right\|^q\right]\\
&=\mathbb{E}\left[\left\|\int_0^T \sum_{j=1}^n K_b\left(t_{j-1}-t_i\right) \mathbbm{1}_{\left\{t_{j-1} \leq s \leq s_j\right\}} \tilde{\sigma}_s \Delta_n d W_s\right\|^q\right]\\
&\leq \mathbb{E}\left[\left(\int_0^T \sum_{j=1}^n K_b\left(t_{j-1}-t_i\right)^2 \mathbbm{1}_{\left\{t_{j-1} \leq s \leq s_j\right\}} \Delta_n^2\left\|\tilde{\sigma}_s\right\|^2 d s\right)^{q / 2}\right]\\
&\leq C\left(\sum_{j=1}^n K_b\left(t_{j-1}-t_i\right)^2 \Delta_n^3\right)^{q / 2}\leq\left(\int_{-\infty}^{\infty} K(u)^2 d u b^{-1} \Delta_n^2\right)^{q / 2}\leq C_q \Delta_n^q b^{-q / 2}.
\end{aligned}
\end{equation}
Finally, since
$$
c_{j-1}^n-c_{i}^n=\int^{t_{j-1}}_{t_{i}} \tilde{\mu}_s d s+\int^{t_{j-1}}_{t_{i}} \tilde{\sigma}_s d W_s,
$$
we have
\begin{equation}
\begin{aligned}
&\mathbb{E}\left[\left\|B_{5}\right\|^q\right]:=\mathbb{E}\left[\left\|\sum_{j=1}^n K_b\left(t_{j-1}-t_i\right)\int^{t_{j-1}}_{t_{i}} \tilde{\mu}_s d s\Delta_n\right\|^q\right]\\
&\leq C\left(\Delta_n\sum_{j=1}^n \left|K_b\left(t_{j-1}-t_i\right)\right|\left|t_{j-1}-t_i\right|\right)^q\leq C \left(\int \left|K(u)u\right|du\right)^q b^q.
\end{aligned}
\end{equation}
and
\begin{equation}\label{b_6}
\begin{aligned}
&\mathbb{E}\left[\left\|B_{6}\right\|^q\right]:=\mathbb{E}\left[\left\|\sum_{j=1}^n K_b\left(t_{j-1}-t_i\right) \int_{t_i}^{t_{j-1}} \tilde{\sigma}_s d W_s \Delta_n\right\|^q\right]\\
&=\mathbb{E}\left[\left\|\int_0^T \sum_{j=1}^n K_b\left(t_{j-1}-t_i\right) \mathbbm{1}_{\left\{t_i \leq s \leq t_{j-1}\right\}} \Delta_n \tilde{\sigma}_s d W_s\right\|^q\right]\\
&\leq \mathbb{E}\left[\left(\int_0^T\left\|\sum_{j=1}^n K_b\left(t_{j-1}-t_i\right) \mathbbm{1}_{\left\{t_i \leq s \leq t_{j-1}\right\}} \Delta_n \tilde{\sigma}_s\right\|^2 d s\right)^{q / 2}\right]\\
&\leq C \left(\int_0^T \sum_{j, l} K_b\left(t_{j-1}-t_i\right) K_b\left(t_{l-1}-t_i\right) \mathbbm{1}_{\left\{t_i \leq s \leq t_{j-1}\right\}} \mathbbm{1}_{\left\{t_i \leq s \leq t_{l-1}\right\}} \Delta_n^2 d s\right)^{q / 2}\\
&\leq C\left(\sum_j K_b\left(t_{j-1}-t_i\right) \sum_{l, l \leq j} K_b\left(t_{l-1}-t_i\right)\left(t_{l-1}-t_i\right) \Delta_n^2\right)^{q/2}\\
&\leq C\left(\int_{-\infty}^{\infty} K(u) d u \int_{-\infty}^{\infty} K(u) u d u b\right)^{q / 2} \leq C_q b^{q / 2}.
\end{aligned}
\end{equation}
Recall the definition of $\beta_i^n$, 
\begin{align*}
    \mathbb{E}\left[\left\|\beta_i^n\right\|^q\right] &= \mathbb{E}\left[\left\|\frac{1}{\bar{K}(t_{i})}\sum_{j=1}^6B_{j}\right\|^q\right]\leq C_q\mathbb{E}\left[\left(\sum_{j=1}^6\left\|B_{j}\right\|\right)^q\right].
\end{align*}
Then, applying Multinomial theorem, the cross terms can be bounded by Cauchy-Schwartz inequality with (\ref{b_1})-(\ref{b_6}). In conclusion, we obtain
$$
\mathbb{E}\left[\left\|\beta_i^n\right\|^q\right] \leq C_q\left(\Delta_n^{q / 2} b^{-q / 2}+b^{q / 2}\right), \quad \text{for}\ q\geq 1.
$$
\end{proof}
\begin{proof}[Proof of (\ref{eq:K_discrete})]
	From Lemma \ref{K_discrete}, we have 
\begin{align*}
   \Delta_n \sum_{i=1}^n K_{{b}}\left(t_{i-1} - \tau \right) &= \int_0^T K_{{b}}(s-\tau) ds +D_1(1)\\
   &= \int_{-\tau/{b}}^{(T-\tau)/{b}} K(u) du+O\left(\frac{\Delta_n}{b}\right)\\
   &=1-\left(\int_{-\infty}^{-\tau/{b}}+\int_{(T-\tau)/{b}}^{\infty}\right) K(u) du+O\left(\frac{\Delta_n}{b}\right).
 \end{align*}
 Therefore, using the Assumption \ref{kernel}-c(ii) and denoting $\underline{\tau}=\inf\{\tau\in K\}>0$ and $\bar{\tau}=\sup\{\tau\in K\}<T$, 
 \begin{align*}
   \sup_{\tau\in K}\left|\Delta_n \sum_{i=1}^n K_{{b}}\left(t_{i-1} - \tau \right)-1\right|&\leq\left(\int_{-\infty}^{-\underline{\tau}/{b}}+\int_{(T-\bar{\tau})/{b}}^{\infty}\right) |K(u)| du+O\left(\frac{\Delta_n}{b}\right)\\
   &\leq C \left(\frac{\underline{\tau}}{b}\right)^{-1-\varepsilon}+C\left(\frac{T-\bar{\tau}}{b}\right)^{-1-\varepsilon}+O\left(\frac{\Delta_n}{b}\right)\\
     &=O\left(\frac{\Delta_n}{b}\right).
\end{align*}
\end{proof}
\begin{proof}[Proof of (\ref{lowerboundofdelta})]
To show $\left|\Bar{K}\left(t\right)\right|=\left|\Delta_{n}\sum_{j=1}^{n}K_{{b}}\left(t_{j-1}-t\right)\right|>\delta$ uniformly on $[0,T]$ for some $\delta>0$, note that by Lemma \ref{K_discrete}, $\forall t$,
$$
\Bar{K}\left(t\right)=1-\left(\int_{-\infty}^{-t / h}+\int_{\left(T-t\right) / h}^{\infty}\right) K(u) d u+D_1(1)
$$
where
$$
D_1(1)=\frac{1}{2} \left[K\left(A^{+}\right)-K\left(B^{-}\right)\right] \frac{\Delta_n}{b}+o\left(\frac{\Delta_n}{b}\right)\to 0
$$
uniformly. For any $\epsilon_{0}\in (0,T/2)$, when $t\in[\epsilon_{0},T-\epsilon_{0}]$,
\begin{align*}
    &\left|\left(\int_{-\infty}^{-t / h}+\int_{\left(T-t\right) / h}^{\infty}\right) K(u) d u\right|\leq \left(\int_{-\infty}^{-t / h}+\int_{\left(T-t\right) / h}^{\infty}\right) \left|K(u)\right| d u\\
    &\leq \left(\int_{-\infty}^{-\epsilon_{0} / h}+\int_{\epsilon_{0} / h}^{\infty}\right) \left|K(u)\right| d u\to 0
\end{align*}
uniformly. When $t\in [0,\epsilon_{0})$, note that
$\left|\int_{\left(T-t\right) / h}^{\infty}K(u)du\right|\to 0 $ uniformly. If $t=0$,
$$
\left|1-\int_{-\infty}^{-t / h}K(u)du\right|=\left|\int_{0}^{\infty}K(u)du\right|>\delta.
$$
If $t\neq 0$, consider the property when $\left|x\right|\to \infty$, $K(x)x^{2}\to 0$. Thus, $\exists x_{0}$, s.t. $\forall \left|x\right|>x_{0}$, we have $\left|K(x)\right|\leq \frac{1}{x^{2}}$. Then, $\forall t\in (0,\epsilon_{0})$, $\exists h$, s.t. $h<\frac{t}{x_{0}}$, and
\begin{align*}
    \left|\int_{-\infty}^{-t / h}K(u)du\right|\leq \int_{-\infty}^{-t / h}\left|K(u)\right|du\leq \int_{-\infty}^{-t / h}\frac{1}{u^{2}}du=\frac{b}{t}\leq \frac{1}{x_{0}}.
\end{align*}
Hence, we have
$$
\left|1-\int_{-\infty}^{-t / h}K(u)du\right|>\delta,
$$
where $\delta$ is not related to $t$. The upper bound can be obtained by (\ref{eq:K_discrete}) and small enough $\delta$.
\end{proof}
\begin{proof}[Proof of Lemma \ref{K_discrete}]
{First, for finite number of undefined points of $K^{\prime}$ in $(A,B)\subset \mathbb{R}$, the same arguments in \ref{a.2.5} can be applied so that we may assume that $K^{\prime}$ is defined in all $(A,B)\subset \mathbb{R}$. For $\tau\in(0,T)$,} the term $D_1(f)$ can be written as
\begin{equation}\label{lemmaa.1.1}
\begin{aligned}
D_{1}(f)
&= \frac{1}{b} \sum_{i=1}^{n} \int_{t_{i-1}}^{t_{i}}\left[K\left(\frac{t_{i-1}-\tau}{b}\right)-K\left(\frac{t-\tau}{b}\right)\right] f(c_t) d t \\
&= \frac{1}{b} \sum_{i=1}^{n} \int_{t_{i-1}}^{t_{i}} K^{\prime}\left(\frac{s_{t}-\tau}{b}\right) \frac{t_{i-1}-t}{b} f(c_t) d t \\
&= \frac{1}{b} \sum_{i=1}^{n} \int_{t_{i-1}}^{t_{i}} K^{\prime}\left(\frac{t_{i-1}-\tau}{b}\right) \frac{t_{i-1}-t}{b} f(c_t) d t \\
&\quad+\frac{1}{b} \sum_{i=1}^{n} \int_{t_{i-1}}^{t_{i}}\left[K^{\prime}\left(\frac{s_{t}-\tau}{b}\right)-K^{\prime}\left(\frac{t_{i-1}-\tau}{b}\right)\right] \frac{t_{i-1}-t}{b} f(c_t) d t, \\
&= D_{11}+D_{12},
\end{aligned}
\end{equation} 
for some $s_t \in (t_{i-1},t_i)$. $D_{12}$ can be controlled as the following: 
\begin{equation}\label{lemmaa.1.2}
\begin{aligned}
\left|D_{12}\right| & \leq \frac{M_{f} \Delta^{2}}{2 b^{2}} \sum_{i=1}^{n} \max _{t \in\left[t_{i-1}, t_{i}\right]}\left|K^{\prime}\left(\frac{t_{i-1}-\tau}{b}\right)-K^{\prime}\left(\frac{t-\tau}{b}\right)\right| \\
& \leq \frac{M_{f} \Delta^{2}}{2 b^{2}} V_{-\infty}^{\infty}\left(K^{\prime}\right)=O\left(\frac{\Delta^{2}}{b^{2}}\right),
\end{aligned}
\end{equation}
where $M_f:= \sup_{t \in[0,T]}|f(c_t)| < \infty$  by (\ref{eq:localization}) and  the assumption on $f$, and $V_{-\infty}^{\infty}\left(K^{\prime}\right)$ is the total variation of $K'$. Next, consider the term $D_{11}$. For any $\delta\in (0,\min(T-\tau,\tau)) $, we have 
\begin{equation*}
\begin{aligned}
D_{11} &=\frac{1}{b}\left(\sum_{\left|t_{i}-\tau\right|<\delta}+\sum_{\delta \leq\left|t_{i}-\tau\right|}
\right) K^{\prime}\left(\frac{t_{i-1}-\tau}{b}\right) \int_{t_{i-1}}^{t_{i}} \frac{t_{i-1}-t}{b} f(c_t) d t \\
& \triangleq D_{111}+D_{112} .
\end{aligned}
\end{equation*}
For an arbitrary, but fixed, $\delta>0$, the term $D_{112}$ is such that 
\begin{equation*}
\begin{aligned}
\left|D_{112}\right| & \leq \frac{1}{b} \sum_{\delta \leq\left|t_{i}-\tau\right|}\left|K^{\prime}\left(\frac{t_{i-1}-\tau}{b}\right)\right| \int_{t_{i-1}}^{t_{i}} \frac{\left|t_{i-1}-t\right|}{b}|f(c_t)| d t \\
& \leq \frac{M_{f} \Delta}{2 b} \sum_{\left|t_{i}-\tau\right| \geq \delta} \frac{\Delta}{b}\left|K^{\prime}\left(\frac{t_{i-1}-\tau}{b}\right)\right|=o\left(\frac{\Delta}{b}\right),
\end{aligned}
\end{equation*}
since, due to {(c) and (d) of Assumption \ref{kernel}}, for some $s_{i-1} \in (t_{i-1},t_{i}) $ and  for large enough $n$, 
\begin{equation}\label{LoOAS}
\begin{aligned}
& \sum_{\left|t_{i}-\tau\right| \geq \delta} \frac{\Delta}{b}\left|K^{\prime}\left(\frac{t_{i-1}-\tau}{b}\right)\right| \\
& \quad\leq \sum_{\left|t_{i}-\tau\right| \geq \delta} \int_{\left(t_{i-1}-\tau\right) / b}^{\left(t_{i}-\tau\right) / b}\left|K^{\prime}(t)\right| d t+\frac{\Delta}{b} \sum_{\left|t_{i}-\tau\right| \geq \delta}\left|K^{\prime}\left(\frac{t_{i-1}-\tau}{b}\right)-K^{\prime}\left(\frac{s_{i-1}-\tau}{b}\right)\right| \\
&\qquad\leq \left(\int_{\delta / 2b}^{+\infty}+\int_{-\infty}^{-\delta / b}\right)\left|K^{\prime}(t)\right| d t+\frac{\Delta}{b} V_{-\infty}^{\infty}\left(K^{\prime}\right)=o(1).
\end{aligned}
\end{equation}
Now consider the first term $D_{111}$. Fix $\varepsilon>0$. Since $c$ and $f$ are continuous,  there exists $\delta=\delta(\varepsilon)>0 $ such that $\left|f(c_t) - f(c_s) \right| \le \varepsilon$ whenever $|t-s|<\delta$. Then,
\begin{equation}
\begin{aligned}
D_{111}=& \frac{1}{b}\left(\sum_{\left|t_{i}-\tau\right|<\delta, K^{\prime} \leq 0}+\sum_{\left|t_{i}-\tau\right|<\delta, K^{\prime}>0}\right) K^{\prime}\left(\frac{t_{i-1}-\tau}{b}\right) \int_{t_{i-1}}^{t_{i}} \frac{t_{i-1}-t}{b} f(c_t) d t \\
\leq & \frac{1}{b} \sum_{\left|t_{i}-\tau\right|<\delta, K^{\prime} \leq 0} K^{\prime}\left(\frac{t_{i-1}-\tau}{b}\right) \int_{t_{i-1}}^{t_{i}} \frac{t_{i-1}-t}{b}(f(c_\tau)+\varepsilon) d t \\
&+\frac{1}{b} \sum_{\left|t_{i}-\tau\right|<\delta, K^{\prime}>0} K^{\prime}\left(\frac{t_{i-1}-\tau}{b}\right) \int_{t_{i-1}}^{t_{i}} \frac{t_{i-1}-t}{b}(f(c_\tau)-\varepsilon) d t \\
=&-\frac{f(c_\tau) \Delta}{2 b} \sum_{\left|t_{i}-\tau\right|<\delta} \frac{\Delta}{b} K^{\prime}\left(\frac{t_{i-1}-\tau}{b}\right)+\frac{\varepsilon \Delta}{2 b} \sum_{\left|t_{i}-\tau\right|<\delta} \frac{\Delta}{b}\left|K^{\prime}\left(\frac{t_{i-1}-\tau}{b}\right)\right|.
\end{aligned}
\end{equation}
Similarly, a lower bound can be written as: 
\begin{equation*}
D_{111} \geq-\frac{f(c_{\tau}) \Delta}{2 b} \sum_{\left|t_{i}-\tau\right|<\delta} \frac{\Delta}{b} K^{\prime}\left(\frac{t_{i-1}-\tau}{b}\right)-\frac{\varepsilon\Delta}{2 b} \sum_{\left|t_{i}-\tau\right|<\delta} \frac{\Delta}{b}\left|K^{\prime}\left(\frac{t_{i-1}-\tau}{b}\right)\right|.
\end{equation*}
Note that, since $\int |K'(x)|dx<\infty$, 
\[
 \sum_{\left|t_{i}-\tau\right|<\delta} \frac{\Delta}{b}\left|K^{\prime}\left(\frac{t_{i-1}-\tau}{b}\right)\right|\leq 
 \sum_{i} \frac{\Delta}{b}\left|K^{\prime}\left(\frac{t_{i-1}-\tau}{b}\right)\right|\longrightarrow\int |K'(x)|dx.
\]
Also, letting $s_{i-1} \in\left(t_{i-1}, t_{i}\right)$ be such that $\int_{\left(t_{i-1}-\tau\right) / b}^{\left(t_{i}-\tau\right) / b} K^{\prime}(x) d x=\frac{\Delta}{b} K^{\prime}\left(\frac{s_{i-1}-\tau}{b}\right)$ and, denoting
$\delta^{+}:=\max \left\{t_{i}-\tau: t_{i}<\tau+\delta\right\}$ and $\delta^{-}:=\min \{t_{i-1}-\tau: t_{i}>\tau-\delta\}$, we have:
\begin{equation} \label{eq:K_B_A}
\begin{aligned}
& \sum_{\left|t_{i}-\tau\right|<\delta} \frac{\Delta}{b} K^{\prime}\left(\frac{t_{i-1}-\tau}{b}\right) \\
&\quad= \sum_{\left|t_{i}-\tau\right|<\delta} \int_{\left(t_{i-1}-\tau\right) / b}^{\left(t_{i}-\tau\right) / b} K^{\prime}(x) d x+\frac{\Delta}{b} \sum_{\left|t_{i}-\tau\right|<\delta}\left[K^{\prime}\left(\frac{t_{i-1}-\tau}{b}\right)-K^{\prime}\left(\frac{s_{i-1}-\tau}{b}\right)\right] \\
&\quad = \int_{\left(\delta^{-} / b, \delta^{+}/ b\right) \cap(A, B)} K^{\prime}(x) d x+\frac{\Delta}{b} \sum_{\left|t_{i}-\tau\right|<\delta}\left[K^{\prime}\left(\frac{t_{i-1}-\tau}{b}\right)-K^{\prime}\left(\frac{s_{i-1}-\tau}{b}\right)\right] \\
&\quad =(K(B-)-K(A+))+o(1), 
\end{aligned}
\end{equation} since 
\begin{align*}
&\frac{\Delta}{b} \sum_{\left|t_{i}-\tau\right|<\delta}\left|K^{\prime}\left(\frac{t_{i-1}-\tau}{b}\right)-K^{\prime}\left(\frac{s_{i-1}-\tau}{b}\right)\right| \leq \frac{\Delta}{b} V_{-\infty}^{\infty}\left(K^{\prime}\right)=O\left(\frac{\Delta}{b}\right)=o(1),\\
&\int_{\left(\delta^{-} / b, \delta^{+} / b\right) \cap(A, B) }K^{\prime}(x) d x \longrightarrow K(B-)-K(A+).
\end{align*} 
Combining the previous inequalities, we have
\begin{align*}
	L-2\int |K'(x)|dx\frac{\varepsilon}{2}\leq\liminf_{n\to\infty} \frac{b}{\Delta}D_{111}\leq 
	\limsup_{n\to\infty} \frac{b}{\Delta}D_{111}
	\leq
	L+\int |K'(x)|dx\frac{\varepsilon}{2}.
\end{align*}
where $L:=(K(A+)-K(B-)) f(c_{\tau})/2$. Since $\varepsilon$ is arbitrary, we conclude that 
\begin{equation*}
D_{111}=\frac{(K(A+)-K(B-)) f(c_{\tau})}{2} \frac{\Delta}{b}+o\left(\frac{\Delta}{b}\right)
\end{equation*}
and thus (\ref{eq:D1f}) is proved. {Note that (\ref{lemmaa.1.1}) and (\ref{lemmaa.1.2}) still hold for $\tau=0$ or $T$. When $\tau=0$, 
\begin{align*}
    \left|D_{11}\right|=\left|\frac{1}{b} \sum_{i=1}^{n} \int_{t_{i-1}}^{t_{i}} K^{\prime}\left(\frac{t_{i-1}}{b}\right) \frac{t_{i-1}-t}{b} f(c_t) d t\right|\leq \frac{M_f\Delta^2}{2b^2}\sum_{i=1}^{n}\left|K^{\prime}\left(\frac{t_{i-1}}{b}\right)\right|,
\end{align*}
where $\frac{\Delta}{b}\sum_{i=1}^{n}\left|K^{\prime}\left(\frac{t_{i-1}}{b}\right)\right|\to \int_{0}^{\infty}\left|K^{\prime}(s)\right|ds$. Thus, $\left|D_{11}\right|=O\left(\frac{\Delta}{b}\right)$. {\Blue The same} result can be proved for $\tau=T$. Then, we have $D_{1}(f)=O\left(\frac{\Delta}{b}\right)$ when $\tau=0$ or $T$.}

%

For (\ref{eq:D2}), first note that by (\ref{eq:localization}) and (\ref{eq:X_c_bound}), {for $\tau\in [0,T]$,}
\begin{equation*}
\begin{split}
    &\mathbb{E}\left| \sum_{i = 1}^n K_{{b}}(t_{i-1}- \tau) f(c_i^n) \Delta_n - \sum_{i = 1}^n K_{{b}}(t_{i-1}- \tau) \int_{t_{i-1}}^{t_i}f(c_s) ds\right| \\
& \quad  \leq \sum_{i = 1}^n \left|K_{{b}}(t_{i-1}- \tau)\right|  \int_{t_{i-1}}^{t_i}\mathbb{E}\left| f(c_i^n)-f(c_s)\right|ds\\
  & \quad \leq C\sum_{i = 1}^n \left|K_{{b}}(t_{i-1}- \tau)\right|  \Delta_n^{3/2}= O(\Delta_n^{1/2}).
\end{split}
\end{equation*}
Therefore, 
\begin{equation*}
\begin{split}
    &D_2(f) = D_1(f) + \sum_{i = 1}^n K_{{b}}(t_{i-1}- \tau) f(c_{t_i}) \Delta_n - \sum_{i = 1}^n K_{{b}}(t_{i-1}- \tau) \int_{t_{i-1}}^{t_i}f(c_s) ds\\
    = & D_1(f) +  O_p\left(\Delta_n^{1/2} \right).
\end{split}
\end{equation*}
\end{proof}
\subsection{Technical Proofs related to Theorem \ref{clt}}\label{Technical2_1}
\begin{proof}[Proof of (\ref{ApprmRSWInt})]
Set
\begin{align*}
	D_n:=
	\Delta_n^{\frac{3}{2}}\sum_{j=u+1}^{n} \sum_{i=1}^{u-1} K_{{b}}\left(t_{j-1}-t_{i}\right) \partial_{\ell m}g\left(c_{i}^{n}\right) - \Delta_n^{-\frac{1}{2}}\int_{t_{u+1}}^{T}\int_0^{t_u}K_{{b}}(s-t) \partial_{\ell m}g(c_t) dt ds=o_p(1).
\end{align*}
To this end, note that, by (\ref{eq:X_c_bound}-iii) {and the local boundedness of the second-order partial derivatives of $g$,} we have
\begin{align*}
\nonumber
\mathbb{E}\left|D_n\right| &\leq\Delta_n^{-1/2}\sum_{j=u+1}^{n-1}\sum_{i=1}^{u}\int_{t_{j}}^{t_{j+1}}\int_{t_{i-1}}^{t_i}\mathbb{E}\left|K_{{b}}(t_{j-1}-t_i )\partial_{\ell m}g(c^n_i) - K_{{b}}(s-t)\partial_{\ell m}g(c_t)\right| dt ds\\
\nonumber
 & \leq \Delta_n^{-1/2}\sum_{j=u+1}^{n-1}\sum_{i=1}^{u}\int_{t_{j}}^{t_{j+1}}\int_{t_{i-1}}^{t_i}\mathbb{E}\left[\left|K_{{b}}(t_{j-1}-t_i )\left(\partial_{\ell m}g(c^n_i)- \partial_{\ell m}g(c_t)\right) \right|\right.\\
 \nonumber
& \qquad \qquad \qquad \qquad \qquad \qquad \qquad \quad + \left.\left| \partial_{\ell m}g(c_t)\left( K_{{b}}(t_{j-1}-t_i )-K_{{b}}(s-t) \right)\right|\right] dt ds\\
\nonumber
 &\leq C  \Delta_n^{-1/2}\sum_{j=u+1}^{n-1}\sum_{i=1}^{u}\int_{t_{j}}^{t_{j+1}}\int_{t_{i-1}}^{t_i}\left[\left|K_{{b}}(t_{j-1}-t_i )\right|\sqrt{\Delta_n} + \left|K_{{b}}(t_{j-1}-t_i )-K_{{b}}(s-t)\right|\right]dt ds\\
 \nonumber
 &\leq C \sum_{j=u+1}^{n-1}\sum_{i=1}^{u}\left[\left|K_{{b}}(t_{j-1}-t_i )\right|\Delta_n^2 + \Delta_n^{-1/2}\int_{t_{j}}^{t_{j+1}}\int_{t_{i-1}}^{t_i}\left|K_{{b}}(t_{j-1}-t_i )-K_{{b}}(s-t)\right|dt ds\right]\\
& =: P_1 + P_2.
\end{align*}
We have 
\begin{equation*}
P_1 \leq C \Delta_n \sum_{j = u+1}^{n-1} \int_{\frac{(t_{j-1}-t_{u})}{b}}^{\frac{t_{j-1}}{b}}|K(v)| dv\leq  C b \int_0^{\frac{(T-t_{u})}{b}} \int_{w}^{w +\frac{t_u}{b}}|K(v)|dv dw\leq C b \int_0^{\infty} \bar{L}(w) dw\to 0,
\end{equation*} 
where $\bar{L}(w) = \int_w^{\infty}|K(v)| dv\mathrm{1}_{w>0}$. 
For $P_2$, note that, for some $ \xi_i \in (t_{i-3}, t_{i})$,  \begin{equation*}
\begin{split}
P_2 \leq& C \Delta_n^{3/2}\sum_{j = u+1}^{n-1} \sum_{i=1}^{u} \left| K'\left(\frac{t_{j-1} - \xi_i }{b} \right)\right|\frac{\Delta_n}{{b}^{2}}\\
\leq & C \Delta_n^{3/2}\sum_{j = u}^n \sum_{i=1}^{u-1} \left| K'\left(\frac{t_{j-1} - t_i }{b} \right)\right|+\Delta_n^{3/2}\sum_{j = u}^n \sum_{i=1}^{u-1} \left| K'\left(\frac{t_{j-1} - t_i }{b} \right)- K'\left(\frac{t_{j-1} - \xi_i }{b} \right)\right| \\
\leq & C \Delta_n^{1/2}b\sum_{j = u}^n \int_{(t_{j-1} - t_u)/{b}}^{t_{j-1}/{b}} \left| K'\left(v\right)\right|dv+\Delta_n^{3/2}\sum_{j = u}^n V_{-\infty}^{\infty}(K') \\
\leq & C \frac{{b}^{2}}{\sqrt{\Delta_n}} \int_{0}^{\infty} \int_{w}^{\infty} \left| K'\left(v\right)\right|dv+\Delta_n^{3/2}\sum_{j = u}^n V_{-\infty}^{\infty}(K') \to  0.
\end{split}
\end{equation*}
\end{proof}

\begin{proof}[Proof of (\ref{eq:K2_discrete})]
{Denote
\begin{equation}
\begin{array}{l}
S_{1}:=\sum_{i=1}^n K_{{b}}^2\left(t_{j-1}-t_i\right) \bar{K}\left(t_i\right)^{-1} \mathbbm{1}_{\left\{t_{j-1}-t_i\geq 0\right\}} \partial g\left(c_i^n\right) \Delta_n, \\
S_{2}:=\sum_{i=1}^n K_{{b}}^2\left(t_{j-1}-t_i\right) \bar{K}\left(t_i\right)^{-1} \mathbbm{1}_{\left\{t_{j-1}-t_i\geq 0\right\}}\int_{t_{i-1}}^{t_i} \partial g\left(c_s\right) d s,\\
S_{3}:=\int_0^T K_{{b}}^2\left(t_{j-1}-s\right) \bar{K}(s)^{-1} \mathbbm{1}_{\left\{t_{j-1}-s \geq 0\right\}} \partial g(s) d s.
\end{array}
\end{equation}
Define
$$
\tilde{K}^{j,n}\left(x\right):=\frac{K(x)^2}{\bar{K}\left(t_{j-1}-bx\right)}\mathbbm{1}_{\left\{x \geq 0\right\}}=\frac{K(x)^2}{\frac{\Delta_{n}}{b} \sum_{l=1}^n K\left(x+\frac{t_{l-1}-t_{j-1}}{b}\right)}\mathbbm{1}_{\left\{x \geq 0\right\}},
$$
and differentiate it as
$$
\partial\tilde{K}^{j,n}\left(x\right)=\frac{2 K(x) K^{\prime}(x) \bar{K}\left(t_{j-1}-b x\right)-K(x)^2 \overline{K^{\prime}}\left(t_{j-1}-b x\right)}{\left[\bar{K}\left(t_{j-1}-b x\right)\right]^2} \mathbbm{1}_{\left\{x \geq 0\right\}}.
$$
Then, similar as the definition of $D_{1}(f)$ in Lemma \ref{K_discrete}, 
\begin{align*}
    &D_1(\partial g)=S_{2}-S_{3}=\frac{1}{{b}^{2}} \sum_{i=1}^n \int_{t_{i-1}}^{t_i}\left[\tilde{K}^{j,n}\left(\frac{t_{j-1}-t_{i}}{b}\right)-\tilde{K}^{j,n}\left(\frac{t_{j-1}-s}{b}\right)\right] \partial g\left(c_s\right) d s\\
    &=\frac{1}{{b}^{2}} \sum_{i=1}^n \int_{t_{i-1}}^{t_i} \partial\tilde{K}^{j,n}\left(\frac{t_{j-1}-s_{t}}{b}\right) \frac{s-t_i}{b} \partial g\left(c_s\right) d s=\frac{1}{{b}^{2}} \sum_{i=1}^n \int_{t_{i-1}}^{t_i} \partial\tilde{K}^{j,n}\left(\frac{t_{j-1}-t_i}{b}\right) \frac{s-t_i}{b} \partial g\left(c_s\right) d s\\
    &+\frac{1}{{b}^{2}} \sum_{i=1}^n \int_{t_{i-1}}^{t_i} \left[\partial\tilde{K}^{j,n}\left(\frac{t_{j-1}-s_{t}}{b}\right)-\partial\tilde{K}^{j,n}\left(\frac{t_{j-1}-t_i}{b}\right)\right] \frac{s-t_i}{b} \partial g\left(c_s\right) d s=:D_{11}+D_{12},
\end{align*}
where $s_t \in\left(s, t_i\right)$. Then,
\begin{equation}\label{d12decomp}
\begin{split}
\left|D_{12}\right|&\leq \frac{M_{\partial g} \Delta_n^2}{2 b^3} \sum_{i=1}^n \max _{s \in\left[t_{i-1}, t_i\right]}\left|\partial\tilde{K}^{j,n}\left(\frac{t_{j-1}-s}{b}\right)-\partial\tilde{K}^{j,n}\left(\frac{t_{j-1}-t_i}{b}\right)\right|\\
&\leq \frac{M_{\partial g} \Delta_n^2}{2 b^3}\left\{\sum_{i=1}^n \max _{s\in\left[t_{i-1}, t_i\right]}\left|\frac{K\left(\frac{t_{j-1}-s}{b}\right) K^{\prime}\left(\frac{t_{j-1}-s}{b}\right) \bar{K}\left(t_i\right)-K\left(\frac{t_{j-1}-t_i}{b}\right) K^{\prime}\left(\frac{t_{j-1}-t_i}{b}\right) \bar{K}(s)}{\bar{K}(s) \bar{K}\left(t_i\right)}\right|\right.\\
&\left.+\sum_{i=1}^n \max _{s \in\left[t_{i-1}, t_i\right]}\left|\frac{K\left(\frac{t_{j-1}-s}{b}\right)^2 \overline{K^{\prime}}(s) \bar{K}\left(t_i\right)^2-K\left(\frac{t_{j-1}-t_{i}}{b}\right)^2 \overline{K^{\prime}}\left(t_i\right) \bar{K}(s)^2}{\bar{K}\left(t_i\right)^2 \bar{K}(s)^2}\right|\right\}.
\end{split}
\end{equation}
Note that 
\begin{equation}\label{d12decomp1}
\begin{split}
    &\sum_{i=1}^n \max _{s\in\left[t_{i-1}, t_i\right]}\left|\frac{K\left(\frac{t_{j-1}-s}{b}\right) K^{\prime}\left(\frac{t_{j-1}-s}{b}\right) \bar{K}\left(t_i\right)-K\left(\frac{t_{j-1}-t_i}{b}\right) K^{\prime}\left(\frac{t_{j-1}-t_i}{b}\right) \bar{K}(s)}{\bar{K}(s) \bar{K}\left(t_i\right)}\right|\\
    &\leq \frac{1}{\delta^2} \sum_{i=1}^n \max _{s \in\left[t_{i-1}, t_i\right]} \left|\left[K\left(\frac{t_{j-1}-s}{b}\right)-K\left(\frac{t_{j-1}-t_i}{b}\right)\right] K^{\prime}\left(\frac{t_{j-1}-s}{b}\right) \bar{K}\left(t_i\right)\right|\\
    &\quad+\left|K\left(\frac{t_{j-1}-t_i}{b}\right)\left[K^{\prime}\left(\frac{t_{j-1}-s}{b}\right)-K^{\prime}\left(\frac{t_{j-1}-t_i}{b}\right)\right] \bar{K}\left(t_i\right)\right|\\
    &\quad+\left|K\left(\frac{t_{j-1}-t_i}{b}\right) K^{\prime}\left(\frac{t_{j-1}-t_i}{b}\right)\left[\bar{K}\left(t_i\right)-\bar{K}(s)\right]\right|\\
    &\leq \frac{C}{\delta^2}\left[V_{-\infty}^{+\infty}(K)+V_{-\infty}^{+\infty}\left(K^{\prime}\right)+\sum_{i=1}^n \max _{s \in\left[t_{i-1}, t_i\right]}\left|\bar{K}\left(t_i\right)-\bar{K}(s)\right|\right].
\end{split}
\end{equation}
Similarly, we have
\begin{equation}\label{d12decomp2}
\begin{split}
    &\sum_{i=1}^n \max _{s \in\left[t_{i-1}, t_i\right]}\left|\frac{K\left(\frac{t_{j-1}-s}{b}\right)^2 \overline{K^{\prime}}(s) \bar{K}\left(t_i\right)^2-K\left(\frac{t_{j-1}-t_{i}}{b}\right)^2 \overline{K^{\prime}}\left(t_i\right) \bar{K}(s)^2}{\bar{K}\left(t_i\right)^2 \bar{K}(s)^2}\right|\\
    &\leq \frac{C}{\delta^4}\left[V_{-\infty}^{+\infty}\left(K^{2}\right)+\sum_{i=1}^n \max _{s \in\left[t_{i-1}, t_i\right]}\left|\overline{K^{\prime}}\left(t_i\right)-\overline{K^{\prime}}(s)\right|+\sum_{i=1}^n \max _{s \in\left[t_{i-1}, t_i\right]}\left|\bar{K}^2\left(t_i\right)-\bar{K}^2(s)\right|\right].
\end{split}
\end{equation}
Consider
$$
\sum_{i=1}^n \max _{s \in\left[t_{i-1}, t_i\right]}\left|\bar{K}\left(t_i\right)-\bar{K}(s)\right|=\sum_{i=1}^n\left|\bar{K}\left(t_i\right)-\bar{K}\left(s_i\right)\right|=\sum_{i=1}^n\left|\bar{K}^{\prime}\left(s_i^*\right)\right|\left|t_i-s_i\right|,
$$
where $s_{i}=\arg\max_{s \in\left[t_{i-1}, t_i\right]}\left|\bar{K}\left(t_i\right)-\bar{K}(s)\right|$ and $s_i^* \in\left[s_i, t_i\right]$. Meanwhile,
\begin{align*}
    \bar{K}^{\prime}(t)&=\frac{\Delta_{n}}{b} \sum_{l=1}^n K^{\prime}\left(\frac{t_{l-1}-t}{b}\right)\left(-\frac{1}{b}\right)\\
    &=\sum_{i=1}^{n}\left[K_{{b}}^{\prime}\left(t_{l-1}-t\right)\Delta_{n}-\int_{t_{i-1}}^{t_{i}}K_{{b}}^{\prime}\left(u-t\right)du\right]\left(-\frac{1}{b}\right)-\frac{1}{{b}^{2}} \int_0^T K^{\prime}\left(\frac{u-t}{b}\right) d u,
\end{align*}
where $K^{\prime}_{b}(x):=\frac{1}{b}K^{\prime}\left(\frac{x}{b}\right)$. By Lemma \ref{K_discrete}, $D_{1}(1)=\sum_{i=1}^{n}\left[K_{{b}}^{\prime}\left(t_{l-1}-t\right)\Delta_{n}-\int_{t_{i-1}}^{t_{i}}K_{{b}}^{\prime}\left(u-t\right)du\right]=O\left(\frac{\Delta_{n}}{b}\right)$. Thus,
$$
\bar{K}^{\prime}(t)=\frac{1}{b}\left[K\left(-\frac{t}{b}\right)-K\left(\frac{T-t}{b}\right)\right]+O\left(\frac{\Delta_{n}}{b^{2}}\right).
$$
$\forall \epsilon^{\prime}>0$, $\exists$ positive constant $C\left(\epsilon^{\prime}\right)$ such that $\forall x>C\left(\epsilon^{\prime}\right)$, $\left|K(x)\right|\left|x\right|^{2+\epsilon}<\epsilon^{\prime}$. Then, for $C k_n+1 \leq i \leq n-C k_n$, since $\left|\frac{s_i^*}{b}\right|,\left|\frac{T-s_i^*}{b}\right|\geq \frac{C k_n \Delta_n}{b}=C$,
\begin{align*}
\sum_{i=1}^n\left|\bar{K}^{\prime}\left(s_i^*\right)\right|\left|t_i-s_i\right|&\leq \left(\sum_{i=C k_n+1}^{n-C k_n}+\sum_{i=1}^{C k_n}+\sum_{i=n-C k_n+1}^n\right) \frac{1}{b}\left[\left|K\left(-\frac{s_i^*}{b}\right)\right|+\left|K\left(\frac{T-s_i^*}{b}\right)\right|\right] \Delta_n+O\left(\frac{\Delta_{n}}{b^{2}}\right)\\
&\leq \sum_{i=C k_n+1}^{n-C k_n} \frac{1}{b}\left[\left|\frac{s_i^*}{b}\right|^{-2}+\left|\frac{T-s_i^*}{b}\right|^{-2}\right] \Delta_n \epsilon^{\prime}+2 C k_n \frac{1}{b} 2 M_K \Delta_n+O\left(\frac{\Delta_{n}}{b^{2}}\right)\\
&\leq 2b \epsilon^{\prime}\Delta_n\sum_{i=C k_n+1}^{n-C k_n}t_{i-1}^{-2} + 4CM_{K}+O\left(\frac{\Delta_{n}}{b^{2}}\right),
\end{align*}
where 
\begin{align*}
    &\Delta_n\sum_{i=C k_n+1}^{n-C k_n}t_{i-1}^{-2}=\sum_{i=C k_n+1}^{n-C k_n}\int_{t_{i-1}}^{t_{i}}\left(t_{i-1}^{-2}-s^{-2}\right)ds+\int_{t_{Ck_n}}^{t_{n-Ck_n}}s^{-2}ds\\
    &\leq \sum_{i=C k_n+1}^{n-C k_n}2\Delta_{n}^{2}t_{i-1}^{-3}+\left(\frac{1}{t_{Ck_{n}}}-\frac{1}{t_{n-Ck_n}}\right)\leq 2C^{-3}\Delta_{n}^{-2}k_{n}^{-3}+C^{-1}\Delta_{n}^{-1}k_{n}^{-1}.
\end{align*}
Hence, $b \epsilon^{\prime}\Delta_n\sum_{i=C k_n+1}^{n-C k_n}t_{i-1}^{-2}=O\left(\Delta_{n}^{-1}k_{n}^{-2}\right)+O(1)$. We conclude that $\sum_{i=1}^n \max _{s \in\left[t_{i-1}, t_i\right]}\left|\bar{K}\left(t_i\right)-\bar{K}(s)\right|<\infty$. Similar proof can be shown for $\sum_{i=1}^n \max _{s \in\left[t_{i-1}, t_i\right]}\left|\overline{K^{\prime}}\left(t_i\right)-\overline{K^{\prime}}(s)\right|$. Also note that 
$$
\sum_{i=1}^n \max _{s \in\left[t_{i-1}, t_i\right]}\left|\bar{K}^2\left(t_i\right)-\bar{K}^2(s)\right|=\sum_{i=1}^n \max _{s \in\left[t_{i-1}, t_i\right]}\left|\bar{K}\left(t_i\right)+\bar{K}(s)\right|\left|\bar{K}\left(t_i\right)-\bar{K}(s)\right|<\infty.
$$
By (\ref{d12decomp}), (\ref{d12decomp1}) and (\ref{d12decomp2}), We can further conclude that $\left|D_{12}\right|=O_{p}\left(\frac{\Delta_{n}^{2}}{b^3}\right)$.\\
For any $\delta^{\prime} \in\left(0, \min \left(T-t_{j-1}, t_{j-1}\right)\right)$,
\begin{align*}
    D_{11}&=\frac{1}{{b}^{2}}\left(\sum_{i:\left|t_i-t_{j-1}\right|<\delta^{\prime}}+\sum_{i:\left|t_i-t_{j-1}\right| \geq \delta^{\prime}}\right)\partial\tilde{K}^{j,n}\left(\frac{t_{j-1}-t_i}{b}\right) \int_{t_{i-1}}^{t_i} \frac{s-t_i}{b} \partial g\left(c_s\right) d s\\
    &=:D_{111}+D_{112},
\end{align*}
where
\begin{align*}
    \left|D_{112}\right|&\leq \frac{1}{{b}^{2}}\sum_{i:\left|t_i-t_{j-1}\right| \geq \delta^{\prime}}\left|\partial\tilde{K}^{j,n}\left(\frac{t_{j-1}-t_i}{b}\right)\right|\int_{t_{i-1}}^{t_i} \frac{\left|s-t_i\right|}{b} \left|\partial g\left(c_s\right)\right| d s\\
    &\leq \frac{M_{\partial g} \Delta_{n}}{2 {b}^{2}} \sum_{i:\left|t_i-t_{j-1}\right| \geq \delta^{\prime}}\frac{\Delta_{n}}{b}\left|\frac{2 K\left(\frac{t_{j-1}-t_i}{b}\right) K^{\prime}\left(\frac{t_{j-1}-t_i}{b}\right)}{\bar{K}\left(t_i\right)}\right|+\left|\frac{K\left(\frac{t_{j-1}-t_i}{b}\right)^2 \overline{K^{\prime}}\left(t_i\right)}{\left[\bar{K}\left(t_i\right)\right]^2}\right|
\end{align*}
Since for some $s_i \in\left[t_{i-1}, t_i\right]$,
\begin{align*}
    &\sum_{i:\left|t_i-t_{j-1}\right| \geq \delta^{\prime}}\frac{\Delta_{n}}{b}\left|\frac{2 K\left(\frac{t_{j-1}-t_i}{b}\right) K^{\prime}\left(\frac{t_{j-1}-t_i}{b}\right)}{\bar{K}\left(t_i\right)}\right|+\left|\frac{K\left(\frac{t_{j-1}-t_i}{b}\right)^2 \overline{K^{\prime}}\left(t_i\right)}{\left[\bar{K}\left(t_i\right)\right]^2}\right|\\
    &\leq \sum_{i:\left|t_i-t_{j-1}\right| \geq \delta^{\prime}}\int_{\left(t_{j-1}-t_i\right) / h}^{\left(t_{j-1}-t_{i-1}\right) / h}\frac{2}{\delta}\left|K(t) K^{\prime}(t)\right|+\frac{C}{\delta^{2}}\left|K(t)^{2}\right| d t\\
    &+\frac{\Delta_{n}}{b}\sum_{i:\left|t_i-t_{j-1}\right| \geq \delta^{\prime}}\frac{2}{\delta}\left|K\left(\frac{t_{j-1}-t_i}{b}\right) K^{\prime}\left(\frac{t_{j-1}-t_i}{b}\right)-K\left(\frac{t_{j-1}-s_i}{b}\right) K^{\prime}\left(\frac{t_{j-1}-s_i}{b}\right) \right|\\
    &+\frac{C}{\delta^{2}}\left|K\left(\frac{t_{j-1}-t_i}{b}\right)^{2}-K\left(\frac{t_{j-1}-s_i}{b}\right)^{2}\right|\\
    &\leq \left(\int_{\delta / h}^{\infty}+\int_{-\infty}^{-\delta / h}\right)\frac{2}{\delta}\left|K(t) K^{\prime}(t)\right|+\frac{C}{\delta^{2}}\left|K(t)^{2}\right| d t+\frac{C}{\delta^{2}} \frac{\Delta_{n}}{b}\left[V_{-\infty}^{+\infty}(K)+V_{-\infty}^{+\infty}\left(K^{\prime}\right)+V_{-\infty}^{+\infty}\left(K^{2}\right)\right]\\
    &=o(1),
\end{align*}
we have $\left|D_{112}\right|=o_{p}\left(\frac{\Delta_{n}}{{b}^{2}}\right)$.\\
Fix $\varepsilon>0$. Since $c$ and $\partial g$ are continuous, $\exists \delta^{\prime}=\delta^{\prime}(\varepsilon)>0$ such that $\left|\partial g\left(c_t\right)-\partial g\left(c_s\right)\right|<\varepsilon$, $ \forall|t-s|<\delta^{\prime}$. Then,
\begin{align*}
    &D_{111}=\frac{1}{{b}^{2}}\left(\sum_{\substack{i:\left|t_i-t_{j-1}\right|<\delta^{\prime} \\ \partial\tilde{K}^{j, n} \leq 0}}+\sum_{\substack{i: \left| t_{i}-t_{j-1}\right|<\delta^{\prime} \\ \partial\tilde{K}^{j, n}>0}}\right)\partial\tilde{K}^{j,n}\left(\frac{t_{j-1}-t_i}{b}\right) \int_{t_{i-1}}^{t_i} \frac{s-t_i}{b} \partial g\left(c_s\right) d s\\
    &\leq \frac{1}{{b}^{2}}\sum_{\substack{i:\left|t_i-t_{j-1}\right|<\delta^{\prime} \\ \partial\tilde{K}^{j, n} \leq 0}}\partial\tilde{K}^{j,n}\left(\frac{t_{j-1}-t_i}{b}\right) \int_{t_{i-1}}^{t_i} \frac{s-t_i}{b} \left[\partial g\left(c_{t_{j-1}}\right)+\varepsilon\right] d s\\
    &+\frac{1}{{b}^{2}}\sum_{\substack{i:\left|t_i-t_{j-1}\right|<\delta^{\prime} \\ \partial\tilde{K}^{j, n} > 0}}\partial\tilde{K}^{j,n}\left(\frac{t_{j-1}-t_i}{b}\right) \int_{t_{i-1}}^{t_i} \frac{s-t_i}{b} \left[\partial g\left(c_{t_{j-1}}\right)-\varepsilon\right] d s\\
    &=-\frac{\partial g\left(c_{t_{j-1}}\right) \Delta_{n}}{2 {b}^{2}} \sum_{i:\left|t_i-t_{j-1}\right|<\delta^{\prime}} \frac{\Delta_{n}}{b} \partial\tilde{K}^{j, n}\left(\frac{t_{j-1}-t_i}{b}\right)+\frac{\varepsilon\Delta_{n}}{2 {b}^{2}} \sum_{i:\left|t_i-t_{j-1}\right|<\delta^{\prime}} \frac{\Delta_{n}}{b} \left|\partial\tilde{K}^{j, n}\left(\frac{t_{j-1}-t_i}{b}\right)\right|.
\end{align*}
Similarly, we have
$$
D_{111}\geq -\frac{\partial g\left(c_{t_{j-1}}\right) \Delta_{n}}{2 {b}^{2}} \sum_{i:\left|t_i-t_{j-1}\right|<\delta^{\prime}} \frac{\Delta_{n}}{b} \partial\tilde{K}^{j, n}\left(\frac{t_{j-1}-t_i}{b}\right)-\frac{\varepsilon\Delta_{n}}{2 {b}^{2}} \sum_{i:\left|t_i-t_{j-1}\right|<\delta^{\prime}} \frac{\Delta_{n}}{b} \left|\partial\tilde{K}^{j, n}\left(\frac{t_{j-1}-t_i}{b}\right)\right|.
$$
Note that
\begin{align*}
    &\sum_{i:\left|t_i-t_{j-1}\right|<\delta^{\prime}} \frac{\Delta_{n}}{b} \left|\partial\tilde{K}^{j, n}\left(\frac{t_{j-1}-t_i}{b}\right)\right|\leq \frac{C}{\delta} \sum_i \frac{\Delta_{n}}{b}\left|K\left(\frac{t_{j-1}-t_i}{b}\right) K^{\prime}\left(\frac{t_{j-1}-t_i}{b}\right)\right|\\
    &+\frac{C}{\delta^2} \sum_i \frac{\Delta_{n}}{b}\left|K\left(\frac{t_{j-1}-t_i}{b}\right)^2\right|\rightarrow \frac{C}{\delta} \int\left|K(x) K^{\prime}(x)\right| d x+\frac{C}{\delta^2} \int K(x)^2 d x.
\end{align*}
We get $D_{111}=O_p\left(\frac{\Delta_{n}}{{b}^{2}}\right)$, and $D_1(\partial g)=S_{2}-S_{3}=O_p\left(\frac{\Delta_{n}}{{b}^{2}}\right)$.\\
For $S_{1}-S_{2}$, note that
\begin{align*}
    &\mathbb{E}\left[\sum_{i=1}^n K_{{b}}^2\left(t_{j-1}-t_i\right) \bar{K}\left(t_i\right)^{-1} \mathbbm{1}_{\left\{t_{j-1}-t_i\geq 0\right\}} \int_{t_{i-1}}^{t_i}\left(\partial g\left(c_i^n\right)- \partial g\left(c_s\right)\right)d s\right]\\
    &\leq \frac{1}{\delta}\sum_{i=1}^n\left|K_{{b}}^2\left(t_{j-1}-t_i\right)\right|\int_{t_{i-1}}^{t_i}\mathbb{E}\left|\partial g\left(c_i^n\right)- \partial g\left(c_s\right)\right| d s\\
    &\leq \frac{C}{\delta}\sum_{i=1}^n\left|K_{{b}}^2\left(t_{j-1}-t_i\right)\right|\Delta_{n}^{3/2}=\frac{C}{\delta} \frac{\Delta_{n}^{1 / 2}}{b} \sum_{i=1}^n \frac{\Delta_{n}}{b}\left|K\left(\frac{t_{j-1}-t_i}{b}\right)^2\right|=O_p\left(\frac{\Delta_{n}^{1/2}}{b}\right).
\end{align*}
Thus, $S_{1}-S_{3}=O_{p}\left(\frac{\Delta_{n}}{{b}^{2}}\right)+O_p\left(\frac{\Delta_{n}^{1/2}}{b}\right)$.}
\end{proof}

\begin{proof}[Proof of (\ref{eq:alphapartialg})]
First note that for $i\geq l$,
\begin{equation}\label{eq:decomalphapartialg}
\left|\mathbb{E}\left[\left.\alpha_l^{n, p q} \partial_{p q, u v}^2 g\left(c_i^n\right) \right| \mathcal{F}_{l-1}^n\right]\right|=\left|\mathbb{E}\left[\left.\alpha_l^{n, p q}\left[\partial_{p q, u v}^2 g\left(c_i^n\right)-\partial_{p q, u v}^2 g\left(c_{l-1}^n\right)\right]\right| \mathcal{F}_{l-1}^n\right]\right|+\left|\partial_{p q, u v}^2 g\left(c_{l-1}^n\right)\mathbb{E}\left[\left.\alpha_l^n\right| \mathcal{F}_{l-1}^n\right]\right|,
\end{equation}
where $\mathbb{E}\left[\left.\alpha_{l}^n \right| \mathcal{F}_{l-1}^n\right] \leq C \Delta_n^{3 / 2}\left(\sqrt{\Delta_n}+\eta_{l-1}^n\right)$. By Itô's formula,
\begin{equation}\label{eq:itoformula}
\begin{split}
&\alpha_l^{n, p q}=\int_{t_{l-1}}^{t_l}\left[2\left(X^{\prime}_{s, p}-X^{\prime}_{t_{t-1, p}}\right) \mu^{\prime}_{s, q}+\left(c_s^{p q}-c_{t_{l-1}}^{p q}\right)\right] d s+2 \int_{t_{l-1}}^{t_l}\left(X^{\prime}_{s, p}-X^{\prime}_{t_{l-1}, p}\right) \sum_{r=1}^d \sigma_s^{q, r} d W_s^r,\\
&\partial_{p q , u v}^2 g\left(c_i^n\right)-\partial_{p q , u v}^2 g\left(c_{l-1}^n\right)=\int_{t_{l-1}}^{t_i}\Big[\sum_{p^{\prime}, q^{\prime}=1}^d \partial_{p q, u v, p^{\prime} q^{\prime}}^3 g\left(c_s\right) \tilde{\mu}_s^{p^{\prime} q^{\prime}}\\
&+\frac{1}{2} \sum_{r^{\prime}=1}^d \sum_{p^{\prime}, q^{\prime}=1}^d \sum_{u^{\prime}, v^{\prime}=1}^d \partial_{p q, u v, p^{\prime} q^{\prime}, u^{\prime} v^{\prime}}^4 g\left(c_s\right) \tilde{\sigma}_s^{p^{\prime} q^{\prime}, r^{\prime}} \tilde{\sigma}_s^{u^{\prime} v^{\prime}, r^{\prime}}\Big]ds+\int_{t_{l-1}}^{t_i} \sum_{r^{\prime}=1}^d \sum_{p^{\prime}, q^{\prime}=1}^d \partial_{p q, u v, p^{\prime} q^{\prime}}^3 g\left(c_s\right) \tilde{\sigma}_s^{p^{\prime} q^{\prime}, r^{\prime}} d W_s^{r^{\prime}}.
\end{split}
\end{equation}
Thus, we have the following decomposition of the term $\mathbb{E}\left[\left.\alpha_l^{n, p q}\left[\partial_{p q, u v}^2 g\left(c_i^n\right)-\partial_{p q, u v}^2 g\left(c_l^n\right)\right]\right| \mathcal{F}_{l-1}^n\right]$ in (\ref{eq:decomalphapartialg}),
\begin{equation}
\begin{aligned}
&\mathbb{E}\left[\left.\alpha_l^{n, p q}\left[\partial_{p q, u v}^2 g\left(c_i^n\right)-\partial_{p q, u v}^2 g\left(c_l^n\right)\right]\right| \mathcal{F}_{l-1}^n\right]\\
&=\mathbb{E}\left[\left.\int_{t_{l-1}}^{t_l}2\left(X^{\prime}_{s, p}-X^{\prime}_{t_{t-1, p}}\right) \mu^{\prime}_{s, q}ds \int_{t_{l-1}}^{t_i}\sum_{p^{\prime}, q^{\prime}=1}^d \partial_{p q, u v, p^{\prime} q^{\prime}}^3 g\left(c_s\right) \tilde{\mu}_s^{p^{\prime} q^{\prime}}ds \right| \mathcal{F}_{l-1}^n\right]\\
&+\mathbb{E}\left[\left. \int_{t_{l-1}}^{t_l}\left(X^{\prime}_{s, p}-X^{\prime}_{t_{l-1, p}}\right) \mu^{\prime}_{s, q} d s \int_{t_{l-1}}^{t_i}\sum_{r^{\prime}=1}^d \sum_{p^{\prime}, q^{\prime}=1}^d \sum_{u^{\prime}, v^{\prime}=1}^d \partial_{p q, u v, p q^{\prime}, u^{\prime} v^{\prime} }^{4}g\left(c_s\right) \tilde{\sigma}_s^{p^{\prime} q^{\prime}, r^{\prime}}\tilde{\sigma}_s^{u^{\prime} v^{\prime}, r^{\prime}} d s \right| \mathcal{F}_{l-1}^n\right]\\
&+\mathbb{E}\left[\left. \int_{t_{l-1}}^{t_l}\left(c_s^{p q}-c_{t_{l-1}}^{p q}\right) d s \int_{t_{l-1}}^{t_i} \sum_{p^{\prime}, q^{\prime}=1}^d\partial_{p q, u v, p^{\prime} q^{\prime}}^3 g\left(c_s\right) \tilde{\mu}_s^{p^{\prime} q^{\prime}} d s \right| \mathcal{F}_{l-1}^n\right]\\
&+\mathbb{E}\left[\left. \int_{t_{l-1}}^{t_l}\left(c_s^{p q}-c_{t_{l-1}}^{p q}\right) d s\int_{t_{l-1}}^{t_i}\frac{1}{2} \sum_{r^{\prime}=1}^d \sum_{p^{\prime}, q^{\prime}=1}^d \sum_{u^{\prime}, v^{\prime}=1}^d \partial_{p q, u v, p q^{\prime}, u^{\prime} v^{\prime} }^{4}g\left(c_s\right) \tilde{\sigma}_s^{p^{\prime} q^{\prime}, r^{\prime}}\tilde{\sigma}_s^{u^{\prime} v^{\prime}, r^{\prime}} d s \right| \mathcal{F}_{l-1}^n\right]\\
&+\mathbb{E}\left[\left. \int_{t_{l-1}}^{t_l}2\left(X^{\prime}_{s, p}-X^{\prime}_{t_{l-1, p}}\right) \mu^{\prime}_{s, q} d s\int_{t_{l-1}}^{t_i}\sum_{r^{\prime}=1}^d \sum_{p^{\prime}, q^{\prime}=1}^d \partial_{p q, u v, p^{\prime} q^{\prime}}^3 g\left(c_s\right) \tilde{\sigma}_s^{p^{\prime} q^{\prime}, r^{\prime}} d W_s^{r^{\prime}} \right| \mathcal{F}_{l-1}^n\right]\\
&+\mathbb{E}\left[\left. \int_{t_{l-1}}^{t_l}\left(c_s^{p q}-c_{t_{l-1}}^{p q}\right) d s \int_{t_{l-1}}^{t_i} \sum_{r^{\prime}=1}^d \sum_{p^{\prime}, q^{\prime}=1}^d\partial_{p q, u v, p^{\prime} q^{\prime}}^3 g\left(c_s\right) \tilde{\sigma}_s^{p^{\prime} q^{\prime}, r^{\prime}} d W_s^{r^{\prime}} \right| \mathcal{F}_{l-1}^n\right]\\
&+\mathbb{E}\left[\left. \int_{t_{l-1}}^{t_l}2\left(X^{\prime}_{s, p}-X^{\prime}_{t_{l-1, p}}\right)\sum_{r=1}^{d} \sigma_s^{q, r} d W_s^r \int_{t_{l-1}}^{t_i} \sum_{p^{\prime}, q^{\prime}=1}^d\partial_{p q, u v, p^{\prime} q^{\prime}}^3 g\left(c_s\right) \tilde{\mu}_s^{p^{\prime} q^{\prime}} d s \right| \mathcal{F}_{l-1}^n\right]\\
&+\mathbb{E}\left[\left. \int_{t_{l-1}}^{t_l}\left(X^{\prime}_{s, p}-X^{\prime}_{t_{l-1, p}}\right) \sum_{r=1}^{d}\sigma_s^{q, r} d W_s^r \int_{t_{l-1}}^{t_i}\sum_{r^{\prime}=1}^d \sum_{p^{\prime}, q^{\prime}=1}^d \sum_{u^{\prime}, v^{\prime}=1}^d \partial_{p q, u v, p^{\prime} q^{\prime}, u^{\prime} v^{\prime}}^4 g\left(c_s\right) \tilde{\sigma}_s^{p^{\prime} q^{\prime}, r^{\prime}} \tilde{\sigma}_s^{u^{\prime} v^{\prime}, r^{\prime}} d s \right| \mathcal{F}_{l-1}^n\right]\\
&+\mathbb{E}\left[\left. \int_{t_{l-1}}^{t_l}2\left(X^{\prime}_{s, p}-X^{\prime}_{t_{l-1, p}}\right)\sum_{r=1}^{d} \sigma_s^{q, r} d W_s^r \int_{t_{l-1}}^{t_i}\sum_{r^{\prime}=1}^d \sum_{p^{\prime}, q^{\prime}=1}^d \partial_{p q, u v, p^{\prime} q^{\prime}}^3 g\left(c_s\right) \tilde{\sigma}_s^{p^{\prime} q^{\prime}, r^{\prime}} d W_s^{r^{\prime}} \right| \mathcal{F}_{l-1}^n\right]\\
&=: \sum_{p^{\prime}, q^{\prime}=1}^d H_{1,1}(p^{\prime}, q^{\prime}) + \sum_{r^{\prime}=1}^d \sum_{p^{\prime}, q^{\prime}=1}^d \sum_{u^{\prime}, v^{\prime}=1}^d H_{1,2}(r^{\prime},p^{\prime}, q^{\prime},u^{\prime}, v^{\prime})+\sum_{p^{\prime}, q^{\prime}=1}^d H_{1,3}(p^{\prime}, q^{\prime})\\
&+\sum_{r^{\prime}=1}^d \sum_{p^{\prime}, q^{\prime}=1}^d \sum_{u^{\prime}, v^{\prime}=1}^d H_{1,4}(r^{\prime},p^{\prime}, q^{\prime},u^{\prime}, v^{\prime})+\sum_{r^{\prime}=1}^d \sum_{p^{\prime}, q^{\prime}=1}^d H_{2,1}(r^{\prime},p^{\prime}, q^{\prime})+\sum_{r^{\prime}=1}^d \sum_{p^{\prime}, q^{\prime}=1}^d H_{2,2}(r^{\prime},p^{\prime}, q^{\prime})\\
&+\sum_{r=1}^d \sum_{p^{\prime}, q^{\prime}=1}^d H_{3,1}(r,p^{\prime}, q^{\prime})+\sum_{r=1}^d \sum_{r^{\prime}=1}^d \sum_{p^{\prime}, q^{\prime}=1}^d \sum_{u^{\prime}, v^{\prime}=1}^d H_{3,2}(r,r^{\prime},p^{\prime}, q^{\prime},u^{\prime}, v^{\prime})+\sum_{r=1}^{d}\sum_{r^{\prime}=1}^d \sum_{p^{\prime}, q^{\prime}=1}^dH_{4,1}(r,r^{\prime},p^{\prime}, q^{\prime}),
\end{aligned}
\end{equation}
which can be divided into four different groups of product of stochastic integrals, namely "$ds\cdot ds$" ($\left\{H_{1,j}\right\}_{j=1}^{4}$); "$ds\cdot dW_s$" ($\left\{H_{2,j}\right\}_{j=1}^{2}$); "$dW_s\cdot ds$" ($\left\{H_{3,j}\right\}_{j=1}^{2}$); and "$dW_s\cdot dW_s$" ($H_{4,1}$). For simplicity, we only show the upper bound of one term for each group.\\
Group 1 ("$ds\cdot ds$"): By Cauchy-Schwartz inequality, the boundedness of $\mu_{s}^{\prime}$, $\tilde{\mu}_{s}$, $c_s$, $\tilde{c}_s$, and the local boundedness of third-order and fourth-order partial derivatives of $g$, together with (\ref{eq:X_c_bound}),
\begin{equation}\label{eq:group1}
\begin{aligned}
    &\left|H_{11}\right|=\left|\mathbb{E}\left[\left.\int_{t_{l-1}}^{t_l}\left(X^{\prime}_{s, p}-X^{\prime}_{t_{l-1, p}}\right) \mu^{\prime}_{s, q} d s \int_{t_{l-1}}^{t_i} \partial_{p q, u v, p^{\prime} q^{\prime}}^3 g\left(c_s\right) \tilde{\mu}_s^{p^{\prime} q^{\prime}} d s\right|\mathcal{F}_{l-1}^n\right]\right|\\
    &\leq C\left\{\mathbb{E}\left[\left.\sup _{t \in\left[0, \Delta_n\right]}\left\|X^{\prime}_{t_{l-1}+t}-X^{\prime}_{t_{l-1}}\right\|^2 \Delta_n^2 \right| \mathcal{F}_{l-1}^n\right]\right\}^{1 / 2}\left(t_i-t_{t-1}\right)\leq C \Delta_n^{3 / 2}\left(t_i-t_{l-1}\right).
\end{aligned}
\end{equation}
Similarly, we have
$$
\left.\begin{array}{l}
\left|\mathbb{E}\left[\left.\int_{t_{l-1}}^{t_l}\left(X^{\prime}_{s, p}-X^{\prime}_{t_{l-1, p}}\right) \mu^{\prime}_{s, q} d s \int_{t_{l-1}}^{t_i} \partial_{p q, u v, p q^{\prime}, u^{\prime} v^{\prime} }^{4}g\left(c_s\right) \tilde{\sigma}_s^{p^{\prime} q^{\prime}, r^{\prime}}\tilde{\sigma}_s^{u^{\prime} v^{\prime}, r^{\prime}} d s\right|\mathcal{F}_{l-1}^n\right]\right| \\
\left|\mathbb{E}\left[\left.\int_{t_{l-1}}^{t_l}\left(c_s^{p q}-c_{t_{l-1}}^{p q}\right) d s \int_{t_{l-1}}^{t_i} \partial_{p q, u v, p^{\prime} q^{\prime}}^3 g\left(c_s\right) \tilde{\mu}_s^{p^{\prime} q^{\prime}} d s\right|\mathcal{F}_{l-1}^n\right]\right| \\
\left|\mathbb{E}\left[\left.\int_{t_{l-1}}^{t_l}\left(c_s^{p q}-c_{t_{l-1}}^{p q}\right) d s\int_{t_{l-1}}^{t_i} \partial_{p q, u v, p q^{\prime}, u^{\prime} v^{\prime} }^{4}g\left(c_s\right) \tilde{\sigma}_s^{p^{\prime} q^{\prime}, r^{\prime}}\tilde{\sigma}_s^{u^{\prime} v^{\prime}, r^{\prime}} d s\right|\mathcal{F}_{l-1}^n\right]\right|
\end{array}\right\}\leq C \Delta_n^{3 / 2}\left(t_i-t_{l-1}\right).
$$
Group 2 ("$ds\cdot dW_s$"): By Cauchy-Schwartz inequality, the boundedness of $\mu_{s}^{\prime}$, $c_s$, $\tilde{c}_s$, and the local boundedness of third-order partial derivatives of $g$, together with (\ref{eq:X_c_bound}),
\begin{equation}\label{eq:group2}
\begin{aligned}
    &\left|H_{21}\right|=\left|\mathbb{E}\left[\left.\int_{t_{l-1}}^{t_l}\left(X^{\prime}_{s, p}-X^{\prime}_{t_{l-1, p}}\right) \mu^{\prime}_{s, q} d s\int_{t_{l-1}}^{t_i} \partial_{p q, u v, p^{\prime} q^{\prime}}^3 g\left(c_s\right) \tilde{\sigma}_s^{p^{\prime} q^{\prime}, r^{\prime}} d W_s^{r^{\prime}}\right|\mathcal{F}_{l-1}^n\right]\right|\\
    &\leq C\left\{\mathbb{E}\left[\left.\sup _{t \in\left[0, \Delta_n\right]}\left\|X^{\prime}_{t_{l-1}+t}-X^{\prime}_{t_{l-1}}\right\|^2 \Delta_n^2 \right| \mathcal{F}_{l-1}^n\right]\right\}^{1 / 2}\left\{\mathbb{E}\left[\left.\int_{t_{t-1}}^{t_i}\left[\partial_{p q, u v, p^{\prime} q^{\prime}}^3 g\left(c_s\right) \tilde{\sigma}_s^{p^{\prime} q^{\prime}, r^{\prime}}\right]^2 d s\right| \mathcal{F}_{l-1}^n\right]\right\}^{1 / 2}\\
    &\leq C \Delta_n^{3 / 2} \sqrt{t_i-t_{l-1}},
\end{aligned}
\end{equation}
where the BDG inequality is applied in the first inequality. Similarly, 
$$
\left|H_{22}\right|=\left|\mathbb{E}\left[\left.\int_{t_{l-1}}^{t_l}\left(c_s^{p q}-c_{t_{l-1}}^{p q}\right) d s \int_{t_{l-1}}^{t_i} \partial_{p q, u v, p^{\prime} q^{\prime}}^3 g\left(c_s\right) \tilde{\sigma}_s^{p^{\prime} q^{\prime}, r^{\prime}} d W_s^{r^{\prime}}\right|\mathcal{F}_{l-1}^n\right]\right|\leq C \Delta_n^{3 / 2} \sqrt{t_i-t_{l-1}}.
$$
Group 3 ("$dW_s\cdot ds$"): 
\begin{equation}\label{eq:decomgroup3}
\begin{aligned}
&\left|H_{31}\right|=\left|\mathbb{E}\left[\left.\int_{t_{l-1}}^{t_l}\left(X^{\prime}_{s, p}-X^{\prime}_{t_{l-1, p}}\right) \sigma_s^{q, r} d W_s^r \int_{t_{l-1}}^{t_i} \partial_{p q, u v, p^{\prime} q^{\prime}}^3 g\left(c_s\right) \tilde{\mu}_s^{p^{\prime} q^{\prime}} d s\right|\mathcal{F}_{l-1}^n\right]\right|\\
    &\leq \left|\mathbb{E}\left[\left. \int_{t_{l-1}}^{t_l}\left(X^{\prime}_{s, p}-X^{\prime}_{t_{l-1, p}}\right) \sigma_s^{q, r} d W_s^r \int_{t_{l-1}}^{t_i}\partial_{p q, u v, p^{\prime} q^{\prime}}^3 g\left(c_{t_{l-1}}\right) \tilde{\mu}_{t_{l-1}}^{p^{\prime} q^{\prime}}ds\right|\mathcal{F}_{l-1}^n\right]\right|\\
    &+\left|\mathbb{E}\left[\left. \int_{t_{l-1}}^{t_l}\left(X^{\prime}_{s, p}-X^{\prime}_{t_{l-1, p}}\right) \sigma_s^{q, r} d W_s^r \int_{t_{l-1}}^{t_i}\left[\partial_{p q, u v, p^{\prime} q^{\prime}}^3 g\left(c_s\right)-\partial_{p q, u v, p^{\prime} q^{\prime}}^3 g\left(c_{t_{l-1}}\right)\right] \tilde{\mu}_{t_{l-1}}^{p^{\prime} q^{\prime}} ds\right|\mathcal{F}_{l-1}^n\right]\right|\\
    &+\left|\mathbb{E}\left[\left. \int_{t_{l-1}}^{t_l}\left(X^{\prime}_{s, p}-X^{\prime}_{t_{l-1, p}}\right) \sigma_s^{q, r} d W_s^r \int_{t_{l-1}}^{t_i}\partial_{p q, u v, p^{\prime} q^{\prime}}^3 g\left(c_{t_{l-1}}\right)\left(\tilde{\mu}_s^{p^{\prime} q^{\prime}}-\tilde{\mu}_{t_{l-1}}^{p^{\prime} q^{\prime}}\right) ds\right|\mathcal{F}_{l-1}^n\right]\right|\\
    &+\left|\mathbb{E}\left[\left. \int_{t_{l-1}}^{t_l}\left(X^{\prime}_{s, p}-X^{\prime}_{t_{l-1, p}}\right) \sigma_s^{q, r} d W_s^r \int_{t_{l-1}}^{t_i} \left[\partial_{p q, u v, p^{\prime} q^{\prime}}^3 g\left(c_s\right)-\partial_{p q, u v, p^{\prime} q^{\prime}}^3 g\left(c_{t_{l-1}}\right)\right]\left(\tilde{\mu}_s^{p^{\prime} q^{\prime}}-\tilde{\mu}_{t_{l-1}}^{p^{\prime} q^{\prime}}\right)ds\right|\mathcal{F}_{l-1}^n\right]\right|.
\end{aligned}
\end{equation}
Note that for the first term, 
\begin{align*}
    &\left|\mathbb{E}\left[\left.\int_{t_{l-1}}^{t_l}\left(X^{\prime}_{s, p}-X^{\prime}_{t_{l-1}, p}\right) \sigma_s^{q, r} d W_s^r \int_{t_{l-1}}^{t_i} \partial_{p q, u v, p^{\prime} q^{\prime}}^3 g\left(c_{t_{l-1}}\right) \tilde{\mu}_{t_{l-1}}^{p^{\prime}q^{\prime}} d s\right|\mathcal{F}_{l-1}^n\right]\right|\\
    &=\left|\int_{t_{l-1}}^{t_i} \partial_{p q, u v, p^{\prime} q^{\prime}}^3 g\left(c_{t_{l-1}}\right) \tilde{\mu}_{t_{l-1}}^{p^{\prime}q^{\prime}} d s\mathbb{E}\left[\left.\int_{t_{l-1}}^{t_l}\left(X^{\prime}_{s, p}-X^{\prime}_{t_{l-1}, p}\right) \sigma_s^{q, r} d W_s^r \right|\mathcal{F}_{l-1}^n\right]\right|=0,
\end{align*}
and by the mean value theorem, for some $\lambda\in [0,1]$, the second term
\begin{align*}
    &\left|\mathbb{E}\left[\left.\int_{t_{l-1}}^{t_l}\left(X^{\prime}_{s, p}-X^{\prime}_{t_{l-1, p}}\right) \sigma_s^{q, r} d W_s^r \int_{t_{l-1}}^{t_i} \left[\partial_{p q, u v, p^{\prime} q^{\prime}}^3 g\left(c_s\right)-\partial_{p q, u v, p^{\prime} q^{\prime}}^3 g\left(c_{t_{l-1}}\right)\right] \tilde{\mu}_{t_{l-1}}^{p^{\prime} q^{\prime}} d s\right|\mathcal{F}_{l-1}^n\right]\right|\\
    &\leq\left\{\mathbb{E}\left[\left.\left(\int_{t_{l-1}}^{t_l}\left(X^{\prime}_{s, p}-X^{\prime}_{t_{l-1, p}}\right) \sigma_s^{q, r} d W_s^r\right)^2 \right| \mathcal{F}_{l-1}^n\right]\right\}^{1 / 2}\\
    &\quad\times\left\{\mathbb{E}\left[\left.\left(\int_{t_{l-1}}^{t_i} \sum_{u^{\prime}, v^{\prime}=1}^d \partial_{p q, u v, p^{\prime} q^{\prime}, u^{\prime} v^{\prime}}^4 g\left(c_s+\lambda\left(c_{t_{l-1}}-c_s\right)\right)\left(c_s^{u^{\prime} v^{\prime}}-c_{t_{l-1}}^{u^{\prime} v^{\prime}}\right) \tilde{\mu}_{t_{l-1}}^{p^{\prime} q^{\prime}} d s\right)^2\right| \mathcal{F}_{l-1}^n\right]\right\}^{1 / 2}\\
    &\leq C\left\{\mathbb{E}\left[\left.\sup _{t \in\left[0, \Delta_n\right]}\left\|X^{\prime}_{t_{l-1}+t}-X^{\prime}_{t_{l-1}}\right\|^2 \Delta_n \right| \mathcal{F}_{l-1}^n\right]\right\}^{1 / 2}\left\{\mathbb{E}\left[\left.\sup _{t \in\left[0,(i-l+1) \Delta_n\right]}\left\|c_{t_{l-1}+t}-c_{t_{l-1}}\right\|^2\left(t_i-t_{l-1}\right)^2 \right| \mathcal{F}_{l-1}^a\right]\right\}^{1 / 2}\\
    &\leq C \Delta_n \eta_{l-1, i-l+1}^n\left(t_i-t_{l-1}\right),
\end{align*}
where the last inequality holds with (\ref{eq:X_c_bound}), (\ref{eq:eta}), and (\ref{etadef}). Meanwhile, with similar argument, we can show the third and fourth terms in the decomposition of (\ref{eq:decomgroup3}) have the same upper bound. Thus,
\begin{equation}\label{eq:group3}
    \left|\mathbb{E}\left[\left.\int_{t_{l-1}}^{t_l}\left(X^{\prime}_{s, p}-X^{\prime}_{t_{l-1, p}}\right) \sigma_s^{q, r} d W_s^r \int_{t_{l-1}}^{t_i} \partial_{p q, u v, p^{\prime} q^{\prime}}^3 g\left(c_s\right) \tilde{\mu}_s^{p^{\prime} q^{\prime}} d s\right|\mathcal{F}_{l-1}^n\right]\right|\leq C \Delta_n \eta_{l-1, i-l+1}^n\left(t_i-t_{l-1}\right).
\end{equation}
Similarly, $H_{32}$ in Group 3 also satisfies
$$
\left|\mathbb{E}\left[\left.\int_{t_{l-1}}^{t_l}\left(X^{\prime}_{s, p}-X^{\prime}_{t_{l-1, p}}\right) \sigma_s^{q, r} d W_s^r\int_{t_{l-1}}^{t_i} \partial^{4}_{p q, u v, p^{\prime} q^{\prime}, u^{\prime} v^{\prime}} g\left(c_s\right) \tilde{\sigma}_s^{p^{\prime} q^{\prime}, r^{\prime}}\tilde{\sigma}_s^{u^{\prime} v^{\prime}, r^{\prime}} d s\right|\mathcal{F}_{l-1}^n\right]\right|\leq C \Delta_n \eta_{l-1, i-l+1}^n\left(t_i-t_{l-1}\right).
$$
Group 4 ("$dW_s\cdot dW_s$"): Since $H_{41}=0$ when $r\neq r^{\prime}$, we only need to check for $r=r^{\prime}$,
\begin{equation}\label{eq:decomgroup4}
\begin{aligned}
&\left|H_{41}\right|=\left|\mathbb{E}\left[\left.\int_{t_{l-1}}^{t_l}\left(X^{\prime}_{s, p}-X^{\prime}_{t_{l-1, p}}\right) \sigma_s^{q, r} d W_s^r \int_{t_{l-1}}^{t_i} \partial_{p q, u v, p^{\prime} q^{\prime}}^3 g\left(c_s\right) \tilde{\sigma}_s^{p^{\prime} q^{\prime}, r} d W_s^r\right|\mathcal{F}_{l-1}^n\right]\right|\\
&=\left|\mathbb{E}\left[\left.\int_{t_{l-1}}^{t_l}\left(X^{\prime}_{s, p}-X^{\prime}_{t_{l-1, p}}\right) \sigma_s^{q, r}\partial_{p q, u v, p^{\prime} q^{\prime}}^3 g\left(c_s\right) \tilde{\sigma}_s^{p^{\prime} q^{\prime}, r}ds\right|\mathcal{F}_{l-1}^n\right]\right|\\
&=\Bigg|\mathbb{E}\Bigg[\int_{t_{l-1}}^{t_l}\left(X^{\prime}_{s, p}-X^{\prime}_{t_{l-1, p}}\right) \left(\sigma_s^{q, r}-\sigma_{t_{l-1}}^{q, r}+\sigma_{t_{l-1}}^{q, r}\right)\\
&\qquad\qquad\times\left[\partial_{p q, u v, p^{\prime} q^{\prime}}^3 g\left(c_s\right)-\partial_{p q, u v, p^{\prime} q^{\prime}}^3 g\left(c_{t_{l-1}}\right)+\partial_{p q, u v, p^{\prime} q^{\prime}}^3 g\left(c_{t_{l-1}}\right)\right]\left(\tilde{\sigma}_s^{p^{\prime} q^{\prime}, r}-\tilde{\sigma}_{t_{l-1}}^{p^{\prime} q^{\prime}, r}+\tilde{\sigma}_{t_{l-1}}^{p^{\prime} q^{\prime}, r}\right)ds\Bigg|\mathcal{F}_{l-1}^n\Bigg]\Bigg|\\
&\leq \left|\mathbb{E}\left[\left.\int_{t_{l-1}}^{t_l}\left(X^{\prime}_{s, p}-X^{\prime}_{t_{l-1, p}}\right) \sigma_{t_{l-1}}^{q, r} \partial_{p q, u v, p^{\prime} q^{\prime}}^{3} g\left(c_{t_{l-1}}\right) \tilde{\sigma}_{t_{l-1}}^{p^{\prime} q^{\prime}, r} d s\right|\mathcal{F}_{l-1}^n\right]\right|\\
&+\left|\mathbb{E}\left[\left.\int_{t_{l-1}}^{t_l}\left(X^{\prime}_{s, p}-X^{\prime}_{t_{l-1, p}}\right) \left(\sigma_s^{q, r}-\sigma_{t_{l-1}}^{q, r}\right) \partial_{p q, u v, p^{\prime} q^{\prime}}^{3} g\left(c_{t_{l-1}}\right) \tilde{\sigma}_{t_{l-1}}^{p^{\prime} q^{\prime}, r} d s\right|\mathcal{F}_{l-1}^n\right]\right|\\
&+\left|\mathbb{E}\left[\left.\int_{t_{l-1}}^{t_l}\left(X^{\prime}_{s, p}-X^{\prime}_{t_{l-1, p}}\right) \sigma_{t_{l-1}}^{q, r} \left[\partial_{p q, u v, p^{\prime} q^{\prime}}^3 g\left(c_s\right)-\partial_{p q, u v, p^{\prime} q^{\prime}}^3 g\left(c_{t_{l-1}}\right)\right] \tilde{\sigma}_{t_{l-1}}^{p^{\prime} q^{\prime}, r} d s\right|\mathcal{F}_{l-1}^n\right]\right|\\
&+\left|\mathbb{E}\left[\left.\int_{t_{l-1}}^{t_l}\left(X^{\prime}_{s, p}-X^{\prime}_{t_{l-1, p}}\right) \sigma_{t_{l-1}}^{q, r} \partial_{p q, u v, p^{\prime} q^{\prime}}^{3} g\left(c_{t_{l-1}}\right) \left(\tilde{\sigma}_s^{p^{\prime} q^{\prime}, r}-\tilde{\sigma}_{t_{l-1}}^{p^{\prime} q^{\prime}, r}\right)
 d s\right|\mathcal{F}_{l-1}^n\right]\right|\\
&+\left|\mathbb{E}\left[\left.\int_{t_{l-1}}^{t_l}\left(X^{\prime}_{s, p}-X^{\prime}_{t_{l-1, p}}\right) \left(\sigma_s^{q, r}-\sigma_{t_{l-1}}^{q, r}\right)\left[\partial_{p q, u v, p^{\prime} q^{\prime}}^3 g\left(c_s\right)-\partial_{p q, u v, p^{\prime} q^{\prime}}^3 g\left(c_{t_{l-1}}\right)\right]\tilde{\sigma}_{t_{l-1}}^{p^{\prime} q^{\prime}, r} d s\right|\mathcal{F}_{l-1}^n\right]\right|\\
&+\left|\mathbb{E}\left[\left.\int_{t_{l-1}}^{t_l}\left(X^{\prime}_{s, p}-X^{\prime}_{t_{l-1, p}}\right) \left(\sigma_s^{q, r}-\sigma_{t_{l-1}}^{q, r}\right) \partial_{p q, u v, p^{\prime} q^{\prime}}^{3} g\left(c_{t_{l-1}}\right) \left(\tilde{\sigma}_s^{p^{\prime} q^{\prime}, r}-\tilde{\sigma}_{t_{l-1}}^{p^{\prime} q^{\prime}, r}\right)
 d s\right|\mathcal{F}_{l-1}^n\right]\right|\\
&+\left|\mathbb{E}\left[\left.\int_{t_{l-1}}^{t_l}\left(X^{\prime}_{s, p}-X^{\prime}_{t_{l-1, p}}\right) \sigma_{t_{l-1}}^{q, r} \left[\partial_{p q, u v, p^{\prime} q^{\prime}}^3 g\left(c_s\right)-\partial_{p q, u v, p^{\prime} q^{\prime}}^3 g\left(c_{t_{l-1}}\right)\right] \left(\tilde{\sigma}_s^{p^{\prime} q^{\prime}, r}-\tilde{\sigma}_{t_{l-1}}^{p^{\prime} q^{\prime}, r}\right) d s\right|\mathcal{F}_{l-1}^n\right]\right|\\
&+\left|\mathbb{E}\left[\left.\int_{t_{l-1}}^{t_l}\left(X^{\prime}_{s, p}-X^{\prime}_{t_{l-1, p}}\right) \left(\sigma_s^{q, r}-\sigma_{t_{l-1}}^{q, r}\right)\left[\partial_{p q, u v, p^{\prime} q^{\prime}}^3 g\left(c_s\right)-\partial_{p q, u v, p^{\prime} q^{\prime}}^3 g\left(c_{t_{l-1}}\right)\right]\left(\tilde{\sigma}_s^{p^{\prime} q^{\prime}, r}-\tilde{\sigma}_{t_{l-1}}^{p^{\prime} q^{\prime}, r}\right) d s\right|\mathcal{F}_{l-1}^n\right]\right|.
\end{aligned}
\end{equation}
The first term of the decomposition in (\ref{eq:decomgroup4}) can be bounded by
$$
\left|\mathbb{E}\left[\left.\int_{t_{l-1}}^{t_l}\left(X^{\prime}_{s, p}-X^{\prime}_{t_{l-1, p}}\right) \sigma_{t_{l-1}}^{q, r} \partial_{p q, u v, p^{\prime} q^{\prime}}^{3} g\left(c_{t_{l-1}}\right) \tilde{\sigma}_{t_{l-1}}^{p^{\prime} q^{\prime}, r} d s\right|\mathcal{F}_{l-1}^n\right]\right|\leq C\left|\int_{t_{l-1}}^{t_l} \mathbb{E}\left[\left.\left(X^{\prime}_{s, p}-X^{\prime}_{t_{l-1, p}}\right) \right| \mathcal{F}_{l-1}^n\right] d s\right|\leq C\Delta_n^2,
$$
with (\ref{eq:X_c_bound}). By the mean value theorem, (\ref{eq:X_c_bound}), (\ref{eq:eta}), and (\ref{etadef}), for the third term,
\begin{align*}
    &\left|\mathbb{E}\left[\left.\int_{t_{l-1}}^{t_l}\left(X^{\prime}_{s, p}-X^{\prime}_{t_{l-1, p}}\right) \sigma_{t_{l-1}}^{q,r}\left(\partial_{p q, u v, p^{\prime} q^{\prime}}^{3} g\left(c_{s}\right)-\partial_{p q, u v, p^{\prime} q^{\prime}}^{3} g\left(c_{t_{l-1}}\right)\right) \tilde{\sigma}_{t_{l-1}}^{p^{\prime} q^{\prime}, r}ds\right|\mathcal{F}_{l-1}^n\right]\right|\\
    &\leq C\int_{t_{l-1}}^{t_l}\left\{\mathbb{E}\left[\left.\left(X^{\prime}_{s, p}-X^{\prime}_{t_{l-1, p}}\right)^2\right|\mathcal{F}_{l-1}^n\right]\right\}^{1/2}\left\{\mathbb{E}\left[\left.\left\|c_s-c_{t_{l-1}}\right\|^2 \right| \mathcal{F}_{l-1}^n\right]\right\}^{1 / 2}ds\leq C \Delta_n^{3/2}\eta^n_{l-1}.
\end{align*}
Other terms in (\ref{eq:decomgroup4}) can be controlled by the same upper bound similarly.
Thus,
\begin{equation}\label{eq:group4}
    \left|\mathbb{E}\left[\left.\int_{t_{l-1}}^{t_l}\left(X^{\prime}_{s, p}-X^{\prime}_{t_{l-1, p}}\right) \sigma_s^{q, r} d W_s^r \int_{t_{l-1}}^{t_i} \partial_{p q, u v, p^{\prime} q^{\prime}}^3 g\left(c_s\right) \tilde{\sigma}_s^{p^{\prime} q^{\prime}, r} d W_s^r\right|\mathcal{F}_{l-1}^n\right]\right|\leq C\Delta_n^{3/2}\eta^n_{l-1}.
\end{equation}
Combining (\ref{eq:decomalphapartialg}), (\ref{eq:group1}), (\ref{eq:group2}), (\ref{eq:group3}), and (\ref{eq:group4}), we conclude that
$$
\left|\mathbb{E}\left[\left.\alpha_l^{n, p q} \partial_{p q, u v}^2 g\left(c_i^n\right) \right| \mathcal{F}_{l-1}^n\right]\right|\leq C\Delta_n \eta_{l-1, i-l+1}^n\left(t_i-t_{l-1}\right).
$$
\end{proof}

\begin{proof}[Proof of (\ref{eq:sucessive})]
We  show that for $i^{\prime} \geq i \geq l$,
\begin{equation}\label{eq:pfsucess}
    \left|\mathbb{E}\left[\left.\partial_{p q, u v}^2 g\left(c_i^n\right) \partial_{p q, u v}^2 g\left(c_{i^{\prime}}^n\right) \alpha_l^{n, p q} \right| \mathcal{F}_{l-1}^n\right]\right|\leq C\Delta_n \eta_{l-1, i^{\prime}-l+1}^n\left(t_{i^{\prime}}-t_{l-1}\right).
\end{equation}
Note that
\begin{align*}
    &\left|\mathbb{E}\left[\left.\partial_{p q, u v}^2 g\left(c_i^n\right) \partial_{p q, u v}^2 g\left(c_{i^{\prime}}^n\right) \alpha_l^{n, p q} \right| \mathcal{F}_{l-1}^n\right]\right|\\
    &=\left|\mathbb{E}\left[\left.\left[\partial_{p q, u v}^2 g\left(c_i^n\right)-\partial_{p q, u v}^2 g\left(c_{l-1}^n\right)+\partial_{p q, u v}^2 g\left(c_{l-1}^n\right)\right] \left[\partial_{p q, u v}^2 g\left(c_{i^{\prime}}^n\right)-\partial_{p q, u v}^2 g\left(c_{l-1}^n\right)+\partial_{p q, u v}^2 g\left(c_{l-1}^n\right)\right] \alpha_l^{n, p q} \right| \mathcal{F}_{l-1}^n\right]\right|\\
    &\leq \left|\mathbb{E}\left[\left.\left[\partial_{p q, u v}^2 g\left(c_i^n\right)-\partial_{p q, u v}^2 g\left(c_{l-1}^n\right)\right] \left[\partial_{p q, u v}^2 g\left(c_{i^{\prime}}^n\right)-\partial_{p q, u v}^2 g\left(c_{l-1}^n\right)\right] \alpha_l^{n, p q} \right| \mathcal{F}_{l-1}^n\right]\right|\\
    &\quad+\left|\partial_{p q, u v}^2 g\left(c_{l-1}^n\right)\mathbb{E}\left[\left.\left[\partial_{p q, u v}^2 g\left(c_i^n\right)-\partial_{p q, u v}^2 g\left(c_{l-1}^n\right)\right] \alpha_l^{n, p q} \right| \mathcal{F}_{l-1}^n\right]\right|\\
    &\quad+\left|\partial_{p q, u v}^2 g\left(c_{l-1}^n\right)\mathbb{E}\left[\left.\left[\partial_{p q, u v}^2 g\left(c_{i^{\prime}}^n\right)-\partial_{p q, u v}^2 g\left(c_{l-1}^n\right)\right] \alpha_l^{n, p q} \right| \mathcal{F}_{l-1}^n\right]\right|\\
    &\quad+\left|\left[\partial_{p q, u v}^2 g\left(c_{l-1}^n\right)\right]^2\mathbb{E}\left[\left.\alpha_l^{n, p q} \right| \mathcal{F}_{l-1}^n\right]\right|\\
    &\leq \left|\mathbb{E}\left[\left.\left[\partial_{p q, u v}^2 g\left(c_i^n\right)-\partial_{p q, u v}^2 g\left(c_{l-1}^n\right)\right] \left[\partial_{p q, u v}^2 g\left(c_{i^{\prime}}^n\right)-\partial_{p q, u v}^2 g\left(c_{l-1}^n\right)\right] \alpha_l^{n, p q} \right| \mathcal{F}_{l-1}^n\right]\right|\\
    &\quad+C\left|\mathbb{E}\left[\left.\left[\partial_{p q, u v}^2 g\left(c_i^n\right)-\partial_{p q, u v}^2 g\left(c_{l-1}^n\right)\right] \alpha_l^{n, p q} \right| \mathcal{F}_{l-1}^n\right]\right|+C\left|\mathbb{E}\left[\left.\left[\partial_{p q, u v}^2 g\left(c_{i^{\prime}}^n\right)-\partial_{p q, u v}^2 g\left(c_{l-1}^n\right)\right] \alpha_l^{n, p q} \right| \mathcal{F}_{l-1}^n\right]\right|\\
    &\quad+\left|\mathbb{E}\left[\left.\alpha_l^{n, p q} \right| \mathcal{F}_{l-1}^n\right]\right|\\
    &\leq C\Delta_n \eta_{l-1, i-l+1}^n\left(t_i-t_{l-1}\right)+C\Delta_n \eta_{l-1, i^{\prime}-l+1}^n\left(t_{i^{\prime}}-t_{l-1}\right)+C \Delta_{n}^{3 / 2}\left(\sqrt{\Delta_{n}}+\eta_{l-1}^{n}\right),
\end{align*}
where in the third inequality we apply (\ref{eq:better_bound_alpha}), (\ref{eq:alphapartialg}), Itô's lemma and the arguments in the proof of (\ref{eq:alphapartialg}).
\end{proof}
\begin{proof}[Proof of (\ref{eq:A.84})]
{When $i, l-1 < j$, by (\ref{eq:better_bound_alpha}),
\begin{align*}
    &\left|\mathbb{E}\left[\left.\alpha_j^{n, u v}\left(c_{l-1}^{n, p q}-c_i^{n, p q}\right) \partial_{p q, u v}^2 g\left(c_i^n\right) \right| \mathcal{F}_{j-1}^n\right]\right|=\left|\left(c_{l-1}^{n, p q}-c_i^{n, p q}\right) \partial_{p q, u v}^2 g\left(c_i^n\right)\mathbb{E}\left[\left.\alpha_j^{n, u v} \right| \mathcal{F}_{j-1}^n\right]\right|\\
    &\leq \left|\left(c_{l-1}^{n, p q}-c_i^{n, p q}\right) \partial_{p q, u v}^2 g\left(c_i^n\right)C\Delta_n^{3 / 2}\left(\sqrt{\Delta_n}+\eta_{j-1}^n\right)\right|\leq C\Delta_n^{3 / 2}\left(\sqrt{\Delta_n}+\eta_{j-1}^n\right).
\end{align*}
When $i < j \leq l-1$, by (\ref{eq:better_bound_alpha}) \& (\ref{eq:alphapartialg}),
\begin{align*}
    &\left|\mathbb{E}\left[\left.\alpha_j^{n, u v}\left(c_{l-1}^{n, p q}-c_i^{n, p q}\right) \partial_{p q, u v}^2 g\left(c_i^n\right) \right| \mathcal{F}_{j-1}^n\right]\right|\leq C\left|\mathbb{E}\left[\left.\alpha_j^{n, u v}c_{l-1}^{n, p q}\right| \mathcal{F}_{j-1}^n\right]\right|+C\left|\mathbb{E}\left[\left.\alpha_j^{n, u v} \right| \mathcal{F}_{j-1}^n\right]\right|\\
    &\leq C \Delta_n \eta_{j-1, l-j}^n\left(t_{l-1}-t_{j-1}\right)+C \Delta_n^{3 / 2}\left(\sqrt{\Delta_n}+\eta_{j-1}^n\right).
\end{align*}
Finally, when $i\geq j$, 
\begin{align*}
    &\left|\mathbb{E}\left[\left.\alpha_j^{n, u v}\left(c_{l-1}^{n, p q}-c_i^{n, p q}\right) \partial_{p q, u v}^2 g\left(c_i^n\right) \right| \mathcal{F}_{j-1}^n\right]\right|\leq \left|\mathbb{E}\left[\left.\alpha_j^{n, u v}\left(c_{l-1}^{n, p q}-c_i^{n, p q}\right) \partial_{p q, u v}^2 g\left(c_{j-1}^n\right) \right| \mathcal{F}_{j-1}^n\right]\right|\\
    &+\left|\mathbb{E}\left[\left.\alpha_j^{n, u v}\left(c_{l-1}^{n, p q}-c_i^{n, p q}\right) \left(\partial_{p q, u v}^2 g\left(c_i^n\right)-\partial_{p q, u v}^2 g\left(c_{j-1}^n\right)\right) \right| \mathcal{F}_{j-1}^n\right]\right|,
\end{align*}
where the second term can be bounded by $C \Delta_n \eta_{j-1, i-j+1}^n\left(t_i-t_{j-1}\right)$ following the proof of (\ref{eq:alphapartialg}) with the boundedness of $c_{l-1}^{n, p q}-c_i^{n, p q}$. The first term
\begin{align*}
    &\left|\mathbb{E}\left[\left.\alpha_j^{n, u v}\left(c_{l-1}^{n, p q}-c_i^{n, p q}\right) \partial_{p q, u v}^2 g\left(c_{j-1}^n\right) \right| \mathcal{F}_{j-1}^n\right]\right|=\left|\partial_{p q, u v}^2 g\left(c_{j-1}^n\right)\mathbb{E}\left[\left.\alpha_j^{n, u v}\left(c_{l-1}^{n, p q}-c_i^{n, p q}\right)  \right| \mathcal{F}_{j-1}^n\right]\right|\\
    &\leq C\left|\mathbb{E}\left[\left.\alpha_j^{n, u v}c_i^{n, p q}  \right| \mathcal{F}_{j-1}^n\right]\right|+C\left|\mathbb{E}\left[\left.\alpha_j^{n, u v}c_{l-1}^{n, p q}  \right| \mathcal{F}_{j-1}^n\right]\right|,
\end{align*}
where $\left|\mathbb{E}\left[\left.\alpha_j^{n, u v}c_i^{n, p q}  \right| \mathcal{F}_{j-1}^n\right]\right|\leq C \Delta_n \eta_{j-1, i-j+1}^n\left(t_i-t_{j-1}\right)$ by (\ref{eq:alphapartialg}). When $l\leq j$,
\begin{align*}
    \left|\mathbb{E}\left[\left.\alpha_j^{n, u v}c_{l-1}^{n, p q}  \right| \mathcal{F}_{j-1}^n\right]\right|=\left|c_{l-1}^{n, p q}\mathbb{E}\left[\left.\alpha_j^{n, u v}  \right| \mathcal{F}_{j-1}^n\right]\right|\leq C \Delta_n^{3 / 2}\left(\sqrt{\Delta_n}+\eta_{j-1}^n\right).
\end{align*}
When $l>j$, by (\ref{eq:alphapartialg}),
$$
\left|\mathbb{E}\left[\left.\alpha_j^{n, u v}c_{l-1}^{n, p q}  \right| \mathcal{F}_{j-1}^n\right]\right|\leq C \Delta_n \eta_{j-1, l-j}^n\left(t_{l-1}-t_{j-1}\right).
$$}
\end{proof}

\subsection{Technical Proofs related to Theorem \ref{thm2.2}}\label{Technical2_2}
\subsection{{\Blue Proof of \eqref{FourthStOfLmts}.}}
{\DR Recall
\begin{align*}
\widehat{A}^{n,3}_{\lambda_1,\lambda_2}&=\frac{\sqrt{\Delta_n}}{8} \sum_{i=1}^{n-2k_n+1}  \sum_{p, q, u, v} \partial_{p q, u v}^2 g(\hat{c}_i^n)
\beta_{i+\lambda_1 k_n}^{n, p q} \beta_{i+\lambda_2 k_n}^{n, u v},
\end{align*}
and define
\begin{align*}
{A}^{n,3}_{\lambda_1,\lambda_2}&:=\frac{\sqrt{\Delta_n}}{8} \sum_{i=1}^{n-2k_n+1}  \sum_{p, q, u, v} \partial_{p q, u v}^2 g(c_i^n)
\beta_{i+\lambda_1 k_n}^{n, p q} \beta_{i+\lambda_2 k_n}^{n, u v},
\end{align*}
where $\lambda_1, \lambda_2 = 0\ \text{or}\ 1$. Note} that for some $\xi=\hat{c}_i^n+\phi\left(c_i^n-\hat{c}_i^n\right)$, where $\phi\in [0,1]$,
\begin{align*}
    \mathbb{E}\left|\widehat{A}_{0,0}^{n,3}-A_{0,0}^{n,3}\right|
    &\leq \frac{\sqrt{\Delta_n}}{8} \sum_{i=1}^{n-2k_n+1}  \sum_{p, q, u, v}
    \sqrt{\mathbb{E}\left[\left(\partial_{p q, u v}^2 g(\hat{c}_i^n)-\partial_{p q, u v}^2 g(c_i^n)\right)^2\right] \mathbb{E}\left[\left\|\beta_i^n\right\|^4\right]}\\
    &\leq \frac{\sqrt{\Delta_n}}{8} \sum_{i=1}^{n-2k_n+1}  \sum_{p, q, u, v}
    \sqrt{\mathbb{E}\left[\left\|\nabla \partial_{p q, u v}^2 g(\xi)\right\|^2\left\|\hat{c}_i^n-c_i^n\right\|^2\right] \Delta_n}\\
    &\leq C\sqrt{\Delta_n}\sum_{i=1}^{n-2k_n+1}  \sum_{p, q, u, v}
    \Delta_{n}^{3/4}\to 0.
\end{align*}
which also holds with $\widehat{A}_{1,0}^{n,3}$, $\widehat{A}_{0,1}^{n,3}$, and $\widehat{A}_{1,1}^{n,3}$ as well. {\DR Then, in order to analyze the asymptotic behavior of $\widehat{A}^{n,3}_{\lambda_1, \lambda_2}$, we only need to consider ${A}_{\lambda_1, \lambda_2}^{n,3}$. To this end, we follow {closely the analysis following \eqref{eq:beta2_decomp}}. Indeed, one may check that}
\begin{equation}
\begin{aligned}
    {A}_{\lambda_1, \lambda_2}^{n,3}&=\frac{\sqrt{\Delta_n}}{8} \sum_{i=1}^{n-2k_n+1}  \sum_{p, q, u, v} \partial_{p q, u v}^2 g(c_i^n)\bar{K}\left(t_{i+\lambda_1 k_n}\right)^{-1}\bar{K}\left(t_{i+\lambda_2 k_n}\right)^{-1} \sum_{r=1}^6 \xi_{r, i}^{n,\lambda_1,\lambda_2}\\
    &=:\sum_{p, q, u, v}\sum_{r=1}^6{A}_{\lambda_1 ,\lambda_2 }^{n,3}(r,p, q, u, v)=:\sum_{p, q, u, v}\sum_{r=1}^6{A}_{\lambda_1 ,\lambda_2 }^{n,3}(r),
\end{aligned}
\end{equation}
where
\begin{align*}
\xi_{1,i}^{n,\lambda_1, \lambda_2}&:=\sum_{j=1}^{n} K_{{b}}\left(t_{j-1}-t_{i+\lambda_1 k_n}\right)K_{{b}}\left(t_{j-1}-t_{i+\lambda_2 k_n}\right)\alpha_{j}^{n,pq}\alpha_{j}^{n,uv},\\
\xi^{n,\lambda_1, \lambda_2}_{2,i} & = \Delta_{n}^2\left(\sum_{j=1}^{n} K_{{b}}\left(t_{j-1}-t_{i+\lambda_1 k_n}\right) \left(c_{j-1}^{n,pq}-c_{i+\lambda_1 k_n}^{n,pq}\right)\right)
\left(\sum_{j=1}^{n} K_{{b}}\left(t_{j-1}-t_{i+\lambda_2 k_n}\right) \left(c_{j-1}^{n,uv}-c_{i+\lambda_2 k_n}^{n,uv}\right)\right)\\
\xi^{n,\lambda_1, \lambda_2}_{3,i}  &=\sum_{j=1}^n\sum_{l=j+1}^{n} 
K_{{b}}\left(t_{j-1}-t_{i+\lambda_1 k_n}\right) K_{{b}}\left(t_{l-1}-t_{i+\lambda_2 k_n}\right) \alpha_{j}^{n,uv} \alpha_{l}^{n,pq},\\
\xi^{n,\lambda_1, \lambda_2}_{4,i}  &=\sum_{l=1}^{n-1}\sum_{j=l+1}^{n} 
K_{{b}}\left(t_{j-1}-t_{i+\lambda_1 k_n}\right) K_{{b}}\left(t_{l-1}-t_{i+\lambda_2 k_n}\right) \alpha_{j}^{n,uv} \alpha_{l}^{n,pq},\\
\xi^{n,\lambda_1, \lambda_2}_{5,i}  & =\sum_{j=1}^{n} \sum_{l=1}^{n} K_{{b}}\left(t_{j-1}-t_{i+\lambda_1 k_n}\right) K_{{b}}\left(t_{l-1}-t_{i+\lambda_2 k_n}\right) \alpha_{j}^{n,uv} \left(c_{l-1}^{n,pq}-c_{i+\lambda_1 k_n}^{n,pq}\right) \Delta_{n},\\
\xi^{n,\lambda_1, \lambda_2}_{6,i}  & =\sum_{j=1}^{n} \sum_{l=1}^{n} K_{{b}}\left(t_{j-1}-t_{i+\lambda_1 k_n}\right) K_{{b}}\left(t_{l-1}-t_{i+\lambda_2 k_n}\right) \alpha_{j}^{n,pq} \left(c_{l-1}^{n,uv}-c_{i+\lambda_2 k_n}^{n,uv}\right) \Delta_{n}.
\end{align*}
When $r=1$, 
\begin{equation*}
\begin{aligned}
    {A}_{\lambda_1,\lambda_2}^{n,3}(1)&=\sum_{j=1}^n \sum_{i=1}^{j-1}\frac{\sqrt{\Delta_n}}{8} \partial_{p q, u v}^2 g(c_i^n)\prod_{\ell=1}^2\frac{K_{{b}}\left(t_{j-1}-t_{i+\lambda_\ell k_n}\right)}{\bar{K}\left(t_{i+\lambda_\ell k_n}\right)}\mathbbm{1}_{\left\{t_{i}\leq t_{n-2 k_n+1}\right\}}\alpha_j^{n, pq} \alpha_j^{n, uv}\\
    &\quad+\sum_{j=1}^n \sum_{i=j}^{n}
    \frac{\sqrt{\Delta_n}}{8} \partial_{p q, u v}^2 g(c_i^n)\prod_{\ell=1}^2\frac{K_{{b}}\left(t_{j-1}-t_{i+\lambda_\ell k_n}\right)}{\bar{K}\left(t_{i+\lambda_\ell k_n}\right)}\mathbbm{1}_{\left\{t_{i}\leq t_{n-2 k_n+1}\right\}}\alpha_j^{n, pq} \alpha_j^{n, uv}\\
    &=:R+L.
\end{aligned}
\end{equation*}
Similar to \eqref{eq:A.50}, define
\begin{align*}
    M&=\sum_{j=1}^n \sum_{i=j}^{n}\frac{\sqrt{\Delta_n}}{8} \partial_{p q, u v}^2 g(c_{j-1}^n)\prod_{\ell=1}^2\frac{K_{{b}}\left(t_{j-1}-t_{i+\lambda_\ell k_n}\right)}{\bar{K}\left(t_{i+\lambda_\ell k_n}\right)}\mathbbm{1}_{\left\{t_{i}\leq t_{n-2 k_n+1}\right\}}\alpha_j^{n, pq} \alpha_j^{n, uv}    
\end{align*}
and note that, {similar to \eqref{eq:L_M_v4}},
\begin{equation}
\begin{aligned}
\mathbb{E}|L-M|
&\leq C \sqrt{\Delta_n} \sum_{j=1}^n \sum_{i=j}^{n-2k_n+1}K_{{b}}\left(t_{j-1}-t_{i+\lambda_1 k_n}\right)K_{{b}}\left(t_{j-1}-t_{i+\lambda_2 k_n}\right) \sqrt{\left|t_{j-1}-t_i\right| \Delta_n^4}\\
&\leq C \sqrt{\Delta_n} \sum_{j=1}^n \Delta_n \int_{-\infty}^0 K(u-\lambda_1)K(u-\lambda_2) \sqrt{-u b} d u \frac{1}{b}\leq C \sqrt{\frac{\Delta_n}{b}} \rightarrow 0.
\end{aligned}
\end{equation}
Then, we aim to find the limit of
\begin{equation}
R+M=: \frac{\sqrt{\Delta_n}}{8} \sum_{j=1}^n \zeta_j^n \alpha_j^{n, p q} \alpha_j^{n, u v},
\end{equation}
where $R+M$ is an adapted triangular array with
\begin{align*}
    \zeta_j^n&:=\sum_{i=1}^{j-1} \partial_{p q, u v}^2 g(c_i^n)\prod_{\ell=1}^2\frac{K_{{b}}\left(t_{j-1}-t_{i+\lambda_\ell k_n}\right)}{\bar{K}\left(t_{i+\lambda_\ell k_n}\right)}\mathbbm{1}_{\left\{t_{i}\leq t_{n-2 k_n+1}\right\}}\\
&\quad+\sum_{i=j}^n \partial_{p q, u v}^2 g(c_{j-1}^n)
\prod_{\ell=1}^2\frac{K_{{b}}\left(t_{j-1}-t_{i+\lambda_\ell k_n}\right)}{\bar{K}\left(t_{i+\lambda_\ell k_n}\right)}\mathbbm{1}_{\left\{t_{i}\leq t_{n-2 k_n+1}\right\}}\\
&=: \zeta_j^{n, 1}+\zeta_j^{n, 2}.
\end{align*}
By (\ref{eq:better_bound_alpha}), 
\begin{equation}\label{TRTIN}
\sum_{j=1}^n \mathbb{E}\left[\zeta_j^n \alpha_j^{n, p q} \alpha_j^{n, u v} \mid {\mathcal{F}_{j-1}^n}\right]=\sum_{j=1}^n \zeta_j^n\left(\check{c}_{j-1}^{n, p q, u v} \Delta_n^2+O_p\left(\Delta_n^{5 / 2}\right)\right),
\end{equation}
{which suggests to consider}
\begin{equation}
	A_{n}:=\sqrt{\Delta_n}\sum_{j=1}^{n}\zeta_{j}^{n,1}\check{c}_{j-1}^{n,pq,uv}\Delta_{n}^{2},
	\quad 
	B_{n}:=\sqrt{\Delta_n}\sum_{j=1}^{n}\zeta_{j}^{n,2}\check{c}_{j-1}^{n,pq,uv}\Delta_{n}^{2}.
\end{equation}
First note that
\begin{equation*}
\begin{aligned}
\zeta_j^{n, 1} \Delta_n&=\sum_{i=1}^{j-1} \partial_{p q, u v} g(c_i^n)\prod_{\ell=1}^2\frac{K_{{b}}\left(t_{j-1}-t_{i+\lambda_\ell k_n}\right)}{\bar{K}\left(t_{i+\lambda_\ell k_n}\right)}\mathbbm{1}_{\left\{t_{i}\leq t_{n-2 k_n+1}\right\}} \Delta_n\\
&=\int_0^{t_{j-1}\wedge t_{n-2 k_n+1}} \partial_{p q, u v} g\left(c_s\right)
\prod_{\ell=1}^{2}\frac{K_b\left(t_{j-1}-s-\lambda_\ell b\right)}{\bar{K}(s+\lambda_\ell b)}d s+O\left(\frac{\Delta_n}{b_n^2}\right)+O_p\left(\frac{\sqrt{\Delta_n}}{b_n}\right),
\end{aligned}
\end{equation*}
which can be proved similarly as (\ref{eq:K2_discrete}) in Appendix \ref{ApOTPrfs}. Therefore, 
\begin{equation*}
\begin{aligned}
A_n&=\sqrt{\Delta_n} \sum_{j=1}^n \check{c}_{j-1}^{n, p q, u v} \int_0^{t_{j-1}\wedge t_{n-2 k_n+1}} \partial_{p q, u v} g\left(c_s\right)
\prod_{\ell=1}^{2}\frac{K_b\left(t_{j-1}-s-\lambda_\ell b\right)}{\bar{K}(s+\lambda_\ell b)}d s\Delta_n+O_p\left(\sqrt{\Delta_n}\right)\\
&=\sqrt{\Delta_n} \sum_{j=1}^n \check{c}_{j-1}^{n, p q, u v} \partial_{p q, u v} g\left(c_{j-1}^n\right)  \Delta_n\\
&\qquad\qquad\quad\times
\int_0^{\frac{t_{j-1}}{b}} 
\prod_{\ell=1}^{2}\left[\frac{\Delta_n}{b} \sum_{l=1}^n K\left(u+\frac{t_{l-1}-t_{j-1}}{b}-\lambda_\ell\right)\right]^{-1}K(u-\lambda_\ell)du\frac{1}{b}\\
&\quad+O_p\left(\sqrt{\Delta_n}\right)+O_p\left(\frac{\sqrt{\Delta_n}}{\sqrt{b}}\right)\\
&=\frac{\sqrt{\Delta_n}}{b} \int_0^{T} \partial^2_{pq, uv} g\left(c_s\right) \check{c}_s^{pq, uv}\\
&\quad\times \int_0^{s / b}\prod_{\ell=1}^{2} \left[\int_{u-\frac{s}{b}}^{u+\frac{T-s}{b}} K(v-\lambda_\ell) d v\right]^{-1}K(u-\lambda_\ell)\mathbbm{1}_{\left\{u\geq \frac{s-t_{n-2k_n+1}}{b_n}\right\}} d u d s+o_p(1)\\
    &\stackrel{n \rightarrow \infty}{\longrightarrow} \frac{1}{\theta} \int_0^T \partial^2_{pq, uv} g\left(c_s\right) \check{c}_s^{pq, uv} d s \int_0^{\infty} K(u-\lambda_1)K(u-\lambda_2) d u,
\end{aligned}
\end{equation*}
where the second equality above can be shown similarly to \eqref{secndEq}.
For $B_n$, we have
\begin{equation*}
\begin{aligned}
B_n&=\Delta_n^{\frac{3}{2}} \sum_{j=1}^n \partial_{p q, u v} g\left(c_{j-1}^n\right) \check{c}_{j-1}^{n, p q, u v} \int_{t_{j-1}}^{T\wedge t_{n-2 k_n+1}} \prod_{\ell=1}^{2}\bar{K}(s+\lambda_\ell b)^{-1}K_b\left(t_{j-1}-s-\lambda_\ell b\right)d s+o_p(1)\\
&=\frac{\sqrt{\Delta_n}}{b} \int_0^{T} \partial^2_{pq, uv} g\left(c_s\right) \check{c}_s^{pq, uv} \int^0_{\frac{s-t_{n-2k_n+1}}{b}} \prod_{\ell=1}^{2}\left[\int_{u-\frac{s}{b}}^{u+\frac{T-s}{b_n}} K(v-\lambda_\ell) d v\right]^{-1}K(u-\lambda_\ell)d u d s+o_p(1)\\
    &\stackrel{n \rightarrow \infty}{\longrightarrow} \frac{1}{\theta} \int_0^T \partial^2_{pq, uv} g\left(c_s\right) \check{c}_s^{pq, uv} d s \int^0_{-\infty} K(u-\lambda_1)K(u-\lambda_2)
 d u.
\end{aligned}
\end{equation*}
Note that 
$$
\left|\zeta_j^n\right| \leq C \sum_{i=1}^{n-2k_n+1} K_{{b}}\left(t_{j-1}-t_{i+\lambda_1 k_n}\right)K_{{b}}\left(t_{j-1}-t_{i+\lambda_2 k_n}\right)
 \leq C \frac{1}{b_n \Delta_n} \int_{-\infty}^{\infty} K(u-\lambda_1)K(u-\lambda_2)
 d u,
$$
and, thus, the remainder term produced by $O(\Delta_n^{5/2})$ in \eqref{TRTIN} will vanish because
\begin{equation*}
\sqrt{\Delta_n} O\left(\Delta_n^{5 / 2}\right) \sum_{j=1}^n\left|\zeta_j^n\right| \leq \sqrt{\Delta_n} O\left(\Delta_n^{5 / 2}\right) \frac{1}{b_n \Delta_n^2} \leq O_p\left(\frac{\Delta_n}{b_n}\right)=o_p(1).
\end{equation*}
Also, due to \eqref{eq:bounds_alpha}, 
\begin{equation*}
\Delta_n \sum_{j=1}^n \mathbb{E}\left[\left(\zeta_j^n \alpha_j^{n, p q} \alpha_j^{n, u v}\right)^2 \mid \mathcal{F}^n_{j-1}\right] \leq C \Delta_n\left|\zeta_j^n\right|^2 \sum_{j=1}^n \mathbb{E}\left[\left\|\alpha_j^n\right\|^4 \mid \mathcal{F}^n_{j-1}\right]=O_p\left(\frac{\Delta_n^5}{b^2 \Delta_n^3}\right)=o_p(1).
\end{equation*}
Then, we conclude that
$$
{A}_{\lambda_1,\lambda_2}^{n,3}(1)\to\frac{1}{8\theta} \int_0^T \partial^2_{pq, uv} g\left(c_s\right) \check{c}_s^{pq, uv} d s \int^{\infty}_{-\infty} K(u-\lambda_1)K(u-\lambda_2)
 d u.
$$
For $r=2$, {we proceed as in \eqref{RFEANa}-\eqref{RFEAN} when dealing with $S_{n,2}$. Specifically, with the notation $c_{\cdot}^{n,1}:=c_{\cdot}^{n,pq}$, $c_{\cdot}^{n,2}:=c_{\cdot}^{n,uv}$, $c_{\cdot}^{(1)}=c_{\cdot}^{pq}$, and $c_{\cdot}^{(2)}=c_{\cdot}^{uv}$,}
\begin{align}
\nonumber
&{A}_{\lambda_1,\lambda_2}^{n,3}(2)=\frac{\Delta_n^{5/2}}{8} \sum_{i=1}^{n-2k_n+1}\partial_{p q, u v}^2 g(c_i^n)
\prod_{\ell=1}^{2}
\bar{K}\left(t_{i+\lambda_\ell k_n}\right)^{-1}\sum_{j=1}^{n} K_{{b}}\left(t_{j-1}-t_{i+\lambda_\ell k_n}\right) \left(c_{j-1}^{n,\ell}-c_{i+\lambda_\ell k_n}^{n,\ell}\right)\\
\nonumber
&=\frac{\sqrt{\Delta_n}}{8} \sum_{i=1}^{n-2k_n+1}\partial_{p q, u v}^2 g(c_i^n)\bar{K}\left(t_{i+\lambda_1 k_n}\right)^{-1}\bar{K}\left(t_{i+\lambda_2 k_n}\right)^{-1}\\
\nonumber
&\qquad\qquad\times\prod_{\ell=1}^{2}\left[\int_0^T K_b\left(s-t_{i+\lambda_\ell k_n}\right)\left(c_s^{(\ell)}-c_{i+\lambda_\ell k_n}^{n, \ell}\right) d s+\frac{\Delta_n\Delta K}{2b}\left(c_{i+\lambda_\ell k_n-1}^{n, \ell}-c_{i+\lambda_\ell k_n}^{n, \ell}\right)\right]+o_p(1)\\
\nonumber
&=\frac{\sqrt{\Delta_n}}{8} \sum_{i=1}^{n-2k_n+1}\partial_{p q, u v}^2 g(c_i^n)
\prod_{\ell=1}^{2}\bar{K}\left(t_{i+\lambda_\ell k_n}\right)^{-1}\left(\int_0^T K_b\left(s-t_{i+\lambda_\ell k_n}\right)\left(c_s^{(\ell)}-c_{i+\lambda_\ell k_n}^{n, \ell}\right) d s\right)\\
\nonumber
&\quad+ \sum_{i=1}^{n-2k_n+1}\frac{\Delta_n^{\frac{3}{2}}}{16b} \partial_{p q, u v}^2 g\left(c_i^n\right)\prod_{\ell=1}^2 \bar{K}\left(t_{i+\lambda_\ell k_n}\right)^{-1}\\
\nonumber
&\qquad\qquad\times
\int_0^T K_b\left(s-t_{i+\lambda_1 k_n}\right)\left(c_s^{ p q}-c_{i+\lambda_1 k_n}^{n, p q}\right) d s\left(c_{i+\lambda_2 k_n-1}^{n, u v}-c_{i+\lambda_2 k_n}^{n, u v}\right) \Delta K\\
\nonumber
&\quad+ \sum_{i=1}^{n-2k_n+1}\frac{\Delta_n^{\frac{3}{2}}}{16b} \partial_{p q, u v}^2 g\left(c_i^n\right)\prod_{\ell=1}^2 \bar{K}\left(t_{i+\lambda_\ell k_n}\right)^{-1}\\
\nonumber
&\qquad\qquad\times
\int_0^T K_b\left(s-t_{i+\lambda_2 k_n}\right)\left(c_s^{uv}-c_{i+\lambda_2 k_n}^{n, uv}\right) d s\left(c_{i+\lambda_1 k_n-1}^{n, pq}-c_{i+\lambda_1 k_n}^{n, pq}\right) \Delta K+o_{p}(1)\\
&=: Q+P_1+P_2+o_p(1).\label{DcmrTL}
\end{align}
As in \eqref{p2p3}, we can show that
\begin{equation}\label{DcmrTLab}
\begin{split}
    P_\ell&=O\left(\left|\sum_{i=1}^{n-2k_n+1}\Delta_n \int_0^T K_b\left(s-t_{i+\lambda_\ell k_n}\right)\left(c_s^{(\ell)}-c_{i+\lambda_\ell k_n}^{n, \ell}\right) d s  \right|\right).
\end{split}
\end{equation}
Now consider $Q$. First, by integration by parts and denoting $\tilde{\sigma}^{(1)}=\tilde{\sigma}^{pq}$ and $\tilde{\sigma}^{(2)}=\tilde{\sigma}^{uv}$, note that:
\begin{equation}\label{DcmrTLb}
\begin{aligned}
& \int_0^T K_b\left(s-t_{i+\lambda_\ell k_n}\right)\left(c_s^{(\ell)}-c_{i+\lambda_\ell k_n}^{n, \ell}\right) d s \\
&\quad= -L\left(\frac{T-t_{i+\lambda_\ell k_n}}{{b}}\right)\left(c_T^{(\ell)} - c_{t_{i+\lambda_\ell k_n}}^{(\ell)} \right) + L\left(\frac{-t_{i+\lambda_\ell k_n}}{{b}}\right)\left(c_0^{(\ell)} - c_{t_{i+\lambda_\ell k_n}}^{(\ell)} \right)\\
&\quad\quad + \int_0^T L\left(\frac{s-t_{i+\lambda_\ell k_n}}{{b}} \right)\tilde{\sigma}^{(\ell)}_s dW_s + o_p({b}^{\frac{1}{2}})\\
&\quad=:  R^{(\ell)}_{1,i} + R^{(\ell)}_{2,i} +\int_0^T L\left(\frac{s-t_{i+\lambda_\ell k_n}}{b} \right)\tilde{\sigma}^{(\ell)}_s dW_s+ o_p({b}^{\frac{1}{2}}).
\end{aligned}
\end{equation}
Then, as in Eq.~\eqref{eq:R1},
\begin{equation*}
\begin{aligned}
&\mathbb{E}\left|\frac{\sqrt{\Delta_n}}{8} \sum_{i=1}^{n-2k_n+1}\partial_{p q, u v}^2 g(c_i^n)\prod_{\ell=1}^{2}\bar{K}\left(t_{i+\lambda_\ell k_n}\right)^{-1}
 R_{1, i}^{(\ell)} \right|\\
&\leq \frac{C}{\delta^2} \sum_{i=1}^{n-2k_n+1} \sqrt{\Delta_n} \prod_{\ell=1}^{2}|L\left(\frac{T-t_{i+\lambda_\ell k_n}}{b}\right)|\left(T-t_{i+\lambda_\ell k_n}\right)\\
& \leq \frac{C}{\delta^2} \frac{b^2}{\sqrt{\Delta_n}} \int_{-\infty}^{\infty} \prod_{\ell=1}^2|L(u-\lambda_\ell)| |u-\lambda_\ell| d u \xrightarrow{n \rightarrow \infty} 0.
\end{aligned}
\end{equation*}
In the same way, $$\mathbb{E}\left|\frac{\sqrt{\Delta_n}}{8} \sum_{i=1}^{n-2k_n+1}\partial_{p q, u v}^2 g(c_i^n)\prod_{\ell=1}^2\bar{K}\left(t_{i+\lambda_\ell k_n}\right)^{-1}
 R_{2, i}^{\ell}\right|\xrightarrow{n \rightarrow \infty} 0.$$
To analyze the term coming from multiplying the $dW$ terms in \eqref{DcmrTLb}, we consider two cases. For $t_i<\sqrt{b}$, by Cauchy's and BDG's inequalities,
\begin{equation*}
\begin{aligned}
&\mathbb{E}\Bigg|\sum_{i=1,t_i<\sqrt{b}}^{n-2k_n+1} \frac{\sqrt{\Delta_n}}{8} \partial_{p q, u v}^2 g(c_i^n)\prod_{\ell=1}^{2}\bar{K}\left(t_{i+\lambda_\ell k_n}\right)^{-1}\int_0^T L\left(\frac{s-t_{i+\lambda_\ell k_n}}{{b}} \right)\tilde{\sigma}^{(\ell)}_s dW_s\Bigg|\\
&\quad \leq \frac{C}{\delta^2} b\sum_{t_i<\sqrt{b}} \sqrt{\Delta_n} \int L^2(u)d u \leq \frac{C}{\delta} \sqrt{b} \rightarrow 0 .
\end{aligned}
\end{equation*}
For $t_i>\sqrt{b}$, by the proof of Theorem 6.2 in \cite{FigLi}, $\int_0^T L\left(\frac{t-t_{i+\lambda_\ell k_n}}{b} \right)\tilde{\sigma}^{(\ell)}_t dW_t=\sum_{r=1}^{d}\tilde{\sigma}^{(\ell),r}_{t_{i+\lambda_\ell k_n} - \sqrt{b}}\int_{t_{i+\lambda_\ell k_n} - \sqrt{b}}^T L\left(\frac{t-t_{i+\lambda_\ell k_n}}{b} \right) dW^{r}_t + o_p({b}^{\frac{1}{2}})$. Since $b\ll\sqrt{b}$, $t_{i+\lambda_\ell k_n}=t_{i}+\lambda_\ell b$, and $b\to{}0$, we can replace $\tilde{\sigma}^{(\ell),r}_{t_{i+\lambda_\ell k_n} - \sqrt{b}}$ and $\partial_{p q, u v}^2 g(c_i^n)$ with $\tilde{\sigma}^{(\ell),r}_{t_{i} - \sqrt{b}}$ and $\partial_{p q, u v}^2 g(c_{t_i-\sqrt{b}})$, respectively. Thus, for any $r_1,r_2\in\{1,\dots,d\}$, it suffices to consider the asymptotic behavior of
\begin{equation*}
\begin{aligned}
&\sum_{i=1,t_i>\sqrt{b}}^{n-2k_n+1} \frac{\sqrt{\Delta_n}}{8} \partial_{p q, u v}^2 g(c_{t_i-\sqrt{b}})
\prod_{\ell=1}^2\bar{K}\left(t_{i+\lambda_\ell k_n}\right)^{-1}\tilde{\sigma}_{t_{i} - \sqrt{b}}^{(\ell), r_\ell} \int_{t_{i+\lambda_\ell k_n} - \sqrt{b}}^T L\left(\frac{t-t_{i+\lambda_\ell k_n}}{b} \right) dW^{r_\ell}_t=:\sum_{i=1,t_i>\sqrt{b}}^{n-2k_n+1}\zeta_i^{n}.
\end{aligned}
\end{equation*}
When $r_1= r_2=:r$,
\begin{equation*}
\begin{aligned}
    &\sum_{i=1,t_i>\sqrt{b}}^{n-2k_n+1}\mathbb{E}\left[\zeta_i^n \mid \mathcal{F}_{t_i-\sqrt{b}}\right]=\sum_{i=1,t_i>\sqrt{b}}^{n-2k_n+1} \frac{\sqrt{\Delta_n}}{8} \partial^2_{pq, uv} g\left(c_{t_i-\sqrt{b}}\right) \prod_{\ell=1}^2\bar{K}\left(t_{i+\lambda_\ell k_n}\right)^{-1}\tilde{\sigma}_{t_i-\sqrt{b}}^{(\ell), r}\\
    &\qquad\qquad\qquad\qquad\qquad \qquad\qquad\times \int_{t_i+(\lambda_1\vee \lambda_2)k_n-\sqrt{b}}^T L\left(\frac{t-t_{i+\lambda_1 k_n}}{b} \right)L\left(\frac{t-t_{i+\lambda_2 k_n}}{b} \right) d t\\
    &=\int_0^{t_{n-2k_n+1}-\sqrt{b}} \frac{\theta}{8} \partial^2_{pq, uv} g\left(c_s\right)\prod_{\ell=1}^2\left(\int_{-\frac{s+\sqrt{b}}{b}}^{\frac{T-s-\sqrt{b}}{b}} K(v-\lambda_\ell) d v\right)^{-1}\tilde{\sigma}_s^{(\ell), r}\int_{(\lambda_1\vee\lambda_2)-\frac{1}{\sqrt{b}}}^{\frac{T-s-\sqrt{b}}{b}} \prod_{\ell=1}^2 L(u-\lambda_\ell)d ud s+o_p(1)\\
    &\quad\stackrel{n \rightarrow \infty}{\longrightarrow}\frac{\theta}{8} \int_0^T \partial^2_{pq, uv} g\left(c_s\right) \tilde{\sigma}_s^{pq, r} \tilde{\sigma}_s^{uv, r} d s\left(\int_{-\infty}^{\infty} L(u-\lambda_1)L(u-\lambda_2) d u\right).
\end{aligned}
\end{equation*}
When $r_1\neq r_2$, $\mathbb{E}\left[\zeta_i^n \mid \mathcal{F}_{t_i-\sqrt{b}}\right]=0$. 
By Cauchy's and BDG's inequalities, we also have:
\begin{align*}
&\mathbb{E}\left(\sum_{i=1,t_i>\sqrt{b}}^{n-2k_n+1}\left(\zeta_i^n-\mathbb{E}\left[\zeta_i^n \mid \mathcal{F}_{t_i-\sqrt{b}}\right]\right)\right)^2\leq C \sum_{t_i>\sqrt{b}} \mathbb{E}\left[\mathbb{E}\left[\left(\zeta_i^n\right)^2 \mid \mathcal{F}_{t_i-\sqrt{b}}\right]\right]\\
&\leq C \sum_{t_i>\sqrt{b}} \Delta_n
\prod_{\ell=1}^2\left(\int_{-\infty}^{\infty} L^2\left(\frac{t-t_{i+\lambda_{\ell}k_n}}{b}\right)d u\right)\xrightarrow{n \rightarrow \infty} 0
\\
&\leq C \sum_{t_i>\sqrt{b}} \Delta_n\prod_{\ell=1}^2\int_{-\infty}^{\infty} L^2(u-\lambda_\ell)du b^2 \xrightarrow{n \rightarrow \infty} 0.
\end{align*}
We then conclude the following asymptotics for the term $Q$ in \eqref{DcmrTL}:
\begin{equation}
Q \xrightarrow{\mathbb{P}} \frac{\theta}{8} \sum_{r=1}^d \int_0^T \partial_{p q, u v} g\left(c_s\right) \tilde{\sigma}_s^{p q, r} \tilde{\sigma}_s^{u v, r} d s\left(\int_{-\infty}^{\infty} L(u-\lambda_1)L(u-\lambda_2) d u\right) .
\end{equation}
As for $P_\ell$, note that 
\begin{align*}
    &\mathbb{E}\left|\sum_{i=1}^{n-2k_n+1} \Delta_n R_{1, i}^{(\ell)}\right| \leq \sum_{i=1}^{n-2k_n+1} \Delta_n \mathbb{E}\left|L\left(\frac{T-t_{i+\lambda_\ell k_n}}{b}\right)\left(c_T^{(\ell)}-c_{t_{i+\lambda_1 k_n}}^{(\ell)}\right)\right|\\ 
    &\leq \sum_{i=1}^{n-2k_n+1}\Delta_n\left|L\left(\frac{T-t_{i+\lambda_\ell k_n}}{b}\right)\right|\left(T-t_{i+\lambda_\ell k_n}\right)^{1/2}=\int_{-\infty}^{\infty}|L(u)| b \sqrt{b u} d u\\
    &=b^{3 / 2} \int_{-\infty}^{\infty}|L (u)| \sqrt{u} d u \rightarrow 0,
\end{align*}
and
\begin{align*}
    &\mathbb{E}\left|\sum_{i=1}^{n-2k_n+1} \Delta_n \int_0^T L\left(\frac{t-t_{i+\lambda_\ell k_n}}{b}\right) \tilde{\sigma}_t^{(\ell)} d W_t\right| \leq \sum_{i=1}^{n-2k_n+1} \Delta_n \sqrt{\mathbb{E}\left(\int_0^T L\left(\frac{t-t_{i+\lambda_\ell k_n}}{b}\right) \tilde{\sigma}_t^{(\ell)} d W_t\right)^2} \\
    &\leq \sum_{i=1}^{n-2k_n+1}\Delta_n \sqrt{\int_0^T L\left(\frac{t-t_{i+\lambda_\ell k_n}}{b}\right)^2 d t}\leq \sum_{i=1}^n \Delta_n \sqrt{\int_{-\infty}^{\infty} L(u)^2 d u b} \rightarrow 0.
\end{align*}
Then, using \eqref{DcmrTLab} and \eqref{DcmrTLb}, we conclude that $P_\ell \to 0$. Thus, as $n \rightarrow \infty$,
\begin{equation*}
{A}_{\lambda_1,\lambda_2}^{n,3}(2)\xrightarrow{\mathbb{P}} \frac{\theta}{8} \sum_{r=1}^d \int_0^T \partial_{p q, u v} g\left(c_s\right) \tilde{\sigma}_s^{p q, r} \tilde{\sigma}_s^{u v, r} d s\left(\int_{-\infty}^{\infty} L(u-\lambda_1)L(u-\lambda_2) d u\right).
\end{equation*}
For $r=3$, we split the inner summation into three parts:
\begin{equation*}
\begin{split}
{A}_{\lambda_1,\lambda_2}^{n,3}(3)
   & = \sum_{i=1}^{n-2k_n+1}\frac{\sqrt{\Delta_n}}{8}\partial_{p q, u v}^2 g(c_i^n)\prod_{\ell=1}^{2}\bar{K}\left(t_{i+\lambda_\ell k_n}\right)^{-1}\\
   &\qquad
\times\left(\sum_{j=i+1}^n\sum_{l=j+1}^n +\sum_{j=1}^{i}\sum_{l=j+1}^{i} +\sum_{j=1}^{i}\sum_{l=i+1}^{n}\right)
K_{{b}}\left(t_{j-1}-t_{i+\lambda_1 k_n}\right) 
K_{{b}}\left(t_{l-1}-t_{i+\lambda_2 k_n}\right)\alpha^{n,uv}_j\alpha^{n,pq}_l\\
&  =:  R + L + Cross.
\end{split}
\end{equation*}
Rewrite
\begin{equation*}
\begin{split}
    R  &=  \sum_{l=3}^n\alpha^{n,pq}_l\sum_{j = 2}^{l-1}\sum_{i=1}^{j-1} \frac{\sqrt{\Delta_n}}{8}\partial_{p q, u v}^2 g(c_i^n)\prod_{\ell=1}^2 \bar{K}\left(t_{i+\lambda_\ell k_n}\right)^{-1}\\
    &\qquad\qquad\times K_{{b}}\left(t_{j-1}-t_{i+\lambda_1 k_n}\right) 
K_{{b}}\left(t_{l-1}-t_{i+\lambda_2 k_n}\right)\mathbbm{1}_{\left\{t_i\leq t_{n-2k_n+1}\right\}} \alpha^{n,uv}_j\\
    &=:  \sum_{l= 3}^n \alpha^{n,pq}_l \zeta^n_l.
\end{split}
\end{equation*}
Then,
\begin{equation*}
\begin{split}
& \mathbb{E}\left|   \sum_{l=3}^n \mathbb{E}\left[\alpha^{n,pq}_l \zeta^{n}_l\mid \mathcal{F}^n_{l-1} \right] \right| \leq C\sum_{l=3}^n  \mathbb{E}\left[|\zeta^n_l|  \Delta_n^{3/2}\left(\sqrt{\Delta_n} + \eta^n_{l-1}\right)\right]\\
  &   \leq \frac{C}{\delta^2} \sum_{l=3}^n \Delta_n^{3/2}\sum_{j = 2}^{l-1}\sum_{i=1}^{(j-1)\wedge(n-2k_n+1)}\sqrt{\Delta_n}\left|
  K_{{b}}\left(t_{j-1}-t_{i+\lambda_1 k_n}\right) 
K_{{b}}\left(t_{l-1}-t_{i+\lambda_2 k_n}\right) \right|\sqrt{\mathbb{E}\left(\alpha^{n,uv}_j\right)^2 \mathbb{E}\left(\eta^n_{l-1}\right)^2 }\\
 &    \leq \frac{C}{\delta^2} \Delta_n \sum_{l=3}^n\sqrt{\mathbb{E}\left(\eta^n_{l-1}\right)^2} \left(\int_{-\infty}^{\infty}|K(x)|dx\right)^2\stackrel{n\to\infty}{\longrightarrow}0,
\end{split}
\end{equation*}
and 
\begin{equation*} 
\begin{split}
& \mathbb{E}\left|   \sum_{l=3}^n \mathbb{E}\left[\left(\alpha^{n,pq}_l \zeta^n_l\right)^2\mid \mathcal{F}^n_{l-1} \right] \right|
\leq   C\Delta_n^2  \sum_{l=3}^n  \mathbb{E}[\left(\zeta^n_l\right)^2]\\
&\leq  C\Delta_n^2  \sum_{l=3}^n \sum_{j=2}^{l-1}\left (\sum_{i=1}^{(j-1)\wedge(n-2k_n+1)}\sqrt{\Delta_n}  |K_{{b}}\left(t_{j-1}-t_{i+\lambda_1 k_n}\right) 
K_{{b}}\left(t_{l-1}-t_{i+\lambda_2 k_n}\right)|\right)^2 \mathbb{E}\left[\left(\alpha^{n,uv}_j\right)^2\right] \\
&\quad + C\Delta_n^2  \sum_{l=3}^n \left (\sum_{j=2}^{l-1}\sum_{i=1}^{(j-1)\wedge(n-2k_n+1)}\sqrt{\Delta_n} |K_{{b}}\left(t_{j-1}-t_{i+\lambda_1 k_n}\right) 
K_{{b}}\left(t_{l-1}-t_{i+\lambda_2 k_n}\right)|\right)^2\max_{j,j': j\neq j'} \left| \mathbb{E}\left(\alpha^{n,uv}_j\alpha^{n,uv}_{j'}\right)\right|\\
 & \leq C\Delta_n b\sum_{l=3}^n \int_{-\infty}^{\infty}\left( \int_{-\infty}^{\infty}\left|K(u) K(v+u)  \right|du\right)^2 dv + C \sum_{l=3}^n\Delta_n^{3/2}\left( \int_{-\infty}^{\infty} \int_{-\infty}^{\infty}\left|K(u) K(v+u)\right|dudv  \right)^2\stackrel{n\to\infty}{\longrightarrow}0.
\end{split}
\end{equation*}
Hence, $R\xrightarrow{\mathbb{P}} 0$. 
Similarly, $L$ can be written as
\begin{equation*}
\begin{split}
L &=  \sum_{j = 1}^{n}\sum_{l=j+1}^{n} \sum_{i = l}^{n-2k_n+1}\frac{\sqrt{\Delta_n}}{8}\partial_{p q, u v}^2 g(c_i^n)\bar{K}\left(t_{i+\lambda_1 k_n}\right)^{-1}\bar{K}\left(t_{i+\lambda_2 k_n}\right)^{-1}\\
&\qquad\qquad\qquad\qquad\times
K_{{b}}\left(t_{j-1}-t_{i+\lambda_1 k_n}\right) K_{{b}}\left(t_{l-1}-t_{i+\lambda_2 k_n}\right) \alpha^{n,uv}_j \alpha^{n,pq}_l =: \sum_{j=1}^{n} \zeta^n_j.
\end{split}
\end{equation*} 
Then, similarly to (\ref{eq:eta_sum}), (\ref{eq:L4p3_mean}) and (\ref{eq:L4p3_2moment}), we can apply (\ref{eq:alphapartialg}) and (\ref{eq:sucessive}) to show that
\begin{equation}\label{TBTrUsd}
\mathbb{E}\left|\sum_{j= 1}^n \mathbb{E}\left[\zeta^n_j\mid \mathcal{F}^n_{j-1}\right] \right|\to 0,\quad \mathbb{E}\left(\sum_{j= 1}^n \left( \zeta^n_j- \mathbb{E}\left[\zeta^n_j\mid \mathcal{F}^n_{j-1}\right] \right)  \right)^2\to 0.
\end{equation}
The Cross term can be analyzed similarly to the case of $L$. We then conclude 
\begin{equation}         
    {A}_{\lambda_1,\lambda_2}^{n,3}(r)\to_p 0, \quad \text{for} \ r=3 \ \text{and} \ r=4.
\end{equation}
When $r=5$, first note that by (\ref{eq:A.84}), we have
\begin{equation}\label{eq:ikn}
\begin{aligned}
    &\left|\mathbb{E}\left[\left.\alpha^{n,uv}_j\left(c_{i+\lambda_1k_n}^{n, p q}-c_i^{n, p q} \right)\partial^2_{pq,uv}g(c^n_i) \right| \mathcal{F}^n_{j-1} \right]\right| \\
    &\leq
\begin{cases}
        C\Delta_n^{3/2}\left(\sqrt{\Delta_n} + \eta^n_{j-1} \right)& \text{ if } i+\lambda_1k_n<j,\\
    C\Delta_n^{3/2}\left(\sqrt{\Delta_n} + \eta^n_{j-1} \right)+C\Delta_n\eta^{n}_{j-1,i+\lambda_1k_n-j+1}\left(t_{i+\lambda_1k_n}-t_{j-1}\right) & \text{ if } i < j \leq i+\lambda_1k_n,\\
    C \Delta_n \eta_{j-1, i+\lambda_1k_n-j+1}^n\left(t_{i+\lambda_1k_n}-t_{j-1}\right)+C\Delta_n\eta^{n}_{j-1,i-j+1}\left(t_{i}-t_{j-1}\right) & \text{ if } i\geq j.
    \end{cases}
\end{aligned}
\end{equation}
Then,
\begin{align*}
    {A}_{\lambda_1,\lambda_2}^{n,3}(5)&=\frac{\sqrt{\Delta_n}}{8} \sum_{i=1}^{n-2k_n+1}\partial_{p q, u v}^2 g(c_i^n)\bar{K}\left(t_{i+\lambda_1 k_n}\right)^{-1}\bar{K}\left(t_{i+\lambda_2 k_n}\right)^{-1}\\
    &\qquad\qquad\qquad\sum_{j=1}^{n} \sum_{l=1}^{n} K_{{b}}\left(t_{j-1}-t_{i+\lambda_1 k_n}\right) K_{{b}}\left(t_{l-1}-t_{i+\lambda_2 k_n}\right) \alpha_{j}^{n,uv} \left(c_{l-1}^{n,pq}-c_{i+\lambda_1 k_n}^{n,pq}\right) \Delta_{n}\\
    &=:\sum_{j=1}^{n}\zeta_{j}^{n}.
\end{align*}
By (\ref{eq:A.84}), (\ref{eq:ikn}), and the following decomposition:
\begin{align*}
    &\left|\mathbb{E}\left[\left.\alpha^{n,uv}_j\left(c_{l-1}^{n, p q}-c_{i+\lambda_1k_n}^{n, p q} \right)\partial^2_{pq,uv}g(c^n_i) \right| \mathcal{F}^n_{j-1} \right]\right|\\
    &\leq \left|\mathbb{E}\left[\left.\alpha^{n,uv}_j\left(c_{l-1}^{n, p q}-c_i^{n, p q} \right)\partial^2_{pq,uv}g(c^n_i) \right| \mathcal{F}^n_{j-1} \right]\right|+\left|\mathbb{E}\left[\left.\alpha^{n,uv}_j\left(c_{i+\lambda_1k_n}^{n, p q}-c_i^{n, p q} \right)\partial^2_{pq,uv}g(c^n_i) \right| \mathcal{F}^n_{j-1} \right]\right|,
\end{align*}
we have 
\begin{equation*}
\begin{aligned}
    &\mathbb{E}\left|\sum_{j= 1}^n \mathbb{E}\left[\zeta^n_j\mid \mathcal{F}^n_{j-1}\right] \right|\\
    &\leq C \Delta_n^{3/2}\sum_{j= 1}^n\sum_{i=1}^{n-2k_n+1}\sum_{l=1}^{n}\left|K_{{b}}\left(t_{j-1}-t_{i+\lambda_1 k_n}\right) K_{{b}}\left(t_{l-1}-t_{i+\lambda_2 k_n}\right)\right|\mathbb{E}\left|\mathbb{E}\left[\left.\alpha^{n,uv}_j\left(c_{l-1}^{n, p q}-c_{i+\lambda_1k_n}^{n, p q} \right)\partial^2_{pq,uv}g(c^n_i) \right| \mathcal{F}^n_{j-1} \right]\right|\\
    &\leq C \Delta_n^{3/2}\sum_{j= 1}^n\sum_{i=1}^{n-2k_n+1}\sum_{l=1}^{n}\left|K_{{b}}\left(t_{j-1}-t_{i+\lambda_1 k_n}\right) K_{{b}}\left(t_{l-1}-t_{i+\lambda_2 k_n}\right)\right|\Big[\Delta_n^{3/2}\left(\sqrt{\Delta_n}+ \mathbb{E}\left[\eta^n_{j-1}\right] \right)\\
    & +\Delta_n \mathbb{E}\left[\eta_{j-1, \left|l-j\right|}^n\right]\left|t_{l-1}-t_{j-1}\right|+\Delta_n\mathbb{E}\left[\eta^{n}_{j-1,\left|i-j+1\right|}\right]\left|t_{i}-t_{j-1}\right|+\Delta_n \mathbb{E}\left[\eta_{j-1, \left|i+\lambda_1k_n-j+1\right|}^n\right]\left|t_{i+\lambda_1k_n}-t_{j-1}\right|\Big]\\
    &\leq C \Delta_n^{3/2}\sum_{j= 1}^n\sum_{i=1}^{n-2k_n+1} \left|K_{{b}}\left(t_{j-1}-t_{i+\lambda_1 k_n}\right)\right|\int_{-t_{i}/b}^{(T-t_i)/b}\left|K(u-\lambda_2)\right|\Big[\Delta_n+ \mathbb{E}\left[\eta^n_{j-1}\right]\sqrt{\Delta_n}\\
    &+ \mathbb{E}\left[\eta_{t_{j-1},\left|t_i-t_{j-1}+ub\right|}\right]\left|t_{i}-t_{j-1}+ub\right|+\mathbb{E}\left[\eta_{t_{j-1},\left|t_i-t_{j-1}\right|}\right]\left|t_{i}-t_{j-1}\right|+\mathbb{E}\left[\eta_{t_{j-1}, \left|t_{i+\lambda_1k_n}-t_{j-1}\right|}\right]\left|t_{i+\lambda_1k_n}-t_{j-1}\right|\Big]du\\
    &\leq C\sqrt{\Delta_n}\sum_{j=1}^{n}\int_{(t_{j-1}-T)/b}^{t_{j-1}/b}\left|K(v-\lambda_1)\right|\int_{v-\frac{t_{j-1}}{b}}^{v+\frac{T-t_{j-1}}{b}}\left|K(u-\lambda_2)\right|\Big[\Delta_n+ \eta^n_{j-1}\sqrt{\Delta_n}\\
    &\qquad\qquad\qquad\qquad+\mathbb{E}\left[\eta_{t_{j-1},\left|u-v\right|b}\right]\left|u-v\right|b+\mathbb{E}\left[\eta_{t_{j-1},\left|v\right|b}\right]\left|v\right|b+\mathbb{E}\left[\eta_{t_{j-1}, \left|v-\lambda_1\right|b}\right]\left|v-\lambda_1\right|b\Big]dudv\\
    &\leq C\sqrt{\Delta_n}b\sum_{j=1}^{n}\int \left|K(v-\lambda_1)\right|\int \left|K(u-\lambda_2)\right|\Bigg[\frac{\Delta_n}{b}+\eta^n_{j-1}\frac{\sqrt{\Delta_n}}{b}\\
    &\qquad\qquad\qquad\qquad\qquad+\mathbb{E}\left[\eta_{t_{j-1},\left|u-v\right|b}\right]\left|u-v\right|+\mathbb{E}\left[\eta_{t_{j-1},\left|v\right|b}\right]\left|v\right|+\mathbb{E}\left[\eta_{t_{j-1}, \left|v-\lambda_1\right|b}\right]\left|v-\lambda_1\right|\Bigg]dudv\to 0.
\end{aligned}
\end{equation*}
Meanwhile, by (\ref{eq:X_c_bound}), (\ref{eq:bounds_alpha}), and Cauchy-Schwartz inequality,
\begin{equation*}
\begin{aligned}
&\mathbb{E}\left(\sum_{j= 1}^n \left( \zeta^n_j- \mathbb{E}\left[\zeta^n_j\mid \mathcal{F}^n_{j-1}\right] \right)  \right)^2\leq \sum_{j= 1}^n\mathbb{E}\left[\left( \zeta^n_j \right)^2\right]\\
&\leq C\Delta_n^3\sum_{j= 1}^n \mathbb{E}\Bigg[\left(\alpha_{j}^{n,uv}\right)^2\Bigg(\sum_{i=1}^{n-2k_n+1}\partial_{p q, u v}^2 g(c_i^n)\bar{K}\left(t_{i+\lambda_1 k_n}\right)^{-1}\bar{K}\left(t_{i+\lambda_2 k_n}\right)^{-1}\\
&\qquad\qquad\qquad\qquad\qquad\qquad\qquad\qquad\times
\sum_{l=1}^{n}K_{{b}}\left(t_{j-1}-t_{i+\lambda_1 k_n}\right) K_{{b}}\left(t_{l-1}-t_{i+\lambda_2 k_n}\right) \left(c_{l-1}^{n,pq}-c_{i+\lambda_1 k_n}^{n,pq}\right)\Bigg)^2\Bigg]\\
&\leq C\Delta_n^3\sum_{j= 1}^n \sum_{i,i^{\prime}}\sum_{l,l^{\prime}}\left|K_{{b}}\left(t_{j-1}-t_{i+\lambda_1 k_n}\right)\right|\left| K_{{b}}\left(t_{l-1}-t_{i+\lambda_2 k_n}\right)\right|\left|K_{{b}}\left(t_{j-1}-t_{i^{\prime}+\lambda_1 k_n}\right)\right|\left| K_{{b}}\left(t_{l^{\prime}-1}-t_{i^{\prime}+\lambda_2 k_n}\right)\right|\\
&\qquad\qquad\qquad\qquad\qquad\qquad\qquad\qquad\qquad\qquad\qquad\qquad\times
\mathbb{E}\left[\left(\alpha_{j}^{n,uv}\right)^2\left|c_{l-1}^{n,pq}-c_{i+\lambda_1 k_n}^{n,pq}\right|\left|c_{l^{\prime}-1}^{n,pq}-c_{i^{\prime}+\lambda_1 k_n}^{n,pq}\right|\right]\\
&\leq C\Delta_n^5\sum_{j= 1}^n \sum_{i,i^{\prime}}\sum_{l,l^{\prime}}\left|K_{{b}}\left(t_{j-1}-t_{i+\lambda_1 k_n}\right)\right|\left| K_{{b}}\left(t_{l-1}-t_{i+\lambda_2 k_n}\right)\right|\left|K_{{b}}\left(t_{j-1}-t_{i^{\prime}+\lambda_1 k_n}\right)\right|\left| K_{{b}}\left(t_{l^{\prime}-1}-t_{i^{\prime}+\lambda_2 k_n}\right)\right|\\
&\qquad\qquad\qquad\qquad\qquad\qquad\qquad\qquad\qquad\qquad\qquad\qquad\qquad\qquad\times \left|t_{l-1}-t_{i+\lambda_1k_n}\right|^{1/4}\left|t_{l^{\prime}-1}-t_{i^{\prime}+\lambda_1k_n}\right|^{1/4}\\
&\leq C\Delta_n\sum_{j= 1}^n\left(\int \left|K(v-\lambda_1)\right|\left|v-\lambda_1\right|^{1/4}dv\right)^2\left(\int \left|K(u-\lambda_2)\right|du\right)^2\sqrt{b}\to 0.
\end{aligned}
\end{equation*}
Thus, ${A}_{\lambda_1,\lambda_2}^{n,3}(5)\to_p 0$. Following the same steps, we can show ${A}_{\lambda_1,\lambda_2}^{n,3}(r)\to_p 0$ when $r=6$. Thus,
\begin{align*}
    {A}_{\lambda_1,\lambda_2}^{n,3}\to &\frac{1}{8\theta} \int_0^T \sum_{p, q, u, v}\partial^2_{pq, uv} g\left(c_s\right) \check{c}_s^{pq, uv} d s \int^{\infty}_{-\infty} K(u-\lambda_1)K(u-\lambda_2) d u\\
    &+\frac{\theta}{8} \int_0^T \sum_{p, q, u, v} \partial^2_{pq, uv} g\left(c_s\right) \tilde{c}_s^{pq, uv} d s\left(\int_{-\infty}^{\infty} L(u-\lambda_1)L(u-\lambda_2) d u\right).
\end{align*}

\subsection{{\Blue Proof of \eqref{SndStOfLmts}.}}\label{Term1Of30}
{In what follows, we shall use the following estimate from \cite{jacod2015estimation} (see Eq.~(3.23) therein):
\begin{equation}\label{EstJR0}
\left|\mathbb{E}\left[\left.\left(c_{i \Delta_n+t}^{p q}-c_{i \Delta_n}^{p q}\right)\left(c_{i \Delta_n+t}^{u v}-c_{i\Delta_n}^{u v}\right) \right| \mathcal{F}_i^n\right]-t \tilde{c}_i^{n, p q, u v}\right| \leq C t \eta_{i, k_n}^n,
\end{equation}
 valid for any $t>0$, where
$
\eta_{i,k_n}^n=\max \left(\eta(Y)_{i,k_n}^n: Y=\mu, \tilde{\mu}, c, \tilde{c}, \hat{c}\right).
$
Recall that}
\begin{equation*}
\begin{aligned}
    &\widehat{A}_{1}^{n,3}:=\frac{\sqrt{\Delta_n}}{8} \sum_{i=1}^{n-2 k_n+1} \sum_{p, q, u, v} \partial_{p q, u v}^2 g(\hat{c}_i^n)\left(c_{i+k_n}^{n, p q}-c_i^{n, p q}\right)\left(c_{i+k_n}^{n, u v}-c_i^{n, u v}\right),
   \end{aligned}
\end{equation*} 
and define
\begin{equation*}
\begin{aligned}
    &A_{1}^{n,3}:=\frac{\sqrt{\Delta_n}}{8} \sum_{i=1}^{n-2 k_n+1} \sum_{p, q, u, v} \partial_{p q, u v}^2 g(c_i^n)\left(c_{i+k_n}^{n, p q}-c_i^{n, p q}\right)\left(c_{i+k_n}^{n, u v}-c_i^{n, u v}\right).
\end{aligned}
\end{equation*}
Note that, by \eqref{eq:X_c_bound} and \eqref{eq:beta_bound},
\begin{align*}
    &\mathbb{E}\left|\widehat{A}_{1}^{n,3}-A_{1}^{n,3}\right|\\
    &=\mathbb{E}\Bigg|\frac{\sqrt{\Delta_n}}{8} \sum_{i=1}^{n-2 k_n+1} \sum_{p, q, u, v} \left(\partial_{p q, u v}^2 g(\hat{c}_i^n)-\partial_{p q, u v}^2 g(c_i^n)\right)\left(c_{i+k_n}^{n, p q}-c_i^{n, p q}\right)\left(c_{i+k_n}^{n, u v}-c_i^{n, u v}\right)\Bigg|\\
    &\quad\leq \frac{\sqrt{\Delta_n}}{8} \sum_{i=1}^{n-2 k_n+1} \sum_{p, q, u, v}\sqrt{\mathbb{E}\left[\left(\partial_{p q, u v}^2 g(\hat{c}_i^n)-\partial_{p q, u v}^2 g(c_i^n)\right)^2\right] \mathbb{E}\left[\left\|c_{i+k_n}^{n}-c_i^{n}\right\|^4\right]}\\
    &\quad\leq \frac{\sqrt{\Delta_n}}{8} \sum_{i=1}^{n-2 k_n+1} \sum_{p, q, u, v}\sqrt{\mathbb{E}\left[\left\|\nabla \partial_{p q, u v}^2 g(\xi)\right\|^2\left\|\hat{c}_i^n-c_i^n\right\|^2\right] \left(k_n\Delta_n\right)^2}\\
    &\quad\leq C\sqrt{\Delta_n}\sum_{i=1}^{n-2 k_n+1} \sum_{p, q, u, v}\left(k_n\Delta_n\right)^{3/2}\to 0.
\end{align*}
Then, we only need to consider $A_{1}^{n,3}$, which can be written as
$$
A_{1}^{n,3}=\sum_{i=1}^{n-2 k_n+1}\zeta^{n}_{i}:=\sum_{i=1}^{n-2 k_n+1}\frac{\sqrt{\Delta_n}}{8}\sum_{p, q, u, v} \partial_{p q, u v}^2 g(c_i^n)\left(c_{i+k_n}^{n, p q}-c_i^{n, p q}\right)\left(c_{i+k_n}^{n, u v}-c_i^{n, u v}\right).
$$
By \eqref{EstJR0}, 
\begin{equation*}
\begin{aligned}
    &\sum_{i=1}^{n-2 k_n+1}\mathbb{E}\left[\zeta^{n}_i\left.\right|\mathcal{F}^n_i\right]\\
    &=\frac{\sqrt{\Delta_n}}{8} \sum_{i=1}^{n} \sum_{p, q, u, v} \partial_{p q, u v}^2 g(c_i^n)\mathbbm{1}_{\left\{t_{i}\leq t_{n-2 k_n+1}\right\}}\mathbb{E}\left[\left.\left(c_{i+k_n}^{n, p q}-c_i^{n, p q}\right)\left(c_{i+k_n}^{n, u v}-c_i^{n, u v}\right)\right| \mathcal{F}_i^n\right]\\
    &=\frac{\sqrt{\Delta_n}}{8} \sum_{i=1}^{n} \sum_{p, q, u, v} \partial_{p q, u v}^2 g(c_i^n)\mathbbm{1}_{\left\{t_{i}\leq t_{n-2 k_n+1}\right\}}k_n\Delta_n \tilde{c}_i^{n, p q, u v}+o_{p}(1)\\
    &=\frac{k_n\sqrt{\Delta_n}}{8} \int_{0}^{t_{n-2k_n+1}} \sum_{p, q, u, v} \partial_{p q, u v}^2 g(c_s) \tilde{c}_s^{p q, u v}ds+o_{p}(1)\stackrel{n \rightarrow \infty}{\longrightarrow} \frac{\theta}{8} \int_{0}^{T} \sum_{p, q, u, v} \partial_{p q, u v}^2 g(c_s) \tilde{c}_s^{p q, u v}ds,
\end{aligned}
\end{equation*}
and
\begin{equation*}
\begin{aligned}
&\mathbb{E}\left(\sum_{i=1}^{n-2 k_n+1}\left(\zeta_i^n-\mathbb{E}\left[\left.\zeta_i^n \right| \mathcal{F}_{i}^n\right]\right)\right)^2\leq \sum_{i=1}^{n-2 k_n+1}\mathbb{E}\left[\left(\zeta_i^n\right)^2\right]\\
&\leq C \sum_{i=1}^{n-2 k_n+1}\Delta_n\mathbb{E}\left[\left(c_{i+k_n}^{n, p q}-c_i^{n, p q}\right)^2\left(c_{i+k_n}^{n, u v}-c_i^{n, u v}\right)^2\right]\\
&\leq C \sum_{i=1}^{n-2 k_n+1}\Delta_n\sqrt{\mathbb{E}\left[\left(c_{i+k_n}^{n, p q}-c_i^{n, p q}\right)^4\right]\mathbb{E}\left[\left(c_{i+k_n}^{n, u v}-c_i^{n, u v}\right)^4\right]}\leq C \sum_{i=1}^{n-2 k_n+1}k_n\Delta_n^2\to 0.
\end{aligned}
\end{equation*}
Thus,
\begin{equation*}
    A_{1}^{n,3}\to \frac{\theta}{8} \int_{0}^{T} \sum_{p, q, u, v} \partial_{p q, u v}^2 g(c_s) \tilde{c}_s^{p q, u v}ds.
\end{equation*}
\subsection{{\Blue Proof of \eqref{TrdStOfLmts}.}}
For simplicity, we only consider the term involving $\beta_i^{n, p q}\left(c_{i+k_n}^{n, u v}-c_i^{n, u v}\right)$. Consider
\begin{equation*}
\begin{array}{l}
\widehat{A}_{2}^{n,3}=\frac{\sqrt{\Delta_n}}{8} \sum_{i=1}^{n-2k_n+1} \sum_{p, q, u, v} \partial_{p q, u v}^2 g(\hat{c}_i^n)\beta_i^{n, p q}\left(c_{i+k_n}^{n, u v}-c_i^{n, u v}\right),\\
A_{2}^{n,3}:=\frac{\sqrt{\Delta_n}}{8} \sum_{i=1}^{n-2k_n+1} \sum_{p, q, u, v} \partial_{p q, u v}^2 g(c_i^n)\beta_i^{n, p q}\left(c_{i+k_n}^{n, u v}-c_i^{n, u v}\right),
\end{array}
\end{equation*}
Again, $\mathbb{E}\left|\widehat{A}_{2}^{n,3}-A_{2}^{n,3}\right|\to 0$, and {using \eqref{Dfnalphabeta0}, we can write $A_{2}^{n,3}$ as}
\begin{align*}
   A_{2}^{n,3}&=\frac{\sqrt{\Delta_n}}{8} \sum_{i=1}^{n-2 k_n+1} \sum_{j=1}^n\sum_{p, q, u, v} \partial_{p q, u v}^2 g(c_i^n)K_{b}\left(t_{j-1}-t_i\right) \bar{K}\left(t_i\right)^{-1}\\
   &\quad\quad\quad\quad\quad\quad\quad\quad\left[\alpha_j^{n, p q}\left(c_{i+k_n}^{n, u v}-c_i^{n, u v}\right)+\left(c_j^{n, p q}-c_i^{n, p q}\right) \left(c_{i+k_n}^{n, u v}-c_i^{n, u v}\right)\Delta_n\right]\\
   &=:\sum_{j=1}^n\zeta^{n,1}_{j}+\sum_{i=1}^{n-2 k_n+1}\zeta^{n,2}_{i},
\end{align*}
where
\begin{equation*}
\begin{aligned}
    \zeta^{n,1}_{j}&:=\frac{\sqrt{\Delta_n}}{8}\sum_{i=1}^{n-2 k_n+1}\sum_{p, q, u, v}\partial_{p q, u v}^2 g(c_i^n)\bar{K}\left(t_i\right)^{-1} K_{b}\left(t_{j-1}-t_i\right) \alpha_j^{n, p q}\left(c_{i+k_n}^{n, u v}-c_i^{n, u v}\right),\\
    \zeta^{n,2}_{i}&:=\frac{\sqrt{\Delta_n}}{8}\sum_{j=1}^n\sum_{p, q, u, v} \partial_{p q, u v}^2 g(c_i^n)K_{b}\left(t_{j-1}-t_i\right) \bar{K}\left(t_i\right)^{-1}\left(c_j^{n, p q}-c_i^{n, p q}\right) \left(c_{i+k_n}^{n, u v}-c_i^{n, u v}\right)\Delta_n.\\
    &=\frac{\sqrt{\Delta_n}}{8}\sum_{j=1}^i\sum_{p, q, u, v} \partial_{p q, u v}^2 g(c_i^n)K_{b}\left(t_{j-1}-t_i\right) \bar{K}\left(t_i\right)^{-1}\left(c_j^{n, p q}-c_i^{n, p q}\right) \left(c_{i+k_n}^{n, u v}-c_i^{n, u v}\right)\Delta_n.\\
    &\quad+\frac{\sqrt{\Delta_n}}{8}\sum_{j=i+1}^n\sum_{p, q, u, v} \partial_{p q, u v}^2 g(c_i^n)K_{b}\left(t_{j-1}-t_i\right) \bar{K}\left(t_i\right)^{-1}\left(c_j^{n, p q}-c_i^{n, p q}\right) \left(c_{i+k_n}^{n, u v}-c_i^{n, u v}\right)\Delta_n.\\
    &=:\zeta^{n,21}_{i}+\zeta^{n,22}_{i}.
\end{aligned}
\end{equation*}
By (\ref{eq:A.84}), 
\begin{equation}\label{zeta1.1}
\begin{aligned}
&\mathbb{E}\left|\sum_{j=1}^n\mathbb{E}\left[\zeta^{n,1}_j\left.\right|\mathcal{F}_{j-1}\right]\right|\\
&\leq C\sqrt{\Delta_n}\sum_{j=1}^n\mathbb{E}\left|\sum_{i=1}^{n-2 k_n+1}\sum_{p, q, u, v} K_{b}\left(t_{j-1}-t_i\right) \mathbb{E}\left[\left.\alpha_j^{n, p q}\left(c_{i+k_n}^{n, u v}-c_i^{n, u v}\right)\partial_{p q, u v}^2 g(c_i^n)\right|\mathcal{F}_{j-1}\right]\right|\\
&\leq C\sqrt{\Delta_n}\sum_{j=1}^n\Bigg(\sum_{i=1}^{j-1}\left|K_{b}\left(t_{j-1}-t_i\right)\right| \mathbb{E}\left[\Delta_n^{3 / 2}\left(\sqrt{\Delta_n}+\eta_{j-1}^n\right)+\Delta_n \eta_{j-1, \left|i+k_n-j+1\right|}^n\left|t_{i+k_n}-t_{j-1}\right|\right]\\
&+\sum_{i=j}^{n} \left|K_{b}\left(t_{j-1}-t_i\right)\right| \mathbb{E}\left[\Delta_n \eta_{j-1, i-j+1}^n\left(t_i-t_{j-1}\right)+\Delta_n \eta_{j-1, i+k_n-j+1}^n\left(t_{i+k_n}-t_{j-1}\right)\right]\Bigg)\\
&\leq C\sqrt{\Delta_n}\sum_{j=1}^n\Bigg(\sqrt{\Delta_n}\left(\sqrt{\Delta_n}+\mathbb{E}\left[\eta_{t_{j-1}}\right]\right)\int^{t_{j-1}/b}_{0}\left|K(u)\right|du+b\int^{t_{j-1}/b}_{(t_{j-1}-T)/b)}\left|K(u)\right|\left|1-u\right|\mathbb{E}\left[\eta_{t_{j-1},\left|1-u\right|b}\right]du\\
&+b\int_{(t_{j-1}-T)/b}^{0}\left|K(u)u\right|\mathbb{E}\left[\eta_{t_{j-1},-ub}\right]du\Bigg)\to 0,
\end{aligned}
\end{equation}
where in the second inequality we apply (\ref{eq:ikn}) with $\lambda_1=1$. By Cauchy-Schwartz inequality, (\ref{eq:X_c_bound}), and (\ref{eq:bounds_alpha}),
\begin{equation}\label{zeta1.2}
\begin{aligned}
&\mathbb{E}\left(\sum_{j=1}^n\left(\zeta_i^{n,1}-\mathbb{E}\left[\left.\zeta_{j}^{n,1} \right| \mathcal{F}_{j-1}\right]\right)\right)^2\leq \sum_{j=1}^n\mathbb{E}\left[\left(\zeta_j^{n,1}\right)^2\right]\\
&\leq C\Delta_n\sum_{j=1}^n\sum_{p, q, u, v}\sum_{i=1}^{n}K^2_{b}\left(t_{j-1}-t_i\right) \mathbb{E}\left[\left(\alpha_j^{n, p q}\right)^2\left(c_{i+k_n}^{n, u v}-c_i^{n, u v}\right)^2\right]\\
&+C\Delta_n\sum_{j=1}^n\sum_{p, q, u, v}\sum_{i,i^{\prime},i\neq i^{\prime}}K_{b}\left(t_{j-1}-t_i\right)K_{b}\left(t_{j-1}-t_{i^{\prime}}\right)\\
&\qquad\qquad\qquad\qquad
\mathbb{E}\left[\left(\alpha_j^{n, p q}\right)^2\partial_{p q, u v}^2 g(c_i^n)\partial_{p q, u v}^2 g(c_{i^{\prime}}^n)\left(c_{i+k_n}^{n, u v}-c_i^{n, u v}\right)\left(c_{{i^{\prime}}+k_n}^{n, u v}-c_{i^{\prime}}^{n, u v}\right)\right]\\
&\leq C\Delta_n\sum_{j=1}^n\sum_{i=1}^{n}K^2_{b}\left(t_{j-1}-t_i\right)\Delta_n^2\sqrt{k_n\Delta_n}\\
&+C\Delta_n\sum_{j=1}^n\sum_{i,i^{\prime},i\neq i^{\prime}}K_{b}\left(t_{j-1}-t_i\right)K_{b}\left(t_{j-1}-t_{i^{\prime}}\right)\Delta_n^2\sqrt{k_n\Delta_n}\\
&\leq C \Delta^2_n\sum_{j=1}^n \int K^2(u)du \frac{\sqrt{k_n\Delta_n}}{b}+C \Delta_n\sum_{j=1}^n\left(\int K(u)du\right)^2\sqrt{k_n\Delta_n} \to 0.
\end{aligned}
\end{equation}
Thus, $\sum_{j=1}^n\zeta^{n,1}_{j}\to 0$. Meanwhile, by (\ref{eq:X_c_bound}),
\begin{equation}\label{zeta21.1}
\begin{aligned}
    &\mathbb{E}\left|\sum_{i=1}^{n-2 k_n+1}\mathbb{E}\left[\zeta^{n,21}_i\left.\right|\mathcal{F}_i\right]\right|\\
    &=\frac{\sqrt{\Delta_n}}{8} \mathbb{E}\left|\sum_{i=1}^{n-2 k_n+1}\sum_{p, q, u, v} \partial_{p q, u v}^2 g(c_i^n)\sum_{j=1}^i K_{b}\left(t_{j-1}-t_i\right) \bar{K}\left(t_i\right)^{-1}\left(c_j^{n, p q}-c_i^{n, p q}\right)\Delta_n\mathbb{E}\left[\left. \left(c_{i+k_n}^{n, u v}-c_i^{n, u v}\right)\right|\mathcal{F}_{i}\right]\right|\\
    &\leq C\sqrt{\Delta_n}\sum_{i=1}^{n-2 k_n+1}\sum_{j=1}^i \left|K_{b}\left(t_{j-1}-t_i\right)\right|\Delta_n\mathbb{E}\left[\left|c_j^{n, p q}-c_i^{n, p q}\right|\left|\mathbb{E}\left[\left. \left(c_{i+k_n}^{n, u v}-c_i^{n, u v}\right)\right|\mathcal{F}_{i}\right]\right|\right]\\
    &\leq C\sqrt{\Delta_n}\sum_{i=1}^{n-2 k_n+1}\sum_{j=1}^i \left|K_{b}\left(t_{j-1}-t_i\right)\right|\sqrt{t_{j-1}-t_{i}}k_n\Delta_n^2\\
    &\leq Ck_n\sqrt{\Delta_n b}\Delta_n\sum_{i=1}^{n-2 k_n+1}\int \left|K(u)\right|\sqrt{\left|u\right|}du\to 0,
\end{aligned}
\end{equation}
and
\begin{equation}\label{zeta21.2}
\begin{aligned}
&\mathbb{E}\left(\sum_{i=1}^{n-2 k_n+1}\left(\zeta_i^{n,21}-\mathbb{E}\left[\left.\zeta_i^{n,21} \right| \mathcal{F}_{i}\right]\right)\right)^2\leq \sum_{i=1}^{n-2 k_n+1}\mathbb{E}\left[\left(\zeta_i^{n,21}\right)^2\right]\\
&\leq C\Delta_n\sum_{i=1}^{n-2 k_n+1}\sum_{j=1}^{i}K^2_{b}\left(t_{j-1}-t_i\right)\mathbb{E}\left[\left(c_j^{n, p q}-c_i^{n, p q}\right)^2 \left(c_{i+k_n}^{n, u v}-c_i^{n, u v}\right)^2\right]\Delta_n^2\\
&+C\Delta_n\sum_{i=1}^{n-2 k_n+1}\sum_{j,j^{\prime}\leq i,j\neq j^{\prime}}K_{b}\left(t_{j-1}-t_i\right)K_{b}\left(t_{j^{\prime}-1}-t_i\right)\mathbb{E}\left[\left(c_j^{n, p q}-c_i^{n, p q}\right)\left(c_{j^{\prime}}^{n, p q}-c_i^{n, p q}\right) \left(c_{i+k_n}^{n, u v}-c_i^{n, u v}\right)^2\right]\Delta_n^2\\
&\leq C\Delta^3_n\sum_{i=1}^{n-2 k_n+1}\sum_{j=1}^{i}K^2_{b}\left(t_{j-1}-t_i\right)\sqrt{\left|t_{j}-t_{i}\right|}\sqrt{k_n\Delta_n}\\
&+C\Delta^3_n\sum_{i=1}^{n-2 k_n+1}\sum_{j,j^{\prime}\leq i,j\neq j^{\prime}}K_{b}\left(t_{j-1}-t_i\right)K_{b}\left(t_{j^{\prime}-1}-t_i\right)\left|t_{j}-t_{i}\right|^{\frac{1}{4}}\left|t_{j^{\prime}}-t_{i}\right|^{\frac{1}{4}}\sqrt{k_n\Delta_n}\\
&\leq C\Delta^2_n\sum_{i=1}^{n-2 k_n+1}\int K^2\left(u\right)\sqrt{\left|u\right|}du\frac{\sqrt{k_n\Delta_n}}{\sqrt{b}}+C\Delta_n\sum_{i=1}^{n-2 k_n+1}\left(\int K(u)\left|u\right|^{\frac{1}{4}}du\right)^2\sqrt{k_n\Delta_n b}\to 0,
\end{aligned}
\end{equation}
which implies $\sum_{i=1}^{n-2 k_n+1}\zeta^{n,21}_{i}\to 0$. Now consider
\begin{equation}\label{zeta2.1}
\begin{aligned}
    &\sum_{i=1}^{n-2 k_n+1}\mathbb{E}\left[\zeta^{n,22}_i\left.\right|\mathcal{F}_i\right]=\frac{\sqrt{\Delta_n}}{8} \sum_{i=1}^{n-2 k_n+1}\sum_{p, q, u, v} \partial_{p q, u v}^2 g(c_i^n)\sum_{j=i+1}^n K_{b}\left(t_{j-1}-t_i\right) \bar{K}\left(t_i\right)^{-1}\Delta_n\\
   &\qquad\qquad\qquad\qquad\qquad\qquad\qquad
\qquad\qquad\qquad
\mathbb{E}\left[\left.\left(c_j^{n, p q}-c_i^{n, p q}\right) \left(c_{i+k_n}^{n, u v}-c_i^{n, u v}\right)\right|\mathcal{F}_{i}\right].
\end{aligned}
\end{equation}
Note that
\begin{equation}\label{cdiff_1}
\begin{aligned}
    &\mathbb{E}\left[\left.\left(c_j^{n, p q}-c_i^{n, p q}\right) \left(c_{i+k_n}^{n, u v}-c_i^{n, u v}\right)\right|\mathcal{F}_{i}\right]\\
    &=\begin{cases}\mathbb{E}\left[\left.\left(c_j^{n, p q}-c_i^{n, p q}\right) \left(c_{j}^{n, u v}-c_i^{n, u v}\right)\right|\mathcal{F}_{i}\right]+\epsilon_{11} & \text{if}\ i<j\leq i+k_n, \\ \mathbb{E}\left[\left.\left(c_{i+k_n}^{n, p q}-c_i^{n, p q}\right) \left(c_{i+k_n}^{n, u v}-c_i^{n, u v}\right)\right|\mathcal{F}_{i}\right]+\epsilon_{12} &\text{if}\ j>i+k_n,\end{cases}
\end{aligned}
\end{equation}
where $\left|\epsilon_{11}\right|\leq C\sqrt{t_{j}-t_{i}}(t_{i+k_n}-t_{j})$, $\left|\epsilon_{12}\right|\leq C\sqrt{k_n\Delta_n}(t_{j}-t_{i+k_n})$. The proof of (\ref{cdiff_1}) can be found in Section \ref{SectionB.3.5}. Also, by (3.23) in \cite{jacod2015estimation},
$$
\begin{array}{l}
\left|\mathbb{E}\left[\left.\left(c_j^{n, p q}-c_i^{n, p q}\right) \left(c_{j}^{n, u v}-c_i^{n, u v}\right)\right|\mathcal{F}_{i}\right]-\tilde{c}_i^{n, p q, u v}(j-i)\Delta_n\right|\leq C (j-i)\Delta_n \eta_{i, k_n}^n,\\
\left|\mathbb{E}\left[\left.\left(c_{i+k_n}^{n, p q}-c_i^{n, p q}\right) \left(c_{i+k_n}^{n, u v}-c_i^{n, u v}\right)\right|\mathcal{F}_{i}\right]-\tilde{c}_i^{n, p q, u v}k_n\Delta_n\right|\leq C k_n\Delta_n \eta_{i, k_n}^n.
\end{array}
$$
Thus,
\begin{equation}\label{zeta2.2}
\begin{aligned}
    &\sum_{i=1}^{n-2 k_n+1}\mathbb{E}\left[\zeta^{n,22}_i\left.\right|\mathcal{F}_i\right]\\
   &=\frac{\sqrt{\Delta_n}}{8} \sum_{i=1}^{n-2 k_n+1}\sum_{p, q, u, v} \partial_{p q, u v}^2 g(c_i^n)\bar{K}\left(t_i\right)^{-1}\Delta_n\sum_{j=i+1}^{i+k_n} K_{b}\left(t_{j-1}-t_i\right)\left(\tilde{c}_i^{n, p q, u v}(t_{j}-t_{i})+\epsilon_{11}\right)\\
   &\quad + \frac{\sqrt{\Delta_n}}{8} \sum_{i=1}^{n-2 k_n+1}\sum_{p, q, u, v} \partial_{p q, u v}^2 g(c_i^n)\bar{K}\left(t_i\right)^{-1}\Delta_n\sum_{j=i+k_n+1}^{n} K_{b}\left(t_{j-1}-t_i\right)\left(\tilde{c}_i^{n, p q, u v}k_n\Delta_n+\epsilon_{12}\right)+o_{p}(1)\\
   &=\frac{\sqrt{\Delta_n}}{8} \sum_{i=1}^{n-2 k_n+1}\sum_{p, q, u, v} \partial_{p q, u v}^2 g(c_i^n)\bar{K}\left(t_i\right)^{-1}\tilde{c}_i^{n, p q, u v}\int_{0}^{1}K(u)ubdu\\
   &\quad + \frac{\sqrt{\Delta_n}}{8} \sum_{i=1}^{n-2 k_n+1}\sum_{p, q, u, v} \partial_{p q, u v}^2 g(c_i^n)\bar{K}\left(t_i\right)^{-1}\tilde{c}_i^{n, p q, u v}b\int_{1}^{\infty}K(u)du+o_{p}(1)\\
   &\to \frac{\theta}{8}\int_{0}^{T}\sum_{p, q, u, v} \partial_{p q, u v}^2 g(c_{s})\tilde{c}_{s}^{p q, u v}ds\left(\int_{1}^{\infty}K(u)du+\int_{0}^{1}K(u)udu\right),
\end{aligned}
\end{equation}
where $\epsilon_{11}, \epsilon_{12}$ are negligible as $o_p(1)$ since
\begin{equation}\label{epsilon}
\begin{aligned}
&\mathbb{E}\left|\frac{\sqrt{\Delta_n}}{8} \sum_{i=1}^{n-2 k_n+1}\sum_{p, q, u, v} \partial_{p q, u v}^2 g(c_i^n)\bar{K}\left(t_i\right)^{-1}\Delta_n\right.\\
&\qquad\qquad\qquad\qquad\left.\left(\sum_{j=i+1}^{i+k_n} K_{b}\left(t_{j-1}-t_i\right)
\epsilon_{11}+\sum_{j=i+k_n+1}^{n} K_{b}\left(t_{j-1}-t_i\right) 
\epsilon_{12}\right)\right|\\
&\leq C \Delta_n^{3/2}\sum_{i=1}^{n-2 k_n+1}\Bigg(\sum_{j=i+1}^{i+k_n} \left|K_{b}\left(t_{j-1}-t_i\right)\right|\sqrt{t_{j}-t_{i}}(t_{i+k_n}-t_{j})\\
&\qquad\qquad\qquad\qquad\qquad+\sum_{j=i+k_n+1}^{n} \left|K_{b}\left(t_{j-1}-t_i\right)\right|\sqrt{k_n\Delta_n}(t_{j}-t_{i+k_n})\Bigg)\\
&\leq C\Delta_n^{1/2}b^{3/2}\sum_{i=1}^{n-2 k_n+1} \left(\int_{0}^{1}\left|K(u)\right|\sqrt{u}(1-u)du+\int_{1}^{\infty}\left|K(u)\right|(u-1)du\right)\to 0.
\end{aligned}
\end{equation}
Meanwhile,
\begin{equation}\label{zeta2.3}
\begin{aligned}
&\mathbb{E}\left(\sum_{i=1}^{n-2 k_n+1}\left(\zeta_i^{n,22}-\mathbb{E}\left[\left.\zeta_i^{n,22} \right| \mathcal{F}_{i}\right]\right)\right)^2\leq \sum_{i=1}^{n-2 k_n+1}\mathbb{E}\left[\left(\zeta_i^{n,22}\right)^2\right]\\
&\leq C\Delta_n\sum_{i=1}^{n-2 k_n+1}\sum_{j=i+1}^{n}K^2_{b}\left(t_{j-1}-t_i\right)\mathbb{E}\left[\left(c_j^{n, p q}-c_i^{n, p q}\right)^2 \left(c_{i+k_n}^{n, u v}-c_i^{n, u v}\right)^2\right]\Delta_n^2\\
&+C\Delta_n\sum_{i=1}^{n-2 k_n+1}\sum_{j,j^{\prime}>i,j\neq j^{\prime}}K_{b}\left(t_{j-1}-t_i\right)K_{b}\left(t_{j^{\prime}-1}-t_i\right)\mathbb{E}\left[\left(c_j^{n, p q}-c_i^{n, p q}\right)\left(c_{j^{\prime}}^{n, p q}-c_i^{n, p q}\right) \left(c_{i+k_n}^{n, u v}-c_i^{n, u v}\right)^2\right]\Delta_n^2\\
&\leq C\Delta^3_n\sum_{i=1}^{n-2 k_n+1}\sum_{j=i+1}^{n}K^2_{b}\left(t_{j-1}-t_i\right)\sqrt{t_{j}-t_{i}}\sqrt{k_n\Delta_n}\\
&+C\Delta^3_n\sum_{i=1}^{n-2 k_n+1}\sum_{j,j^{\prime}>i,j\neq j^{\prime}}K_{b}\left(t_{j-1}-t_i\right)K_{b}\left(t_{j^{\prime}-1}-t_i\right)\left(t_{j}-t_{i}\right)^{\frac{1}{4}}\left(t_{j^{\prime}}-t_{i}\right)^{\frac{1}{4}}\sqrt{k_n\Delta_n}\\
&\leq C\Delta^2_n\sum_{i=1}^{n-2 k_n+1}\int K^2\left(u\right)\sqrt{\left|u\right|}du\frac{\sqrt{k_n\Delta_n}}{\sqrt{b}}+C\Delta_n\sum_{i=1}^{n-2 k_n+1}\left(\int K(u)\left|u\right|^{\frac{1}{4}}du\right)^2\sqrt{k_n\Delta_n b}\to 0.
\end{aligned}
\end{equation}
Thus, we conclude that
\begin{equation}
    A_{2}^{n,3}\stackrel{n \rightarrow \infty}{\longrightarrow} \frac{\theta}{8}\int_{0}^{T}\sum_{p, q, u, v} \partial_{p q, u v}^2 g(c_{s})\tilde{c}_{s}^{p q, u v}ds\left(\int_{1}^{\infty}K(u)du+\int_{0}^{1}K(u)udu\right).
\end{equation}
The convergence of $\widehat{A}_{3}^{n,3}:=\frac{\sqrt{\Delta_n}}{8} \sum_{i=1}^{n-2k_n+1} \sum_{p, q, u, v} \partial_{p q, u v}^2 g(\hat{c}_i^n)\beta_i^{n, u v}\left(c_{i+k_n}^{n, p q}-c_i^{n, p q}\right)$ can be shown in the same way.

\subsection{{\Blue Proof of \eqref{FourthStOfLmts}.}}
We only need to consider
\begin{equation*}
\begin{array}{l}
\widehat{A}_{4}^{n,3}=\frac{\sqrt{\Delta_n}}{8} \sum_{i=1}^{n-2 k_n+1} \sum_{p, q, u, v} \partial_{p q, u v}^2 g(\hat{c}_i^n)\beta_{i+k_n}^{n, p q}\left(c_{i+k_n}^{n, u v}-c_i^{n, u v}\right),\\
A_{4}^{n,3}:=\frac{\sqrt{\Delta_n}}{8} \sum_{i=1}^{n-2 k_n+1} \sum_{p, q, u, v} \partial_{p q, u v}^2 g(c_i^n)\beta_{i+k_n}^{n, p q}\left(c_{i+k_n}^{n, u v}-c_i^{n, u v}\right),\\
\end{array}
\end{equation*}
Again, $\mathbb{E}\left|\widehat{A}_{4}^{n,3}-A_{4}^{n,3}\right|\to 0$, and, by \eqref{Dfnalphabeta0},
\begin{align*}
   A_{4}^{n,3}&=\frac{\sqrt{\Delta_n}}{8} \sum_{i=1}^{n-2 k_n+1} \sum_{j=1}^n\sum_{p, q, u, v} \partial_{p q, u v}^2 g(c_i^n)K_{b}\left(t_{j-1}-t_{i+k_n}\right) \bar{K}\left(t_{i+k_n}\right)^{-1}\\
   &\quad\quad\quad\quad\quad\quad\quad\quad\left[\alpha_j^{n, p q}\left(c_{i+k_n}^{n, u v}-c_i^{n, u v}\right)+\left(c_j^{n, p q}-c_{i+k_n}^{n, p q}\right) \left(c_{i+k_n}^{n, u v}-c_i^{n, u v}\right)\Delta_n\right]\\
   &=:\sum_{j=1}^n\zeta^{n,1}_{j}+\sum_{i=1}^{n-2 k_n+1}\zeta^{n,2}_{i},
\end{align*}
where $\sum_{j=1}^n\zeta^{n,1}_{j}\to 0$ can be shown similarly as in (\ref{zeta1.1}) and (\ref{zeta1.2}). We decompose $\zeta^{n,2}_{i}$ as $\zeta^{n,2}_{i}
    =\zeta^{n,21}_{i}+\zeta^{n,22}_{i}$,
where
\begin{align*}
    \zeta^{n,21}_{i}&:=\frac{\sqrt{\Delta_n}}{8}\sum_{j=1}^i\sum_{p, q, u, v} \partial_{p q, u v}^2 g(c_i^n)K_{b}\left(t_{j-1}-t_{i+k_n}\right) \bar{K}\left(t_{i+k_n}\right)^{-1}\left(c_j^{n, p q}-c_{i}^{n, p q}\right) \left(c_{i+k_n}^{n, u v}-c_i^{n, u v}\right)\Delta_n,\\
    \zeta^{n,22}_{i}&:=\frac{\sqrt{\Delta_n}}{8}\sum_{j=1}^i\sum_{p, q, u, v} \partial_{p q, u v}^2 g(c_i^n)K_{b}\left(t_{j-1}-t_{i+k_n}\right) \bar{K}\left(t_{i+k_n}\right)^{-1}\left(c_i^{n, p q}-c_{i+k_n}^{n, p q}\right) \left(c_{i+k_n}^{n, u v}-c_i^{n, u v}\right)\Delta_n\\
    &+\frac{\sqrt{\Delta_n}}{8}\sum_{j=i+1}^n\sum_{p, q, u, v} \partial_{p q, u v}^2 g(c_i^n)K_{b}\left(t_{j-1}-t_{i+k_n}\right) \bar{K}\left(t_{i+k_n}\right)^{-1}\left(c_j^{n, p q}-c_{i+k_n}^{n, p q}\right) \left(c_{i+k_n}^{n, u v}-c_i^{n, u v}\right)\Delta_n.
\end{align*}
We can show that 
$\sum_{i=1}^{n-2 k_n+1}\zeta^{n,21}_{i}\to 0$ in a similarly as in (\ref{zeta21.1}) and (\ref{zeta21.2}). Consider
\begin{equation}
\begin{aligned}
    &\sum_{i=1}^{n-2 k_n+1}\mathbb{E}\left[\zeta^{n,22}_i\left.\right|\mathcal{F}_i\right]\\
    &=\frac{\sqrt{\Delta_n}}{8} \sum_{i=1}^{n-2 k_n+1} \sum_{p, q, u, v} \partial_{p q, u v}^2 g(c_i^n)\sum_{j=1}^iK_{b}\left(t_{j-1}-t_{i+k_n}\right) \bar{K}\left(t_{i+k_n}\right)^{-1}\Delta_n\mathbb{E}\left[\left.\left(c_i^{n, p q}-c_{i+k_n}^{n, p q}\right) \left(c_{i+k_n}^{n, u v}-c_i^{n, u v}\right)\right|\mathcal{F}_{i}\right]\\
    &+\frac{\sqrt{\Delta_n}}{8} \sum_{i=1}^{n-2 k_n+1} \sum_{p, q, u, v} \partial_{p q, u v}^2 g(c_i^n)\sum_{j=i+1}^nK_{b}\left(t_{j-1}-t_{i+k_n}\right) \bar{K}\left(t_{i+k_n}\right)^{-1}\Delta_n\mathbb{E}\left[\left.\left(c_j^{n, p q}-c_{i+k_n}^{n, p q}\right) \left(c_{i+k_n}^{n, u v}-c_i^{n, u v}\right)\right|\mathcal{F}_{i}\right].
\end{aligned}
\end{equation}
Note that
\begin{equation}\label{cdiff_2}
\begin{aligned}
     &\mathbb{E}\left[\left.\left(c_j^{n, p q}-c_{i+k_n}^{n, p q}\right) \left(c_{i+k_n}^{n, u v}-c_i^{n, u v}\right)\right|\mathcal{F}_{i}\right]\\
     &=\begin{cases}\mathbb{E}\left[\left.\left(c_j^{n, p q}-c_{i+k_n}^{n, p q}\right) \left(c_{i+k_n}^{n, u v}-c_j^{n, u v}\right)\right|\mathcal{F}_{i}\right]+\epsilon_{21} & \text{if}\ i<j< i+k_n, \\ \epsilon_{22} &\text{if}\ j\geq i+k_n, \end{cases}
\end{aligned}
\end{equation}
where $\left|\epsilon_{21}\right|\leq C\sqrt{t_{j}-t_{i}}(t_{i+k_n}-t_{j})$, $\left|\epsilon_{22}\right|\leq C\sqrt{k_n\Delta_n}(t_{j}-t_{i+k_n})$. The proof of (\ref{cdiff_2}) can be found in Section \ref{SectionB.3.5}. By (3.23) in \cite{jacod2015estimation},
$$
\begin{array}{l}
\left|\mathbb{E}\left[\left.\left(c_i^{n, p q}-c_{i+k_n}^{n, p q}\right) \left(c_{i+k_n}^{n, u v}-c_i^{n, u v}\right)\right|\mathcal{F}_{i}\right]+\tilde{c}_i^{n, p q, u v}k_n\Delta_n\right|\leq C k_n\Delta_n \eta_{i, k_n}^n,\\
\left|\mathbb{E}\left[\left.\left(c_j^{n, p q}-c_{i+k_n}^{n, p q}\right) \left(c_{i+k_n}^{n, u v}-c_j^{n, u v}\right)\right|\mathcal{F}_{j}\right]+\tilde{c}_j^{n, p q, u v}(i+k_n-j)\Delta_n\right|\leq C (i+k_n-j)\Delta_n \eta_{j, k_n}^n.
\end{array}
$$
Thus,
\begin{align*}
    &\sum_{i=1}^{n-2 k_n+1}\mathbb{E}\left[\zeta^{n,22}_i\left.\right|\mathcal{F}_i\right]\\
   &=\frac{\sqrt{\Delta_n}}{8} \sum_{i=1}^{n-2 k_n+1} \sum_{p, q, u, v} \partial_{p q, u v}^2 g(c_i^n)\bar{K}\left(t_{i+k_n}\right)^{-1}\Delta_n\sum_{j=1}^iK_{b}\left(t_{j-1}-t_{i+k_n}\right) \left(-\tilde{c}_i^{n, p q, u v}k_n\Delta_n\right)\\
    &\quad+\frac{\sqrt{\Delta_n}}{8} \sum_{i=1}^{n-2 k_n+1} \sum_{p, q, u, v} \partial_{p q, u v}^2 g(c_i^n)\bar{K}\left(t_{i+k_n}\right)^{-1}\Delta_n\sum_{j=i+1}^{i+k_n-1}K_{b}\left(t_{j-1}-t_{i+k_n}\right) \left(-\tilde{c}_j^{n, p q, u v}(t_{i+k_n}-t_{j})+\epsilon_{21}\right)\\
    &\quad+\frac{\sqrt{\Delta_n}}{8} \sum_{i=1}^{n-2 k_n+1} \sum_{p, q, u, v} \partial_{p q, u v}^2 g(c_i^n)\bar{K}\left(t_{i+k_n}\right)^{-1}\Delta_n\sum_{j=i+k_n}^{n}K_{b}\left(t_{j-1}-t_{i+k_n}\right)\epsilon_{22}+o_{p}(1)\\
    &=\frac{\sqrt{\Delta_n}}{8} \sum_{i=1}^{n-2 k_n+1} \sum_{p, q, u, v} \partial_{p q, u v}^2 g(c_i^n)\bar{K}\left(t_{i+k_n}\right)^{-1}\left(-\tilde{c}_i^{n, p q, u v}b\right)\int_{-t_{i}/b}^{0}K(u-1)du\\
    &\quad+\frac{\sqrt{\Delta_n}}{8} \sum_{i=1}^{n-2 k_n+1} \sum_{p, q, u, v} \partial_{p q, u v}^2 g(c_i^n)\bar{K}\left(t_{i+k_n}\right)^{-1}\int_{0}^{1}K(u-1)\tilde{c}_{t_{i}+ub}^{p q, u v}(u-1)bdu+o_{p}(1)\\
    &\to \frac{\theta}{8}\int_{0}^{T}\sum_{p, q, u, v} \partial_{p q, u v}^2 g(c_{s})\tilde{c}_{s}^{p q, u v}ds\left(\int_{0}^{1}K(u-1)(u-1)du-\int_{-\infty}^{0}K(u-1)du\right),
\end{align*}
where $\epsilon_{21}, \epsilon_{22}$ are negligible similar as in (\ref{epsilon}). Also, $\mathbb{E}\left(\sum_{i=1}^{n-2 k_n+1}\left(\zeta_i^{n,22}-\mathbb{E}\left[\left.\zeta_i^{n,22} \right| \mathcal{F}_{i}\right]\right)\right)^2\to 0$ can be shown similarly as in (\ref{zeta2.3}). Thus, we conclude that
\begin{equation}
    A_{4}^{n,3}\stackrel{n \rightarrow \infty}{\longrightarrow} \frac{\theta}{8}\int_{0}^{T}\sum_{p, q, u, v} \partial_{p q, u v}^2 g(c_{s})\tilde{c}_{s}^{p q, u v}ds\left(\int_{0}^{1}K(u-1)(u-1)du-\int_{-\infty}^{0}K(u-1)du\right).
\end{equation}
The convergence of $\widehat{A}_{5}^{n,3}$ can be shown in the same way.

\subsubsection{Proof of (\ref{cdiff_1}) \& (\ref{cdiff_2})}\label{SectionB.3.5}
We proceed to show that when $i<j\leq i+k_n$,
\begin{align*}
    \mathbb{E}\left[\left.\left(c_j^{n, p q}-c_i^{n, p q}\right) \left(c_{i+k_n}^{n, u v}-c_i^{n, u v}\right)\right|\mathcal{F}_{i}\right]&=\mathbb{E}\left[\left.\left(c_j^{n, p q}-c_i^{n, p q}\right) \left(c_{j}^{n, u v}-c_i^{n, u v}\right)\right|\mathcal{F}_{i}\right]+\epsilon_{11},
\end{align*}
where $\left|\epsilon_{11}\right|\leq C\sqrt{t_{j}-t_{i}}(t_{i+k_n}-t_{j})$. Other cases can be proved similarly. Note that
\begin{align*}
    \mathbb{E}\left[\left.\left(c_j^{n, p q}-c_i^{n, p q}\right) \left(c_{i+k_n}^{n, u v}-c_i^{n, u v}\right)\right|\mathcal{F}_{i}\right]&=\mathbb{E}\left[\left.\left(c_j^{n, p q}-c_i^{n, p q}\right) \left(c_{j}^{n, u v}-c_i^{n, u v}\right)\right|\mathcal{F}_{i}\right]\\
    &+\mathbb{E}\left[\left.\left(c_j^{n, p q}-c_i^{n, p q}\right) \left(c_{i+k_n}^{n, u v}-c_j^{n, u v}\right)\right|\mathcal{F}_{i}\right],
\end{align*}
where by (\ref{eq:X_c_bound}),
\begin{align*}
    &\left|\mathbb{E}\left[\left.\left(c_j^{n, p q}-c_i^{n, p q}\right) \left(c_{i+k_n}^{n, u v}-c_j^{n, u v}\right)\right|\mathcal{F}_{i}\right]\right|=\left|\mathbb{E}\left[\left.\left(c_j^{n, p q}-c_i^{n, p q}\right) \mathbb{E}\left[\left.\left(c_{i+k_n}^{n, u v}-c_j^{n, u v}\right)\right|\mathcal{F}_{j}\right]\right|\mathcal{F}_{i}\right]\right|\\
    &\leq \mathbb{E}\left[\left.\left|c_j^{n, p q}-c_i^{n, p q}\right| \left|\mathbb{E}\left[\left.\left(c_{i+k_n}^{n, u v}-c_j^{n, u v}\right)\right|\mathcal{F}_{j}\right]\right|\right|\mathcal{F}_{i}\right]\leq C\mathbb{E}\left[\left.\left|c_j^{n, p q}-c_i^{n, p q}\right|\right|\mathcal{F}_{i}\right](t_{i+k_n}-t_{j})\\
    &\leq C \sqrt{t_{j}-t_{i}}(t_{i+k_n}-t_{j}).
\end{align*}

\bibliographystyle{jtbnew}

\end{document}